\newif\ifextended
\extendedtrue
\ifextended
\documentclass[sigconf,dvipsnames,authorversion]{acmart}
\else
\documentclass[sigconf,dvipsnames]{acmart}
\fi
\newcommand{\makefitcolumn}[1]{
\noindent
\resizebox{\linewidth}{!}{
\begin{minipage}{\linewidth}
#1
\end{minipage}
}
}

\copyrightyear{2025}
\acmYear{2025}
\setcopyright{cc}
\setcctype{by-nc}
\acmConference[CCS '25]{Proceedings of the 2025 ACM SIGSAC Conference on Computer and Communications Security}{October 13--17, 2025}{Taipei, Taiwan}
\acmBooktitle{Proceedings of the 2025 ACM SIGSAC Conference on Computer and Communications Security (CCS '25), October 13--17, 2025, Taipei, Taiwan}\acmDOI{10.1145/3719027.3765156}
\acmISBN{979-8-4007-1525-9/2025/10}



\usepackage{soul}
\usepackage{stfloats}
\usepackage{placeins}
\usepackage[utf8]{inputenc}
\usepackage{tensor}
\usepackage{bm}
\usepackage{xspace}
\usepackage{todonotes}
\usepackage{multifootnote}
\usepackage{mathtools}
\usepackage{caption}
\usepackage{subcaption}
\usepackage{horst}
\usepackage{mathpartir}
\usepackage{booktabs}
\usepackage{tikz}
\usepackage[inline]{enumitem}
\usepackage{lipsum}
\usepackage{framed}
\usepackage{listings}
\usepackage{quoting}
\usepackage{fontawesome5}

\usetikzlibrary{shapes.misc, positioning, decorations.pathmorphing, calc,arrows.meta}

\newcommand{\shortened}[2]{#2}

\newtheorem{theorem}{Theorem}
\newtheorem{lemma}{Lemma}
\newtheorem{corollary}{Corollary}

\theoremstyle{definition}
\newtheorem{defthm}{Definition}

\newtheorem{examplethm}{Example}

\usepackage{pifont}

\ifextended
\title{\textsc{Wanilla}: Sound Noninterference Analysis for WebAssembly (Technical Report)}
\else
\title{\textsc{Wanilla}: Sound Noninterference Analysis for WebAssembly}
\fi

\author{Markus Scherer}
\orcid{0009-0005-6826-3493}
\affiliation{%
  \institution{TU Wien}
  \city{Vienna}
  \country{Austria}
  }
\additionalaffiliation{
  \institution{Christian Doppler Laboratory Blockchain Technologies for the Internet of Things}
  \city{Vienna}
  \country{Austria}
}

\author{Jeppe Fredsgaard Blaabjerg}
\orcid{0000-0001-6228-6137}
\affiliation{%
  \institution{Aarhus University}
  \city{Aarhus}
  \country{Denmark}
  }

\author{Alexander Sjösten}
\orcid{0000-0001-7620-5799}

\author{Matteo Maffei}
\orcid{0000-0001-8061-1685}
\affiliation{%
  \institution{TU Wien}
  \city{Vienna}
  \country{Austria}
  }
\authornotemark[1]

\begin{document}

\lstdefinestyle{mystyle}{
    commentstyle=\color{green!70!black},
    keywordstyle=\color{magenta},
    stringstyle=\color{blue!65!white},
    basicstyle=\ttfamily\footnotesize,
    breakatwhitespace=false,
    breaklines=true,
    captionpos=b,
    keepspaces=true,
    numbers=left,
    numbersep=5pt,
    showspaces=false,
    showstringspaces=false,
    showtabs=false,
    tabsize=2,
    escapeinside={|@}{@|}
}
\lstset{
  literate={\_}{}{0\discretionary{\_}{}{\_}}%
}
\lstset{style=mystyle}

\def\ccode{\lstinline[language=C, breakatwhitespace, basicstyle=\ttfamily\normalsize]}
\def\fnccode{\lstinline[language=C, breakatwhitespace, basicstyle=\ttfamily\footnotesize]}

\newcounter{linenumber}
\newcounter{indent}

\newcommand{\TInstr}[1]{
  \draw (1,\arabic{linenumber}) node [anchor=north east,yshift=-1] { \footnotesize \arabic{linenumber} \texttt{\vphantom{A}} };
  \draw (\arabic{indent} * 0.6 + 1,\arabic{linenumber}) node [anchor=north west] { #1 \vphantom{A} };
  \stepcounter{linenumber}
}
\newcommand{\TInstrNN}[1]{
  \draw (\arabic{indent} + 1,\arabic{linenumber}) node [anchor=north west] { #1 \vphantom{A} };
  \stepcounter{linenumber}
}
  

\newcommand{\TInstrBlock}[1]{
  \TInstr{$\Block{#1}$}
  \stepcounter{indent}
}
\newcommand{\TInstrLoop}{
  \TInstr{$\Loop{}$}
  \stepcounter{indent}
}
\newcommand{\TInstrEnd}{
  \addtocounter{indent}{-1}
  \TInstr{$\End{}$}
}

\newcommand{\li}[1]{\texttt{L#1}}

\newcommand{\wanillaBenchmarkModuleCount}{57\xspace}
\newcommand{\wanillaBenchmarkTestCaseCount}{311\xspace}
\newcommand{\wanillaBenchmarkMinimalInstructionCount}{1\xspace}
\newcommand{\wanillaBenchmarkMaximalInstructionCount}{60\xspace}
\newcommand{\wanillaBenchmarkAverageInstructionCount}{17.26\xspace}

\newcommand{\cmark}{\ding{51}}%
\newcommand{\xmark}{\ding{55}}%
\newcommand{\nimark}{{\small \faShield*}}
\newcommand{\inmark}{{\small \faBolt}}


\newcommand{\rapidBenchmarkAsmtPreciseQueries}{8}
\newcommand{\rapidBenchmarkAsmtImpreciseQueries}{9}
\newcommand{\rapidBenchmarkAsmtPreciseTestcases}{4}
\newcommand{\rapidBenchmarkAsmtImpreciseTestcases}{6}

\newcommand{\wanillaBenchmarkAsmtPreciseQueries}{251}
\newcommand{\wanillaBenchmarkAsmtImpreciseQueries}{135}
\newcommand{\wanillaBenchmarkAsmtPreciseTestcases}{198}
\newcommand{\wanillaBenchmarkAsmtImpreciseTestcases}{113}

\include{macros/background}
\newcommand{\itodo}[1]{\colorbox{orange}{#1}}
\newcommand{\checkme}[1]{#1}

\newcommand{\ST}{\textit{HH}\xspace}
\newcommand{\PT}{\textit{LH}\xspace}
\newcommand{\SU}{\textit{HL}\xspace}
\newcommand{\PU}{\textit{LL}\xspace}

\newcommand{\callstackprefix}[2]{#1 \prec #2}
\newcommand{\loweq}[2]{#1 \mathrel{=_L} #2}
\newcommand{\lowneq}[2]{#1 \mathrel{\neq_L} #2}

\newcommand{\pos}{\mathit{pos}}
\newcommand{\absloweq}[2]{#1 \mathrel{=_L} #2}

\newcommand{\stepSYMBOL}{\hookrightarrow\mkern-17mu\hookrightarrow}
\newcommand{\step}[2]{#1 \mathbin{\stepSYMBOL} #2} 
\newcommand{\nstep}[3]{#2 \mathbin{\overset{#1}{\stepSYMBOL}} #3} 

\newcommand{\addedInWanilla}[1]{{\color{orange!80!black} #1}}

\newcommand{\pcof}[1]{\mathit{pos}\mleft(#1\mright)}
\newcommand{\ctoa}[1]{\alpha\mleft(#1\mright)}
\newcommand{\absmodule}{\ctoa{m}}
\newcommand{\derive}[2]{#1 \mathbin{\vdash} #2}
\newcommand{\nderive}[2]{#1 \mathbin{\nvdash} #2}
\newcommand{\abslowneq}[2]{#1 \mathrel{\neq_L} #2}
\newcommand{\D}{\Delta}
\newcommand{\poseq}[2]{#1 \mathrel{=_{pos}} #2}
\newcommand{\posneq}[2]{#1 \mathrel{\neq_{pos}} #2}
\newcommand{\labeleq}[2]{#1 \mathrel{=_\ell} #2}
\newcommand{\valeq}[2]{#1 \mathrel{=_v} #2}
\newcommand{\EC}[1]{E[#1]}

\newcommand{\Nat}{\mathbb{N}}
\newcommand{\MemPos}{\textit{MemPos}}

\newcommand{\horst}{\textsc{HoRSt}\xspace}
\newcommand{\wappler}{\textsc{Wappler}\xspace}
\newcommand{\tool}{\textsc{Wanilla}\xspace}
\newcommand{\java}{\textsc{Java}\textsuperscript{TM}\xspace}

\newcommand{\deftupaccess}[2]{#1\widetilde{\left[#2\right]}}
\newcommand{\leak}[2]{\call{leak_{#1}}{#2}}
\newcommand{\cleak}[3]{\call{leak_{#1}^{#2}}{#3}}
\newcommand{\pbr}{\h{br}<p>}
\newcommand{\ppc}{\h{pc}<p>}

\newcommand{\cd}{\dot{c}}
\newcommand{\ci}{\overline{c}}
\newcommand{\di}{\overline{\Delta}}

\newcommand{\disj}[4]{\call{disj_{#1,#2}}{#3,#4}}

\newcommand{\annotate}[1]{\mathit{annotate}\mleft(#1\mright)}

\newcommand{\cssOf}[1]{\access{#1}{cs}}
\newcommand{\att}{\ensuremath{\textit{att}}}
\newcommand{\flowsto}{\sqsubseteq}
\newcommand{\notflowsto}{\not\sqsubseteq}
\newcommand{\lub}{\sqcup}
\newcommand{\nat}{\mathbb{N}}
\newcommand{\mpaccess}[2]{\horstACCESS{#1}{#2}}
\newcommand{\mpset}[3]{#1 \mleft[ #2 \leftarrow #3 \mright]}
\newcommand{\elleq}[2]{#1 \mathrel{=_\ell} #2}
\newcommand{\ellneq}[2]{#1 \mathrel{\neq_\ell} #2}
\newcommand{\ipp}{\textit{pp}}
\newcommand{\fid}{\textit{fid}}

\newcommand{\ppReturn}[1]{\rho_{#1}}

\NewDocumentCommand{\abstracts}{O{} O{} O{} O{}} {
  \mathbin{\prescript{#1}{#2}{\ge}^{#3}_{#4}}  
}

\newcommand{\abstractshvalue}[4]{
  #3 \mathbin{\prescript{#1}{}{\ge^{#2}_H}} #4
}
\newcommand{\abstractsatcontext}[4]{
  #3 \mathbin{\prescript{#1}{}{\ge^{\textit{ctx}}_H}} #4
}
\newcommand{\lcomplete}[1]{\call{C_\ell}{#1}}
\newcommand{\abstractsforell}[2]{
  #1 \mathbin{\ge^\ell} #2
}

\newcommand{\subsetofmps}[3]{
  #1 \mathbin{\downarrow^{#3}_{#2}}
}

\newcommand{\es}{[]}

\ExplSyntaxOn
\NewDocumentCommand{\lctx}{o}{
  \Mlctx
  \IfNoValueF{#1}{
    ^{#1}
  }
}
\NewDocumentCommand{\hctx}{o}{
  \textit{hctx}^\ell
  \IfNoValueF{#1}{
    \c_math_subscript_token {#1}
  }
}
\ExplSyntaxOff
\newcommand{\CLAP}[1]{\call{CLAP^\ell}{#1}}
\newcommand{\LCLAP}[1]{\call{LCLAP^\ell}{#1}}
\newcommand{\WCLAP}[1]{\call{CLAP^\ell_W}{#1}}
\newcommand{\Tr}[1]{\mathit{Tr}\left(#1\right)}
\NewDocumentCommand{\aaind}{m m O{\ipp}}{#1 \mathrel{\hat{=}^\ell_{#3}} #2}
\NewDocumentCommand{\aadis}{m m O{\ipp}}{#1 \mathrel{\hat{\neq}^\ell_{#3}} #2}
\newcommand{\ppOf}[1]{\access{#1}{pp}}
\newcommand{\haspp}[2]{\ppOf{#1} = #2}

\newcommand{\Lat}{\ensuremath{\mathcal{L}}}

\makeatletter
\newenvironment{smallalign*}
{
  \vspace{-1.4ex}
  \begingroup
  \small
  \start@align\@ne\st@rredtrue\m@ne
}
{
  \endalign
  \endgroup
}
\makeatother

\NewDocumentCommand{\terminalconf}{o}{
  \IfNoValueTF{#1}{
    \textit{term}
  }{
    \call{term}{#1}
  }
}

\NewDocumentCommand{\newfunc}{m o}{
  \IfNoValueTF{#2}{
    #1
  }{
    \call{#1}{#2}
  }
}

\newcommand{\hbr}{\h{br}<p>}
\newcommand{\helse}{\h{else}<p>}
\newcommand{\hidx}{\h{idx}<p>}
\newcommand{\hpc}{\h{pc}<p>}
\newcommand{\aseq}[1]{\call{\alpha_{\textit{seq}}}{#1}}

\newcommand{\exconftuple}{\newfunc{\eta_c}}
\newcommand{\exconfpp}{\newfunc{\eta_\textit{pp}}}
\newcommand{\exconfmem}{\newfunc{\eta_\textit{mem}}}
\newcommand{\exstack}[1]{\call{\eta_S}{#1}}
\newcommand{\exglobals}[1]{\call{\eta_G}{#1}}
\newcommand{\exlocals}[1]{\call{\eta_L}{#1}}
\newcommand{\exargs}[1]{\call{\eta_A}{#1}}
\newcommand{\exmemsize}[1]{\call{\eta_{MS}}{#1}}
\newcommand{\exmem}[1]{\call{\eta_{M}}{#1}}
\newcommand{\aframe}[1]{\call{\alpha_F}{#1}}
\newcommand{\atable}[1]{\call{\alpha_t}{#1}}
\newcommand{\athread}[1]{\call{\alpha_T}{#1}}
\newcommand{\athreadcur}[1]{\call{\alpha_{T_{\textit{cur}}}}{#1}}
\newcommand{\athreadcf}[1]{\call{\alpha_{T_{\textit{cf}}}}{#1}}
\newcommand{\ctoacur}[1]{\call{\alpha_{\textit{cur}}}{#1}}
\newcommand{\ctoacf}[1]{\call{\alpha_{\textit{cf}}}{#1}}
\newcommand{\atrap}[1]{\call{\alpha_{trap}}{#1}}
\newcommand{\avalue}[1]{\call{\alpha_V}{#1}}
\newcommand{\afunctions}[1]{\call{\alpha_C}{#1}}

\newcommand{\MPStore}{\textit{Store}}
\newcommand{\MPLocal}[1]{\call{Frame}{#1}}
\newcommand{\MPAnnot}[1]{\call{Annot}{#1}}

\renewcommand{\st}{\h{st}<f>}
\newcommand{\gt}{\h{gt}<f>}
\newcommand{\lt}{\h{lt}<f>}
\newcommand{\mem}{\h{mem}<f>}
\newcommand{\omem}{\h{mem}<f>_0}
\newcommand{\tbl}{\h{tbl}<f>}
\newcommand{\ctx}{\h{ctx}<f>}
\newcommand{\oat}{\h{at}<f>_0}
\newcommand{\ogt}{\h{gt}<f>_0}
\newcommand{\MState}[2]{\h{MState}[#1][#2]}
\newcommand{\hpcI}{\h{pc_1}<p>}
\newcommand{\hpcII}{\h{pc_2}<p>}
\newcommand{\hfid}{\h{fid}<p>}
\newcommand{\hnext}{\h{next}<p>}
\newcommand{\amem}[1]{\call{\alpha_M}{#1}}
\newcommand{\sz}{\h{sz}<f>}
\newcommand{\osz}{\h{sz}<f>_0}
\newcommand{\vm}{\h{v}<f>}
\newcommand{\ovm}{\h{v}<f>_0}
\newcommand{\ainstr}[2]{(\mkern-3.4mu| #2 |\mkern-3.3mu)^{\h{fid}<p>}_{#1}} 
\newcommand{\imp}{\textit{mp}}

\newcommand{\CmdI}{\Cmd_{\hpcI}}
\newcommand{\CmdII}{\Cmd_{\hpcII}}
\newcommand{\constprefix}[1]{\call{cp}{#1}}
\newcommand{\constprefixP}[1]{\call{cp'}{#1}}
\newcommand{\cppc}[1]{\call{cp_{pc}}{#1}}
\newcommand{\cei}[1]{\access{#1}{cmd}}

\newcommand{\concretemod}[1]{\call{\gamma_m}{#1}}
\newcommand{\preframes}{\newfunc{\gamma_{M,F}}}
\newcommand{\memaccess}[2]{\arraccess{#1}{#2}}

\newcommand{\pca}{\textit{pc}_\textit{a}}
\newcommand{\pcb}{\textit{pc}_\textit{b}}
\newcommand{\fida}{\textit{fid}_\textit{a}}
\newcommand{\fidb}{\textit{fid}_\textit{b}}

\newcommand{\pposs}[2]{
\ensuremath{(#1, #2) \cdot \textit{cs}}
}
\newcommand{\memmp}[1]{
\ensuremath{\mu(#1)}
}
\newcommand{\ea}{
\textit{ea}
}
\newcommand{\ii}{
\textit{i}
}
\NewDocumentCommand{\tos}{o}{%
    \IfNoValueTF{#1}
    { \mpst{z} }
    { \mpst{z#1} }
}
\newcommand{\memarg}{\textit{memarg}}

\newcommand{\getframe}{\Mfunc{get\_frame}}
\newcommand{\getstack}{\Mfunc{get\_stack}}

\newcommand{\abstractcallgraph}{\newfunc{G_{M}}}

\newcommand{\concrete}[1]{\call{\gamma_M}{#1}}
\newcommand{\concretestack}[1]{\call{\gamma_{M,st}}{#1}}
\newcommand{\concretestackp}[1]{\call{\gamma'_{M,st}}{#1}}
\newcommand{\concretestacklocal}[1]{\call{\gamma_{t}}{#1}}
\newcommand{\concretemem}[1]{\call{\gamma_{M,\textit{mem}}}{#1}}
\newcommand{\concretememlocal}[1]{\call{\gamma'_{M,\textit{mem}}}{#1}}
\newcommand{\concretetable}[1]{\call{\gamma_{t}}{#1}}
\newcommand{\concretestore}[1]{\call{\gamma_{M,S}}{#1}}
\newcommand{\concretestorelocal}[1]{\call{\gamma'_{M,S}}{#1}}
\newcommand{\concreteglobals}[1]{\call{\gamma_{M,g}}{#1}}
\newcommand{\concreteglobalslocal}[1]{\call{\gamma'_{M,g}}{#1}}
\newcommand{\icnt}{\textit{cnt}}
\newcommand{\ist}{\textit{st}}
\newcommand{\iinstrsrest}{\iinstrs_\textit{rest}}

\newenvironment{myalign}
 {\par\nopagebreak\small\noindent\ignorespaces\csname align*\endcsname}
 {\nopagebreak\ignorespacesafterend\csname endalign*\endcsname}

\newcommand{\mycode}[1]{\ensuremath{\bm{\mathsf{#1}}}}
\newcommand{\wasmimmediate}[1]{\textit{#1}}

\ExplSyntaxOn
\NewDocumentCommand{\InstructionWithImmediate}{m m o}{
  \ensuremath{
  \mycode{#1}
  \IfNoValueTF{#3}{
  }{
    \c_math_subscript_token {#3}
  }
  \tl_if_blank:nF{#2}{\text{~}}\mathit{#2}
  }
}
\NewDocumentCommand{\InstructionWithOptionalImmediate}{m o o}{
  \ensuremath{
  \mycode{#1}
  \IfNoValueTF{#2}{
  }{
    \IfNoValueTF{#3}{
    }{
      \c_math_subscript_token {#3}
    }
    \tl_if_blank:nF{#2}{\text{~}}\mathit{#2}
  }
  }
}

\ExplSyntaxOff

\newcommand{\watype}[1]{\ensuremath{\mathsf{#1}}}
\newcommand{\intt}{\watype{i32}}
\newcommand{\FuncVoid}[1]{\mycode{func}~\wasmimmediate{#1}}
\newcommand{\FuncNoParam}[2]{\mycode{func}~\wasmimmediate{#1}~(\mycode{result}~#2)}
\newcommand{\Block}[1]{\InstructionWithImmediate{block}{#1}}
\newcommand{\Loop}[1]{\InstructionWithImmediate{loop}{#1}}
\NewDocumentCommand{\End}{o}{\mycode{end}\IfNoValueTF{#1}{}{_{#1}}}
\NewDocumentCommand{\Const}{m m o}{\ensuremath{\watype{#1}.\mycode{const}\IfNoValueTF{#3}{}{_{#3}}~#2}}
\newcommand{\LocalDecl}[2]{(\mycode{local}~\mathit{#1}~#2)}
\newcommand{\LocalSet}[1]{\InstructionWithImmediate{local.set}{#1}}
\newcommand{\LocalTee}[1]{\InstructionWithImmediate{local.tee}{#1}}
\newcommand{\LocalGet}[1]{\InstructionWithImmediate{local.get}{#1}}
\newcommand{\GlobalSet}[1]{\InstructionWithImmediate{global.set}{#1}}
\newcommand{\GlobalGet}[1]{\InstructionWithImmediate{global.get}{#1}}
\newcommand{\Br}[1]{\InstructionWithImmediate{br}{#1}}
\newcommand{\BrIf}[1]{\InstructionWithImmediate{br\_if}{#1}}
\newcommand{\BrTable}[1]{\InstructionWithImmediate{br\_table}{#1}}
\newcommand{\Eq}[1]{\watype{#1}.\mycode{eq}}
\newcommand{\Eqz}[1]{\watype{#1}.\mycode{eqz}}
\newcommand{\Add}[1]{\watype{#1}.\mycode{add}}
\newcommand{\Sub}[1]{\watype{#1}.\mycode{sub}}
\newcommand{\Mul}[1]{\watype{#1}.\mycode{mul}}
\newcommand{\Div}[1]{\watype{#1}.\mycode{div}}
\newcommand{\Mod}[1]{\watype{#1}.\mycode{mod}}
\newcommand{\Store}[1]{\watype{#1}.\mycode{store}}
\newcommand{\StoreN}[1]{\watype{#1}.\mycode{store}\textit{N}~\textit{memarg}}
\newcommand{\Load}[1]{#1.\mycode{load}}
\newcommand{\GeneralLoad}[1]{\AnyCmd{#1}{loadN\_sx}}
\newcommand{\GeneralStore}[1]{\AnyCmd{t}{storeN\_sx}}
\newcommand{\If}{\InstructionWithOptionalImmediate{if}}
\newcommand{\Then}{\mycode{then}}
\NewDocumentCommand{\Else}{o}{\mycode{else}\IfNoValueTF{#1}{}{_{#1}}}
\newcommand{\Result}[1]{(\mycode{result}~#1)}
\newcommand{\IfResult}[1]{\If~\Result{#1}}
\newcommand{\Trap}{\mycode{trap}}
\newcommand{\Label}[2]{\ensuremath{\mycode{label}_{#1}\{#2\}}}
\newcommand{\Frame}[2]{\mycode{frame}_{#1}\{#2\}}
\ExplSyntaxOn
\NewDocumentCommand\Call{m o}{
  \IfNoValueTF{#2}{
    \mycode{call}
  }{
    \mycode{call \c_math_subscript_token {#2}}
  }
  \tl_if_blank:nF{#1}{\hspace{1mm}{#1}}
}
\ExplSyntaxOff
\ExplSyntaxOn
\NewDocumentCommand\CallIndirect{o o}{
\IfNoValueTF{#1}{
\mycode{call_indirect}
}{
  \IfNoValueTF{#2}{
    \mycode{call_indirect}
  }{
    \mycode{call_indirect \c_math_subscript_token {#2}}
  }
  \tl_if_blank:nF{#1}{\hspace{1mm}{#1}}
}
}
\ExplSyntaxOff
\newcommand{\Grow}{\mycode{memory.grow}}
\newcommand{\Size}{\mycode{memory.size}}
\newcommand{\Return}{\mycode{return}}
\newcommand{\Nop}{\mycode{nop}}
\newcommand{\Unreachable}{\mycode{unreachable}}
\newcommand{\Drop}{\mycode{drop}}
\newcommand{\Select}{\mycode{select}}
\NewDocumentCommand{\Cmd}{o}{\mycode{cmd}\IfNoValueTF{#1}{}{_{#1}}}
\newcommand{\AnyCmd}[2]{\ensuremath{\watype{#1}.\mycode{#2}}}

\newcommand{\Invoke}{\mycode{invoke}} 
\newcommand{\Invokea}{\mycode{invoke}~a}

\newcommand{\BlockContext}[2]{B^{#1}\left[#2\right]}


\newcommand{\size}[1]{|{#1}|}
\newcommand{\record}[1]{\left\{{#1}\right\}}
\newcommand{\opt}[1]{{#1}^{?}}
\newcommand{\seq}[1]{{{#1}^{*}}}
\newcommand{\recordentry}[2]{\textsf{#1}\ #2}

\newcommand{\origconfig}{\textit{config}} 
\newcommand{\ival}{\textit{val}} 
\ExplSyntaxOn
\NewDocumentCommand\iinstr{o}{
  \IfNoValueTF{#1}{
    \textit{instr}
  }{
    \ensuremath{\textit{instr}\c_math_subscript_token#1}
  }
}
\NewDocumentCommand\iinstrs{o}{
  \IfNoValueTF{#1}{
    \seq{\iinstr}
  }{
    \seq{\iinstr[#1]}
  }
}
\ExplSyntaxOff
\newcommand{\ianninstr}{\textit{anninstr}} 
\newcommand{\iframe}{\textit{frame}} 
\newcommand{\imoduleinst}{\textit{moduleinst}} 
\newcommand{\ifuncaddr}{\textit{funcaddr}} 
\newcommand{\istore}{\textit{store}} 

\newcommand{\ostep}{\mathrel{\hookrightarrow}}
\newcommand{\bstep}{\mathrel{\underset{b}{\hookrightarrow}}}
\newcommand{\bnstep}{\mathrel{\overset{*}{\bstep}}}
\newcommand{\call}[2]{\mathit{#1}\mkern-2mu\left(#2\right)} 
\newcommand{\mcall}[2]{#1\mkern-2mu\left(#2\right)} 

\newcommand{\pluseq}{\mathrel{+}=}
\newcommand{\concat}{\mathrel{+}\mathrel{+}}
\newcommand{\access}[2]{#1.\mathsf{#2}} 
\newcommand{\arraccess}[2]{#1 \mleft[ #2 \mright]} 
\newcommand{\modify}[3]{#1\left[\textsf{#2} \mapsto #3\right]} 
\newcommand{\increment}[2]{#1\left[\textsf{#2} \pluseq 1\right]} 

\newcommand{\true}{\textit{true}}
\newcommand{\false}{\textit{false}}
\newcommand{\opCONS}{\horstCONS{}{}}

\newcommand{\bconfig}{\textit{bconfig}} 
\newcommand{\ibframe}{\textit{bframe}} 
\newcommand{\iblabel}{\textit{blabel}} 
\newcommand{\ibool}{\textit{bool}} 
\newcommand{\ipc}{\textit{pc}} 
\newcommand{\cond}[3]{\horstCOND{#1}{#2}{#3}} 
\newcommand{\first}[1]{\call{first}{#1}} 

\begin{horstSelectorFunction}{memorySizeSafe}
  \horstName{memorySizeSafe}
  \horstReturnType{\horstTypeint}
\end{horstSelectorFunction}
\begin{horstSelectorFunction}{getAmountOfReturnValuesInBlock}
  \horstName{getAmountOfReturnValuesInBlock}
  \horstParameterType{\horstTypeint}
  \horstParameterType{\horstTypeint}
  \horstReturnType{\horstTypeint}
\end{horstSelectorFunction}
\begin{horstSelectorFunction}{tableSafe}
  \horstName{tableSafe}
  \horstReturnType{\horstTypeint}
\end{horstSelectorFunction}
\begin{horstSelectorFunction}{memoryStoreInLoopBlock}
  \horstName{memoryStoreInLoopBlock}
  \horstParameterType{\horstTypeint}
  \horstParameterType{\horstTypeint}
  \horstReturnType{\horstTypebool}
\end{horstSelectorFunction}
\begin{horstSelectorFunction}{memoryDataSafe}
  \horstName{memoryDataSafe}
  \horstReturnType{\horstTypeint}
\end{horstSelectorFunction}
\begin{horstSelectorFunction}{isMemoryPresent}
  \horstName{isMemoryPresent}
  \horstReturnType{\horstTypebool}
\end{horstSelectorFunction}
\begin{horstSelectorFunction}{tableOutLabel}
  \horstName{tableOutLabel}
  \horstReturnType{\horstTypebool}
  \horstReturnType{\horstTypebool}
\end{horstSelectorFunction}
\begin{horstSelectorFunction}{memoryDataInLabelForImportedFunction}
  \horstName{memoryDataInLabelForImportedFunction}
  \horstParameterType{\horstTypeint}
  \horstReturnType{\horstTypebool}
  \horstReturnType{\horstTypebool}
\end{horstSelectorFunction}
\begin{horstSelectorFunction}{breakTableDestinations}
  \horstName{breakTableDestinations}
  \horstParameterType{\horstTypeint}
  \horstParameterType{\horstTypeint}
  \horstParameterType{\horstTypeint}
  \horstReturnType{\horstTypeint}
\end{horstSelectorFunction}
\begin{horstSelectorFunction}{importCallContextSafe}
  \horstName{importCallContextSafe}
  \horstParameterType{\horstTypeint}
  \horstReturnType{\horstTypeint}
\end{horstSelectorFunction}
\begin{horstSelectorFunction}{binOps}
  \horstName{binOps}
  \horstReturnType{\horstTypeint}
\end{horstSelectorFunction}
\begin{horstSelectorFunction}{globalSafe}
  \horstName{globalSafe}
  \horstReturnType{\horstTypeint}
\end{horstSelectorFunction}
\begin{horstSelectorFunction}{joinsForFunctionId}
  \horstName{joinsForFunctionId}
  \horstParameterType{\horstTypeint}
  \horstReturnType{\horstTypeint}
\end{horstSelectorFunction}
\begin{horstSelectorFunction}{importedFunctionIds}
  \horstName{importedFunctionIds}
  \horstReturnType{\horstTypeint}
\end{horstSelectorFunction}
\begin{horstSelectorFunction}{tableLeak}
  \horstName{tableLeak}
  \horstReturnType{\horstTypeint}
\end{horstSelectorFunction}
\begin{horstSelectorFunction}{returnCountForFunctionId}
  \horstName{returnCountForFunctionId}
  \horstParameterType{\horstTypeint}
  \horstReturnType{\horstTypeint}
\end{horstSelectorFunction}
\begin{horstSelectorFunction}{blocksForFunctionId}
  \horstName{blocksForFunctionId}
  \horstParameterType{\horstTypeint}
  \horstReturnType{\horstTypeint}
\end{horstSelectorFunction}
\begin{horstSelectorFunction}{resultLabelForImportedFunctionAndPosition}
  \horstName{resultLabelForImportedFunctionAndPosition}
  \horstParameterType{\horstTypeint}
  \horstParameterType{\horstTypeint}
  \horstReturnType{\horstTypebool}
  \horstReturnType{\horstTypebool}
\end{horstSelectorFunction}
\begin{horstSelectorFunction}{importCallGlobalLeak}
  \horstName{importCallGlobalLeak}
  \horstParameterType{\horstTypeint}
  \horstReturnType{\horstTypeint}
\end{horstSelectorFunction}
\begin{horstSelectorFunction}{importCallContextLeak}
  \horstName{importCallContextLeak}
  \horstParameterType{\horstTypeint}
  \horstReturnType{\horstTypeint}
\end{horstSelectorFunction}
\begin{horstSelectorFunction}{memoryAreaLeak}
  \horstName{memoryAreaLeak}
  \horstReturnType{\horstTypeint}
  \horstReturnType{\horstTypeint}
  \horstReturnType{\horstTypeint}
\end{horstSelectorFunction}
\begin{horstSelectorFunction}{argumentLabelForPosition}
  \horstName{argumentLabelForPosition}
  \horstParameterType{\horstTypeint}
  \horstReturnType{\horstTypebool}
  \horstReturnType{\horstTypebool}
\end{horstSelectorFunction}
\begin{horstSelectorFunction}{memoryDataOutLabelForImportedFunction}
  \horstName{memoryDataOutLabelForImportedFunction}
  \horstParameterType{\horstTypeint}
  \horstReturnType{\horstTypebool}
  \horstReturnType{\horstTypebool}
  \horstReturnType{\horstTypebool}
\end{horstSelectorFunction}
\begin{horstSelectorFunction}{cvtOps}
  \horstName{cvtOps}
  \horstReturnType{\horstTypeint}
\end{horstSelectorFunction}
\begin{horstSelectorFunction}{interval}
  \horstName{interval}
  \horstParameterType{\horstTypeint}
  \horstParameterType{\horstTypeint}
  \horstReturnType{\horstTypeint}
\end{horstSelectorFunction}
\begin{horstSelectorFunction}{functionIds}
  \horstName{functionIds}
  \horstReturnType{\horstTypeint}
\end{horstSelectorFunction}
\begin{horstSelectorFunction}{globalOutLabelForImportedFunctionAndPosition}
  \horstName{globalOutLabelForImportedFunctionAndPosition}
  \horstParameterType{\horstTypeint}
  \horstParameterType{\horstTypeint}
  \horstReturnType{\horstTypebool}
  \horstReturnType{\horstTypebool}
  \horstReturnType{\horstTypebool}
\end{horstSelectorFunction}
\begin{horstSelectorFunction}{memoryDataLeak}
  \horstName{memoryDataLeak}
  \horstReturnType{\horstTypeint}
\end{horstSelectorFunction}
\begin{horstSelectorFunction}{onlyIf}
  \horstName{onlyIf}
  \horstParameterType{\horstTypebool}
  \horstReturnType{\horstTypebool}
\end{horstSelectorFunction}
\begin{horstSelectorFunction}{stackSizeForFunctionIdAndPc}
  \horstName{stackSizeForFunctionIdAndPc}
  \horstParameterType{\horstTypeint}
  \horstParameterType{\horstTypeint}
  \horstReturnType{\horstTypeint}
\end{horstSelectorFunction}
\begin{horstSelectorFunction}{trappingCvtOps}
  \horstName{trappingCvtOps}
  \horstReturnType{\horstTypeint}
\end{horstSelectorFunction}
\begin{horstSelectorFunction}{valueAndTopOfGlobal}
  \horstName{valueAndTopOfGlobal}
  \horstParameterType{\horstTypeint}
  \horstReturnType{\horstTypeint}
  \horstReturnType{\horstTypebool}
\end{horstSelectorFunction}
\begin{horstSelectorFunction}{storeOps}
  \horstName{storeOps}
  \horstReturnType{\horstTypeint}
\end{horstSelectorFunction}
\begin{horstSelectorFunction}{trappingBinOps}
  \horstName{trappingBinOps}
  \horstReturnType{\horstTypeint}
\end{horstSelectorFunction}
\begin{horstSelectorFunction}{importCallTableLeak}
  \horstName{importCallTableLeak}
  \horstParameterType{\horstTypeint}
  \horstReturnType{\horstTypeint}
\end{horstSelectorFunction}
\begin{horstSelectorFunction}{importCallArgumentSafe}
  \horstName{importCallArgumentSafe}
  \horstParameterType{\horstTypeint}
  \horstReturnType{\horstTypeint}
\end{horstSelectorFunction}
\begin{horstSelectorFunction}{importCallGlobalSafe}
  \horstName{importCallGlobalSafe}
  \horstParameterType{\horstTypeint}
  \horstReturnType{\horstTypeint}
\end{horstSelectorFunction}
\begin{horstSelectorFunction}{sizeOfBreakTable}
  \horstName{sizeOfBreakTable}
  \horstParameterType{\horstTypeint}
  \horstParameterType{\horstTypeint}
  \horstReturnType{\horstTypeint}
\end{horstSelectorFunction}
\begin{horstSelectorFunction}{memoryAreaSafe}
  \horstName{memoryAreaSafe}
  \horstReturnType{\horstTypeint}
  \horstReturnType{\horstTypeint}
  \horstReturnType{\horstTypeint}
\end{horstSelectorFunction}
\begin{horstSelectorFunction}{memoryDataInLabel}
  \horstName{memoryDataInLabel}
  \horstReturnType{\horstTypebool}
  \horstReturnType{\horstTypebool}
\end{horstSelectorFunction}
\begin{horstSelectorFunction}{exitPointsForFunctionId}
  \horstName{exitPointsForFunctionId}
  \horstParameterType{\horstTypeint}
  \horstReturnType{\horstTypeint}
\end{horstSelectorFunction}
\begin{horstSelectorFunction}{globalCount}
  \horstName{globalCount}
  \horstReturnType{\horstTypeint}
\end{horstSelectorFunction}
\begin{horstSelectorFunction}{memorySizeLeak}
  \horstName{memorySizeLeak}
  \horstReturnType{\horstTypeint}
\end{horstSelectorFunction}
\begin{horstSelectorFunction}{endsForFunctionId}
  \horstName{endsForFunctionId}
  \horstParameterType{\horstTypeint}
  \horstReturnType{\horstTypeint}
  \horstReturnType{\horstTypeint}
\end{horstSelectorFunction}
\begin{horstSelectorFunction}{globalLeak}
  \horstName{globalLeak}
  \horstReturnType{\horstTypeint}
\end{horstSelectorFunction}
\begin{horstSelectorFunction}{possibleCallTargets}
  \horstName{possibleCallTargets}
  \horstParameterType{\horstTypeint}
  \horstParameterType{\horstTypeint}
  \horstReturnType{\horstTypeint}
  \horstReturnType{\horstTypeint}
\end{horstSelectorFunction}
\begin{horstSelectorFunction}{tableInLabelForImportedFunction}
  \horstName{tableInLabelForImportedFunction}
  \horstParameterType{\horstTypeint}
  \horstReturnType{\horstTypebool}
  \horstReturnType{\horstTypebool}
\end{horstSelectorFunction}
\begin{horstSelectorFunction}{maxPcForFunctionId}
  \horstName{maxPcForFunctionId}
  \horstParameterType{\horstTypeint}
  \horstReturnType{\horstTypeint}
\end{horstSelectorFunction}
\begin{horstSelectorFunction}{memorySizeOutLabel}
  \horstName{memorySizeOutLabel}
  \horstReturnType{\horstTypebool}
  \horstReturnType{\horstTypebool}
\end{horstSelectorFunction}
\begin{horstSelectorFunction}{importCallArgumentLeak}
  \horstName{importCallArgumentLeak}
  \horstParameterType{\horstTypeint}
  \horstReturnType{\horstTypeint}
\end{horstSelectorFunction}
\begin{horstSelectorFunction}{localModifiedInLoopBlock}
  \horstName{localModifiedInLoopBlock}
  \horstParameterType{\horstTypeint}
  \horstParameterType{\horstTypeint}
  \horstParameterType{\horstTypeint}
  \horstReturnType{\horstTypebool}
\end{horstSelectorFunction}
\begin{horstSelectorFunction}{loopsForFunctionId}
  \horstName{loopsForFunctionId}
  \horstParameterType{\horstTypeint}
  \horstReturnType{\horstTypeint}
\end{horstSelectorFunction}
\begin{horstSelectorFunction}{resultLabelForPosition}
  \horstName{resultLabelForPosition}
  \horstParameterType{\horstTypeint}
  \horstReturnType{\horstTypebool}
  \horstReturnType{\horstTypebool}
\end{horstSelectorFunction}
\begin{horstSelectorFunction}{tableOutLabelForImportedFunction}
  \horstName{tableOutLabelForImportedFunction}
  \horstParameterType{\horstTypeint}
  \horstReturnType{\horstTypebool}
  \horstReturnType{\horstTypebool}
  \horstReturnType{\horstTypebool}
\end{horstSelectorFunction}
\begin{horstSelectorFunction}{argumentLabelForImportedFunctionAndPosition}
  \horstName{argumentLabelForImportedFunctionAndPosition}
  \horstParameterType{\horstTypeint}
  \horstParameterType{\horstTypeint}
  \horstReturnType{\horstTypebool}
  \horstReturnType{\horstTypebool}
\end{horstSelectorFunction}
\begin{horstSelectorFunction}{getMemoryMin}
  \horstName{getMemoryMin}
  \horstReturnType{\horstTypeint}
\end{horstSelectorFunction}
\begin{horstSelectorFunction}{memoryGrowInLoopBlock}
  \horstName{memoryGrowInLoopBlock}
  \horstParameterType{\horstTypeint}
  \horstParameterType{\horstTypeint}
  \horstReturnType{\horstTypebool}
\end{horstSelectorFunction}
\begin{horstSelectorFunction}{immediateForFunctionIdAndPc}
  \horstName{immediateForFunctionIdAndPc}
  \horstParameterType{\horstTypeint}
  \horstParameterType{\horstTypeint}
  \horstReturnType{\horstTypeint}
\end{horstSelectorFunction}
\begin{horstSelectorFunction}{startFunctionId}
  \horstName{startFunctionId}
  \horstReturnType{\horstTypeint}
\end{horstSelectorFunction}
\begin{horstSelectorFunction}{memoryDataInLabelExceptions}
  \horstName{memoryDataInLabelExceptions}
  \horstReturnType{\horstTypeint}
  \horstReturnType{\horstTypeint}
  \horstReturnType{\horstTypebool}
  \horstReturnType{\horstTypebool}
\end{horstSelectorFunction}
\begin{horstSelectorFunction}{memoryDataOutLabel}
  \horstName{memoryDataOutLabel}
  \horstReturnType{\horstTypebool}
  \horstReturnType{\horstTypebool}
\end{horstSelectorFunction}
\begin{horstSelectorFunction}{tableInLabel}
  \horstName{tableInLabel}
  \horstReturnType{\horstTypebool}
  \horstReturnType{\horstTypebool}
\end{horstSelectorFunction}
\begin{horstSelectorFunction}{contextLabelForImportedFunction}
  \horstName{contextLabelForImportedFunction}
  \horstParameterType{\horstTypeint}
  \horstReturnType{\horstTypebool}
  \horstReturnType{\horstTypebool}
\end{horstSelectorFunction}
\begin{horstSelectorFunction}{resultSafe}
  \horstName{resultSafe}
  \horstReturnType{\horstTypeint}
\end{horstSelectorFunction}
\begin{horstSelectorFunction}{pcsForFunctionId}
  \horstName{pcsForFunctionId}
  \horstParameterType{\horstTypeint}
  \horstReturnType{\horstTypeint}
\end{horstSelectorFunction}
\begin{horstSelectorFunction}{argumentCountForFunctionId}
  \horstName{argumentCountForFunctionId}
  \horstParameterType{\horstTypeint}
  \horstReturnType{\horstTypeint}
\end{horstSelectorFunction}
\begin{horstSelectorFunction}{datasegmentsWithPositions}
  \horstName{datasegmentsWithPositions}
  \horstReturnType{\horstTypeint}
  \horstReturnType{\horstTypeint}
\end{horstSelectorFunction}
\begin{horstSelectorFunction}{getMemoryMax}
  \horstName{getMemoryMax}
  \horstReturnType{\horstTypeint}
\end{horstSelectorFunction}
\begin{horstSelectorFunction}{isMemoryImported}
  \horstName{isMemoryImported}
  \horstReturnType{\horstTypebool}
\end{horstSelectorFunction}
\begin{horstSelectorFunction}{importCallMemorySizeLeak}
  \horstName{importCallMemorySizeLeak}
  \horstParameterType{\horstTypeint}
  \horstReturnType{\horstTypeint}
\end{horstSelectorFunction}
\begin{horstSelectorFunction}{memorySizeInLabelForImportedFunction}
  \horstName{memorySizeInLabelForImportedFunction}
  \horstParameterType{\horstTypeint}
  \horstReturnType{\horstTypebool}
  \horstReturnType{\horstTypebool}
\end{horstSelectorFunction}
\begin{horstSelectorFunction}{importCallMemoryDataSafe}
  \horstName{importCallMemoryDataSafe}
  \horstParameterType{\horstTypeint}
  \horstReturnType{\horstTypeint}
\end{horstSelectorFunction}
\begin{horstSelectorFunction}{importCallMemorySizeSafe}
  \horstName{importCallMemorySizeSafe}
  \horstParameterType{\horstTypeint}
  \horstReturnType{\horstTypeint}
\end{horstSelectorFunction}
\begin{horstSelectorFunction}{resultLeak}
  \horstName{resultLeak}
  \horstReturnType{\horstTypeint}
\end{horstSelectorFunction}
\begin{horstSelectorFunction}{globalInLabelForImportedFunctionAndPosition}
  \horstName{globalInLabelForImportedFunctionAndPosition}
  \horstParameterType{\horstTypeint}
  \horstParameterType{\horstTypeint}
  \horstReturnType{\horstTypebool}
  \horstReturnType{\horstTypebool}
\end{horstSelectorFunction}
\begin{horstSelectorFunction}{globalOutLabelForPosition}
  \horstName{globalOutLabelForPosition}
  \horstParameterType{\horstTypeint}
  \horstReturnType{\horstTypebool}
  \horstReturnType{\horstTypebool}
\end{horstSelectorFunction}
\begin{horstSelectorFunction}{getHighestDataSegmentIndex}
  \horstName{getHighestDataSegmentIndex}
  \horstReturnType{\horstTypeint}
\end{horstSelectorFunction}
\begin{horstSelectorFunction}{bitwidthForLocal}
  \horstName{bitwidthForLocal}
  \horstParameterType{\horstTypeint}
  \horstParameterType{\horstTypeint}
  \horstReturnType{\horstTypeint}
\end{horstSelectorFunction}
\begin{horstSelectorFunction}{memoryOffsetForFunctionIdAndPc}
  \horstName{memoryOffsetForFunctionIdAndPc}
  \horstParameterType{\horstTypeint}
  \horstParameterType{\horstTypeint}
  \horstReturnType{\horstTypeint}
\end{horstSelectorFunction}
\begin{horstSelectorFunction}{memorySizeOutLabelForImportedFunction}
  \horstName{memorySizeOutLabelForImportedFunction}
  \horstParameterType{\horstTypeint}
  \horstReturnType{\horstTypebool}
  \horstReturnType{\horstTypebool}
  \horstReturnType{\horstTypebool}
\end{horstSelectorFunction}
\begin{horstSelectorFunction}{bitwidthForGlobal}
  \horstName{bitwidthForGlobal}
  \horstParameterType{\horstTypeint}
  \horstReturnType{\horstTypeint}
\end{horstSelectorFunction}
\begin{horstSelectorFunction}{globalModifiedInLoopBlock}
  \horstName{globalModifiedInLoopBlock}
  \horstParameterType{\horstTypeint}
  \horstParameterType{\horstTypeint}
  \horstParameterType{\horstTypeint}
  \horstReturnType{\horstTypebool}
\end{horstSelectorFunction}
\begin{horstSelectorFunction}{importCallTableSafe}
  \horstName{importCallTableSafe}
  \horstParameterType{\horstTypeint}
  \horstReturnType{\horstTypeint}
\end{horstSelectorFunction}
\begin{horstSelectorFunction}{isImportedFunction}
  \horstName{isImportedFunction}
  \horstParameterType{\horstTypeint}
  \horstReturnType{\horstTypebool}
\end{horstSelectorFunction}
\begin{horstSelectorFunction}{endForIf}
  \horstName{endForIf}
  \horstParameterType{\horstTypeint}
  \horstParameterType{\horstTypeint}
  \horstReturnType{\horstTypeint}
\end{horstSelectorFunction}
\begin{horstSelectorFunction}{elseForIf}
  \horstName{elseForIf}
  \horstParameterType{\horstTypeint}
  \horstParameterType{\horstTypeint}
  \horstReturnType{\horstTypeint}
\end{horstSelectorFunction}
\begin{horstSelectorFunction}{pcsAndValueAndTopOfConstsForFunctionId}
  \horstName{pcsAndValueAndTopOfConstsForFunctionId}
  \horstParameterType{\horstTypeint}
  \horstReturnType{\horstTypeint}
  \horstReturnType{\horstTypeint}
  \horstReturnType{\horstTypebool}
\end{horstSelectorFunction}
\begin{horstSelectorFunction}{globalInLabelForPosition}
  \horstName{globalInLabelForPosition}
  \horstParameterType{\horstTypeint}
  \horstReturnType{\horstTypebool}
  \horstReturnType{\horstTypebool}
\end{horstSelectorFunction}
\begin{horstSelectorFunction}{pcsForFunctionIdAndOpcode}
  \horstName{pcsForFunctionIdAndOpcode}
  \horstParameterType{\horstTypeint}
  \horstParameterType{\horstTypeint}
  \horstReturnType{\horstTypeint}
\end{horstSelectorFunction}
\begin{horstSelectorFunction}{breakDestinations}
  \horstName{breakDestinations}
  \horstParameterType{\horstTypeint}
  \horstParameterType{\horstTypeint}
  \horstReturnType{\horstTypeint}
\end{horstSelectorFunction}
\begin{horstSelectorFunction}{bitwidthForResult}
  \horstName{bitwidthForResult}
  \horstParameterType{\horstTypeint}
  \horstParameterType{\horstTypeint}
  \horstReturnType{\horstTypeint}
\end{horstSelectorFunction}
\begin{horstSelectorFunction}{importCallMemoryDataLeak}
  \horstName{importCallMemoryDataLeak}
  \horstParameterType{\horstTypeint}
  \horstReturnType{\horstTypeint}
\end{horstSelectorFunction}
\begin{horstSelectorFunction}{possibleHavokCallTargets}
  \horstName{possibleHavokCallTargets}
  \horstParameterType{\horstTypeint}
  \horstParameterType{\horstTypeint}
  \horstReturnType{\horstTypeint}
\end{horstSelectorFunction}
\begin{horstSelectorFunction}{memorySizeInLabel}
  \horstName{memorySizeInLabel}
  \horstReturnType{\horstTypebool}
  \horstReturnType{\horstTypebool}
\end{horstSelectorFunction}
\begin{horstSelectorFunction}{localCountForFunctionId}
  \horstName{localCountForFunctionId}
  \horstParameterType{\horstTypeint}
  \horstReturnType{\horstTypeint}
\end{horstSelectorFunction}
\begin{horstSelectorFunction}{ifs}
  \horstName{ifs}
  \horstParameterType{\horstTypeint}
  \horstReturnType{\horstTypeint}
\end{horstSelectorFunction}
\begin{horstSelectorFunction}{bitwidthForArgument}
  \horstName{bitwidthForArgument}
  \horstParameterType{\horstTypeint}
  \horstParameterType{\horstTypeint}
  \horstReturnType{\horstTypeint}
\end{horstSelectorFunction}
\begin{horstSelectorFunction}{isTableImprecise}
  \horstName{isTableImprecise}
  \horstReturnType{\horstTypebool}
\end{horstSelectorFunction}
\begin{horstSelectorFunction}{unOps}
  \horstName{unOps}
  \horstReturnType{\horstTypeint}
\end{horstSelectorFunction}
\begin{horstSelectorFunction}{loadOps}
  \horstName{loadOps}
  \horstReturnType{\horstTypeint}
\end{horstSelectorFunction}
\begin{horstSelectorFunction}{unit}
  \horstName{unit}
\end{horstSelectorFunction}
\begin{horstPredicate}{Init}
  \horstName{Init}
  \horstParVar{fid}{\horstTypeint}
  \horstArgumentType{\horstTypeLabel}
  \horstArgumentType{\horstTypeHomInit{\horstTypeLValue}{\horstOpAppas{\horstParVarfid}{}}}
  \horstArgumentType{\horstTypeHomInit{\horstTypeLValue}{\horstOpAppgs{}{}}}
  \horstArgumentType{\horstTypeMemory}
\end{horstPredicate}
\begin{horstPredicate}{Return}
  \horstName{Return}
  \horstParVar{fid}{\horstTypeint}
  \horstArgumentType{\horstTypeContext}
  \horstArgumentType{\horstTypeHomInit{\horstTypeLValue}{\horstOpApprs{\horstParVarfid}{}}}
  \horstArgumentType{\horstTypeHomInit{\horstTypeLValue}{\horstOpAppgs{}{}}}
  \horstArgumentType{\horstTypeMemory}
  \horstArgumentType{\horstTypeTable}
  \horstArgumentType{\horstTypeHomInit{\horstTypeLValue}{\horstOpAppas{\horstParVarfid}{}}}
  \horstArgumentType{\horstTypeHomInit{\horstTypeLValue}{\horstOpAppgs{}{}}}
  \horstArgumentType{\horstTypeMemory}
\end{horstPredicate}
\begin{horstPredicate}{ReturnCall}
  \horstName{ReturnCall}
  \horstParVar{fid}{\horstTypeint}
  \horstArgumentType{\horstTypeContext}
  \horstArgumentType{\horstTypeHomInit{\horstTypeLValue}{\horstOpApprs{\horstParVarfid}{}}}
  \horstArgumentType{\horstTypeHomInit{\horstTypeLValue}{\horstOpAppgs{}{}}}
  \horstArgumentType{\horstTypeMemory}
  \horstArgumentType{\horstTypeTable}
  \horstArgumentType{\horstTypeHomInit{\horstTypeLValue}{\horstOpAppas{\horstParVarfid}{}}}
  \horstArgumentType{\horstTypeHomInit{\horstTypeLValue}{\horstOpAppgs{}{}}}
  \horstArgumentType{\horstTypeMemory}
\end{horstPredicate}
\begin{horstPredicate}{MStateToJoin}
  \horstName{MStateToJoin}
  \horstParVar{fid}{\horstTypeint}
  \horstParVar{pc}{\horstTypeint}
  \horstArgumentType{\horstTypeContext}
  \horstArgumentType{\horstTypeHomInit{\horstTypeLValue}{\horstOpAppss{\horstParVarfid,\horstParVarpc}{}}}
  \horstArgumentType{\horstTypeHomInit{\horstTypeLValue}{\horstOpAppgs{}{}}}
  \horstArgumentType{\horstTypeHomInit{\horstTypeLValue}{\horstOpAppls{\horstParVarfid}{}}}
  \horstArgumentType{\horstTypeMemory}
  \horstArgumentType{\horstTypeTable}
  \horstArgumentType{\horstTypeHomInit{\horstTypeLValue}{\horstOpAppas{\horstParVarfid}{}}}
  \horstArgumentType{\horstTypeHomInit{\horstTypeLValue}{\horstOpAppgs{}{}}}
  \horstArgumentType{\horstTypeMemory}
\end{horstPredicate}
\begin{horstPredicate}{MState}
  \horstName{MState}
  \horstParVar{fid}{\horstTypeint}
  \horstParVar{pc}{\horstTypeint}
  \horstArgumentType{\horstTypeContext}
  \horstArgumentType{\horstTypeHomInit{\horstTypeLValue}{\horstOpAppss{\horstParVarfid,\horstParVarpc}{}}}
  \horstArgumentType{\horstTypeHomInit{\horstTypeLValue}{\horstOpAppgs{}{}}}
  \horstArgumentType{\horstTypeHomInit{\horstTypeLValue}{\horstOpAppls{\horstParVarfid}{}}}
  \horstArgumentType{\horstTypeMemory}
  \horstArgumentType{\horstTypeTable}
  \horstArgumentType{\horstTypeHomInit{\horstTypeLValue}{\horstOpAppas{\horstParVarfid}{}}}
  \horstArgumentType{\horstTypeHomInit{\horstTypeLValue}{\horstOpAppgs{}{}}}
  \horstArgumentType{\horstTypeMemory}
\end{horstPredicate}
\begin{horstPredicate}{ReturnToJoin}
  \horstName{ReturnToJoin}
  \horstParVar{fid}{\horstTypeint}
  \horstArgumentType{\horstTypeContext}
  \horstArgumentType{\horstTypeHomInit{\horstTypeLValue}{\horstOpApprs{\horstParVarfid}{}}}
  \horstArgumentType{\horstTypeHomInit{\horstTypeLValue}{\horstOpAppgs{}{}}}
  \horstArgumentType{\horstTypeMemory}
  \horstArgumentType{\horstTypeTable}
  \horstArgumentType{\horstTypeHomInit{\horstTypeLValue}{\horstOpAppas{\horstParVarfid}{}}}
  \horstArgumentType{\horstTypeHomInit{\horstTypeLValue}{\horstOpAppgs{}{}}}
  \horstArgumentType{\horstTypeMemory}
\end{horstPredicate}
\begin{horstPredicate}{ScopeExtend}
  \horstName{ScopeExtend}
  \horstParVar{fid}{\horstTypeint}
  \horstArgumentType{\horstTypeint}
  \horstArgumentType{\horstTypeint}
\end{horstPredicate}
\begin{horstOperation}{pow}
  \horstName{pow}
  \horstParVar{n}{\horstTypeint}
  \horstVar{a}{\horstTypeint}
  \horstReturnType{\horstTypeint}
  \horstBody{\horstSimpleSumExp{MUL}{\horstSelectorFunctionAppinterval{\horstParVari}{0,\horstParVarn}}{\horstVara}{{i}}}
\end{horstOperation}
\begin{horstOperation}{abs}
  \horstName{abs}
  \horstVar{i}{\horstTypeint}
  \horstReturnType{\horstTypeint}
  \horstBody{\horstCOND{\horstLT{\horstVari}{0}}{\horstSUB{0}{\horstVari}}{\horstVari}}
\end{horstOperation}
\begin{horstOperation}{min}
  \horstName{min}
  \horstVar{a}{\horstTypeint}
  \horstVar{b}{\horstTypeint}
  \horstReturnType{\horstTypeint}
  \horstBody{\horstCOND{\horstLT{\horstVara}{\horstVarb}}{\horstVara}{\horstVarb}}
\end{horstOperation}
\begin{horstOperation}{max}
  \horstName{max}
  \horstVar{a}{\horstTypeint}
  \horstVar{b}{\horstTypeint}
  \horstReturnType{\horstTypeint}
  \horstBody{\horstCOND{\horstGT{\horstVara}{\horstVarb}}{\horstVara}{\horstVarb}}
\end{horstOperation}
\begin{horstOperation}{pcmax}
  \horstName{pcmax}
  \horstParVar{fid}{\horstTypeint}
  \horstReturnType{\horstTypeint}
  \horstBody{\horstSimpleSumExp{ADD}{\horstSelectorFunctionAppmaxPcForFunctionId{\horstParVarpc}{\horstParVarfid}}{\horstParVarpc}{{pc}}}
\end{horstOperation}
\begin{horstOperation}{ss}
  \horstName{ss}
  \horstParVar{fid}{\horstTypeint}
  \horstParVar{pc}{\horstTypeint}
  \horstReturnType{\horstTypeint}
  \horstBody{\horstSimpleSumExp{ADD}{\horstSelectorFunctionAppstackSizeForFunctionIdAndPc{\horstParVarss}{\horstParVarfid,\horstParVarpc}}{\horstParVarss}{{ss}}}
\end{horstOperation}
\begin{horstOperation}{ls}
  \horstName{ls}
  \horstParVar{fid}{\horstTypeint}
  \horstReturnType{\horstTypeint}
  \horstBody{\horstSimpleSumExp{ADD}{\horstSelectorFunctionApplocalCountForFunctionId{\horstParVarls}{\horstParVarfid}}{\horstParVarls}{{ls}}}
\end{horstOperation}
\begin{horstOperation}{as}
  \horstName{as}
  \horstParVar{fid}{\horstTypeint}
  \horstReturnType{\horstTypeint}
  \horstBody{\horstSimpleSumExp{ADD}{\horstSelectorFunctionAppargumentCountForFunctionId{\horstParVaras}{\horstParVarfid}}{\horstParVaras}{{as}}}
\end{horstOperation}
\begin{horstOperation}{rs}
  \horstName{rs}
  \horstParVar{fid}{\horstTypeint}
  \horstReturnType{\horstTypeint}
  \horstBody{\horstSimpleSumExp{ADD}{\horstSelectorFunctionAppreturnCountForFunctionId{\horstParVarrs}{\horstParVarfid}}{\horstParVarrs}{{rs}}}
\end{horstOperation}
\begin{horstOperation}{gs}
  \horstName{gs}
  \horstReturnType{\horstTypeint}
  \horstBody{\horstSimpleSumExp{ADD}{\horstSelectorFunctionAppglobalCount{\horstParVargs}{}}{\horstParVargs}{{gs}}}
\end{horstOperation}
\begin{horstOperation}{ds}
  \horstName{ds}
  \horstReturnType{\horstTypeint}
  \horstBody{\horstCustomSumExp{\horstSelectorFunctionAppgetHighestDataSegmentIndex{\horstParVari}{}}{x}{\horstParVari}{0}{{i}}}
\end{horstOperation}
\begin{horstOperation}{mms}
  \horstName{mms}
  \horstReturnType{\horstTypeint}
  \horstBody{\horstCustomSumExp{\horstSelectorFunctionAppgetMemoryMax{\horstParVarmax}{}}{x}{\horstParVarmax}{0}{{max}}}
\end{horstOperation}
\begin{horstOperation}{freeBaseValue}
  \horstName{freeBaseValue}
  \horstReturnType{\horstTypeBVLXIV}
\end{horstOperation}
\begin{horstOperation}{bvIIint}
  \horstName{bv2int}
  \horstVar{v}{\horstTypeBVLXIV}
  \horstReturnType{\horstTypeint}
\end{horstOperation}
\begin{horstOperation}{intIIbv}
  \horstName{int2bv}
  \horstParVar{bw}{\horstTypeint}
  \horstVar{v}{\horstTypeint}
  \horstReturnType{\horstTypeBVLXIV}
\end{horstOperation}
\begin{horstOperation}{cintIIbv}
  \horstName{cint2bv}
  \horstParVar{bw}{\horstTypeint}
  \horstParVar{v}{\horstTypeint}
  \horstReturnType{\horstTypeBVLXIV}
\end{horstOperation}
\begin{horstOperation}{base}
  \horstName{base}
  \horstVar{v}{\horstTypeValue}
  \horstReturnType{\horstTypeBVLXIV}
  \horstBody{\horstMatchExp{{{{x}},\horstMATCHIF{\horstEQ{\horstVarv}{\horstConstructorAppVal{\horstVarx}}},\horstVarx},{{},\horstOTHERWISE,\horstOpAppfreeBaseValue{}{}}}}
\end{horstOperation}
\begin{horstOperation}{cinj}
  \horstName{cinj}
  \horstParVar{i}{\horstTypeint}
  \horstReturnType{\horstTypeBVLXIV}
  \horstBody{\horstOpAppcintIIbv{64,\horstParVari}{}}
\end{horstOperation}
\begin{horstOperation}{inj}
  \horstName{inj}
  \horstVar{i}{\horstTypeint}
  \horstReturnType{\horstTypeBVLXIV}
  \horstBody{\horstOpAppintIIbv{64}{\horstVari}}
\end{horstOperation}
\begin{horstOperation}{proj}
  \horstName{proj}
  \horstVar{i}{\horstTypeBVLXIV}
  \horstReturnType{\horstTypeint}
  \horstBody{\horstOpAppbvIIint{}{\horstVari}}
\end{horstOperation}
\begin{horstOperation}{freeValOrTop}
  \horstName{freeValOrTop}
  \horstReturnType{\horstTypeValue}
  \horstBody{\horstConstructorAppVal{\horstOpAppfreeBaseValue{}{}}}
\end{horstOperation}
\begin{horstOperation}{freeBool}
  \horstName{freeBool}
  \horstReturnType{\horstTypebool}
\end{horstOperation}
\begin{horstOperation}{freeInt}
  \horstName{freeInt}
  \horstReturnType{\horstTypeint}
\end{horstOperation}
\begin{horstOperation}{mkConst}
  \horstName{mkConst}
  \horstParVar{i}{\horstTypeint}
  \horstReturnType{\horstTypeValue}
  \horstBody{\horstConstructorAppVal{\horstOpAppcinj{\horstParVari}{}}}
\end{horstOperation}
\begin{horstOperation}{mkValue}
  \horstName{mkValue}
  \horstVar{i}{\horstTypeint}
  \horstReturnType{\horstTypeValue}
  \horstBody{\horstConstructorAppVal{\horstOpAppinj{}{\horstVari}}}
\end{horstOperation}
\begin{horstOperation}{mkIntOfValue}
  \horstName{mkIntOfValue}
  \horstVar{val}{\horstTypeValue}
  \horstReturnType{\horstTypeint}
  \horstBody{\horstMatchExp{{{{v}},\horstMATCHIF{\horstEQ{\horstVarval}{\horstConstructorAppVal{\horstVarv}}},\horstOpAppproj{}{\horstVarv}},{{},\horstOTHERWISE,\horstOpAppfreeInt{}{}}}}
\end{horstOperation}
\begin{horstOperation}{val}
  \horstName{val}
  \horstParVar{top}{\horstTypebool}
  \horstParVar{v}{\horstTypeint}
  \horstReturnType{\horstTypeValue}
  \horstBody{\horstCOND{\horstParVartop}{\horstOpAppfreeValOrTop{}{}}{\horstOpAppmkConst{\horstParVarv}{}}}
\end{horstOperation}
\begin{horstOperation}{isJustV}
  \horstName{isJustV}
  \horstVar{mv}{\horstTypeMaybeValue}
  \horstReturnType{\horstTypebool}
  \horstBody{\horstMatchExp{{{},\horstMATCHIF{\horstEQ{\horstVarmv}{\horstConstructorAppJustV{\horstVarWILDCARD}}},\horstTrue},{{},\horstOTHERWISE,\horstFalse}}}
\end{horstOperation}
\begin{horstOperation}{fromJustV}
  \horstName{fromJustV}
  \horstVar{mv}{\horstTypeMaybeValue}
  \horstReturnType{\horstTypeValue}
  \horstBody{\horstMatchExp{{{{v}},\horstMATCHIF{\horstEQ{\horstVarmv}{\horstConstructorAppJustV{\horstVarv}}},\horstVarv},{{},\horstOTHERWISE,\horstOpAppfreeValOrTop{}{}}}}
\end{horstOperation}
\begin{horstOperation}{bvextendu}
  \horstName{bvextendu}
  \horstParVar{m}{\horstTypeint}
  \horstParVar{n}{\horstTypeint}
  \horstVar{x}{\horstTypeBVLXIV}
  \horstReturnType{\horstTypeBVLXIV}
\end{horstOperation}
\begin{horstOperation}{bvextends}
  \horstName{bvextends}
  \horstParVar{m}{\horstTypeint}
  \horstParVar{n}{\horstTypeint}
  \horstVar{x}{\horstTypeBVLXIV}
  \horstReturnType{\horstTypeBVLXIV}
\end{horstOperation}
\begin{horstOperation}{bvextract}
  \horstName{bvextract}
  \horstParVar{m}{\horstTypeint}
  \horstParVar{n}{\horstTypeint}
  \horstVar{x}{\horstTypeBVLXIV}
  \horstReturnType{\horstTypeBVLXIV}
\end{horstOperation}
\begin{horstOperation}{csigned}
  \horstName{csigned}
  \horstParVar{bw}{\horstTypeint}
  \horstVar{i}{\horstTypeBVLXIV}
  \horstReturnType{\horstTypeBVLXIV}
  \horstBody{\horstVari}
\end{horstOperation}
\begin{horstOperation}{cunsigned}
  \horstName{cunsigned}
  \horstParVar{bw}{\horstTypeint}
  \horstVar{i}{\horstTypeBVLXIV}
  \horstReturnType{\horstTypeBVLXIV}
  \horstBody{\horstVari}
\end{horstOperation}
\begin{horstOperation}{cextendu}
  \horstName{cextendu}
  \horstParVar{m}{\horstTypeint}
  \horstParVar{n}{\horstTypeint}
  \horstVar{x}{\horstTypeBVLXIV}
  \horstReturnType{\horstTypeBVLXIV}
  \horstBody{\horstOpAppbvextendu{\horstParVarm,\horstParVarn}{\horstVarx}}
\end{horstOperation}
\begin{horstOperation}{cextends}
  \horstName{cextends}
  \horstParVar{m}{\horstTypeint}
  \horstParVar{n}{\horstTypeint}
  \horstVar{x}{\horstTypeBVLXIV}
  \horstReturnType{\horstTypeBVLXIV}
  \horstBody{\horstOpAppbvextends{\horstParVarm,\horstParVarn}{\horstVarx}}
\end{horstOperation}
\begin{horstOperation}{cwrap}
  \horstName{cwrap}
  \horstParVar{m}{\horstTypeint}
  \horstParVar{n}{\horstTypeint}
  \horstVar{x}{\horstTypeBVLXIV}
  \horstReturnType{\horstTypeBVLXIV}
  \horstBody{\horstOpAppbvextract{\horstParVarm,\horstParVarn}{\horstVarx}}
\end{horstOperation}
\begin{horstOperation}{signed}
  \horstName{signed}
  \horstParVar{bw}{\horstTypeint}
  \horstVar{v}{\horstTypeValue}
  \horstReturnType{\horstTypeValue}
  \horstBody{\horstMatchExp{{{{x}},\horstMATCHIF{\horstEQ{\horstVarv}{\horstConstructorAppVal{\horstVarx}}},\horstConstructorAppVal{\horstOpAppcsigned{\horstParVarbw}{\horstVarx}}},{{},\horstOTHERWISE,\horstOpAppfreeValOrTop{}{}}}}
\end{horstOperation}
\begin{horstOperation}{unsigned}
  \horstName{unsigned}
  \horstParVar{bw}{\horstTypeint}
  \horstVar{v}{\horstTypeValue}
  \horstReturnType{\horstTypeValue}
  \horstBody{\horstMatchExp{{{{x}},\horstMATCHIF{\horstEQ{\horstVarv}{\horstConstructorAppVal{\horstVarx}}},\horstConstructorAppVal{\horstOpAppcunsigned{\horstParVarbw}{\horstVarx}}},{{},\horstOTHERWISE,\horstOpAppfreeValOrTop{}{}}}}
\end{horstOperation}
\begin{horstOperation}{extendu}
  \horstName{extendu}
  \horstParVar{m}{\horstTypeint}
  \horstParVar{bw}{\horstTypeint}
  \horstVar{i}{\horstTypeValue}
  \horstReturnType{\horstTypeValue}
  \horstBody{\horstMatchExp{{{{x}},\horstMATCHIF{\horstEQ{\horstVari}{\horstConstructorAppVal{\horstVarx}}},\horstConstructorAppVal{\horstOpAppcextendu{\horstParVarm,\horstParVarbw}{\horstVarx}}},{{},\horstOTHERWISE,\horstOpAppfreeValOrTop{}{}}}}
\end{horstOperation}
\begin{horstOperation}{extends}
  \horstName{extends}
  \horstParVar{m}{\horstTypeint}
  \horstParVar{bw}{\horstTypeint}
  \horstVar{i}{\horstTypeValue}
  \horstReturnType{\horstTypeValue}
  \horstBody{\horstMatchExp{{{{x}},\horstMATCHIF{\horstEQ{\horstVari}{\horstConstructorAppVal{\horstVarx}}},\horstConstructorAppVal{\horstOpAppcextends{\horstParVarm,\horstParVarbw}{\horstVarx}}},{{},\horstOTHERWISE,\horstOpAppfreeValOrTop{}{}}}}
\end{horstOperation}
\begin{horstOperation}{wrap}
  \horstName{wrap}
  \horstParVar{m}{\horstTypeint}
  \horstParVar{bw}{\horstTypeint}
  \horstVar{i}{\horstTypeValue}
  \horstReturnType{\horstTypeValue}
  \horstBody{\horstMatchExp{{{{x}},\horstMATCHIF{\horstEQ{\horstVari}{\horstConstructorAppVal{\horstVarx}}},\horstConstructorAppVal{\horstOpAppcwrap{\horstParVarm,\horstParVarbw}{\horstVarx}}},{{},\horstOTHERWISE,\horstOpAppfreeValOrTop{}{}}}}
\end{horstOperation}
\begin{horstOperation}{truncu}
  \horstName{truncu}
  \horstParVar{m}{\horstTypeint}
  \horstParVar{n}{\horstTypeint}
  \horstVar{v}{\horstTypeValue}
  \horstReturnType{\horstTypeValue}
  \horstBody{\horstOpAppfreeValOrTop{}{}}
\end{horstOperation}
\begin{horstOperation}{truncs}
  \horstName{truncs}
  \horstParVar{m}{\horstTypeint}
  \horstParVar{n}{\horstTypeint}
  \horstVar{v}{\horstTypeValue}
  \horstReturnType{\horstTypeValue}
  \horstBody{\horstOpAppfreeValOrTop{}{}}
\end{horstOperation}
\begin{horstOperation}{promote}
  \horstName{promote}
  \horstParVar{m}{\horstTypeint}
  \horstParVar{n}{\horstTypeint}
  \horstVar{v}{\horstTypeValue}
  \horstReturnType{\horstTypeValue}
  \horstBody{\horstOpAppfreeValOrTop{}{}}
\end{horstOperation}
\begin{horstOperation}{demote}
  \horstName{demote}
  \horstParVar{m}{\horstTypeint}
  \horstParVar{n}{\horstTypeint}
  \horstVar{v}{\horstTypeValue}
  \horstReturnType{\horstTypeValue}
  \horstBody{\horstOpAppfreeValOrTop{}{}}
\end{horstOperation}
\begin{horstOperation}{convertu}
  \horstName{convertu}
  \horstParVar{m}{\horstTypeint}
  \horstParVar{n}{\horstTypeint}
  \horstVar{v}{\horstTypeValue}
  \horstReturnType{\horstTypeValue}
  \horstBody{\horstOpAppfreeValOrTop{}{}}
\end{horstOperation}
\begin{horstOperation}{converts}
  \horstName{converts}
  \horstParVar{m}{\horstTypeint}
  \horstParVar{n}{\horstTypeint}
  \horstVar{v}{\horstTypeValue}
  \horstReturnType{\horstTypeValue}
  \horstBody{\horstOpAppfreeValOrTop{}{}}
\end{horstOperation}
\begin{horstOperation}{reinterpret}
  \horstName{reinterpret}
  \horstParVar{tI}{\horstTypeint}
  \horstParVar{tII}{\horstTypeint}
  \horstVar{v}{\horstTypeValue}
  \horstReturnType{\horstTypeValue}
  \horstBody{\horstOpAppfreeValOrTop{}{}}
\end{horstOperation}
\begin{horstOperation}{bvneg}
  \horstName{bvneg}
  \horstParVar{bw}{\horstTypeint}
  \horstVar{iI}{\horstTypeBVLXIV}
  \horstReturnType{\horstTypeBVLXIV}
\end{horstOperation}
\begin{horstOperation}{bvadd}
  \horstName{bvadd}
  \horstParVar{bw}{\horstTypeint}
  \horstVar{iI}{\horstTypeBVLXIV}
  \horstVar{iII}{\horstTypeBVLXIV}
  \horstReturnType{\horstTypeBVLXIV}
\end{horstOperation}
\begin{horstOperation}{bvsub}
  \horstName{bvsub}
  \horstParVar{bw}{\horstTypeint}
  \horstVar{iI}{\horstTypeBVLXIV}
  \horstVar{iII}{\horstTypeBVLXIV}
  \horstReturnType{\horstTypeBVLXIV}
\end{horstOperation}
\begin{horstOperation}{bvmul}
  \horstName{bvmul}
  \horstParVar{bw}{\horstTypeint}
  \horstVar{iI}{\horstTypeBVLXIV}
  \horstVar{iII}{\horstTypeBVLXIV}
  \horstReturnType{\horstTypeBVLXIV}
\end{horstOperation}
\begin{horstOperation}{bvudiv}
  \horstName{bvudiv}
  \horstParVar{bw}{\horstTypeint}
  \horstVar{iI}{\horstTypeBVLXIV}
  \horstVar{iII}{\horstTypeBVLXIV}
  \horstReturnType{\horstTypeBVLXIV}
\end{horstOperation}
\begin{horstOperation}{bvurem}
  \horstName{bvurem}
  \horstParVar{bw}{\horstTypeint}
  \horstVar{iI}{\horstTypeBVLXIV}
  \horstVar{iII}{\horstTypeBVLXIV}
  \horstReturnType{\horstTypeBVLXIV}
\end{horstOperation}
\begin{horstOperation}{bvand}
  \horstName{bvand}
  \horstParVar{bw}{\horstTypeint}
  \horstVar{iI}{\horstTypeBVLXIV}
  \horstVar{iII}{\horstTypeBVLXIV}
  \horstReturnType{\horstTypeBVLXIV}
\end{horstOperation}
\begin{horstOperation}{bvor}
  \horstName{bvor}
  \horstParVar{bw}{\horstTypeint}
  \horstVar{iI}{\horstTypeBVLXIV}
  \horstVar{iII}{\horstTypeBVLXIV}
  \horstReturnType{\horstTypeBVLXIV}
\end{horstOperation}
\begin{horstOperation}{bvxor}
  \horstName{bvxor}
  \horstParVar{bw}{\horstTypeint}
  \horstVar{iI}{\horstTypeBVLXIV}
  \horstVar{iII}{\horstTypeBVLXIV}
  \horstReturnType{\horstTypeBVLXIV}
\end{horstOperation}
\begin{horstOperation}{bvshl}
  \horstName{bvshl}
  \horstParVar{bw}{\horstTypeint}
  \horstVar{iI}{\horstTypeBVLXIV}
  \horstVar{iII}{\horstTypeBVLXIV}
  \horstReturnType{\horstTypeBVLXIV}
\end{horstOperation}
\begin{horstOperation}{bvlshr}
  \horstName{bvlshr}
  \horstParVar{bw}{\horstTypeint}
  \horstVar{iI}{\horstTypeBVLXIV}
  \horstVar{iII}{\horstTypeBVLXIV}
  \horstReturnType{\horstTypeBVLXIV}
\end{horstOperation}
\begin{horstOperation}{bvashr}
  \horstName{bvashr}
  \horstParVar{bw}{\horstTypeint}
  \horstVar{iI}{\horstTypeBVLXIV}
  \horstVar{iII}{\horstTypeBVLXIV}
  \horstReturnType{\horstTypeBVLXIV}
\end{horstOperation}
\begin{horstOperation}{bvult}
  \horstName{bvult}
  \horstParVar{bw}{\horstTypeint}
  \horstVar{iI}{\horstTypeBVLXIV}
  \horstVar{iII}{\horstTypeBVLXIV}
  \horstReturnType{\horstTypebool}
\end{horstOperation}
\begin{horstOperation}{bvugt}
  \horstName{bvugt}
  \horstParVar{bw}{\horstTypeint}
  \horstVar{iI}{\horstTypeBVLXIV}
  \horstVar{iII}{\horstTypeBVLXIV}
  \horstReturnType{\horstTypebool}
\end{horstOperation}
\begin{horstOperation}{bvule}
  \horstName{bvule}
  \horstParVar{bw}{\horstTypeint}
  \horstVar{iI}{\horstTypeBVLXIV}
  \horstVar{iII}{\horstTypeBVLXIV}
  \horstReturnType{\horstTypebool}
\end{horstOperation}
\begin{horstOperation}{bvuge}
  \horstName{bvuge}
  \horstParVar{bw}{\horstTypeint}
  \horstVar{iI}{\horstTypeBVLXIV}
  \horstVar{iII}{\horstTypeBVLXIV}
  \horstReturnType{\horstTypebool}
\end{horstOperation}
\begin{horstOperation}{bvslt}
  \horstName{bvslt}
  \horstParVar{bw}{\horstTypeint}
  \horstVar{iI}{\horstTypeBVLXIV}
  \horstVar{iII}{\horstTypeBVLXIV}
  \horstReturnType{\horstTypebool}
\end{horstOperation}
\begin{horstOperation}{bvsgt}
  \horstName{bvsgt}
  \horstParVar{bw}{\horstTypeint}
  \horstVar{iI}{\horstTypeBVLXIV}
  \horstVar{iII}{\horstTypeBVLXIV}
  \horstReturnType{\horstTypebool}
\end{horstOperation}
\begin{horstOperation}{bvsle}
  \horstName{bvsle}
  \horstParVar{bw}{\horstTypeint}
  \horstVar{iI}{\horstTypeBVLXIV}
  \horstVar{iII}{\horstTypeBVLXIV}
  \horstReturnType{\horstTypebool}
\end{horstOperation}
\begin{horstOperation}{bvsge}
  \horstName{bvsge}
  \horstParVar{bw}{\horstTypeint}
  \horstVar{iI}{\horstTypeBVLXIV}
  \horstVar{iII}{\horstTypeBVLXIV}
  \horstReturnType{\horstTypebool}
\end{horstOperation}
\begin{horstOperation}{ciadd}
  \horstName{ciadd}
  \horstParVar{bw}{\horstTypeint}
  \horstVar{x}{\horstTypeBVLXIV}
  \horstVar{y}{\horstTypeBVLXIV}
  \horstReturnType{\horstTypeBVLXIV}
  \horstBody{\horstOpAppbvadd{\horstParVarbw}{\horstVarx,\horstVary}}
\end{horstOperation}
\begin{horstOperation}{cisub}
  \horstName{cisub}
  \horstParVar{bw}{\horstTypeint}
  \horstVar{x}{\horstTypeBVLXIV}
  \horstVar{y}{\horstTypeBVLXIV}
  \horstReturnType{\horstTypeBVLXIV}
  \horstBody{\horstOpAppbvsub{\horstParVarbw}{\horstVarx,\horstVary}}
\end{horstOperation}
\begin{horstOperation}{cimul}
  \horstName{cimul}
  \horstParVar{bw}{\horstTypeint}
  \horstVar{x}{\horstTypeBVLXIV}
  \horstVar{y}{\horstTypeBVLXIV}
  \horstReturnType{\horstTypeBVLXIV}
  \horstBody{\horstOpAppbvmul{\horstParVarbw}{\horstVarx,\horstVary}}
\end{horstOperation}
\begin{horstOperation}{cidivs}
  \horstName{cidivs}
  \horstParVar{bw}{\horstTypeint}
  \horstVar{s}{\horstTypeBVLXIV}
  \horstVar{t}{\horstTypeBVLXIV}
  \horstReturnType{\horstTypeBVLXIV}
  \horstBody{\horstMatchExp{{{},\horstMATCHIF{\horstAND{\horstEQ{\horstOpAppbvslt{\horstParVarbw}{\horstVars,\horstOpAppcintIIbv{\horstParVarbw,0}{}}}{\horstFalse}}{\horstEQ{\horstOpAppbvslt{\horstParVarbw}{\horstVart,\horstOpAppcintIIbv{\horstParVarbw,0}{}}}{\horstFalse}}},\horstOpAppbvudiv{\horstParVarbw}{\horstVars,\horstVart}},{{},\horstMATCHIF{\horstAND{\horstEQ{\horstOpAppbvslt{\horstParVarbw}{\horstVars,\horstOpAppcintIIbv{\horstParVarbw,0}{}}}{\horstTrue}}{\horstEQ{\horstOpAppbvslt{\horstParVarbw}{\horstVart,\horstOpAppcintIIbv{\horstParVarbw,0}{}}}{\horstFalse}}},\horstOpAppbvneg{\horstParVarbw}{\horstOpAppbvudiv{\horstParVarbw}{\horstOpAppbvneg{\horstParVarbw}{\horstVars},\horstVart}}},{{},\horstMATCHIF{\horstAND{\horstEQ{\horstOpAppbvslt{\horstParVarbw}{\horstVars,\horstOpAppcintIIbv{\horstParVarbw,0}{}}}{\horstFalse}}{\horstEQ{\horstOpAppbvslt{\horstParVarbw}{\horstVart,\horstOpAppcintIIbv{\horstParVarbw,0}{}}}{\horstTrue}}},\horstOpAppbvneg{\horstParVarbw}{\horstOpAppbvudiv{\horstParVarbw}{\horstVars,\horstOpAppbvneg{\horstParVarbw}{\horstVart}}}},{{},\horstOTHERWISE,\horstOpAppbvudiv{\horstParVarbw}{\horstOpAppbvneg{\horstParVarbw}{\horstVars},\horstOpAppbvneg{\horstParVarbw}{\horstVart}}}}}
\end{horstOperation}
\begin{horstOperation}{cidivu}
  \horstName{cidivu}
  \horstParVar{bw}{\horstTypeint}
  \horstVar{x}{\horstTypeBVLXIV}
  \horstVar{y}{\horstTypeBVLXIV}
  \horstReturnType{\horstTypeBVLXIV}
  \horstBody{\horstOpAppbvudiv{\horstParVarbw}{\horstVarx,\horstVary}}
\end{horstOperation}
\begin{horstOperation}{ciremu}
  \horstName{ciremu}
  \horstParVar{bw}{\horstTypeint}
  \horstVar{x}{\horstTypeBVLXIV}
  \horstVar{y}{\horstTypeBVLXIV}
  \horstReturnType{\horstTypeBVLXIV}
  \horstBody{\horstOpAppbvurem{\horstParVarbw}{\horstVarx,\horstVary}}
\end{horstOperation}
\begin{horstOperation}{cirems}
  \horstName{cirems}
  \horstParVar{bw}{\horstTypeint}
  \horstVar{s}{\horstTypeBVLXIV}
  \horstVar{t}{\horstTypeBVLXIV}
  \horstReturnType{\horstTypeBVLXIV}
  \horstBody{\horstMatchExp{{{},\horstMATCHIF{\horstAND{\horstEQ{\horstOpAppbvslt{\horstParVarbw}{\horstVars,\horstOpAppcintIIbv{\horstParVarbw,0}{}}}{\horstFalse}}{\horstEQ{\horstOpAppbvslt{\horstParVarbw}{\horstVart,\horstOpAppcintIIbv{\horstParVarbw,0}{}}}{\horstFalse}}},\horstOpAppbvurem{\horstParVarbw}{\horstVars,\horstVart}},{{},\horstMATCHIF{\horstAND{\horstEQ{\horstOpAppbvslt{\horstParVarbw}{\horstVars,\horstOpAppcintIIbv{\horstParVarbw,0}{}}}{\horstTrue}}{\horstEQ{\horstOpAppbvslt{\horstParVarbw}{\horstVart,\horstOpAppcintIIbv{\horstParVarbw,0}{}}}{\horstFalse}}},\horstOpAppbvneg{\horstParVarbw}{\horstOpAppbvurem{\horstParVarbw}{\horstOpAppbvneg{\horstParVarbw}{\horstVars},\horstVart}}},{{},\horstMATCHIF{\horstAND{\horstEQ{\horstOpAppbvslt{\horstParVarbw}{\horstVars,\horstOpAppcintIIbv{\horstParVarbw,0}{}}}{\horstFalse}}{\horstEQ{\horstOpAppbvslt{\horstParVarbw}{\horstVart,\horstOpAppcintIIbv{\horstParVarbw,0}{}}}{\horstTrue}}},\horstOpAppbvurem{\horstParVarbw}{\horstVars,\horstOpAppbvneg{\horstParVarbw}{\horstVart}}},{{},\horstOTHERWISE,\horstOpAppbvneg{\horstParVarbw}{\horstOpAppbvurem{\horstParVarbw}{\horstOpAppbvneg{\horstParVarbw}{\horstVars},\horstOpAppbvneg{\horstParVarbw}{\horstVart}}}}}}
\end{horstOperation}
\begin{horstOperation}{ciand}
  \horstName{ciand}
  \horstParVar{bw}{\horstTypeint}
  \horstVar{x}{\horstTypeBVLXIV}
  \horstVar{y}{\horstTypeBVLXIV}
  \horstReturnType{\horstTypeBVLXIV}
  \horstBody{\horstOpAppbvand{\horstParVarbw}{\horstVarx,\horstVary}}
\end{horstOperation}
\begin{horstOperation}{cior}
  \horstName{cior}
  \horstParVar{bw}{\horstTypeint}
  \horstVar{x}{\horstTypeBVLXIV}
  \horstVar{y}{\horstTypeBVLXIV}
  \horstReturnType{\horstTypeBVLXIV}
  \horstBody{\horstOpAppbvor{\horstParVarbw}{\horstVarx,\horstVary}}
\end{horstOperation}
\begin{horstOperation}{cixor}
  \horstName{cixor}
  \horstParVar{bw}{\horstTypeint}
  \horstVar{x}{\horstTypeBVLXIV}
  \horstVar{y}{\horstTypeBVLXIV}
  \horstReturnType{\horstTypeBVLXIV}
  \horstBody{\horstOpAppbvxor{\horstParVarbw}{\horstVarx,\horstVary}}
\end{horstOperation}
\begin{horstOperation}{cishl}
  \horstName{cishl}
  \horstParVar{bw}{\horstTypeint}
  \horstVar{x}{\horstTypeBVLXIV}
  \horstVar{s}{\horstTypeBVLXIV}
  \horstReturnType{\horstTypeBVLXIV}
  \horstBody{\horstOpAppbvshl{\horstParVarbw}{\horstVarx,\horstOpAppbvurem{\horstParVarbw}{\horstVars,\horstOpAppcintIIbv{\horstParVarbw,\horstParVarbw}{}}}}
\end{horstOperation}
\begin{horstOperation}{cishru}
  \horstName{cishru}
  \horstParVar{bw}{\horstTypeint}
  \horstVar{x}{\horstTypeBVLXIV}
  \horstVar{s}{\horstTypeBVLXIV}
  \horstReturnType{\horstTypeBVLXIV}
  \horstBody{\horstOpAppbvlshr{\horstParVarbw}{\horstVarx,\horstOpAppbvurem{\horstParVarbw}{\horstVars,\horstOpAppcintIIbv{\horstParVarbw,\horstParVarbw}{}}}}
\end{horstOperation}
\begin{horstOperation}{cilshr}
  \horstName{cilshr}
  \horstParVar{bw}{\horstTypeint}
  \horstVar{x}{\horstTypeBVLXIV}
  \horstVar{s}{\horstTypeBVLXIV}
  \horstReturnType{\horstTypeBVLXIV}
  \horstBody{\horstOpAppbvlshr{\horstParVarbw}{\horstVarx,\horstVars}}
\end{horstOperation}
\begin{horstOperation}{cishrs}
  \horstName{cishrs}
  \horstParVar{bw}{\horstTypeint}
  \horstVar{x}{\horstTypeBVLXIV}
  \horstVar{s}{\horstTypeBVLXIV}
  \horstReturnType{\horstTypeBVLXIV}
  \horstBody{\horstOpAppbvashr{\horstParVarbw}{\horstVarx,\horstOpAppbvurem{\horstParVarbw}{\horstVars,\horstOpAppcintIIbv{\horstParVarbw,\horstParVarbw}{}}}}
\end{horstOperation}
\begin{horstOperation}{cirotl}
  \horstName{cirotl}
  \horstParVar{bw}{\horstTypeint}
  \horstVar{x}{\horstTypeBVLXIV}
  \horstVar{r}{\horstTypeBVLXIV}
  \horstReturnType{\horstTypeBVLXIV}
  \horstBody{\horstOpAppbvor{\horstParVarbw}{\horstOpAppbvshl{\horstParVarbw}{\horstVarx,\horstOpAppbvurem{\horstParVarbw}{\horstVarr,\horstOpAppcintIIbv{\horstParVarbw,\horstParVarbw}{}}},\horstOpAppbvlshr{\horstParVarbw}{\horstVarx,\horstOpAppbvsub{\horstParVarbw}{\horstOpAppcintIIbv{\horstParVarbw,\horstParVarbw}{},\horstOpAppbvurem{\horstParVarbw}{\horstVarr,\horstOpAppcintIIbv{\horstParVarbw,\horstParVarbw}{}}}}}}
\end{horstOperation}
\begin{horstOperation}{cirotr}
  \horstName{cirotr}
  \horstParVar{bw}{\horstTypeint}
  \horstVar{x}{\horstTypeBVLXIV}
  \horstVar{r}{\horstTypeBVLXIV}
  \horstReturnType{\horstTypeBVLXIV}
  \horstBody{\horstOpAppbvor{\horstParVarbw}{\horstOpAppbvlshr{\horstParVarbw}{\horstVarx,\horstOpAppbvurem{\horstParVarbw}{\horstVarr,\horstOpAppcintIIbv{\horstParVarbw,\horstParVarbw}{}}},\horstOpAppbvshl{\horstParVarbw}{\horstVarx,\horstOpAppbvsub{\horstParVarbw}{\horstOpAppcintIIbv{\horstParVarbw,\horstParVarbw}{},\horstOpAppbvurem{\horstParVarbw}{\horstVarr,\horstOpAppcintIIbv{\horstParVarbw,\horstParVarbw}{}}}}}}
\end{horstOperation}
\begin{horstOperation}{cieq}
  \horstName{cieq}
  \horstParVar{bw}{\horstTypeint}
  \horstVar{iI}{\horstTypeBVLXIV}
  \horstVar{iII}{\horstTypeBVLXIV}
  \horstReturnType{\horstTypebool}
  \horstBody{\horstEQ{\horstVariI}{\horstVariII}}
\end{horstOperation}
\begin{horstOperation}{cine}
  \horstName{cine}
  \horstParVar{bw}{\horstTypeint}
  \horstVar{x}{\horstTypeBVLXIV}
  \horstVar{y}{\horstTypeBVLXIV}
  \horstReturnType{\horstTypebool}
  \horstBody{\horstNEQ{\horstVarx}{\horstVary}}
\end{horstOperation}
\begin{horstOperation}{ciltu}
  \horstName{ciltu}
  \horstParVar{bw}{\horstTypeint}
  \horstVar{x}{\horstTypeBVLXIV}
  \horstVar{y}{\horstTypeBVLXIV}
  \horstReturnType{\horstTypebool}
  \horstBody{\horstOpAppbvult{\horstParVarbw}{\horstVarx,\horstVary}}
\end{horstOperation}
\begin{horstOperation}{cilts}
  \horstName{cilts}
  \horstParVar{bw}{\horstTypeint}
  \horstVar{x}{\horstTypeBVLXIV}
  \horstVar{y}{\horstTypeBVLXIV}
  \horstReturnType{\horstTypebool}
  \horstBody{\horstOpAppbvslt{\horstParVarbw}{\horstVarx,\horstVary}}
\end{horstOperation}
\begin{horstOperation}{cigtu}
  \horstName{cigtu}
  \horstParVar{bw}{\horstTypeint}
  \horstVar{x}{\horstTypeBVLXIV}
  \horstVar{y}{\horstTypeBVLXIV}
  \horstReturnType{\horstTypebool}
  \horstBody{\horstOpAppbvugt{\horstParVarbw}{\horstVarx,\horstVary}}
\end{horstOperation}
\begin{horstOperation}{cigts}
  \horstName{cigts}
  \horstParVar{bw}{\horstTypeint}
  \horstVar{x}{\horstTypeBVLXIV}
  \horstVar{y}{\horstTypeBVLXIV}
  \horstReturnType{\horstTypebool}
  \horstBody{\horstOpAppbvsgt{\horstParVarbw}{\horstVarx,\horstVary}}
\end{horstOperation}
\begin{horstOperation}{cileu}
  \horstName{cileu}
  \horstParVar{bw}{\horstTypeint}
  \horstVar{x}{\horstTypeBVLXIV}
  \horstVar{y}{\horstTypeBVLXIV}
  \horstReturnType{\horstTypebool}
  \horstBody{\horstOpAppbvule{\horstParVarbw}{\horstVarx,\horstVary}}
\end{horstOperation}
\begin{horstOperation}{ciles}
  \horstName{ciles}
  \horstParVar{bw}{\horstTypeint}
  \horstVar{x}{\horstTypeBVLXIV}
  \horstVar{y}{\horstTypeBVLXIV}
  \horstReturnType{\horstTypebool}
  \horstBody{\horstOpAppbvsle{\horstParVarbw}{\horstVarx,\horstVary}}
\end{horstOperation}
\begin{horstOperation}{cigeu}
  \horstName{cigeu}
  \horstParVar{bw}{\horstTypeint}
  \horstVar{x}{\horstTypeBVLXIV}
  \horstVar{y}{\horstTypeBVLXIV}
  \horstReturnType{\horstTypebool}
  \horstBody{\horstOpAppbvuge{\horstParVarbw}{\horstVarx,\horstVary}}
\end{horstOperation}
\begin{horstOperation}{ciges}
  \horstName{ciges}
  \horstParVar{bw}{\horstTypeint}
  \horstVar{x}{\horstTypeBVLXIV}
  \horstVar{y}{\horstTypeBVLXIV}
  \horstReturnType{\horstTypebool}
  \horstBody{\horstOpAppbvsge{\horstParVarbw}{\horstVarx,\horstVary}}
\end{horstOperation}
\begin{horstOperation}{ciclz}
  \horstName{ciclz}
  \horstParVar{bw}{\horstTypeint}
  \horstVar{x}{\horstTypeBVLXIV}
  \horstReturnType{\horstTypeBVLXIV}
  \horstBody{\horstCustomSumExp{\horstSelectorFunctionAppinterval{\horstParVarp}{0,\horstParVarbw}}{acc}{\horstCOND{\horstOpAppbvult{\horstParVarbw}{\horstVarx,\horstOpAppcintIIbv{\horstParVarbw,\horstOpApppow{\horstSUB{\horstSUB{\horstParVarbw}{1}}{\horstParVarp}}{2}}{}}}{\horstOpAppcintIIbv{\horstParVarbw,\horstADD{\horstParVarp}{1}}{}}{\horstVaracc}}{\horstOpAppcintIIbv{\horstParVarbw,0}{}}{{p}}}
\end{horstOperation}
\begin{horstOperation}{cictz}
  \horstName{cictz}
  \horstParVar{bw}{\horstTypeint}
  \horstVar{x}{\horstTypeBVLXIV}
  \horstReturnType{\horstTypeBVLXIV}
  \horstBody{\horstCustomSumExp{\horstSelectorFunctionAppinterval{\horstParVarp}{0,\horstParVarbw}}{res}{\horstCOND{\horstNEQ{\horstOpAppbvurem{\horstParVarbw}{\horstVarx,\horstOpAppcintIIbv{\horstParVarbw,\horstOpApppow{\horstSUB{\horstParVarbw}{\horstParVarp}}{2}}{}}}{\horstOpAppcintIIbv{\horstParVarbw,0}{}}}{\horstOpAppcintIIbv{\horstParVarbw,\horstSUB{\horstSUB{\horstParVarbw}{\horstParVarp}}{1}}{}}{\horstVarres}}{\horstOpAppcintIIbv{\horstParVarbw,\horstParVarbw}{}}{{p}}}
\end{horstOperation}
\begin{horstOperation}{cipopcnt}
  \horstName{cipopcnt}
  \horstParVar{bw}{\horstTypeint}
  \horstVar{x}{\horstTypeBVLXIV}
  \horstReturnType{\horstTypeBVLXIV}
  \horstBody{\horstCustomSumExp{\horstSelectorFunctionAppinterval{\horstParVarp}{0,\horstParVarbw}}{y}{\horstOpAppbvadd{\horstParVarbw}{\horstCOND{\horstEQ{\horstOpAppbvand{\horstParVarbw}{\horstOpAppbvlshr{\horstParVarbw}{\horstVarx,\horstOpAppcintIIbv{\horstParVarbw,\horstParVarp}{}},\horstOpAppcintIIbv{\horstParVarbw,1}{}}}{\horstOpAppcintIIbv{\horstParVarbw,1}{}}}{\horstOpAppcintIIbv{\horstParVarbw,1}{}}{\horstOpAppcintIIbv{\horstParVarbw,0}{}},\horstVary}}{\horstOpAppcintIIbv{\horstParVarbw,0}{}}{{p}}}
\end{horstOperation}
\begin{horstOperation}{cieqz}
  \horstName{cieqz}
  \horstParVar{bw}{\horstTypeint}
  \horstVar{x}{\horstTypeBVLXIV}
  \horstReturnType{\horstTypebool}
  \horstBody{\horstEQ{\horstVarx}{\horstOpAppcintIIbv{64,0}{}}}
\end{horstOperation}
\begin{horstOperation}{cinot}
  \horstName{cinot}
  \horstParVar{bw}{\horstTypeint}
  \horstVar{x}{\horstTypeBVLXIV}
  \horstReturnType{\horstTypeBVLXIV}
  \horstBody{\horstOpAppbvsub{64}{\horstOpAppbvneg{64}{\horstVarx},\horstOpAppcintIIbv{64,1}{}}}
\end{horstOperation}
\begin{horstOperation}{abseq}
  \horstName{abseq}
  \horstVar{a}{\horstTypeValue}
  \horstVar{b}{\horstTypeValue}
  \horstReturnType{\horstTypebool}
  \horstBody{\horstMatchExp{{{{x}{y}},\horstMATCHIF{\horstAND{\horstEQ{\horstVara}{\horstConstructorAppVal{\horstVarx}}}{\horstEQ{\horstVarb}{\horstConstructorAppVal{\horstVary}}}},\horstEQ{\horstVarx}{\horstVary}},{{},\horstOTHERWISE,\horstTrue}}}
\end{horstOperation}
\begin{horstOperation}{absneq}
  \horstName{absneq}
  \horstVar{a}{\horstTypeValue}
  \horstVar{b}{\horstTypeValue}
  \horstReturnType{\horstTypebool}
  \horstBody{\horstMatchExp{{{{x}{y}},\horstMATCHIF{\horstAND{\horstEQ{\horstVara}{\horstConstructorAppVal{\horstVarx}}}{\horstEQ{\horstVarb}{\horstConstructorAppVal{\horstVary}}}},\horstNEQ{\horstVarx}{\horstVary}},{{},\horstOTHERWISE,\horstTrue}}}
\end{horstOperation}
\begin{horstOperation}{absgt}
  \horstName{absgt}
  \horstVar{a}{\horstTypeValue}
  \horstVar{b}{\horstTypeValue}
  \horstReturnType{\horstTypebool}
  \horstBody{\horstMatchExp{{{{x}{y}},\horstMATCHIF{\horstAND{\horstEQ{\horstVara}{\horstConstructorAppVal{\horstVarx}}}{\horstEQ{\horstVarb}{\horstConstructorAppVal{\horstVary}}}},\horstOpAppbvugt{32}{\horstVarx,\horstVary}},{{},\horstOTHERWISE,\horstTrue}}}
\end{horstOperation}
\begin{horstOperation}{absge}
  \horstName{absge}
  \horstVar{a}{\horstTypeValue}
  \horstVar{b}{\horstTypeValue}
  \horstReturnType{\horstTypebool}
  \horstBody{\horstMatchExp{{{{x}{y}},\horstMATCHIF{\horstAND{\horstEQ{\horstVara}{\horstConstructorAppVal{\horstVarx}}}{\horstEQ{\horstVarb}{\horstConstructorAppVal{\horstVary}}}},\horstOpAppbvuge{32}{\horstVarx,\horstVary}},{{},\horstOTHERWISE,\horstTrue}}}
\end{horstOperation}
\begin{horstOperation}{abslt}
  \horstName{abslt}
  \horstVar{a}{\horstTypeValue}
  \horstVar{b}{\horstTypeValue}
  \horstReturnType{\horstTypebool}
  \horstBody{\horstMatchExp{{{{x}{y}},\horstMATCHIF{\horstAND{\horstEQ{\horstVara}{\horstConstructorAppVal{\horstVarx}}}{\horstEQ{\horstVarb}{\horstConstructorAppVal{\horstVary}}}},\horstOpAppbvult{32}{\horstVarx,\horstVary}},{{},\horstOTHERWISE,\horstTrue}}}
\end{horstOperation}
\begin{horstOperation}{absle}
  \horstName{absle}
  \horstVar{a}{\horstTypeValue}
  \horstVar{b}{\horstTypeValue}
  \horstReturnType{\horstTypebool}
  \horstBody{\horstMatchExp{{{{x}{y}},\horstMATCHIF{\horstAND{\horstEQ{\horstVara}{\horstConstructorAppVal{\horstVarx}}}{\horstEQ{\horstVarb}{\horstConstructorAppVal{\horstVary}}}},\horstOpAppbvule{32}{\horstVarx,\horstVary}},{{},\horstOTHERWISE,\horstTrue}}}
\end{horstOperation}
\begin{horstOperation}{iadd}
  \horstName{iadd}
  \horstParVar{bw}{\horstTypeint}
  \horstVar{iI}{\horstTypeValue}
  \horstVar{iII}{\horstTypeValue}
  \horstReturnType{\horstTypeValue}
  \horstBody{\horstMatchExp{{{{x}{y}},\horstMATCHIF{\horstAND{\horstEQ{\horstVariI}{\horstConstructorAppVal{\horstVarx}}}{\horstEQ{\horstVariII}{\horstConstructorAppVal{\horstVary}}}},\horstConstructorAppVal{\horstOpAppciadd{\horstParVarbw}{\horstVarx,\horstVary}}},{{},\horstOTHERWISE,\horstOpAppfreeValOrTop{}{}}}}
\end{horstOperation}
\begin{horstOperation}{isub}
  \horstName{isub}
  \horstParVar{bw}{\horstTypeint}
  \horstVar{iI}{\horstTypeValue}
  \horstVar{iII}{\horstTypeValue}
  \horstReturnType{\horstTypeValue}
  \horstBody{\horstMatchExp{{{{x}{y}},\horstMATCHIF{\horstAND{\horstEQ{\horstVariI}{\horstConstructorAppVal{\horstVarx}}}{\horstEQ{\horstVariII}{\horstConstructorAppVal{\horstVary}}}},\horstConstructorAppVal{\horstOpAppcisub{\horstParVarbw}{\horstVarx,\horstVary}}},{{},\horstOTHERWISE,\horstOpAppfreeValOrTop{}{}}}}
\end{horstOperation}
\begin{horstOperation}{imul}
  \horstName{imul}
  \horstParVar{bw}{\horstTypeint}
  \horstVar{iI}{\horstTypeValue}
  \horstVar{iII}{\horstTypeValue}
  \horstReturnType{\horstTypeValue}
  \horstBody{\horstMatchExp{{{{x}{y}},\horstMATCHIF{\horstAND{\horstEQ{\horstVariI}{\horstConstructorAppVal{\horstVarx}}}{\horstEQ{\horstVariII}{\horstConstructorAppVal{\horstVary}}}},\horstConstructorAppVal{\horstOpAppcimul{\horstParVarbw}{\horstVarx,\horstVary}}},{{},\horstOTHERWISE,\horstOpAppfreeValOrTop{}{}}}}
\end{horstOperation}
\begin{horstOperation}{idivuaux}
  \horstName{idivuaux}
  \horstParVar{bw}{\horstTypeint}
  \horstVar{iI}{\horstTypeValue}
  \horstVar{iII}{\horstTypeValue}
  \horstReturnType{\horstTypeValue}
  \horstBody{\horstMatchExp{{{{x}{y}},\horstMATCHIF{\horstAND{\horstEQ{\horstVariI}{\horstConstructorAppVal{\horstVarx}}}{\horstEQ{\horstVariII}{\horstConstructorAppVal{\horstVary}}}},\horstConstructorAppVal{\horstOpAppcidivu{\horstParVarbw}{\horstVarx,\horstVary}}},{{},\horstOTHERWISE,\horstOpAppfreeValOrTop{}{}}}}
\end{horstOperation}
\begin{horstOperation}{idivu}
  \horstName{idivu}
  \horstParVar{bw}{\horstTypeint}
  \horstVar{iI}{\horstTypeValue}
  \horstVar{iII}{\horstTypeValue}
  \horstReturnType{\horstTypeMaybeValue}
  \horstBody{\horstCOND{\horstOpAppabseq{}{\horstVariII,\horstOpAppmkConst{0}{}}}{\horstConstructorAppNothingV{}}{\horstConstructorAppJustV{\horstOpAppidivuaux{\horstParVarbw}{\horstVariI,\horstVariII}}}}
\end{horstOperation}
\begin{horstOperation}{idivs}
  \horstName{idivs}
  \horstParVar{bw}{\horstTypeint}
  \horstVar{iI}{\horstTypeValue}
  \horstVar{iII}{\horstTypeValue}
  \horstReturnType{\horstTypeMaybeValue}
  \horstBody{\horstCOND{\horstOpAppabseq{}{\horstVariII,\horstOpAppmkConst{0}{}}}{\horstConstructorAppNothingV{}}{\horstMatchExp{{{{x}{y}},\horstMATCHIF{\horstAND{\horstEQ{\horstVariI}{\horstConstructorAppVal{\horstVarx}}}{\horstEQ{\horstVariII}{\horstConstructorAppVal{\horstVary}}}},\horstCOND{\horstAND{\horstEQ{\horstVarx}{\horstOpAppcinj{\horstOpApppow{\horstSUB{\horstParVarbw}{1}}{2}}{}}}{\horstEQ{\horstVary}{\horstOpAppcinj{\horstSUB{\horstOpApppow{\horstParVarbw}{2}}{1}}{}}}}{\horstConstructorAppNothingV{}}{\horstConstructorAppJustV{\horstConstructorAppVal{\horstOpAppcidivs{\horstParVarbw}{\horstVarx,\horstVary}}}}},{{},\horstOTHERWISE,\horstConstructorAppJustV{\horstOpAppfreeValOrTop{}{}}}}}}
\end{horstOperation}
\begin{horstOperation}{iremu}
  \horstName{iremu}
  \horstParVar{bw}{\horstTypeint}
  \horstVar{iI}{\horstTypeValue}
  \horstVar{iII}{\horstTypeValue}
  \horstReturnType{\horstTypeMaybeValue}
  \horstBody{\horstCOND{\horstOpAppabseq{}{\horstVariII,\horstOpAppmkConst{0}{}}}{\horstConstructorAppNothingV{}}{\horstMatchExp{{{{x}{y}},\horstMATCHIF{\horstAND{\horstEQ{\horstVariI}{\horstConstructorAppVal{\horstVarx}}}{\horstEQ{\horstVariII}{\horstConstructorAppVal{\horstVary}}}},\horstConstructorAppJustV{\horstConstructorAppVal{\horstOpAppciremu{\horstParVarbw}{\horstVarx,\horstVary}}}},{{},\horstOTHERWISE,\horstConstructorAppJustV{\horstOpAppfreeValOrTop{}{}}}}}}
\end{horstOperation}
\begin{horstOperation}{irems}
  \horstName{irems}
  \horstParVar{bw}{\horstTypeint}
  \horstVar{iI}{\horstTypeValue}
  \horstVar{iII}{\horstTypeValue}
  \horstReturnType{\horstTypeMaybeValue}
  \horstBody{\horstCOND{\horstOpAppabseq{}{\horstVariII,\horstOpAppmkConst{0}{}}}{\horstConstructorAppNothingV{}}{\horstMatchExp{{{{jI}{jII}},\horstMATCHIF{\horstAND{\horstEQ{\horstVariI}{\horstConstructorAppVal{\horstVarjI}}}{\horstEQ{\horstVariII}{\horstConstructorAppVal{\horstVarjII}}}},\horstConstructorAppJustV{\horstConstructorAppVal{\horstOpAppcirems{\horstParVarbw}{\horstVarjI,\horstVarjII}}}},{{},\horstOTHERWISE,\horstConstructorAppJustV{\horstOpAppfreeValOrTop{}{}}}}}}
\end{horstOperation}
\begin{horstOperation}{iand}
  \horstName{iand}
  \horstParVar{bw}{\horstTypeint}
  \horstVar{iI}{\horstTypeValue}
  \horstVar{iII}{\horstTypeValue}
  \horstReturnType{\horstTypeValue}
  \horstBody{\horstMatchExp{{{{x}{y}},\horstMATCHIF{\horstAND{\horstEQ{\horstVariI}{\horstConstructorAppVal{\horstVarx}}}{\horstEQ{\horstVariII}{\horstConstructorAppVal{\horstVary}}}},\horstConstructorAppVal{\horstOpAppciand{\horstParVarbw}{\horstVarx,\horstVary}}},{{},\horstOTHERWISE,\horstOpAppfreeValOrTop{}{}}}}
\end{horstOperation}
\begin{horstOperation}{ior}
  \horstName{ior}
  \horstParVar{bw}{\horstTypeint}
  \horstVar{iI}{\horstTypeValue}
  \horstVar{iII}{\horstTypeValue}
  \horstReturnType{\horstTypeValue}
  \horstBody{\horstMatchExp{{{{x}{y}},\horstMATCHIF{\horstAND{\horstEQ{\horstVariI}{\horstConstructorAppVal{\horstVarx}}}{\horstEQ{\horstVariII}{\horstConstructorAppVal{\horstVary}}}},\horstConstructorAppVal{\horstOpAppcior{\horstParVarbw}{\horstVarx,\horstVary}}},{{},\horstOTHERWISE,\horstOpAppfreeValOrTop{}{}}}}
\end{horstOperation}
\begin{horstOperation}{ixor}
  \horstName{ixor}
  \horstParVar{bw}{\horstTypeint}
  \horstVar{iI}{\horstTypeValue}
  \horstVar{iII}{\horstTypeValue}
  \horstReturnType{\horstTypeValue}
  \horstBody{\horstMatchExp{{{{x}{y}},\horstMATCHIF{\horstAND{\horstEQ{\horstVariI}{\horstConstructorAppVal{\horstVarx}}}{\horstEQ{\horstVariII}{\horstConstructorAppVal{\horstVary}}}},\horstConstructorAppVal{\horstOpAppcixor{\horstParVarbw}{\horstVarx,\horstVary}}},{{},\horstOTHERWISE,\horstOpAppfreeValOrTop{}{}}}}
\end{horstOperation}
\begin{horstOperation}{ishl}
  \horstName{ishl}
  \horstParVar{bw}{\horstTypeint}
  \horstVar{iI}{\horstTypeValue}
  \horstVar{iII}{\horstTypeValue}
  \horstReturnType{\horstTypeValue}
  \horstBody{\horstMatchExp{{{{x}{y}},\horstMATCHIF{\horstAND{\horstEQ{\horstVariI}{\horstConstructorAppVal{\horstVarx}}}{\horstEQ{\horstVariII}{\horstConstructorAppVal{\horstVary}}}},\horstConstructorAppVal{\horstOpAppcishl{\horstParVarbw}{\horstVarx,\horstVary}}},{{},\horstOTHERWISE,\horstOpAppfreeValOrTop{}{}}}}
\end{horstOperation}
\begin{horstOperation}{ishru}
  \horstName{ishru}
  \horstParVar{bw}{\horstTypeint}
  \horstVar{iI}{\horstTypeValue}
  \horstVar{iII}{\horstTypeValue}
  \horstReturnType{\horstTypeValue}
  \horstBody{\horstMatchExp{{{{x}{y}},\horstMATCHIF{\horstAND{\horstEQ{\horstVariI}{\horstConstructorAppVal{\horstVarx}}}{\horstEQ{\horstVariII}{\horstConstructorAppVal{\horstVary}}}},\horstConstructorAppVal{\horstOpAppcishru{\horstParVarbw}{\horstVarx,\horstVary}}},{{},\horstOTHERWISE,\horstOpAppfreeValOrTop{}{}}}}
\end{horstOperation}
\begin{horstOperation}{ilshr}
  \horstName{ilshr}
  \horstParVar{bw}{\horstTypeint}
  \horstVar{iI}{\horstTypeValue}
  \horstVar{iII}{\horstTypeValue}
  \horstReturnType{\horstTypeValue}
  \horstBody{\horstMatchExp{{{{x}{y}},\horstMATCHIF{\horstAND{\horstEQ{\horstVariI}{\horstConstructorAppVal{\horstVarx}}}{\horstEQ{\horstVariII}{\horstConstructorAppVal{\horstVary}}}},\horstConstructorAppVal{\horstOpAppcilshr{\horstParVarbw}{\horstVarx,\horstVary}}},{{},\horstOTHERWISE,\horstOpAppfreeValOrTop{}{}}}}
\end{horstOperation}
\begin{horstOperation}{ishrs}
  \horstName{ishrs}
  \horstParVar{bw}{\horstTypeint}
  \horstVar{iI}{\horstTypeValue}
  \horstVar{iII}{\horstTypeValue}
  \horstReturnType{\horstTypeValue}
  \horstBody{\horstMatchExp{{{{x}{y}},\horstMATCHIF{\horstAND{\horstEQ{\horstVariI}{\horstConstructorAppVal{\horstVarx}}}{\horstEQ{\horstVariII}{\horstConstructorAppVal{\horstVary}}}},\horstConstructorAppVal{\horstOpAppcishrs{\horstParVarbw}{\horstVarx,\horstVary}}},{{},\horstOTHERWISE,\horstOpAppfreeValOrTop{}{}}}}
\end{horstOperation}
\begin{horstOperation}{irotl}
  \horstName{irotl}
  \horstParVar{bw}{\horstTypeint}
  \horstVar{iI}{\horstTypeValue}
  \horstVar{iII}{\horstTypeValue}
  \horstReturnType{\horstTypeValue}
  \horstBody{\horstMatchExp{{{{x}{y}},\horstMATCHIF{\horstAND{\horstEQ{\horstVariI}{\horstConstructorAppVal{\horstVarx}}}{\horstEQ{\horstVariII}{\horstConstructorAppVal{\horstVary}}}},\horstConstructorAppVal{\horstOpAppcirotl{\horstParVarbw}{\horstVarx,\horstVary}}},{{},\horstOTHERWISE,\horstOpAppfreeValOrTop{}{}}}}
\end{horstOperation}
\begin{horstOperation}{irotr}
  \horstName{irotr}
  \horstParVar{bw}{\horstTypeint}
  \horstVar{iI}{\horstTypeValue}
  \horstVar{iII}{\horstTypeValue}
  \horstReturnType{\horstTypeValue}
  \horstBody{\horstMatchExp{{{{x}{y}},\horstMATCHIF{\horstAND{\horstEQ{\horstVariI}{\horstConstructorAppVal{\horstVarx}}}{\horstEQ{\horstVariII}{\horstConstructorAppVal{\horstVary}}}},\horstConstructorAppVal{\horstOpAppcirotr{\horstParVarbw}{\horstVarx,\horstVary}}},{{},\horstOTHERWISE,\horstOpAppfreeValOrTop{}{}}}}
\end{horstOperation}
\begin{horstOperation}{ieq}
  \horstName{ieq}
  \horstParVar{bw}{\horstTypeint}
  \horstVar{iI}{\horstTypeValue}
  \horstVar{iII}{\horstTypeValue}
  \horstReturnType{\horstTypebool}
  \horstBody{\horstMatchExp{{{{x}{y}},\horstMATCHIF{\horstAND{\horstEQ{\horstVariI}{\horstConstructorAppVal{\horstVarx}}}{\horstEQ{\horstVariII}{\horstConstructorAppVal{\horstVary}}}},\horstOpAppcieq{\horstParVarbw}{\horstVarx,\horstVary}},{{},\horstOTHERWISE,\horstOpAppfreeBool{}{}}}}
\end{horstOperation}
\begin{horstOperation}{ine}
  \horstName{ine}
  \horstParVar{bw}{\horstTypeint}
  \horstVar{iI}{\horstTypeValue}
  \horstVar{iII}{\horstTypeValue}
  \horstReturnType{\horstTypebool}
  \horstBody{\horstMatchExp{{{{x}{y}},\horstMATCHIF{\horstAND{\horstEQ{\horstVariI}{\horstConstructorAppVal{\horstVarx}}}{\horstEQ{\horstVariII}{\horstConstructorAppVal{\horstVary}}}},\horstOpAppcine{\horstParVarbw}{\horstVarx,\horstVary}},{{},\horstOTHERWISE,\horstOpAppfreeBool{}{}}}}
\end{horstOperation}
\begin{horstOperation}{iltu}
  \horstName{iltu}
  \horstParVar{bw}{\horstTypeint}
  \horstVar{iI}{\horstTypeValue}
  \horstVar{iII}{\horstTypeValue}
  \horstReturnType{\horstTypebool}
  \horstBody{\horstMatchExp{{{{x}{y}},\horstMATCHIF{\horstAND{\horstEQ{\horstVariI}{\horstConstructorAppVal{\horstVarx}}}{\horstEQ{\horstVariII}{\horstConstructorAppVal{\horstVary}}}},\horstOpAppciltu{\horstParVarbw}{\horstVarx,\horstVary}},{{},\horstOTHERWISE,\horstOpAppfreeBool{}{}}}}
\end{horstOperation}
\begin{horstOperation}{ilts}
  \horstName{ilts}
  \horstParVar{bw}{\horstTypeint}
  \horstVar{iI}{\horstTypeValue}
  \horstVar{iII}{\horstTypeValue}
  \horstReturnType{\horstTypebool}
  \horstBody{\horstMatchExp{{{{x}{y}},\horstMATCHIF{\horstAND{\horstEQ{\horstVariI}{\horstConstructorAppVal{\horstVarx}}}{\horstEQ{\horstVariII}{\horstConstructorAppVal{\horstVary}}}},\horstOpAppcilts{\horstParVarbw}{\horstVarx,\horstVary}},{{},\horstOTHERWISE,\horstOpAppfreeBool{}{}}}}
\end{horstOperation}
\begin{horstOperation}{igtu}
  \horstName{igtu}
  \horstParVar{bw}{\horstTypeint}
  \horstVar{iI}{\horstTypeValue}
  \horstVar{iII}{\horstTypeValue}
  \horstReturnType{\horstTypebool}
  \horstBody{\horstMatchExp{{{{x}{y}},\horstMATCHIF{\horstAND{\horstEQ{\horstVariI}{\horstConstructorAppVal{\horstVarx}}}{\horstEQ{\horstVariII}{\horstConstructorAppVal{\horstVary}}}},\horstOpAppcigtu{\horstParVarbw}{\horstVarx,\horstVary}},{{},\horstOTHERWISE,\horstOpAppfreeBool{}{}}}}
\end{horstOperation}
\begin{horstOperation}{igts}
  \horstName{igts}
  \horstParVar{bw}{\horstTypeint}
  \horstVar{iI}{\horstTypeValue}
  \horstVar{iII}{\horstTypeValue}
  \horstReturnType{\horstTypebool}
  \horstBody{\horstMatchExp{{{{x}{y}},\horstMATCHIF{\horstAND{\horstEQ{\horstVariI}{\horstConstructorAppVal{\horstVarx}}}{\horstEQ{\horstVariII}{\horstConstructorAppVal{\horstVary}}}},\horstOpAppcigts{\horstParVarbw}{\horstVarx,\horstVary}},{{},\horstOTHERWISE,\horstOpAppfreeBool{}{}}}}
\end{horstOperation}
\begin{horstOperation}{ileu}
  \horstName{ileu}
  \horstParVar{bw}{\horstTypeint}
  \horstVar{iI}{\horstTypeValue}
  \horstVar{iII}{\horstTypeValue}
  \horstReturnType{\horstTypebool}
  \horstBody{\horstMatchExp{{{{x}{y}},\horstMATCHIF{\horstAND{\horstEQ{\horstVariI}{\horstConstructorAppVal{\horstVarx}}}{\horstEQ{\horstVariII}{\horstConstructorAppVal{\horstVary}}}},\horstOpAppcileu{\horstParVarbw}{\horstVarx,\horstVary}},{{},\horstOTHERWISE,\horstOpAppfreeBool{}{}}}}
\end{horstOperation}
\begin{horstOperation}{iles}
  \horstName{iles}
  \horstParVar{bw}{\horstTypeint}
  \horstVar{iI}{\horstTypeValue}
  \horstVar{iII}{\horstTypeValue}
  \horstReturnType{\horstTypebool}
  \horstBody{\horstMatchExp{{{{x}{y}},\horstMATCHIF{\horstAND{\horstEQ{\horstVariI}{\horstConstructorAppVal{\horstVarx}}}{\horstEQ{\horstVariII}{\horstConstructorAppVal{\horstVary}}}},\horstOpAppciles{\horstParVarbw}{\horstVarx,\horstVary}},{{},\horstOTHERWISE,\horstOpAppfreeBool{}{}}}}
\end{horstOperation}
\begin{horstOperation}{igeu}
  \horstName{igeu}
  \horstParVar{bw}{\horstTypeint}
  \horstVar{iI}{\horstTypeValue}
  \horstVar{iII}{\horstTypeValue}
  \horstReturnType{\horstTypebool}
  \horstBody{\horstMatchExp{{{{x}{y}},\horstMATCHIF{\horstAND{\horstEQ{\horstVariI}{\horstConstructorAppVal{\horstVarx}}}{\horstEQ{\horstVariII}{\horstConstructorAppVal{\horstVary}}}},\horstOpAppcigeu{\horstParVarbw}{\horstVarx,\horstVary}},{{},\horstOTHERWISE,\horstOpAppfreeBool{}{}}}}
\end{horstOperation}
\begin{horstOperation}{iges}
  \horstName{iges}
  \horstParVar{bw}{\horstTypeint}
  \horstVar{iI}{\horstTypeValue}
  \horstVar{iII}{\horstTypeValue}
  \horstReturnType{\horstTypebool}
  \horstBody{\horstMatchExp{{{{x}{y}},\horstMATCHIF{\horstAND{\horstEQ{\horstVariI}{\horstConstructorAppVal{\horstVarx}}}{\horstEQ{\horstVariII}{\horstConstructorAppVal{\horstVary}}}},\horstOpAppciges{\horstParVarbw}{\horstVarx,\horstVary}},{{},\horstOTHERWISE,\horstOpAppfreeBool{}{}}}}
\end{horstOperation}
\begin{horstOperation}{iclz}
  \horstName{iclz}
  \horstParVar{bw}{\horstTypeint}
  \horstVar{i}{\horstTypeValue}
  \horstReturnType{\horstTypeValue}
  \horstBody{\horstMatchExp{{{{x}},\horstMATCHIF{\horstEQ{\horstVari}{\horstConstructorAppVal{\horstVarx}}},\horstConstructorAppVal{\horstOpAppciclz{\horstParVarbw}{\horstVarx}}},{{},\horstOTHERWISE,\horstOpAppfreeValOrTop{}{}}}}
\end{horstOperation}
\begin{horstOperation}{ictz}
  \horstName{ictz}
  \horstParVar{bw}{\horstTypeint}
  \horstVar{i}{\horstTypeValue}
  \horstReturnType{\horstTypeValue}
  \horstBody{\horstMatchExp{{{{x}},\horstMATCHIF{\horstEQ{\horstVari}{\horstConstructorAppVal{\horstVarx}}},\horstConstructorAppVal{\horstOpAppcictz{\horstParVarbw}{\horstVarx}}},{{},\horstOTHERWISE,\horstOpAppfreeValOrTop{}{}}}}
\end{horstOperation}
\begin{horstOperation}{ipopcnt}
  \horstName{ipopcnt}
  \horstParVar{bw}{\horstTypeint}
  \horstVar{i}{\horstTypeValue}
  \horstReturnType{\horstTypeValue}
  \horstBody{\horstMatchExp{{{{x}},\horstMATCHIF{\horstEQ{\horstVari}{\horstConstructorAppVal{\horstVarx}}},\horstConstructorAppVal{\horstOpAppcipopcnt{\horstParVarbw}{\horstVarx}}},{{},\horstOTHERWISE,\horstOpAppfreeValOrTop{}{}}}}
\end{horstOperation}
\begin{horstOperation}{ieqz}
  \horstName{ieqz}
  \horstParVar{bw}{\horstTypeint}
  \horstVar{i}{\horstTypeValue}
  \horstReturnType{\horstTypebool}
  \horstBody{\horstMatchExp{{{{x}},\horstMATCHIF{\horstEQ{\horstVari}{\horstConstructorAppVal{\horstVarx}}},\horstOpAppcieqz{\horstParVarbw}{\horstVarx}},{{},\horstOTHERWISE,\horstOpAppfreeBool{}{}}}}
\end{horstOperation}
\begin{horstOperation}{inot}
  \horstName{inot}
  \horstParVar{bw}{\horstTypeint}
  \horstVar{iI}{\horstTypeValue}
  \horstReturnType{\horstTypeValue}
  \horstBody{\horstMatchExp{{{{x}},\horstMATCHIF{\horstEQ{\horstVariI}{\horstConstructorAppVal{\horstVarx}}},\horstConstructorAppVal{\horstOpAppcinot{\horstParVarbw}{\horstVarx}}},{{},\horstOTHERWISE,\horstOpAppfreeValOrTop{}{}}}}
\end{horstOperation}
\begin{horstOperation}{fadd}
  \horstName{fadd}
  \horstParVar{bw}{\horstTypeint}
  \horstVar{zI}{\horstTypeValue}
  \horstVar{zII}{\horstTypeValue}
  \horstReturnType{\horstTypeValue}
  \horstBody{\horstOpAppfreeValOrTop{}{}}
\end{horstOperation}
\begin{horstOperation}{fsub}
  \horstName{fsub}
  \horstParVar{bw}{\horstTypeint}
  \horstVar{zI}{\horstTypeValue}
  \horstVar{zII}{\horstTypeValue}
  \horstReturnType{\horstTypeValue}
  \horstBody{\horstOpAppfreeValOrTop{}{}}
\end{horstOperation}
\begin{horstOperation}{fmul}
  \horstName{fmul}
  \horstParVar{bw}{\horstTypeint}
  \horstVar{zI}{\horstTypeValue}
  \horstVar{zII}{\horstTypeValue}
  \horstReturnType{\horstTypeValue}
  \horstBody{\horstOpAppfreeValOrTop{}{}}
\end{horstOperation}
\begin{horstOperation}{fdiv}
  \horstName{fdiv}
  \horstParVar{bw}{\horstTypeint}
  \horstVar{zI}{\horstTypeValue}
  \horstVar{zII}{\horstTypeValue}
  \horstReturnType{\horstTypeValue}
  \horstBody{\horstOpAppfreeValOrTop{}{}}
\end{horstOperation}
\begin{horstOperation}{fmin}
  \horstName{fmin}
  \horstParVar{bw}{\horstTypeint}
  \horstVar{zI}{\horstTypeValue}
  \horstVar{zII}{\horstTypeValue}
  \horstReturnType{\horstTypeValue}
  \horstBody{\horstOpAppfreeValOrTop{}{}}
\end{horstOperation}
\begin{horstOperation}{fmax}
  \horstName{fmax}
  \horstParVar{bw}{\horstTypeint}
  \horstVar{zI}{\horstTypeValue}
  \horstVar{zII}{\horstTypeValue}
  \horstReturnType{\horstTypeValue}
  \horstBody{\horstOpAppfreeValOrTop{}{}}
\end{horstOperation}
\begin{horstOperation}{fcopysign}
  \horstName{fcopysign}
  \horstParVar{bw}{\horstTypeint}
  \horstVar{zI}{\horstTypeValue}
  \horstVar{zII}{\horstTypeValue}
  \horstReturnType{\horstTypeValue}
  \horstBody{\horstOpAppfreeValOrTop{}{}}
\end{horstOperation}
\begin{horstOperation}{feq}
  \horstName{feq}
  \horstParVar{bw}{\horstTypeint}
  \horstVar{zI}{\horstTypeValue}
  \horstVar{zII}{\horstTypeValue}
  \horstReturnType{\horstTypebool}
  \horstBody{\horstOpAppfreeBool{}{}}
\end{horstOperation}
\begin{horstOperation}{fne}
  \horstName{fne}
  \horstParVar{bw}{\horstTypeint}
  \horstVar{zI}{\horstTypeValue}
  \horstVar{zII}{\horstTypeValue}
  \horstReturnType{\horstTypebool}
  \horstBody{\horstOpAppfreeBool{}{}}
\end{horstOperation}
\begin{horstOperation}{flt}
  \horstName{flt}
  \horstParVar{bw}{\horstTypeint}
  \horstVar{zI}{\horstTypeValue}
  \horstVar{zII}{\horstTypeValue}
  \horstReturnType{\horstTypebool}
  \horstBody{\horstOpAppfreeBool{}{}}
\end{horstOperation}
\begin{horstOperation}{fgt}
  \horstName{fgt}
  \horstParVar{bw}{\horstTypeint}
  \horstVar{zI}{\horstTypeValue}
  \horstVar{zII}{\horstTypeValue}
  \horstReturnType{\horstTypebool}
  \horstBody{\horstOpAppfreeBool{}{}}
\end{horstOperation}
\begin{horstOperation}{fle}
  \horstName{fle}
  \horstParVar{bw}{\horstTypeint}
  \horstVar{zI}{\horstTypeValue}
  \horstVar{zII}{\horstTypeValue}
  \horstReturnType{\horstTypebool}
  \horstBody{\horstOpAppfreeBool{}{}}
\end{horstOperation}
\begin{horstOperation}{fge}
  \horstName{fge}
  \horstParVar{bw}{\horstTypeint}
  \horstVar{zI}{\horstTypeValue}
  \horstVar{zII}{\horstTypeValue}
  \horstReturnType{\horstTypebool}
  \horstBody{\horstOpAppfreeBool{}{}}
\end{horstOperation}
\begin{horstOperation}{fabs}
  \horstName{fabs}
  \horstParVar{bw}{\horstTypeint}
  \horstVar{z}{\horstTypeValue}
  \horstReturnType{\horstTypeValue}
  \horstBody{\horstOpAppfreeValOrTop{}{}}
\end{horstOperation}
\begin{horstOperation}{fneg}
  \horstName{fneg}
  \horstParVar{bw}{\horstTypeint}
  \horstVar{z}{\horstTypeValue}
  \horstReturnType{\horstTypeValue}
  \horstBody{\horstOpAppfreeValOrTop{}{}}
\end{horstOperation}
\begin{horstOperation}{fsqrt}
  \horstName{fsqrt}
  \horstParVar{bw}{\horstTypeint}
  \horstVar{z}{\horstTypeValue}
  \horstReturnType{\horstTypeValue}
  \horstBody{\horstOpAppfreeValOrTop{}{}}
\end{horstOperation}
\begin{horstOperation}{fceil}
  \horstName{fceil}
  \horstParVar{bw}{\horstTypeint}
  \horstVar{z}{\horstTypeValue}
  \horstReturnType{\horstTypeValue}
  \horstBody{\horstOpAppfreeValOrTop{}{}}
\end{horstOperation}
\begin{horstOperation}{ffloor}
  \horstName{ffloor}
  \horstParVar{bw}{\horstTypeint}
  \horstVar{z}{\horstTypeValue}
  \horstReturnType{\horstTypeValue}
  \horstBody{\horstOpAppfreeValOrTop{}{}}
\end{horstOperation}
\begin{horstOperation}{ftrunc}
  \horstName{ftrunc}
  \horstParVar{bw}{\horstTypeint}
  \horstVar{z}{\horstTypeValue}
  \horstReturnType{\horstTypeValue}
  \horstBody{\horstOpAppfreeValOrTop{}{}}
\end{horstOperation}
\begin{horstOperation}{fnearest}
  \horstName{fnearest}
  \horstParVar{bw}{\horstTypeint}
  \horstVar{z}{\horstTypeValue}
  \horstReturnType{\horstTypeValue}
  \horstBody{\horstOpAppfreeValOrTop{}{}}
\end{horstOperation}
\begin{horstOperation}{unOp}
  \horstName{unOp}
  \horstParVar{op}{\horstTypeint}
  \horstVar{a}{\horstTypeValue}
  \horstReturnType{\horstTypeValue}
  \horstBody{\horstMatchExp{{{},\horstMATCHIF{\horstEQ{\horstParVarop}{\horstConstIXXXIIEQZ}},\horstCOND{\horstOpAppieqz{32}{\horstVara}}{\horstOpAppmkConst{1}{}}{\horstOpAppmkConst{0}{}}},{{},\horstMATCHIF{\horstEQ{\horstParVarop}{\horstConstIXXXIICLZ}},\horstOpAppiclz{32}{\horstVara}},{{},\horstMATCHIF{\horstEQ{\horstParVarop}{\horstConstIXXXIICTZ}},\horstOpAppictz{32}{\horstVara}},{{},\horstMATCHIF{\horstEQ{\horstParVarop}{\horstConstIXXXIIPOPCNT}},\horstOpAppipopcnt{32}{\horstVara}},{{},\horstMATCHIF{\horstEQ{\horstParVarop}{\horstConstILXIVEQZ}},\horstCOND{\horstOpAppieqz{64}{\horstVara}}{\horstOpAppmkConst{1}{}}{\horstOpAppmkConst{0}{}}},{{},\horstMATCHIF{\horstEQ{\horstParVarop}{\horstConstILXIVCLZ}},\horstOpAppiclz{64}{\horstVara}},{{},\horstMATCHIF{\horstEQ{\horstParVarop}{\horstConstILXIVCTZ}},\horstOpAppictz{64}{\horstVara}},{{},\horstMATCHIF{\horstEQ{\horstParVarop}{\horstConstILXIVPOPCNT}},\horstOpAppipopcnt{64}{\horstVara}},{{},\horstMATCHIF{\horstEQ{\horstParVarop}{\horstConstFXXXIIABS}},\horstOpAppfabs{32}{\horstVara}},{{},\horstMATCHIF{\horstEQ{\horstParVarop}{\horstConstFXXXIINEG}},\horstOpAppfneg{32}{\horstVara}},{{},\horstMATCHIF{\horstEQ{\horstParVarop}{\horstConstFXXXIICEIL}},\horstOpAppfceil{32}{\horstVara}},{{},\horstMATCHIF{\horstEQ{\horstParVarop}{\horstConstFXXXIIFLOOR}},\horstOpAppffloor{32}{\horstVara}},{{},\horstMATCHIF{\horstEQ{\horstParVarop}{\horstConstFXXXIITRUNC}},\horstOpAppftrunc{32}{\horstVara}},{{},\horstMATCHIF{\horstEQ{\horstParVarop}{\horstConstFXXXIINEAREST}},\horstOpAppfnearest{32}{\horstVara}},{{},\horstMATCHIF{\horstEQ{\horstParVarop}{\horstConstFXXXIISQRT}},\horstOpAppfsqrt{32}{\horstVara}},{{},\horstMATCHIF{\horstEQ{\horstParVarop}{\horstConstFLXIVABS}},\horstOpAppfabs{64}{\horstVara}},{{},\horstMATCHIF{\horstEQ{\horstParVarop}{\horstConstFLXIVNEG}},\horstOpAppfneg{64}{\horstVara}},{{},\horstMATCHIF{\horstEQ{\horstParVarop}{\horstConstFLXIVCEIL}},\horstOpAppfceil{64}{\horstVara}},{{},\horstMATCHIF{\horstEQ{\horstParVarop}{\horstConstFLXIVFLOOR}},\horstOpAppffloor{64}{\horstVara}},{{},\horstMATCHIF{\horstEQ{\horstParVarop}{\horstConstFLXIVTRUNC}},\horstOpAppftrunc{64}{\horstVara}},{{},\horstMATCHIF{\horstEQ{\horstParVarop}{\horstConstFLXIVNEAREST}},\horstOpAppfnearest{64}{\horstVara}},{{},\horstMATCHIF{\horstEQ{\horstParVarop}{\horstConstFLXIVSQRT}},\horstOpAppfsqrt{64}{\horstVara}},{{},\horstOTHERWISE,\horstOpAppfreeValOrTop{}{}}}}
\end{horstOperation}
\begin{horstOperation}{binOp}
  \horstName{binOp}
  \horstParVar{op}{\horstTypeint}
  \horstVar{a}{\horstTypeValue}
  \horstVar{b}{\horstTypeValue}
  \horstReturnType{\horstTypeValue}
  \horstBody{\horstMatchExp{{{},\horstMATCHIF{\horstEQ{\horstParVarop}{\horstConstIXXXIIEQ}},\horstCOND{\horstOpAppieq{32}{\horstVara,\horstVarb}}{\horstOpAppmkConst{1}{}}{\horstOpAppmkConst{0}{}}},{{},\horstMATCHIF{\horstEQ{\horstParVarop}{\horstConstIXXXIINE}},\horstCOND{\horstOpAppine{32}{\horstVara,\horstVarb}}{\horstOpAppmkConst{1}{}}{\horstOpAppmkConst{0}{}}},{{},\horstMATCHIF{\horstEQ{\horstParVarop}{\horstConstIXXXIILTS}},\horstCOND{\horstOpAppilts{32}{\horstVara,\horstVarb}}{\horstOpAppmkConst{1}{}}{\horstOpAppmkConst{0}{}}},{{},\horstMATCHIF{\horstEQ{\horstParVarop}{\horstConstIXXXIILTU}},\horstCOND{\horstOpAppiltu{32}{\horstVara,\horstVarb}}{\horstOpAppmkConst{1}{}}{\horstOpAppmkConst{0}{}}},{{},\horstMATCHIF{\horstEQ{\horstParVarop}{\horstConstIXXXIIGTS}},\horstCOND{\horstOpAppigts{32}{\horstVara,\horstVarb}}{\horstOpAppmkConst{1}{}}{\horstOpAppmkConst{0}{}}},{{},\horstMATCHIF{\horstEQ{\horstParVarop}{\horstConstIXXXIIGTU}},\horstCOND{\horstOpAppigtu{32}{\horstVara,\horstVarb}}{\horstOpAppmkConst{1}{}}{\horstOpAppmkConst{0}{}}},{{},\horstMATCHIF{\horstEQ{\horstParVarop}{\horstConstIXXXIILES}},\horstCOND{\horstOpAppiles{32}{\horstVara,\horstVarb}}{\horstOpAppmkConst{1}{}}{\horstOpAppmkConst{0}{}}},{{},\horstMATCHIF{\horstEQ{\horstParVarop}{\horstConstIXXXIILEU}},\horstCOND{\horstOpAppileu{32}{\horstVara,\horstVarb}}{\horstOpAppmkConst{1}{}}{\horstOpAppmkConst{0}{}}},{{},\horstMATCHIF{\horstEQ{\horstParVarop}{\horstConstIXXXIIGES}},\horstCOND{\horstOpAppiges{32}{\horstVara,\horstVarb}}{\horstOpAppmkConst{1}{}}{\horstOpAppmkConst{0}{}}},{{},\horstMATCHIF{\horstEQ{\horstParVarop}{\horstConstIXXXIIGEU}},\horstCOND{\horstOpAppigeu{32}{\horstVara,\horstVarb}}{\horstOpAppmkConst{1}{}}{\horstOpAppmkConst{0}{}}},{{},\horstMATCHIF{\horstEQ{\horstParVarop}{\horstConstIXXXIIADD}},\horstOpAppiadd{32}{\horstVara,\horstVarb}},{{},\horstMATCHIF{\horstEQ{\horstParVarop}{\horstConstIXXXIISUB}},\horstOpAppisub{32}{\horstVara,\horstVarb}},{{},\horstMATCHIF{\horstEQ{\horstParVarop}{\horstConstIXXXIIMUL}},\horstOpAppimul{32}{\horstVara,\horstVarb}},{{},\horstMATCHIF{\horstEQ{\horstParVarop}{\horstConstIXXXIIAND}},\horstOpAppiand{32}{\horstVara,\horstVarb}},{{},\horstMATCHIF{\horstEQ{\horstParVarop}{\horstConstIXXXIIIOR}},\horstOpAppior{32}{\horstVara,\horstVarb}},{{},\horstMATCHIF{\horstEQ{\horstParVarop}{\horstConstIXXXIIXOR}},\horstOpAppixor{32}{\horstVara,\horstVarb}},{{},\horstMATCHIF{\horstEQ{\horstParVarop}{\horstConstIXXXIISHL}},\horstOpAppishl{32}{\horstVara,\horstVarb}},{{},\horstMATCHIF{\horstEQ{\horstParVarop}{\horstConstIXXXIISHRS}},\horstOpAppishrs{32}{\horstVara,\horstVarb}},{{},\horstMATCHIF{\horstEQ{\horstParVarop}{\horstConstIXXXIISHRU}},\horstOpAppishru{32}{\horstVara,\horstVarb}},{{},\horstMATCHIF{\horstEQ{\horstParVarop}{\horstConstIXXXIIROTL}},\horstOpAppirotl{32}{\horstVara,\horstVarb}},{{},\horstMATCHIF{\horstEQ{\horstParVarop}{\horstConstIXXXIIROTR}},\horstOpAppirotr{32}{\horstVara,\horstVarb}},{{},\horstMATCHIF{\horstEQ{\horstParVarop}{\horstConstILXIVEQ}},\horstCOND{\horstOpAppieq{64}{\horstVara,\horstVarb}}{\horstOpAppmkConst{1}{}}{\horstOpAppmkConst{0}{}}},{{},\horstMATCHIF{\horstEQ{\horstParVarop}{\horstConstILXIVNE}},\horstCOND{\horstOpAppine{64}{\horstVara,\horstVarb}}{\horstOpAppmkConst{1}{}}{\horstOpAppmkConst{0}{}}},{{},\horstMATCHIF{\horstEQ{\horstParVarop}{\horstConstILXIVLTS}},\horstCOND{\horstOpAppilts{64}{\horstVara,\horstVarb}}{\horstOpAppmkConst{1}{}}{\horstOpAppmkConst{0}{}}},{{},\horstMATCHIF{\horstEQ{\horstParVarop}{\horstConstILXIVLTU}},\horstCOND{\horstOpAppiltu{64}{\horstVara,\horstVarb}}{\horstOpAppmkConst{1}{}}{\horstOpAppmkConst{0}{}}},{{},\horstMATCHIF{\horstEQ{\horstParVarop}{\horstConstILXIVGTS}},\horstCOND{\horstOpAppigts{64}{\horstVara,\horstVarb}}{\horstOpAppmkConst{1}{}}{\horstOpAppmkConst{0}{}}},{{},\horstMATCHIF{\horstEQ{\horstParVarop}{\horstConstILXIVGTU}},\horstCOND{\horstOpAppigtu{64}{\horstVara,\horstVarb}}{\horstOpAppmkConst{1}{}}{\horstOpAppmkConst{0}{}}},{{},\horstMATCHIF{\horstEQ{\horstParVarop}{\horstConstILXIVLES}},\horstCOND{\horstOpAppiles{64}{\horstVara,\horstVarb}}{\horstOpAppmkConst{1}{}}{\horstOpAppmkConst{0}{}}},{{},\horstMATCHIF{\horstEQ{\horstParVarop}{\horstConstILXIVLEU}},\horstCOND{\horstOpAppileu{64}{\horstVara,\horstVarb}}{\horstOpAppmkConst{1}{}}{\horstOpAppmkConst{0}{}}},{{},\horstMATCHIF{\horstEQ{\horstParVarop}{\horstConstILXIVGES}},\horstCOND{\horstOpAppiges{64}{\horstVara,\horstVarb}}{\horstOpAppmkConst{1}{}}{\horstOpAppmkConst{0}{}}},{{},\horstMATCHIF{\horstEQ{\horstParVarop}{\horstConstILXIVGEU}},\horstCOND{\horstOpAppigeu{64}{\horstVara,\horstVarb}}{\horstOpAppmkConst{1}{}}{\horstOpAppmkConst{0}{}}},{{},\horstMATCHIF{\horstEQ{\horstParVarop}{\horstConstILXIVADD}},\horstOpAppiadd{64}{\horstVara,\horstVarb}},{{},\horstMATCHIF{\horstEQ{\horstParVarop}{\horstConstILXIVSUB}},\horstOpAppisub{64}{\horstVara,\horstVarb}},{{},\horstMATCHIF{\horstEQ{\horstParVarop}{\horstConstILXIVMUL}},\horstOpAppimul{64}{\horstVara,\horstVarb}},{{},\horstMATCHIF{\horstEQ{\horstParVarop}{\horstConstILXIVAND}},\horstOpAppiand{64}{\horstVara,\horstVarb}},{{},\horstMATCHIF{\horstEQ{\horstParVarop}{\horstConstILXIVIOR}},\horstOpAppior{64}{\horstVara,\horstVarb}},{{},\horstMATCHIF{\horstEQ{\horstParVarop}{\horstConstILXIVXOR}},\horstOpAppixor{64}{\horstVara,\horstVarb}},{{},\horstMATCHIF{\horstEQ{\horstParVarop}{\horstConstILXIVSHL}},\horstOpAppishl{64}{\horstVara,\horstVarb}},{{},\horstMATCHIF{\horstEQ{\horstParVarop}{\horstConstILXIVSHRS}},\horstOpAppishrs{64}{\horstVara,\horstVarb}},{{},\horstMATCHIF{\horstEQ{\horstParVarop}{\horstConstILXIVSHRU}},\horstOpAppishru{64}{\horstVara,\horstVarb}},{{},\horstMATCHIF{\horstEQ{\horstParVarop}{\horstConstILXIVROTL}},\horstOpAppirotl{64}{\horstVara,\horstVarb}},{{},\horstMATCHIF{\horstEQ{\horstParVarop}{\horstConstILXIVROTR}},\horstOpAppirotr{64}{\horstVara,\horstVarb}},{{},\horstMATCHIF{\horstEQ{\horstParVarop}{\horstConstFXXXIIEQ}},\horstCOND{\horstOpAppfeq{32}{\horstVara,\horstVarb}}{\horstOpAppmkConst{1}{}}{\horstOpAppmkConst{0}{}}},{{},\horstMATCHIF{\horstEQ{\horstParVarop}{\horstConstFXXXIINE}},\horstCOND{\horstOpAppfne{32}{\horstVara,\horstVarb}}{\horstOpAppmkConst{1}{}}{\horstOpAppmkConst{0}{}}},{{},\horstMATCHIF{\horstEQ{\horstParVarop}{\horstConstFXXXIILT}},\horstCOND{\horstOpAppflt{32}{\horstVara,\horstVarb}}{\horstOpAppmkConst{1}{}}{\horstOpAppmkConst{0}{}}},{{},\horstMATCHIF{\horstEQ{\horstParVarop}{\horstConstFXXXIIGT}},\horstCOND{\horstOpAppfgt{32}{\horstVara,\horstVarb}}{\horstOpAppmkConst{1}{}}{\horstOpAppmkConst{0}{}}},{{},\horstMATCHIF{\horstEQ{\horstParVarop}{\horstConstFXXXIILE}},\horstCOND{\horstOpAppfle{32}{\horstVara,\horstVarb}}{\horstOpAppmkConst{1}{}}{\horstOpAppmkConst{0}{}}},{{},\horstMATCHIF{\horstEQ{\horstParVarop}{\horstConstFXXXIIGE}},\horstCOND{\horstOpAppfge{32}{\horstVara,\horstVarb}}{\horstOpAppmkConst{1}{}}{\horstOpAppmkConst{0}{}}},{{},\horstMATCHIF{\horstEQ{\horstParVarop}{\horstConstFXXXIIADD}},\horstOpAppfadd{32}{\horstVara,\horstVarb}},{{},\horstMATCHIF{\horstEQ{\horstParVarop}{\horstConstFXXXIISUB}},\horstOpAppfsub{32}{\horstVara,\horstVarb}},{{},\horstMATCHIF{\horstEQ{\horstParVarop}{\horstConstFXXXIIMUL}},\horstOpAppfmul{32}{\horstVara,\horstVarb}},{{},\horstMATCHIF{\horstEQ{\horstParVarop}{\horstConstFXXXIIDIV}},\horstOpAppfdiv{32}{\horstVara,\horstVarb}},{{},\horstMATCHIF{\horstEQ{\horstParVarop}{\horstConstFXXXIIMIN}},\horstOpAppfmin{32}{\horstVara,\horstVarb}},{{},\horstMATCHIF{\horstEQ{\horstParVarop}{\horstConstFXXXIIMAX}},\horstOpAppfmax{32}{\horstVara,\horstVarb}},{{},\horstMATCHIF{\horstEQ{\horstParVarop}{\horstConstFXXXIICOPYSIGN}},\horstOpAppfcopysign{32}{\horstVara,\horstVarb}},{{},\horstMATCHIF{\horstEQ{\horstParVarop}{\horstConstFLXIVEQ}},\horstCOND{\horstOpAppfeq{64}{\horstVara,\horstVarb}}{\horstOpAppmkConst{1}{}}{\horstOpAppmkConst{0}{}}},{{},\horstMATCHIF{\horstEQ{\horstParVarop}{\horstConstFLXIVNE}},\horstCOND{\horstOpAppfne{64}{\horstVara,\horstVarb}}{\horstOpAppmkConst{1}{}}{\horstOpAppmkConst{0}{}}},{{},\horstMATCHIF{\horstEQ{\horstParVarop}{\horstConstFLXIVLT}},\horstCOND{\horstOpAppflt{64}{\horstVara,\horstVarb}}{\horstOpAppmkConst{1}{}}{\horstOpAppmkConst{0}{}}},{{},\horstMATCHIF{\horstEQ{\horstParVarop}{\horstConstFLXIVGT}},\horstCOND{\horstOpAppfgt{64}{\horstVara,\horstVarb}}{\horstOpAppmkConst{1}{}}{\horstOpAppmkConst{0}{}}},{{},\horstMATCHIF{\horstEQ{\horstParVarop}{\horstConstFLXIVLE}},\horstCOND{\horstOpAppfle{64}{\horstVara,\horstVarb}}{\horstOpAppmkConst{1}{}}{\horstOpAppmkConst{0}{}}},{{},\horstMATCHIF{\horstEQ{\horstParVarop}{\horstConstFLXIVGE}},\horstCOND{\horstOpAppfge{64}{\horstVara,\horstVarb}}{\horstOpAppmkConst{1}{}}{\horstOpAppmkConst{0}{}}},{{},\horstMATCHIF{\horstEQ{\horstParVarop}{\horstConstFLXIVADD}},\horstOpAppfadd{64}{\horstVara,\horstVarb}},{{},\horstMATCHIF{\horstEQ{\horstParVarop}{\horstConstFLXIVSUB}},\horstOpAppfsub{64}{\horstVara,\horstVarb}},{{},\horstMATCHIF{\horstEQ{\horstParVarop}{\horstConstFLXIVMUL}},\horstOpAppfmul{64}{\horstVara,\horstVarb}},{{},\horstMATCHIF{\horstEQ{\horstParVarop}{\horstConstFLXIVDIV}},\horstOpAppfdiv{64}{\horstVara,\horstVarb}},{{},\horstMATCHIF{\horstEQ{\horstParVarop}{\horstConstFLXIVMIN}},\horstOpAppfmin{64}{\horstVara,\horstVarb}},{{},\horstMATCHIF{\horstEQ{\horstParVarop}{\horstConstFLXIVMAX}},\horstOpAppfmax{64}{\horstVara,\horstVarb}},{{},\horstMATCHIF{\horstEQ{\horstParVarop}{\horstConstFLXIVCOPYSIGN}},\horstOpAppfcopysign{64}{\horstVara,\horstVarb}},{{},\horstOTHERWISE,\horstOpAppfreeValOrTop{}{}}}}
\end{horstOperation}
\begin{horstOperation}{trappingBinOp}
  \horstName{trappingBinOp}
  \horstParVar{op}{\horstTypeint}
  \horstVar{a}{\horstTypeValue}
  \horstVar{b}{\horstTypeValue}
  \horstReturnType{\horstTypeMaybeValue}
  \horstBody{\horstMatchExp{{{},\horstMATCHIF{\horstEQ{\horstParVarop}{\horstConstIXXXIIDIVS}},\horstOpAppidivs{32}{\horstVara,\horstVarb}},{{},\horstMATCHIF{\horstEQ{\horstParVarop}{\horstConstIXXXIIDIVU}},\horstOpAppidivu{32}{\horstVara,\horstVarb}},{{},\horstMATCHIF{\horstEQ{\horstParVarop}{\horstConstIXXXIIREMS}},\horstOpAppirems{32}{\horstVara,\horstVarb}},{{},\horstMATCHIF{\horstEQ{\horstParVarop}{\horstConstIXXXIIREMU}},\horstOpAppiremu{32}{\horstVara,\horstVarb}},{{},\horstMATCHIF{\horstEQ{\horstParVarop}{\horstConstILXIVDIVS}},\horstOpAppidivs{64}{\horstVara,\horstVarb}},{{},\horstMATCHIF{\horstEQ{\horstParVarop}{\horstConstILXIVDIVU}},\horstOpAppidivu{64}{\horstVara,\horstVarb}},{{},\horstMATCHIF{\horstEQ{\horstParVarop}{\horstConstILXIVREMS}},\horstOpAppirems{64}{\horstVara,\horstVarb}},{{},\horstMATCHIF{\horstEQ{\horstParVarop}{\horstConstILXIVREMU}},\horstOpAppiremu{64}{\horstVara,\horstVarb}},{{},\horstOTHERWISE,\horstConstructorAppNothingV{}}}}
\end{horstOperation}
\begin{horstOperation}{cvtOp}
  \horstName{cvtOp}
  \horstParVar{op}{\horstTypeint}
  \horstVar{a}{\horstTypeValue}
  \horstReturnType{\horstTypeValue}
  \horstBody{\horstMatchExp{{{},\horstMATCHIF{\horstEQ{\horstParVarop}{\horstConstIXXXIIWRAPILXIV}},\horstOpAppwrap{32,32}{\horstVara}},{{},\horstMATCHIF{\horstEQ{\horstParVarop}{\horstConstIXXXIITRUNCSFXXXII}},\horstOpApptruncs{32,32}{\horstVara}},{{},\horstMATCHIF{\horstEQ{\horstParVarop}{\horstConstIXXXIITRUNCUFXXXII}},\horstOpApptruncu{32,32}{\horstVara}},{{},\horstMATCHIF{\horstEQ{\horstParVarop}{\horstConstIXXXIITRUNCSFLXIV}},\horstOpApptruncs{32,64}{\horstVara}},{{},\horstMATCHIF{\horstEQ{\horstParVarop}{\horstConstIXXXIITRUNCUFLXIV}},\horstOpApptruncu{32,64}{\horstVara}},{{},\horstMATCHIF{\horstEQ{\horstParVarop}{\horstConstILXIVEXTENDIXXXIIS}},\horstOpAppextends{32,64}{\horstVara}},{{},\horstMATCHIF{\horstEQ{\horstParVarop}{\horstConstILXIVEXTENDIXXXIIU}},\horstOpAppextendu{32,64}{\horstVara}},{{},\horstMATCHIF{\horstEQ{\horstParVarop}{\horstConstILXIVTRUNCFXXXIIS}},\horstOpApptruncs{64,32}{\horstVara}},{{},\horstMATCHIF{\horstEQ{\horstParVarop}{\horstConstILXIVTRUNCFXXXIIU}},\horstOpApptruncu{64,32}{\horstVara}},{{},\horstMATCHIF{\horstEQ{\horstParVarop}{\horstConstILXIVTRUNCFLXIVS}},\horstOpApptruncs{64,64}{\horstVara}},{{},\horstMATCHIF{\horstEQ{\horstParVarop}{\horstConstILXIVTRUNCFLXIVU}},\horstOpApptruncu{64,64}{\horstVara}},{{},\horstMATCHIF{\horstEQ{\horstParVarop}{\horstConstFXXXIICONVERTIXXXIIS}},\horstOpAppconverts{32,32}{\horstVara}},{{},\horstMATCHIF{\horstEQ{\horstParVarop}{\horstConstFXXXIICONVERTIXXXIIU}},\horstOpAppconvertu{32,32}{\horstVara}},{{},\horstMATCHIF{\horstEQ{\horstParVarop}{\horstConstFXXXIICONVERTILXIVS}},\horstOpAppconverts{32,64}{\horstVara}},{{},\horstMATCHIF{\horstEQ{\horstParVarop}{\horstConstFXXXIICONVERTILXIVU}},\horstOpAppconvertu{32,64}{\horstVara}},{{},\horstMATCHIF{\horstEQ{\horstParVarop}{\horstConstFXXXIIDEMOTEFLXIV}},\horstOpAppdemote{32,64}{\horstVara}},{{},\horstMATCHIF{\horstEQ{\horstParVarop}{\horstConstFLXIVCONVERTIXXXIIS}},\horstOpAppconverts{64,32}{\horstVara}},{{},\horstMATCHIF{\horstEQ{\horstParVarop}{\horstConstFLXIVCONVERTIXXXIIU}},\horstOpAppconvertu{64,32}{\horstVara}},{{},\horstMATCHIF{\horstEQ{\horstParVarop}{\horstConstFLXIVCONVERTILXIVS}},\horstOpAppconverts{64,64}{\horstVara}},{{},\horstMATCHIF{\horstEQ{\horstParVarop}{\horstConstFLXIVCONVERTILXIVU}},\horstOpAppconvertu{64,64}{\horstVara}},{{},\horstMATCHIF{\horstEQ{\horstParVarop}{\horstConstFLXIVPROMOTEFXXXII}},\horstOpApppromote{64,32}{\horstVara}},{{},\horstMATCHIF{\horstEQ{\horstParVarop}{\horstConstIXXXIIREINTERPRETFXXXII}},\horstOpAppreinterpret{32,32}{\horstVara}},{{},\horstMATCHIF{\horstEQ{\horstParVarop}{\horstConstILXIVREINTERPRETFLXIV}},\horstOpAppreinterpret{64,64}{\horstVara}},{{},\horstMATCHIF{\horstEQ{\horstParVarop}{\horstConstFXXXIIREINTERPRETIXXXII}},\horstOpAppreinterpret{32,32}{\horstVara}},{{},\horstMATCHIF{\horstEQ{\horstParVarop}{\horstConstFLXIVREINTERPRETILXIV}},\horstOpAppreinterpret{64,64}{\horstVara}},{{},\horstOTHERWISE,\horstOpAppfreeValOrTop{}{}}}}
\end{horstOperation}
\begin{horstOperation}{isInRangeAux}
  \horstName{isInRangeAux}
  \horstParVar{bw}{\horstTypeint}
  \horstVar{x}{\horstTypeBVLXIV}
  \horstReturnType{\horstTypebool}
  \horstBody{\horstOR{\horstGE{\horstParVarbw}{64}}{\horstOpAppbvult{64}{\horstVarx,\horstOpAppcinj{\horstOpApppow{\horstParVarbw}{2}}{}}}}
\end{horstOperation}
\begin{horstOperation}{isVal}
  \horstName{isVal}
  \horstVar{v}{\horstTypeValue}
  \horstReturnType{\horstTypebool}
  \horstBody{\horstMatchExp{{{{x}},\horstMATCHIF{\horstEQ{\horstVarv}{\horstConstructorAppVal{\horstVarx}}},\horstTrue},{{},\horstOTHERWISE,\horstFalse}}}
\end{horstOperation}
\begin{horstOperation}{isInRange}
  \horstName{isInRange}
  \horstParVar{bw}{\horstTypeint}
  \horstVar{v}{\horstTypeValue}
  \horstReturnType{\horstTypebool}
  \horstBody{\horstMatchExp{{{{x}},\horstMATCHIF{\horstEQ{\horstVarv}{\horstConstructorAppVal{\horstVarx}}},\horstOpAppisInRangeAux{\horstParVarbw}{\horstVarx}},{{},\horstOTHERWISE,\horstFalse}}}
\end{horstOperation}
\begin{horstOperation}{boundOKAux}
  \horstName{boundOKAux}
  \horstParVar{bw}{\horstTypeint}
  \horstParVar{offset}{\horstTypeint}
  \horstVar{size}{\horstTypeBVLXIV}
  \horstVar{start}{\horstTypeBVLXIV}
  \horstReturnType{\horstTypebool}
  \horstBody{\horstOpAppbvule{64}{\horstOpAppbvadd{64}{\horstVarstart,\horstOpAppcintIIbv{64,\horstADD{\horstParVaroffset}{\horstDIV{\horstParVarbw}{8}}}{}},\horstOpAppbvmul{64}{\horstVarsize,\horstOpAppcintIIbv{64,\horstOpApppow{16}{2}}{}}}}
\end{horstOperation}
\begin{horstOperation}{storeValueAux}
  \horstName{storeValueAux}
  \horstParVar{i}{\horstTypeint}
  \horstVar{x}{\horstTypeBVLXIV}
  \horstReturnType{\horstTypeBVLXIV}
  \horstBody{\horstOpAppbvextract{64,8}{\horstOpAppbvlshr{64}{\horstVarx,\horstOpAppcintIIbv{64,\horstMUL{8}{\horstParVari}}{}}}}
\end{horstOperation}
\begin{horstOperation}{loadValueAux}
  \horstName{loadValueAux}
  \horstParVar{i}{\horstTypeint}
  \horstVar{x}{\horstTypeBVLXIV}
  \horstVar{y}{\horstTypeBVLXIV}
  \horstReturnType{\horstTypeBVLXIV}
  \horstBody{\horstOpAppbvor{64}{\horstVarx,\horstOpAppbvshl{64}{\horstVary,\horstOpAppcintIIbv{64,\horstMUL{\horstParVari}{8}}{}}}}
\end{horstOperation}
\begin{horstOperation}{growAdd}
  \horstName{growAdd}
  \horstVar{p}{\horstTypeBVLXIV}
  \horstVar{size}{\horstTypeBVLXIV}
  \horstReturnType{\horstTypeBVLXIV}
  \horstBody{\horstOpAppbvadd{64}{\horstVarp,\horstVarsize}}
\end{horstOperation}
\begin{horstOperation}{growOK}
  \horstName{growOK}
  \horstVar{p}{\horstTypeBVLXIV}
  \horstVar{size}{\horstTypeBVLXIV}
  \horstVar{max}{\horstTypeBVLXIV}
  \horstReturnType{\horstTypebool}
  \horstBody{\horstAND{\horstOpAppbvule{64}{\horstVarp,\horstOpAppcintIIbv{64,\horstOpApppow{16}{2}}{}}}{\horstOpAppbvule{64}{\horstOpAppgrowAdd{}{\horstVarp,\horstVarsize},\horstVarmax}}}
\end{horstOperation}
\begin{horstOperation}{wasmStore}
  \horstName{wasmStore}
  \horstParVar{bw}{\horstTypeint}
  \horstParVar{n}{\horstTypeint}
  \horstParVar{offset}{\horstTypeint}
  \horstVar{mem}{\horstTypeArray{\horstTypeValue}}
  \horstVar{i}{\horstTypeint}
  \horstVar{x}{\horstTypeValue}
  \horstReturnType{\horstTypeArray{\horstTypeValue}}
  \horstBody{\horstARRAYINIT{\horstOpAppfreeValOrTop{}{}}}
\end{horstOperation}
\begin{horstOperation}{wasmLoad}
  \horstName{wasmLoad}
  \horstParVar{bw}{\horstTypeint}
  \horstParVar{n}{\horstTypeint}
  \horstParVar{signed}{\horstTypebool}
  \horstParVar{offset}{\horstTypeint}
  \horstVar{mem}{\horstTypeArray{\horstTypeValue}}
  \horstVar{i}{\horstTypeint}
  \horstReturnType{\horstTypeValue}
  \horstBody{\horstOpAppfreeValOrTop{}{}}
\end{horstOperation}
\begin{horstOperation}{boundOK}
  \horstName{boundOK}
  \horstParVar{bw}{\horstTypeint}
  \horstParVar{offset}{\horstTypeint}
  \horstVar{size}{\horstTypeValue}
  \horstVar{start}{\horstTypeValue}
  \horstReturnType{\horstTypebool}
  \horstBody{\horstOpAppboundOKAux{\horstParVarbw,\horstParVaroffset}{\horstOpAppbase{}{\horstVarsize},\horstOpAppbase{}{\horstVarstart}}}
\end{horstOperation}
\begin{horstOperation}{grow}
  \horstName{grow}
  \horstVar{pages}{\horstTypeValue}
  \horstVar{size}{\horstTypeValue}
  \horstVar{max}{\horstTypeValue}
  \horstReturnType{\horstTypeValue}
  \horstBody{\horstMatchExp{{{{p}{s}{m}},\horstMATCHIF{\horstAND{\horstAND{\horstEQ{\horstVarpages}{\horstConstructorAppVal{\horstVarp}}}{\horstEQ{\horstVarsize}{\horstConstructorAppVal{\horstVars}}}}{\horstEQ{\horstVarmax}{\horstConstructorAppVal{\horstVarm}}}},\horstCOND{\horstOpAppgrowOK{}{\horstVarp,\horstVars,\horstVarm}}{\horstConstructorAppVal{\horstOpAppgrowAdd{}{\horstVarp,\horstVars}}}{\horstVarsize}},{{},\horstOTHERWISE,\horstOpAppfreeValOrTop{}{}}}}
\end{horstOperation}
\begin{horstOperation}{memOpBw}
  \horstName{memOpBw}
  \horstParVar{op}{\horstTypeint}
  \horstReturnType{\horstTypeint}
  \horstBody{\horstMatchExp{{{},\horstMATCHIF{\horstEQ{\horstParVarop}{\horstConstIXXXIILOADMEM}},32},{{},\horstMATCHIF{\horstEQ{\horstParVarop}{\horstConstFXXXIILOADMEM}},32},{{},\horstMATCHIF{\horstEQ{\horstParVarop}{\horstConstIXXXIILOADMEMVIIIS}},32},{{},\horstMATCHIF{\horstEQ{\horstParVarop}{\horstConstIXXXIILOADMEMVIIIU}},32},{{},\horstMATCHIF{\horstEQ{\horstParVarop}{\horstConstIXXXIILOADMEMXVIU}},32},{{},\horstMATCHIF{\horstEQ{\horstParVarop}{\horstConstIXXXIILOADMEMXVIS}},32},{{},\horstMATCHIF{\horstEQ{\horstParVarop}{\horstConstIXXXIISTOREMEM}},32},{{},\horstMATCHIF{\horstEQ{\horstParVarop}{\horstConstFXXXIISTOREMEM}},32},{{},\horstMATCHIF{\horstEQ{\horstParVarop}{\horstConstIXXXIISTOREMEMVIII}},32},{{},\horstMATCHIF{\horstEQ{\horstParVarop}{\horstConstIXXXIISTOREMEMXVI}},32},{{},\horstOTHERWISE,64}}}
\end{horstOperation}
\begin{horstOperation}{memOpTbw}
  \horstName{memOpTbw}
  \horstParVar{op}{\horstTypeint}
  \horstReturnType{\horstTypeint}
  \horstBody{\horstMatchExp{{{},\horstMATCHIF{\horstEQ{\horstParVarop}{\horstConstIXXXIILOADMEM}},32},{{},\horstMATCHIF{\horstEQ{\horstParVarop}{\horstConstILXIVLOADMEM}},64},{{},\horstMATCHIF{\horstEQ{\horstParVarop}{\horstConstFXXXIILOADMEM}},32},{{},\horstMATCHIF{\horstEQ{\horstParVarop}{\horstConstIXXXIILOADMEMVIIIS}},8},{{},\horstMATCHIF{\horstEQ{\horstParVarop}{\horstConstIXXXIILOADMEMVIIIU}},8},{{},\horstMATCHIF{\horstEQ{\horstParVarop}{\horstConstIXXXIILOADMEMXVIS}},16},{{},\horstMATCHIF{\horstEQ{\horstParVarop}{\horstConstIXXXIILOADMEMXVIU}},16},{{},\horstMATCHIF{\horstEQ{\horstParVarop}{\horstConstILXIVLOADMEMVIIIS}},8},{{},\horstMATCHIF{\horstEQ{\horstParVarop}{\horstConstILXIVLOADMEMVIIIU}},8},{{},\horstMATCHIF{\horstEQ{\horstParVarop}{\horstConstILXIVLOADMEMXVIS}},16},{{},\horstMATCHIF{\horstEQ{\horstParVarop}{\horstConstILXIVLOADMEMXVIU}},16},{{},\horstMATCHIF{\horstEQ{\horstParVarop}{\horstConstILXIVLOADMEMXXXIIS}},32},{{},\horstMATCHIF{\horstEQ{\horstParVarop}{\horstConstILXIVLOADMEMXXXIIU}},32},{{},\horstMATCHIF{\horstEQ{\horstParVarop}{\horstConstIXXXIISTOREMEM}},32},{{},\horstMATCHIF{\horstEQ{\horstParVarop}{\horstConstILXIVSTOREMEM}},64},{{},\horstMATCHIF{\horstEQ{\horstParVarop}{\horstConstFXXXIISTOREMEM}},32},{{},\horstMATCHIF{\horstEQ{\horstParVarop}{\horstConstFLXIVSTOREMEM}},64},{{},\horstMATCHIF{\horstEQ{\horstParVarop}{\horstConstIXXXIISTOREMEMVIII}},8},{{},\horstMATCHIF{\horstEQ{\horstParVarop}{\horstConstIXXXIISTOREMEMXVI}},16},{{},\horstMATCHIF{\horstEQ{\horstParVarop}{\horstConstILXIVSTOREMEMVIII}},8},{{},\horstMATCHIF{\horstEQ{\horstParVarop}{\horstConstILXIVSTOREMEMXVI}},16},{{},\horstMATCHIF{\horstEQ{\horstParVarop}{\horstConstILXIVSTOREMEMXXXII}},32},{{},\horstOTHERWISE,64}}}
\end{horstOperation}
\begin{horstOperation}{memLoadSigned}
  \horstName{memLoadSigned}
  \horstParVar{op}{\horstTypeint}
  \horstReturnType{\horstTypebool}
  \horstBody{\horstMatchExp{{{},\horstMATCHIF{\horstEQ{\horstParVarop}{\horstConstIXXXIILOADMEMVIIIS}},\horstTrue},{{},\horstMATCHIF{\horstEQ{\horstParVarop}{\horstConstIXXXIILOADMEMXVIS}},\horstTrue},{{},\horstMATCHIF{\horstEQ{\horstParVarop}{\horstConstILXIVLOADMEMVIIIS}},\horstTrue},{{},\horstMATCHIF{\horstEQ{\horstParVarop}{\horstConstILXIVLOADMEMXVIS}},\horstTrue},{{},\horstMATCHIF{\horstEQ{\horstParVarop}{\horstConstILXIVLOADMEMXXXIIS}},\horstTrue},{{},\horstOTHERWISE,\horstFalse}}}
\end{horstOperation}
\begin{horstOperation}{memLoad}
  \horstName{memLoad}
  \horstParVar{op}{\horstTypeint}
  \horstParVar{offset}{\horstTypeint}
  \horstVar{mem}{\horstTypeArray{\horstTypeValue}}
  \horstVar{i}{\horstTypeValue}
  \horstReturnType{\horstTypeValue}
  \horstBody{\horstMatchExp{{{},\horstMATCHIF{\horstEQ{\horstParVarop}{\horstConstFXXXIILOADMEM}},\horstOpAppfreeValOrTop{}{}},{{},\horstMATCHIF{\horstEQ{\horstParVarop}{\horstConstFLXIVLOADMEM}},\horstOpAppfreeValOrTop{}{}},{{},\horstOTHERWISE,\horstOpAppwasmLoad{\horstOpAppmemOpBw{\horstParVarop}{},\horstOpAppmemOpTbw{\horstParVarop}{},\horstOpAppmemLoadSigned{\horstParVarop}{},\horstParVaroffset}{\horstVarmem,\horstVari}}}}
\end{horstOperation}
\begin{horstOperation}{memStore}
  \horstName{memStore}
  \horstParVar{op}{\horstTypeint}
  \horstParVar{offset}{\horstTypeint}
  \horstVar{mem}{\horstTypeArray{\horstTypeValue}}
  \horstVar{i}{\horstTypeValue}
  \horstVar{x}{\horstTypeValue}
  \horstReturnType{\horstTypeArray{\horstTypeValue}}
  \horstBody{\horstMatchExp{{{},\horstMATCHIF{\horstEQ{\horstParVarop}{\horstConstFXXXIISTOREMEM}},\horstOpAppwasmStore{\horstOpAppmemOpBw{\horstParVarop}{},\horstOpAppmemOpTbw{\horstParVarop}{},\horstParVaroffset}{\horstVarmem,\horstVari,\horstOpAppfreeValOrTop{}{}}},{{},\horstMATCHIF{\horstEQ{\horstParVarop}{\horstConstFLXIVSTOREMEM}},\horstOpAppwasmStore{\horstOpAppmemOpBw{\horstParVarop}{},\horstOpAppmemOpTbw{\horstParVarop}{},\horstParVaroffset}{\horstVarmem,\horstVari,\horstOpAppfreeValOrTop{}{}}},{{},\horstOTHERWISE,\horstOpAppwasmStore{\horstOpAppmemOpBw{\horstParVarop}{},\horstOpAppmemOpTbw{\horstParVarop}{},\horstParVaroffset}{\horstVarmem,\horstVari,\horstVarx}}}}
\end{horstOperation}
\begin{horstOperation}{mkLabel}
  \horstName{mkLabel}
  \horstVar{c}{\horstTypebool}
  \horstVar{i}{\horstTypebool}
  \horstReturnType{\horstTypeLabel}
  \horstBody{\horstConstructorAppLbl{\horstVarc,\horstVari}}
\end{horstOperation}
\begin{horstOperation}{mkTop}
  \horstName{mkTop}
  \horstReturnType{\horstTypeLabel}
  \horstBody{\horstConstructorAppLbl{\horstTrue,\horstFalse}}
\end{horstOperation}
\begin{horstOperation}{mkUntrusted}
  \horstName{mkUntrusted}
  \horstReturnType{\horstTypeLabel}
  \horstBody{\horstConstructorAppLbl{\horstFalse,\horstFalse}}
\end{horstOperation}
\begin{horstOperation}{mkSecret}
  \horstName{mkSecret}
  \horstReturnType{\horstTypeLabel}
  \horstBody{\horstConstructorAppLbl{\horstTrue,\horstTrue}}
\end{horstOperation}
\begin{horstOperation}{mkBot}
  \horstName{mkBot}
  \horstReturnType{\horstTypeLabel}
  \horstBody{\horstConstructorAppLbl{\horstFalse,\horstTrue}}
\end{horstOperation}
\begin{horstOperation}{glb}
  \horstName{glb}
  \horstParVar{sz}{\horstTypeint}
  \horstVar{x}{\horstTypeHomInit{\horstTypeLabel}{\horstParVarsz}}
  \horstReturnType{\horstTypeLabel}
  \horstBody{\horstCustomSumExp{\horstSelectorFunctionAppinterval{\horstParVari}{1,\horstParVarsz}}{acc}{\horstMatchExp{{{{iI}{iII}{cI}{cII}},\horstMATCHIF{\horstAND{\horstEQ{\horstVaracc}{\horstConstructorAppLbl{\horstVarcI,\horstVariI}}}{\horstEQ{\horstACCESS{\horstVarx}{\horstParVari}}{\horstConstructorAppLbl{\horstVarcII,\horstVariII}}}},\horstConstructorAppLbl{\horstAND{\horstVarcI}{\horstVarcII},\horstOR{\horstVariI}{\horstVariII}}},{{},\horstOTHERWISE,\horstOpAppmkTop{}{}}}}{\horstACCESS{\horstVarx}{0}}{{i}}}
\end{horstOperation}
\begin{horstOperation}{lub}
  \horstName{lub}
  \horstParVar{sz}{\horstTypeint}
  \horstVar{x}{\horstTypeHomInit{\horstTypeLabel}{\horstParVarsz}}
  \horstReturnType{\horstTypeLabel}
  \horstBody{\horstCustomSumExp{\horstSelectorFunctionAppinterval{\horstParVari}{1,\horstParVarsz}}{acc}{\horstMatchExp{{{{iI}{iII}{cI}{cII}},\horstMATCHIF{\horstAND{\horstEQ{\horstVaracc}{\horstConstructorAppLbl{\horstVarcI,\horstVariI}}}{\horstEQ{\horstACCESS{\horstVarx}{\horstParVari}}{\horstConstructorAppLbl{\horstVarcII,\horstVariII}}}},\horstConstructorAppLbl{\horstOR{\horstVarcI}{\horstVarcII},\horstAND{\horstVariI}{\horstVariII}}},{{},\horstOTHERWISE,\horstOpAppmkTop{}{}}}}{\horstACCESS{\horstVarx}{0}}{{i}}}
\end{horstOperation}
\begin{horstOperation}{flowsTo}
  \horstName{flowsTo}
  \horstVar{a}{\horstTypeLabel}
  \horstVar{b}{\horstTypeLabel}
  \horstReturnType{\horstTypebool}
  \horstBody{\horstMatchExp{{{{iI}{iII}{cI}{cII}},\horstMATCHIF{\horstAND{\horstEQ{\horstVara}{\horstConstructorAppLbl{\horstVarcI,\horstVariI}}}{\horstEQ{\horstVarb}{\horstConstructorAppLbl{\horstVarcII,\horstVariII}}}},\horstAND{\horstOR{\horstNOT{\horstVarcI}}{\horstVarcII}}{\horstOR{\horstVariI}{\horstNOT{\horstVariII}}}},{{},\horstOTHERWISE,\horstFalse}}}
\end{horstOperation}
\begin{horstOperation}{mkFlowLabel}
  \horstName{mkFlowLabel}
  \horstVar{a}{\horstTypeLabel}
  \horstVar{b}{\horstTypeLabel}
  \horstReturnType{\horstTypeFlowLabel}
  \horstBody{\horstCOND{\horstOpAppflowsTo{}{\horstVara,\horstVarb}}{\horstConstructorAppLegal{}}{\horstConstructorAppIllegal{}}}
\end{horstOperation}
\begin{horstOperation}{flub}
  \horstName{flub}
  \horstParVar{sz}{\horstTypeint}
  \horstVar{x}{\horstTypeHomInit{\horstTypeFlowLabel}{\horstParVarsz}}
  \horstReturnType{\horstTypeFlowLabel}
  \horstBody{\horstCustomSumExp{\horstSelectorFunctionAppinterval{\horstParVari}{1,\horstParVarsz}}{acc}{\horstCOND{\horstEQ{\horstVaracc}{\horstConstructorAppIllegal{}}}{\horstConstructorAppIllegal{}}{\horstACCESS{\horstVarx}{\horstParVari}}}{\horstACCESS{\horstVarx}{0}}{{i}}}
\end{horstOperation}
\begin{horstOperation}{mkCtx}
  \horstName{mkCtx}
  \horstVar{p}{\horstTypeLabel}
  \horstVar{l}{\horstTypeFlowLabel}
  \horstReturnType{\horstTypeContext}
  \horstBody{\horstConstructorAppCtx{\horstVarp,\horstVarl,-1}}
\end{horstOperation}
\begin{horstOperation}{perspectiveOfCtx}
  \horstName{perspectiveOfCtx}
  \horstVar{ctx}{\horstTypeContext}
  \horstReturnType{\horstTypeLabel}
  \horstBody{\horstMatchExp{{{{p}},\horstMATCHIF{\horstEQ{\horstVarctx}{\horstConstructorAppCtx{\horstVarp,\horstVarWILDCARD,\horstVarWILDCARD}}},\horstVarp},{{},\horstOTHERWISE,\horstOpAppmkBot{}{}}}}
\end{horstOperation}
\begin{horstOperation}{labelOfCtx}
  \horstName{labelOfCtx}
  \horstVar{ctx}{\horstTypeContext}
  \horstReturnType{\horstTypeFlowLabel}
  \horstBody{\horstMatchExp{{{{l}},\horstMATCHIF{\horstEQ{\horstVarctx}{\horstConstructorAppCtx{\horstVarWILDCARD,\horstVarl,\horstVarWILDCARD}}},\horstVarl},{{},\horstOTHERWISE,\horstConstructorAppIllegal{}}}}
\end{horstOperation}
\begin{horstOperation}{raiseCtxTo}
  \horstName{raiseCtxTo}
  \horstParVar{pc}{\horstTypeint}
  \horstVar{ctx}{\horstTypeContext}
  \horstVar{l}{\horstTypeFlowLabel}
  \horstReturnType{\horstTypeContext}
  \horstBody{\horstMatchExp{{{{p}{cl}{from}},\horstMATCHIF{\horstEQ{\horstVarctx}{\horstConstructorAppCtx{\horstVarp,\horstVarcl,\horstVarfrom}}},\horstConstructorAppCtx{\horstVarp,\horstOpAppflub{2}{\horstTUPINIT{\horstVarcl,\horstVarl}},\horstCOND{\horstOR{\horstEQ{\horstVarl}{\horstConstructorAppLegal{}}}{\horstEQ{\horstVarcl}{\horstConstructorAppIllegal{}}}}{\horstVarfrom}{\horstParVarpc}}},{{},\horstOTHERWISE,\horstVarctx}}}
\end{horstOperation}
\begin{horstOperation}{mkLConst}
  \horstName{mkLConst}
  \horstParVar{i}{\horstTypeint}
  \horstReturnType{\horstTypeLValue}
  \horstBody{\horstConstructorAppLVal{\horstOpAppmkConst{\horstParVari}{},\horstConstructorAppLegal{}}}
\end{horstOperation}
\begin{horstOperation}{mkLValue}
  \horstName{mkLValue}
  \horstVar{i}{\horstTypeint}
  \horstReturnType{\horstTypeLValue}
  \horstBody{\horstConstructorAppLVal{\horstOpAppmkValue{}{\horstVari},\horstConstructorAppLegal{}}}
\end{horstOperation}
\begin{horstOperation}{lval}
  \horstName{lval}
  \horstParVar{top}{\horstTypebool}
  \horstParVar{v}{\horstTypeint}
  \horstReturnType{\horstTypeLValue}
  \horstBody{\horstConstructorAppLVal{\horstOpAppval{\horstParVartop,\horstParVarv}{},\horstConstructorAppLegal{}}}
\end{horstOperation}
\begin{horstOperation}{valueOf}
  \horstName{valueOf}
  \horstVar{x}{\horstTypeLValue}
  \horstReturnType{\horstTypeValue}
  \horstBody{\horstMatchExp{{{{v}},\horstMATCHIF{\horstEQ{\horstVarx}{\horstConstructorAppLVal{\horstVarv,\horstVarWILDCARD}}},\horstVarv},{{},\horstOTHERWISE,\horstOpAppfreeValOrTop{}{}}}}
\end{horstOperation}
\begin{horstOperation}{labelOf}
  \horstName{labelOf}
  \horstVar{x}{\horstTypeLValue}
  \horstReturnType{\horstTypeFlowLabel}
  \horstBody{\horstMatchExp{{{{l}},\horstMATCHIF{\horstEQ{\horstVarx}{\horstConstructorAppLVal{\horstVarWILDCARD,\horstVarl}}},\horstVarl},{{},\horstOTHERWISE,\horstConstructorAppIllegal{}}}}
\end{horstOperation}
\begin{horstOperation}{raiseTo}
  \horstName{raiseTo}
  \horstVar{x}{\horstTypeLValue}
  \horstVar{l}{\horstTypeFlowLabel}
  \horstReturnType{\horstTypeLValue}
  \horstBody{\horstConstructorAppLVal{\horstOpAppvalueOf{}{\horstVarx},\horstOpAppflub{2}{\horstTUPINIT{\horstVarl,\horstOpApplabelOf{}{\horstVarx}}}}}
\end{horstOperation}
\begin{horstOperation}{join}
  \horstName{join}
  \horstVar{x}{\horstTypeLValue}
  \horstVar{y}{\horstTypeLValue}
  \horstVar{z}{\horstTypeLValue}
  \horstReturnType{\horstTypebool}
  \horstBody{\horstEQ{\horstVarz}{\horstConstructorAppLVal{\horstOpAppvalueOf{}{\horstVarx},\horstCOND{\horstAND{\horstEQ{\horstOpAppflub{2}{\horstTUPINIT{\horstOpApplabelOf{}{\horstVarx},\horstOpApplabelOf{}{\horstVary}}}}{\horstConstructorAppIllegal{}}}{\horstOpAppabsneq{}{\horstOpAppvalueOf{}{\horstVarx},\horstOpAppvalueOf{}{\horstVary}}}}{\horstConstructorAppIllegal{}}{\horstConstructorAppLegal{}}}}}
\end{horstOperation}
\begin{horstOperation}{set}
  \horstName{set}
  \horstParVar{size}{\horstTypeint}
  \horstParVar{idx}{\horstTypeint}
  \horstVar{x}{\horstTypeLValue}
  \horstVar{xs}{\horstTypeHomInit{\horstTypeLValue}{\horstParVarsize}}
  \horstReturnType{\horstTypeHomInit{\horstTypeLValue}{\horstParVarsize}}
  \horstBody{\horstCONCAT{\horstSLICE{\horstVarxs}{}{\horstParVaridx}}{\horstCONCAT{\horstTUPINIT{\horstVarx}}{\horstSLICE{\horstVarxs}{\horstOpAppmin{}{\horstADD{\horstParVaridx}{1},\horstParVarsize}}{\horstParVarsize}}}}
\end{horstOperation}
\begin{horstOperation}{drop}
  \horstName{drop}
  \horstParVar{size}{\horstTypeint}
  \horstParVar{from}{\horstTypeint}
  \horstParVar{n}{\horstTypeint}
  \horstVar{xs}{\horstTypeHomInit{\horstTypeLValue}{\horstParVarsize}}
  \horstReturnType{\horstTypeHomInit{\horstTypeLValue}{\horstSUB{\horstParVarsize}{\horstParVarn}}}
  \horstBody{\horstCONCAT{\horstSLICE{\horstVarxs}{}{\horstParVarfrom}}{\horstSLICE{\horstVarxs}{\horstADD{\horstParVarfrom}{\horstParVarn}}{}}}
\end{horstOperation}
\begin{horstOperation}{reverse}
  \horstName{reverse}
  \horstParVar{sz}{\horstTypeint}
  \horstVar{x}{\horstTypeHomInit{\horstTypeLValue}{\horstParVarsz}}
  \horstReturnType{\horstTypeHomInit{\horstTypeLValue}{\horstParVarsz}}
  \horstBody{\horstCustomSumExp{\horstSelectorFunctionAppinterval{\horstParVari}{0,\horstParVarsz}}{a}{\horstOpAppset{\horstParVarsz,\horstParVari}{\horstACCESS{\horstVarx}{\horstSUB{\horstSUB{\horstParVarsz}{1}}{\horstParVari}},\horstVara}}{\horstVarx}{{i}}}
\end{horstOperation}
\begin{horstOperation}{labelsOf}
  \horstName{labelsOf}
  \horstParVar{sz}{\horstTypeint}
  \horstVar{x}{\horstTypeHomInit{\horstTypeLValue}{\horstParVarsz}}
  \horstReturnType{\horstTypeHomInit{\horstTypeFlowLabel}{\horstParVarsz}}
  \horstBody{\horstCustomSumExp{\horstSelectorFunctionAppinterval{\horstParVari}{0,\horstParVarsz}}{a}{\horstCONCAT{\horstSLICE{\horstVara}{}{\horstParVari}}{\horstCONCAT{\horstTUPINIT{\horstOpApplabelOf{}{\horstACCESS{\horstVarx}{\horstParVari}}}}{\horstSLICE{\horstVara}{\horstOpAppmin{}{\horstADD{\horstParVari}{1},\horstParVarsz}}{\horstParVarsz}}}}{\horstHOMINIT{\horstConstructorAppLegal{}}{\horstParVarsz}}{{i}}}
\end{horstOperation}
\begin{horstOperation}{valuesOf}
  \horstName{valuesOf}
  \horstParVar{sz}{\horstTypeint}
  \horstVar{x}{\horstTypeHomInit{\horstTypeLValue}{\horstParVarsz}}
  \horstReturnType{\horstTypeHomInit{\horstTypeValue}{\horstParVarsz}}
  \horstBody{\horstCustomSumExp{\horstSelectorFunctionAppinterval{\horstParVari}{0,\horstParVarsz}}{a}{\horstCONCAT{\horstSLICE{\horstVara}{}{\horstParVari}}{\horstCONCAT{\horstTUPINIT{\horstOpAppvalueOf{}{\horstACCESS{\horstVarx}{\horstParVari}}}}{\horstSLICE{\horstVara}{\horstOpAppmin{}{\horstADD{\horstParVari}{1},\horstParVarsz}}{\horstParVarsz}}}}{\horstHOMINIT{\horstOpAppmkConst{0}{}}{\horstParVarsz}}{{i}}}
\end{horstOperation}
\begin{horstOperation}{lowEq}
  \horstName{lowEq}
  \horstParVar{sz}{\horstTypeint}
  \horstVar{x}{\horstTypeHomInit{\horstTypeLValue}{\horstParVarsz}}
  \horstVar{y}{\horstTypeHomInit{\horstTypeLValue}{\horstParVarsz}}
  \horstReturnType{\horstTypebool}
  \horstBody{\horstSimpleSumExp{AND}{\horstSelectorFunctionAppinterval{\horstParVari}{0,\horstParVarsz}}{\horstCOND{\horstEQ{\horstOpApplabelOf{}{\horstACCESS{\horstVarx}{\horstParVari}}}{\horstConstructorAppLegal{}}}{\horstEQ{\horstACCESS{\horstVarx}{\horstParVari}}{\horstACCESS{\horstVary}{\horstParVari}}}{\horstEQ{\horstOpApplabelOf{}{\horstACCESS{\horstVarx}{\horstParVari}}}{\horstOpApplabelOf{}{\horstACCESS{\horstVary}{\horstParVari}}}}}{{i}}}
\end{horstOperation}
\begin{horstOperation}{joinTuples}
  \horstName{joinTuples}
  \horstParVar{sz}{\horstTypeint}
  \horstVar{x}{\horstTypeHomInit{\horstTypeLValue}{\horstParVarsz}}
  \horstVar{y}{\horstTypeHomInit{\horstTypeLValue}{\horstParVarsz}}
  \horstVar{res}{\horstTypeHomInit{\horstTypeLValue}{\horstParVarsz}}
  \horstReturnType{\horstTypebool}
  \horstBody{\horstSimpleSumExp{AND}{\horstSelectorFunctionAppinterval{\horstParVari}{0,\horstParVarsz}}{\horstOpAppjoin{}{\horstACCESS{\horstVarx}{\horstParVari},\horstACCESS{\horstVary}{\horstParVari},\horstACCESS{\horstVarres}{\horstParVari}}}{{i}}}
\end{horstOperation}
\begin{horstOperation}{indexOfMem}
  \horstName{indexOfMem}
  \horstVar{memory}{\horstTypeMemory}
  \horstReturnType{\horstTypeValue}
  \horstBody{\horstMatchExp{{{{i}},\horstMATCHIF{\horstEQ{\horstVarmemory}{\horstConstructorAppMem{\horstVari,\horstVarWILDCARD,\horstVarWILDCARD}}},\horstVari},{{},\horstOTHERWISE,\horstOpAppfreeValOrTop{}{}}}}
\end{horstOperation}
\begin{horstOperation}{valueOfMem}
  \horstName{valueOfMem}
  \horstVar{memory}{\horstTypeMemory}
  \horstReturnType{\horstTypeLValue}
  \horstBody{\horstMatchExp{{{{v}},\horstMATCHIF{\horstEQ{\horstVarmemory}{\horstConstructorAppMem{\horstVarWILDCARD,\horstVarv,\horstVarWILDCARD}}},\horstVarv},{{},\horstOTHERWISE,\horstConstructorAppLVal{\horstOpAppfreeValOrTop{}{},\horstConstructorAppIllegal{}}}}}
\end{horstOperation}
\begin{horstOperation}{sizeOfMem}
  \horstName{sizeOfMem}
  \horstVar{memory}{\horstTypeMemory}
  \horstReturnType{\horstTypeLValue}
  \horstBody{\horstMatchExp{{{{n}},\horstMATCHIF{\horstEQ{\horstVarmemory}{\horstConstructorAppMem{\horstVarWILDCARD,\horstVarWILDCARD,\horstVarn}}},\horstVarn},{{},\horstOTHERWISE,\horstConstructorAppLVal{\horstOpAppfreeValOrTop{}{},\horstConstructorAppIllegal{}}}}}
\end{horstOperation}
\begin{horstOperation}{lowEqMem}
  \horstName{lowEqMem}
  \horstVar{x}{\horstTypeMemory}
  \horstVar{y}{\horstTypeMemory}
  \horstReturnType{\horstTypebool}
  \horstBody{\horstMatchExp{{{{iI}{iII}{vI}{vII}{sI}{sII}},\horstMATCHIF{\horstAND{\horstEQ{\horstVarx}{\horstConstructorAppMem{\horstVariI,\horstVarvI,\horstVarsI}}}{\horstEQ{\horstVary}{\horstConstructorAppMem{\horstVariII,\horstVarvII,\horstVarsII}}}},\horstAND{\horstEQ{\horstVariI}{\horstVariII}}{\horstOpApplowEq{2}{\horstTUPINIT{\horstVarvI,\horstVarsI},\horstTUPINIT{\horstVarvII,\horstVarsII}}}},{{},\horstOTHERWISE,\horstFalse}}}
\end{horstOperation}
\begin{horstOperation}{joinMem}
  \horstName{joinMem}
  \horstVar{x}{\horstTypeMemory}
  \horstVar{y}{\horstTypeMemory}
  \horstVar{z}{\horstTypeMemory}
  \horstReturnType{\horstTypebool}
  \horstBody{\horstMatchExp{{{{sIII}{iI}{iII}{iIII}{vI}{vII}{vIII}{sI}{sII}},\horstMATCHIF{\horstAND{\horstAND{\horstEQ{\horstVarx}{\horstConstructorAppMem{\horstVariI,\horstVarvI,\horstVarsI}}}{\horstEQ{\horstVary}{\horstConstructorAppMem{\horstVariII,\horstVarvII,\horstVarsII}}}}{\horstEQ{\horstVarz}{\horstConstructorAppMem{\horstVariIII,\horstVarvIII,\horstVarsIII}}}},\horstAND{\horstAND{\horstAND{\horstEQ{\horstVariI}{\horstVariII}}{\horstEQ{\horstVariI}{\horstVariIII}}}{\horstOpAppjoin{}{\horstVarvI,\horstVarvII,\horstVarvIII}}}{\horstOpAppjoin{}{\horstVarsI,\horstVarsII,\horstVarsIII}}},{{},\horstOTHERWISE,\horstFalse}}}
\end{horstOperation}
\begin{horstOperation}{overApproximateLoopGlobals}
  \horstName{overApproximateLoopGlobals}
  \horstParVar{fid}{\horstTypeint}
  \horstParVar{pc}{\horstTypeint}
  \horstVar{gt}{\horstTypeHomInit{\horstTypeLValue}{\horstOpAppgs{}{}}}
  \horstVar{ngt}{\horstTypeHomInit{\horstTypeLValue}{\horstOpAppgs{}{}}}
  \horstReturnType{\horstTypebool}
  \horstBody{\horstSimpleSumExp{AND}{\horstSelectorFunctionAppinterval{\horstParVari}{0,\horstOpAppgs{}{}},\horstSelectorFunctionAppbitwidthForGlobal{\horstParVarbw}{\horstParVari},\horstSelectorFunctionAppglobalModifiedInLoopBlock{\horstParVarg}{\horstParVarfid,\horstParVarpc,\horstParVari}}{\horstCOND{\horstParVarg}{\horstAND{\horstOpAppisInRange{\horstParVarbw}{\horstOpAppvalueOf{}{\horstACCESS{\horstVarngt}{\horstParVari}}}}{\horstEQ{\horstOpApplabelOf{}{\horstACCESS{\horstVargt}{\horstParVari}}}{\horstOpApplabelOf{}{\horstACCESS{\horstVarngt}{\horstParVari}}}}}{\horstEQ{\horstACCESS{\horstVargt}{\horstParVari}}{\horstACCESS{\horstVarngt}{\horstParVari}}}}{{i}{bw}{g}}}
\end{horstOperation}
\begin{horstOperation}{overApproximateLoopLocals}
  \horstName{overApproximateLoopLocals}
  \horstParVar{fid}{\horstTypeint}
  \horstParVar{pc}{\horstTypeint}
  \horstVar{lt}{\horstTypeHomInit{\horstTypeLValue}{\horstOpAppls{\horstParVarfid}{}}}
  \horstVar{nlt}{\horstTypeHomInit{\horstTypeLValue}{\horstOpAppls{\horstParVarfid}{}}}
  \horstReturnType{\horstTypebool}
  \horstBody{\horstSimpleSumExp{AND}{\horstSelectorFunctionAppinterval{\horstParVari}{0,\horstOpAppls{\horstParVarfid}{}},\horstSelectorFunctionAppbitwidthForLocal{\horstParVarbw}{\horstParVarfid,\horstParVari},\horstSelectorFunctionApplocalModifiedInLoopBlock{\horstParVarl}{\horstParVarfid,\horstParVarpc,\horstParVari}}{\horstCOND{\horstParVarl}{\horstAND{\horstOpAppisInRange{\horstParVarbw}{\horstOpAppvalueOf{}{\horstACCESS{\horstVarnlt}{\horstParVari}}}}{\horstEQ{\horstOpApplabelOf{}{\horstACCESS{\horstVarlt}{\horstParVari}}}{\horstOpApplabelOf{}{\horstACCESS{\horstVarnlt}{\horstParVari}}}}}{\horstEQ{\horstACCESS{\horstVarlt}{\horstParVari}}{\horstACCESS{\horstVarnlt}{\horstParVari}}}}{{i}{bw}{l}}}
\end{horstOperation}
\begin{horstOperation}{overApproximateLoopMemory}
  \horstName{overApproximateLoopMemory}
  \horstParVar{fid}{\horstTypeint}
  \horstParVar{pc}{\horstTypeint}
  \horstVar{mem}{\horstTypeMemory}
  \horstVar{nmem}{\horstTypeMemory}
  \horstReturnType{\horstTypebool}
  \horstBody{\horstAND{\horstEQ{\horstOpAppindexOfMem{}{\horstVarmem}}{\horstOpAppindexOfMem{}{\horstVarnmem}}}{\horstSimpleSumExp{AND}{\horstSelectorFunctionAppmemoryStoreInLoopBlock{\horstParVars}{\horstParVarfid,\horstParVarpc}}{\horstAND{\horstCOND{\horstParVars}{\horstAND{\horstOpAppisInRange{8}{\horstOpAppvalueOf{}{\horstOpAppvalueOfMem{}{\horstVarnmem}}}}{\horstEQ{\horstOpApplabelOf{}{\horstOpAppvalueOfMem{}{\horstVarmem}}}{\horstOpApplabelOf{}{\horstOpAppvalueOfMem{}{\horstVarnmem}}}}}{\horstEQ{\horstOpAppvalueOfMem{}{\horstVarmem}}{\horstOpAppvalueOfMem{}{\horstVarnmem}}}}{\horstSimpleSumExp{AND}{\horstSelectorFunctionAppmemoryGrowInLoopBlock{\horstParVarg}{\horstParVarfid,\horstParVarpc}}{\horstCOND{\horstParVarg}{\horstAND{\horstAND{\horstOpAppileu{64}{\horstOpAppvalueOf{}{\horstOpAppsizeOfMem{}{\horstVarmem}},\horstOpAppvalueOf{}{\horstOpAppsizeOfMem{}{\horstVarnmem}}}}{\horstOpAppileu{64}{\horstOpAppvalueOf{}{\horstOpAppsizeOfMem{}{\horstVarnmem}},\horstOpAppmkConst{\horstOpAppmms{}{}}{}}}}{\horstEQ{\horstOpApplabelOf{}{\horstOpAppsizeOfMem{}{\horstVarmem}}}{\horstOpApplabelOf{}{\horstOpAppsizeOfMem{}{\horstVarnmem}}}}}{\horstEQ{\horstOpAppsizeOfMem{}{\horstVarmem}}{\horstOpAppsizeOfMem{}{\horstVarnmem}}}}{{g}}}}{{s}}}}
\end{horstOperation}
\begin{horstOperation}{overApproximateCallArguments}
  \horstName{overApproximateCallArguments}
  \horstParVar{fid}{\horstTypeint}
  \horstVar{at}{\horstTypeHomInit{\horstTypeLValue}{\horstOpAppas{\horstParVarfid}{}}}
  \horstVar{nat}{\horstTypeHomInit{\horstTypeLValue}{\horstOpAppas{\horstParVarfid}{}}}
  \horstReturnType{\horstTypebool}
  \horstBody{\horstSimpleSumExp{AND}{\horstSelectorFunctionAppinterval{\horstParVari}{0,\horstOpAppas{\horstParVarfid}{}},\horstSelectorFunctionAppbitwidthForArgument{\horstParVarbw}{\horstParVarfid,\horstParVari}}{\horstAND{\horstOpAppisInRange{\horstParVarbw}{\horstOpAppvalueOf{}{\horstACCESS{\horstVarnat}{\horstParVari}}}}{\horstEQ{\horstOpApplabelOf{}{\horstACCESS{\horstVarat}{\horstParVari}}}{\horstOpApplabelOf{}{\horstACCESS{\horstVarnat}{\horstParVari}}}}}{{i}{bw}}}
\end{horstOperation}
\begin{horstOperation}{overApproximateCallGlobals}
  \horstName{overApproximateCallGlobals}
  \horstParVar{fid}{\horstTypeint}
  \horstVar{gt}{\horstTypeHomInit{\horstTypeLValue}{\horstOpAppgs{}{}}}
  \horstVar{ngt}{\horstTypeHomInit{\horstTypeLValue}{\horstOpAppgs{}{}}}
  \horstReturnType{\horstTypebool}
  \horstBody{\horstSimpleSumExp{AND}{\horstSelectorFunctionAppinterval{\horstParVari}{0,\horstOpAppgs{}{}},\horstSelectorFunctionAppbitwidthForGlobal{\horstParVarbw}{\horstParVari}}{\horstAND{\horstOpAppisInRange{\horstParVarbw}{\horstOpAppvalueOf{}{\horstACCESS{\horstVarngt}{\horstParVari}}}}{\horstEQ{\horstOpApplabelOf{}{\horstACCESS{\horstVargt}{\horstParVari}}}{\horstOpApplabelOf{}{\horstACCESS{\horstVarngt}{\horstParVari}}}}}{{i}{bw}}}
\end{horstOperation}
\begin{horstOperation}{overApproximateCallMemory}
  \horstName{overApproximateCallMemory}
  \horstParVar{fid}{\horstTypeint}
  \horstVar{mem}{\horstTypeMemory}
  \horstVar{nmem}{\horstTypeMemory}
  \horstReturnType{\horstTypebool}
  \horstBody{\horstAND{\horstAND{\horstAND{\horstAND{\horstAND{\horstEQ{\horstOpAppindexOfMem{}{\horstVarmem}}{\horstOpAppindexOfMem{}{\horstVarnmem}}}{\horstOpAppisInRange{8}{\horstOpAppvalueOf{}{\horstOpAppvalueOfMem{}{\horstVarnmem}}}}}{\horstEQ{\horstOpApplabelOf{}{\horstOpAppvalueOfMem{}{\horstVarmem}}}{\horstOpApplabelOf{}{\horstOpAppvalueOfMem{}{\horstVarnmem}}}}}{\horstOpAppileu{64}{\horstOpAppvalueOf{}{\horstOpAppsizeOfMem{}{\horstVarmem}},\horstOpAppvalueOf{}{\horstOpAppsizeOfMem{}{\horstVarnmem}}}}}{\horstOpAppileu{64}{\horstOpAppvalueOf{}{\horstOpAppsizeOfMem{}{\horstVarnmem}},\horstOpAppmkConst{\horstOpAppmms{}{}}{}}}}{\horstEQ{\horstOpApplabelOf{}{\horstOpAppsizeOfMem{}{\horstVarmem}}}{\horstOpApplabelOf{}{\horstOpAppsizeOfMem{}{\horstVarnmem}}}}}
\end{horstOperation}
\begin{horstOperation}{labelledUnOp}
  \horstName{labelledUnOp}
  \horstParVar{op}{\horstTypeint}
  \horstVar{a}{\horstTypeLValue}
  \horstReturnType{\horstTypeLValue}
  \horstBody{\horstMatchExp{{{{i}{l}},\horstMATCHIF{\horstEQ{\horstVara}{\horstConstructorAppLVal{\horstVari,\horstVarl}}},\horstConstructorAppLVal{\horstOpAppunOp{\horstParVarop}{\horstVari},\horstVarl}},{{},\horstOTHERWISE,\horstConstructorAppLVal{\horstOpAppfreeValOrTop{}{},\horstConstructorAppIllegal{}}}}}
\end{horstOperation}
\begin{horstOperation}{labelledBinOp}
  \horstName{labelledBinOp}
  \horstParVar{op}{\horstTypeint}
  \horstVar{a}{\horstTypeLValue}
  \horstVar{b}{\horstTypeLValue}
  \horstReturnType{\horstTypeLValue}
  \horstBody{\horstMatchExp{{{{lI}{lII}{iI}{iII}},\horstMATCHIF{\horstAND{\horstEQ{\horstVara}{\horstConstructorAppLVal{\horstVariI,\horstVarlI}}}{\horstEQ{\horstVarb}{\horstConstructorAppLVal{\horstVariII,\horstVarlII}}}},\horstConstructorAppLVal{\horstOpAppbinOp{\horstParVarop}{\horstVariI,\horstVariII},\horstOpAppflub{2}{\horstTUPINIT{\horstVarlI,\horstVarlII}}}},{{},\horstOTHERWISE,\horstConstructorAppLVal{\horstOpAppfreeValOrTop{}{},\horstConstructorAppIllegal{}}}}}
\end{horstOperation}
\begin{horstOperation}{labelledCvtOp}
  \horstName{labelledCvtOp}
  \horstParVar{op}{\horstTypeint}
  \horstVar{a}{\horstTypeLValue}
  \horstReturnType{\horstTypeLValue}
  \horstBody{\horstMatchExp{{{{i}{l}},\horstMATCHIF{\horstEQ{\horstVara}{\horstConstructorAppLVal{\horstVari,\horstVarl}}},\horstConstructorAppLVal{\horstOpAppcvtOp{\horstParVarop}{\horstVari},\horstVarl}},{{},\horstOTHERWISE,\horstConstructorAppLVal{\horstOpAppfreeValOrTop{}{},\horstConstructorAppIllegal{}}}}}
\end{horstOperation}
\begin{horstOperation}{load}
  \horstName{load}
  \horstParVar{sz}{\horstTypeint}
  \horstVar{x}{\horstTypeHomInit{\horstTypeValue}{\horstParVarsz}}
  \horstReturnType{\horstTypeValue}
  \horstBody{\horstCustomSumExp{\horstSelectorFunctionAppinterval{\horstParVari}{1,\horstParVarsz}}{a}{\horstConstructorAppVal{\horstOpApploadValueAux{\horstParVari}{\horstOpAppbase{}{\horstVara},\horstOpAppbase{}{\horstACCESS{\horstVarx}{\horstParVari}}}}}{\horstACCESS{\horstVarx}{0}}{{i}}}
\end{horstOperation}
\begin{horstType}{int}
  \horstName{int}
\end{horstType}
\begin{horstType}{bool}
  \horstName{bool}
\end{horstType}
\begin{horstType}{BVLXIV}
  \horstName{BV64}
\end{horstType}
\begin{horstType}{Value}
  \horstName{Value}
  \begin{horstConstructor}{Val}
    \horstName{Val}
    \horstTypeParameter{\horstTypeBVLXIV}
  \end{horstConstructor}
\end{horstType}
\begin{horstType}{Label}
  \horstName{Label}
  \begin{horstConstructor}{Lbl}
    \horstName{Lbl}
    \horstTypeParameter{\horstTypebool}
    \horstTypeParameter{\horstTypebool}
  \end{horstConstructor}
\end{horstType}
\begin{horstType}{FlowLabel}
  \horstName{FlowLabel}
  \begin{horstConstructor}{Legal}
    \horstName{Legal}
  \end{horstConstructor}
  \begin{horstConstructor}{Illegal}
    \horstName{Illegal}
  \end{horstConstructor}
\end{horstType}
\begin{horstType}{Context}
  \horstName{Context}
  \begin{horstConstructor}{Ctx}
    \horstName{Ctx}
    \horstTypeParameter{\horstTypeLabel}
    \horstTypeParameter{\horstTypeFlowLabel}
    \horstTypeParameter{\horstTypeint}
  \end{horstConstructor}
\end{horstType}
\begin{horstType}{TablePrecision}
  \horstName{TablePrecision}
  \begin{horstConstructor}{TblPrecise}
    \horstName{TblPrecise}
  \end{horstConstructor}
  \begin{horstConstructor}{TblImprecise}
    \horstName{TblImprecise}
  \end{horstConstructor}
\end{horstType}
\begin{horstType}{Table}
  \horstName{Table}
  \begin{horstConstructor}{Tbl}
    \horstName{Tbl}
    \horstTypeParameter{\horstTypeTablePrecision}
    \horstTypeParameter{\horstTypeFlowLabel}
  \end{horstConstructor}
\end{horstType}
\begin{horstType}{LValue}
  \horstName{LValue}
  \begin{horstConstructor}{LVal}
    \horstName{LVal}
    \horstTypeParameter{\horstTypeValue}
    \horstTypeParameter{\horstTypeFlowLabel}
  \end{horstConstructor}
\end{horstType}
\begin{horstType}{Memory}
  \horstName{Memory}
  \begin{horstConstructor}{Mem}
    \horstName{Mem}
    \horstTypeParameter{\horstTypeValue}
    \horstTypeParameter{\horstTypeLValue}
    \horstTypeParameter{\horstTypeLValue}
  \end{horstConstructor}
\end{horstType}
\begin{horstType}{MaybeValue}
  \horstName{MaybeValue}
  \begin{horstConstructor}{JustV}
    \horstName{JustV}
    \horstTypeParameter{\horstTypeValue}
  \end{horstConstructor}
  \begin{horstConstructor}{NothingV}
    \horstName{NothingV}
  \end{horstConstructor}
\end{horstType}
\begin{horstType}{MaybeMemory}
  \horstName{MaybeMemory}
  \begin{horstConstructor}{JustM}
    \horstName{JustM}
    \horstTypeParameter{\horstTypeMemory}
  \end{horstConstructor}
  \begin{horstConstructor}{NothingM}
    \horstName{NothingM}
  \end{horstConstructor}
\end{horstType}
\horstConstant{IXXXIINE}{I32NE}{71}
\horstConstant{ILXIVLTU}{I64LTU}{84}
\horstConstant{FXXXIIMUL}{F32MUL}{148}
\horstConstant{THROW}{THROW}{8}
\horstConstant{ILXIVLTS}{I64LTS}{83}
\horstConstant{ILXIVROTR}{I64ROTR}{138}
\horstConstant{FLXIVLE}{F64LE}{101}
\horstConstant{FXXXIITRUNC}{F32TRUNC}{143}
\horstConstant{IXXXIIREMS}{I32REMS}{111}
\horstConstant{FLXIVLOADMEM}{F64LOADMEM}{43}
\horstConstant{IXXXIIREMU}{I32REMU}{112}
\horstConstant{FXXXIINEAREST}{F32NEAREST}{144}
\horstConstant{CALLREF}{CALLREF}{20}
\horstConstant{IXXXIISHRS}{I32SHRS}{117}
\horstConstant{ILXIVAND}{I64AND}{131}
\horstConstant{FLXIVFLOOR}{F64FLOOR}{156}
\horstConstant{FLXIVSTOREMEM}{F64STOREMEM}{57}
\horstConstant{ILXIVPOPCNT}{I64POPCNT}{123}
\horstConstant{IXXXIILOADMEMVIIIS}{I32LOADMEM8S}{44}
\horstConstant{REFISNULL}{REFISNULL}{209}
\horstConstant{IXXXIILOADMEMVIIIU}{I32LOADMEM8U}{45}
\horstConstant{ILXIVCTZ}{I64CTZ}{122}
\horstConstant{ILXIVREINTERPRETFLXIV}{I64REINTERPRETF64}{189}
\horstConstant{IXXXIIROTL}{I32ROTL}{119}
\horstConstant{IXXXIISEXTENDIVIII}{I32SEXTENDI8}{192}
\horstConstant{RETHROW}{RETHROW}{9}
\horstConstant{ILXIVSHRS}{I64SHRS}{135}
\horstConstant{ILXIVSHRU}{I64SHRU}{136}
\horstConstant{FLXIVCOPYSIGN}{F64COPYSIGN}{166}
\horstConstant{BR}{BR}{12}
\horstConstant{IXXXIISHRU}{I32SHRU}{118}
\horstConstant{IXXXIISHL}{I32SHL}{116}
\horstConstant{ILXIVSUB}{I64SUB}{125}
\horstConstant{IXXXIIADD}{I32ADD}{106}
\horstConstant{FLXIVCONVERTILXIVU}{F64CONVERTI64U}{186}
\horstConstant{FLXIVCONVERTILXIVS}{F64CONVERTI64S}{185}
\horstConstant{ILXIVIOR}{I64IOR}{132}
\horstConstant{RETURNCALLREF}{RETURNCALLREF}{21}
\horstConstant{FXXXIINE}{F32NE}{92}
\horstConstant{ILXIVLES}{I64LES}{87}
\horstConstant{FXXXIINEG}{F32NEG}{140}
\horstConstant{ILXIVROTL}{I64ROTL}{137}
\horstConstant{ILXIVLEU}{I64LEU}{88}
\horstConstant{IXXXIIROTR}{I32ROTR}{120}
\horstConstant{UNREACHABLE}{UNREACHABLE}{0}
\horstConstant{ILXIVSTOREMEM}{I64STOREMEM}{55}
\horstConstant{FLXIVGT}{F64GT}{100}
\horstConstant{FXXXIISUB}{F32SUB}{147}
\horstConstant{LOOP}{LOOP}{3}
\horstConstant{GLOBALSET}{GLOBALSET}{36}
\horstConstant{ILXIVEXTENDIXXXIIU}{I64EXTENDI32U}{173}
\horstConstant{CATCHALL}{CATCHALL}{25}
\horstConstant{ILXIVEXTENDIXXXIIS}{I64EXTENDI32S}{172}
\horstConstant{ILXIVSEXTENDIXXXII}{I64SEXTENDI32}{196}
\horstConstant{ILXIVCONST}{I64CONST}{66}
\horstConstant{FLXIVCONST}{F64CONST}{68}
\horstConstant{FLXIVGE}{F64GE}{102}
\horstConstant{CALLINDIRECT}{CALLINDIRECT}{17}
\horstConstant{TABLESET}{TABLESET}{38}
\horstConstant{CALLFUNCTION}{CALLFUNCTION}{16}
\horstConstant{FLXIVSUB}{F64SUB}{161}
\horstConstant{FXXXIILT}{F32LT}{93}
\horstConstant{FLXIVCEIL}{F64CEIL}{155}
\horstConstant{ILXIVNE}{I64NE}{82}
\horstConstant{ILXIVEQZ}{I64EQZ}{80}
\horstConstant{IXXXIIWRAPILXIV}{I32WRAPI64}{167}
\horstConstant{FLXIVMIN}{F64MIN}{164}
\horstConstant{FXXXIILE}{F32LE}{95}
\horstConstant{ILXIVGES}{I64GES}{89}
\horstConstant{ILXIVREMU}{I64REMU}{130}
\horstConstant{NOP}{NOP}{1}
\horstConstant{ILXIVREMS}{I64REMS}{129}
\horstConstant{ILXIVLOADMEMXXXIIU}{I64LOADMEM32U}{53}
\horstConstant{ILXIVTRUNCFXXXIIS}{I64TRUNCF32S}{174}
\horstConstant{RETURNCALLINDIRECT}{RETURNCALLINDIRECT}{19}
\horstConstant{ILXIVLOADMEMXXXIIS}{I64LOADMEM32S}{52}
\horstConstant{ILXIVGEU}{I64GEU}{90}
\horstConstant{FXXXIIMIN}{F32MIN}{150}
\horstConstant{IXXXIIGES}{I32GES}{78}
\horstConstant{IXXXIIGEU}{I32GEU}{79}
\horstConstant{ILXIVTRUNCFXXXIIU}{I64TRUNCF32U}{175}
\horstConstant{IXXXIIREINTERPRETFXXXII}{I32REINTERPRETF32}{188}
\horstConstant{IXXXIIPOPCNT}{I32POPCNT}{105}
\horstConstant{FLXIVSQRT}{F64SQRT}{159}
\horstConstant{END}{END}{11}
\horstConstant{IXXXIITRUNCUFLXIV}{I32TRUNCUF64}{171}
\horstConstant{IXXXIIEQZ}{I32EQZ}{69}
\horstConstant{FLXIVEQ}{F64EQ}{97}
\horstConstant{FXXXIIDEMOTEFLXIV}{F32DEMOTEF64}{182}
\horstConstant{FXXXIILOADMEM}{F32LOADMEM}{42}
\horstConstant{FLXIVABS}{F64ABS}{153}
\horstConstant{IXXXIILOADMEM}{I32LOADMEM}{40}
\horstConstant{DROP}{DROP}{26}
\horstConstant{REFFUNC}{REFFUNC}{210}
\horstConstant{RETURN}{RETURN}{15}
\horstConstant{IXXXIICLZ}{I32CLZ}{103}
\horstConstant{IXXXIIGTU}{I32GTU}{75}
\horstConstant{ILXIVMUL}{I64MUL}{126}
\horstConstant{FLXIVCONVERTIXXXIIU}{F64CONVERTI32U}{184}
\horstConstant{FLXIVCONVERTIXXXIIS}{F64CONVERTI32S}{183}
\horstConstant{LOCALSET}{LOCALSET}{33}
\horstConstant{IXXXIIGTS}{I32GTS}{74}
\horstConstant{FLXIVDIV}{F64DIV}{163}
\horstConstant{FLXIVADD}{F64ADD}{160}
\horstConstant{DELEGATE}{DELEGATE}{24}
\horstConstant{FXXXIISTOREMEM}{F32STOREMEM}{56}
\horstConstant{REFNULL}{REFNULL}{208}
\horstConstant{IXXXIIEQ}{I32EQ}{70}
\horstConstant{FXXXIICONVERTILXIVS}{F32CONVERTI64S}{180}
\horstConstant{FLXIVREINTERPRETILXIV}{F64REINTERPRETI64}{191}
\horstConstant{ILXIVSEXTENDIVIII}{I64SEXTENDI8}{194}
\horstConstant{FXXXIICONVERTILXIVU}{F32CONVERTI64U}{181}
\horstConstant{ELSE}{ELSE}{5}
\horstConstant{IXXXIILTS}{I32LTS}{72}
\horstConstant{IF}{IF}{4}
\horstConstant{IXXXIILTU}{I32LTU}{73}
\horstConstant{FXXXIIGT}{F32GT}{94}
\horstConstant{FLXIVMUL}{F64MUL}{162}
\horstConstant{FXXXIIFLOOR}{F32FLOOR}{142}
\horstConstant{FXXXIISQRT}{F32SQRT}{145}
\horstConstant{RETURNCALL}{RETURNCALL}{18}
\horstConstant{FLXIVTRUNC}{F64TRUNC}{157}
\horstConstant{FXXXIIGE}{F32GE}{96}
\horstConstant{ILXIVTRUNCFLXIVU}{I64TRUNCF64U}{177}
\horstConstant{ILXIVTRUNCFLXIVS}{I64TRUNCF64S}{176}
\horstConstant{IXXXIIAND}{I32AND}{113}
\horstConstant{LOCALTEE}{LOCALTEE}{34}
\horstConstant{ILXIVSEXTENDIXVI}{I64SEXTENDI16}{195}
\horstConstant{ILXIVGTU}{I64GTU}{86}
\horstConstant{IXXXIIMUL}{I32MUL}{108}
\horstConstant{REFASNONNULL}{REFASNONNULL}{211}
\horstConstant{ILXIVCLZ}{I64CLZ}{121}
\horstConstant{FXXXIIABS}{F32ABS}{139}
\horstConstant{CATCH}{CATCH}{7}
\horstConstant{BRTABLE}{BRTABLE}{14}
\horstConstant{IXXXIITRUNCUFXXXII}{I32TRUNCUF32}{169}
\horstConstant{ILXIVSTOREMEMXXXII}{I64STOREMEM32}{62}
\horstConstant{FXXXIIEQ}{F32EQ}{91}
\horstConstant{ILXIVSTOREMEMXVI}{I64STOREMEM16}{61}
\horstConstant{FXXXIIADD}{F32ADD}{146}
\horstConstant{SELECT}{SELECT}{27}
\horstConstant{IXXXIISTOREMEMVIII}{I32STOREMEM8}{58}
\horstConstant{FXXXIICEIL}{F32CEIL}{141}
\horstConstant{BRIF}{BRIF}{13}
\horstConstant{FXXXIIDIV}{F32DIV}{149}
\horstConstant{ILXIVGTS}{I64GTS}{85}
\horstConstant{ILXIVDIVS}{I64DIVS}{127}
\horstConstant{IXXXIISEXTENDIXVI}{I32SEXTENDI16}{193}
\horstConstant{ILXIVDIVU}{I64DIVU}{128}
\horstConstant{GLOBALGET}{GLOBALGET}{35}
\horstConstant{FXXXIICONVERTIXXXIIU}{F32CONVERTI32U}{179}
\horstConstant{IXXXIISTOREMEM}{I32STOREMEM}{54}
\horstConstant{FXXXIICONVERTIXXXIIS}{F32CONVERTI32S}{178}
\horstConstant{TABLEGET}{TABLEGET}{37}
\horstConstant{IXXXIITRUNCSFXXXII}{I32TRUNCSF32}{168}
\horstConstant{ILXIVLOADMEM}{I64LOADMEM}{41}
\horstConstant{TRY}{TRY}{6}
\horstConstant{FXXXIIMAX}{F32MAX}{151}
\horstConstant{BLOCK}{BLOCK}{2}
\horstConstant{FXXXIICONST}{F32CONST}{67}
\horstConstant{IXXXIIXOR}{I32XOR}{115}
\horstConstant{ILXIVEQ}{I64EQ}{81}
\horstConstant{ILXIVSTOREMEMVIII}{I64STOREMEM8}{60}
\horstConstant{ILXIVLOADMEMXVIU}{I64LOADMEM16U}{51}
\horstConstant{ILXIVXOR}{I64XOR}{133}
\horstConstant{ILXIVLOADMEMXVIS}{I64LOADMEM16S}{50}
\horstConstant{IXXXIISUB}{I32SUB}{107}
\horstConstant{FXXXIICOPYSIGN}{F32COPYSIGN}{152}
\horstConstant{FLXIVMAX}{F64MAX}{165}
\horstConstant{IXXXIILOADMEMXVIS}{I32LOADMEM16S}{46}
\horstConstant{MEMORYSIZE}{MEMORYSIZE}{63}
\horstConstant{FLXIVNE}{F64NE}{98}
\horstConstant{IXXXIILOADMEMXVIU}{I32LOADMEM16U}{47}
\horstConstant{MEMORYGROW}{MEMORYGROW}{64}
\horstConstant{ILXIVSHL}{I64SHL}{134}
\horstConstant{IXXXIICTZ}{I32CTZ}{104}
\horstConstant{FLXIVNEAREST}{F64NEAREST}{158}
\horstConstant{IXXXIICONST}{I32CONST}{65}
\horstConstant{IXXXIIDIVU}{I32DIVU}{110}
\horstConstant{FLXIVPROMOTEFXXXII}{F64PROMOTEF32}{187}
\horstConstant{IXXXIIDIVS}{I32DIVS}{109}
\horstConstant{IXXXIISTOREMEMXVI}{I32STOREMEM16}{59}
\horstConstant{FLXIVNEG}{F64NEG}{154}
\horstConstant{IXXXIILES}{I32LES}{76}
\horstConstant{IXXXIILEU}{I32LEU}{77}
\horstConstant{IXXXIITRUNCSFLXIV}{I32TRUNCSF64}{170}
\horstConstant{ILXIVLOADMEMVIIIS}{I64LOADMEM8S}{48}
\horstConstant{ILXIVADD}{I64ADD}{124}
\horstConstant{LOCALGET}{LOCALGET}{32}
\horstConstant{IXXXIIIOR}{I32IOR}{114}
\horstConstant{FXXXIIREINTERPRETIXXXII}{F32REINTERPRETI32}{190}
\horstConstant{ILXIVLOADMEMVIIIU}{I64LOADMEM8U}{49}
\horstConstant{FLXIVLT}{F64LT}{99}

\begin{horstRule}{brTableDefaultRule}
  \horstParVar{fid}{\horstTypeint}
  \horstParVar{pc}{\horstTypeint}
  \horstParVar{sz}{\horstTypeint}
  \horstParVar{n}{\horstTypeint}
  \horstParVar{br}{\horstTypeint}
  \horstSelectorFunctionInvocation{\horstSelectorFunctionAppfunctionIds{\horstParVarfid}{},\horstSelectorFunctionApppcsForFunctionIdAndOpcode{\horstParVarpc}{\horstParVarfid,\horstConstBRTABLE},\horstSelectorFunctionAppsizeOfBreakTable{\horstParVarsz}{\horstParVarfid,\horstParVarpc},\horstSelectorFunctionAppgetAmountOfReturnValuesInBlock{\horstParVarn}{\horstParVarfid,\horstParVarpc},\horstSelectorFunctionAppbreakTableDestinations{\horstParVarbr}{\horstParVarfid,\horstParVarpc,\horstSUB{\horstParVarsz}{1}}}
  \begin{horstClause}
    \horstFreeVar{memN}{\horstTypeMemory}
    \horstFreeVar{atN}{\horstTypeHomInit{\horstTypeLValue}{\horstOpAppas{\horstParVarfid}{}}}
    \horstFreeVar{st}{\horstTypeHomInit{\horstTypeLValue}{\horstSUB{\horstOpAppss{\horstParVarfid,\horstParVarpc}{}}{1}}}
    \horstFreeVar{tbl}{\horstTypeTable}
    \horstFreeVar{ctx}{\horstTypeContext}
    \horstFreeVar{lt}{\horstTypeHomInit{\horstTypeLValue}{\horstOpAppls{\horstParVarfid}{}}}
    \horstFreeVar{mem}{\horstTypeMemory}
    \horstFreeVar{x}{\horstTypeLValue}
    \horstFreeVar{gtN}{\horstTypeHomInit{\horstTypeLValue}{\horstOpAppgs{}{}}}
    \horstFreeVar{gt}{\horstTypeHomInit{\horstTypeLValue}{\horstOpAppgs{}{}}}
    \horstPremise{\horstPredAppMState{\horstParVarfid,\horstParVarpc}{\horstFreeVarctx,\horstCONS{\horstFreeVarx}{\horstFreeVarst},\horstFreeVargt,\horstFreeVarlt,\horstFreeVarmem,\horstFreeVartbl,\horstFreeVaratN,\horstFreeVargtN,\horstFreeVarmemN}}
    \horstPremise{\horstOpAppabsge{}{\horstOpAppvalueOf{}{\horstFreeVarx},\horstOpAppmkConst{\horstSUB{\horstParVarsz}{1}}{}}}
    \horstConclusion{\horstPredAppMState{\horstParVarfid,\horstParVarbr}{\horstOpAppraiseCtxTo{\horstParVarpc}{\horstFreeVarctx,\horstOpApplabelOf{}{\horstFreeVarx}},\horstOpAppdrop{\horstSUB{\horstOpAppss{\horstParVarfid,\horstParVarpc}{}}{1},\horstParVarn,\horstSUB{\horstSUB{\horstOpAppss{\horstParVarfid,\horstParVarpc}{}}{1}}{\horstOpAppss{\horstParVarfid,\horstParVarbr}{}}}{\horstFreeVarst},\horstFreeVargt,\horstFreeVarlt,\horstFreeVarmem,\horstFreeVartbl,\horstFreeVaratN,\horstFreeVargtN,\horstFreeVarmemN}}
  \end{horstClause}
  \begin{horstClause}
    \horstFreeVar{memN}{\horstTypeMemory}
    \horstFreeVar{atN}{\horstTypeHomInit{\horstTypeLValue}{\horstOpAppas{\horstParVarfid}{}}}
    \horstFreeVar{st}{\horstTypeHomInit{\horstTypeLValue}{\horstSUB{\horstOpAppss{\horstParVarfid,\horstParVarpc}{}}{1}}}
    \horstFreeVar{p}{\horstTypeLabel}
    \horstFreeVar{tbl}{\horstTypeTable}
    \horstFreeVar{ctx}{\horstTypeContext}
    \horstFreeVar{lt}{\horstTypeHomInit{\horstTypeLValue}{\horstOpAppls{\horstParVarfid}{}}}
    \horstFreeVar{mem}{\horstTypeMemory}
    \horstFreeVar{x}{\horstTypeLValue}
    \horstFreeVar{gtN}{\horstTypeHomInit{\horstTypeLValue}{\horstOpAppgs{}{}}}
    \horstFreeVar{gt}{\horstTypeHomInit{\horstTypeLValue}{\horstOpAppgs{}{}}}
    \horstFreeVar{from}{\horstTypeint}
    \horstPremise{\horstPredAppMState{\horstParVarfid,\horstParVarpc}{\horstFreeVarctx,\horstCONS{\horstFreeVarx}{\horstFreeVarst},\horstFreeVargt,\horstFreeVarlt,\horstFreeVarmem,\horstFreeVartbl,\horstFreeVaratN,\horstFreeVargtN,\horstFreeVarmemN}}
    \horstPremise{\horstOpAppabsge{}{\horstOpAppvalueOf{}{\horstFreeVarx},\horstOpAppmkConst{\horstSUB{\horstParVarsz}{1}}{}}}
    \horstPremise{\horstEQ{\horstOpAppraiseCtxTo{\horstParVarpc}{\horstFreeVarctx,\horstOpApplabelOf{}{\horstFreeVarx}}}{\horstConstructorAppCtx{\horstFreeVarp,\horstConstructorAppIllegal{},\horstFreeVarfrom}}}
    \horstConclusion{\horstPredAppScopeExtend{\horstParVarfid}{\horstFreeVarfrom,\horstOpAppmax{}{\horstParVarpc,\horstParVarbr}}}
  \end{horstClause}
\end{horstRule}
\begin{horstRule}{unOpRule}
  \horstParVar{fid}{\horstTypeint}
  \horstParVar{op}{\horstTypeint}
  \horstParVar{pc}{\horstTypeint}
  \horstSelectorFunctionInvocation{\horstSelectorFunctionAppfunctionIds{\horstParVarfid}{},\horstSelectorFunctionAppunOps{\horstParVarop}{},\horstSelectorFunctionApppcsForFunctionIdAndOpcode{\horstParVarpc}{\horstParVarfid,\horstParVarop}}
  \begin{horstClause}
    \horstFreeVar{memN}{\horstTypeMemory}
    \horstFreeVar{atN}{\horstTypeHomInit{\horstTypeLValue}{\horstOpAppas{\horstParVarfid}{}}}
    \horstFreeVar{st}{\horstTypeHomInit{\horstTypeLValue}{\horstSUB{\horstOpAppss{\horstParVarfid,\horstParVarpc}{}}{1}}}
    \horstFreeVar{tbl}{\horstTypeTable}
    \horstFreeVar{ctx}{\horstTypeContext}
    \horstFreeVar{lt}{\horstTypeHomInit{\horstTypeLValue}{\horstOpAppls{\horstParVarfid}{}}}
    \horstFreeVar{mem}{\horstTypeMemory}
    \horstFreeVar{x}{\horstTypeLValue}
    \horstFreeVar{gtN}{\horstTypeHomInit{\horstTypeLValue}{\horstOpAppgs{}{}}}
    \horstFreeVar{gt}{\horstTypeHomInit{\horstTypeLValue}{\horstOpAppgs{}{}}}
    \horstPremise{\horstPredAppMState{\horstParVarfid,\horstParVarpc}{\horstFreeVarctx,\horstCONS{\horstFreeVarx}{\horstFreeVarst},\horstFreeVargt,\horstFreeVarlt,\horstFreeVarmem,\horstFreeVartbl,\horstFreeVaratN,\horstFreeVargtN,\horstFreeVarmemN}}
    \horstConclusion{\horstPredAppMState{\horstParVarfid,\horstADD{\horstParVarpc}{1}}{\horstFreeVarctx,\horstCONS{\horstOpAppraiseTo{}{\horstOpApplabelledUnOp{\horstParVarop}{\horstFreeVarx},\horstOpApplabelOfCtx{}{\horstFreeVarctx}}}{\horstFreeVarst},\horstFreeVargt,\horstFreeVarlt,\horstFreeVarmem,\horstFreeVartbl,\horstFreeVaratN,\horstFreeVargtN,\horstFreeVarmemN}}
  \end{horstClause}
\end{horstRule}
\begin{horstRule}{loadRule}
  \horstParVar{fid}{\horstTypeint}
  \horstParVar{op}{\horstTypeint}
  \horstParVar{pc}{\horstTypeint}
  \horstParVar{offset}{\horstTypeint}
  \horstSelectorFunctionInvocation{\horstSelectorFunctionAppfunctionIds{\horstParVarfid}{},\horstSelectorFunctionApploadOps{\horstParVarop}{},\horstSelectorFunctionApppcsForFunctionIdAndOpcode{\horstParVarpc}{\horstParVarfid,\horstParVarop},\horstSelectorFunctionAppmemoryOffsetForFunctionIdAndPc{\horstParVaroffset}{\horstParVarfid,\horstParVarpc}}
  \begin{horstClause}
    \horstFreeVar{atN}{\horstTypeHomInit{\horstTypeLValue}{\horstOpAppas{\horstParVarfid}{}}}
    \horstFreeVar{st}{\horstTypeHomInit{\horstTypeLValue}{\horstSUB{\horstOpAppss{\horstParVarfid,\horstParVarpc}{}}{1}}}
    \horstFreeVar{tbl}{\horstTypeTable}
    \horstFreeVar{memNN}{\horstTypeMemory}
    \horstFreeVar{lt}{\horstTypeHomInit{\horstTypeLValue}{\horstOpAppls{\horstParVarfid}{}}}
    \horstFreeVar{i}{\horstTypeValue}
    \horstFreeVar{gt}{\horstTypeHomInit{\horstTypeLValue}{\horstOpAppgs{}{}}}
    \horstFreeVar{l}{\horstTypeFlowLabel}
    \horstFreeVar{size}{\horstTypeLValue}
    \horstFreeVar{memN}{\horstTypeMemory}
    \horstFreeVar{u}{\horstTypeValue}
    \horstFreeVar{ctx}{\horstTypeContext}
    \horstFreeVar{v}{\horstTypeLValue}
    \horstFreeVar{w}{\horstTypeValue}
    \horstFreeVar{x}{\horstTypeLValue}
    \horstFreeVar{gtN}{\horstTypeHomInit{\horstTypeLValue}{\horstOpAppgs{}{}}}
    \horstFreeVar{vs}{\horstTypeHomInit{\horstTypeLValue}{1}}
    \horstPremise{\horstEQ{\horstOpAppmemOpTbw{\horstParVarop}{}}{8}}
    \horstPremise{\horstOpAppiltu{64}{\horstOpAppiadd{64}{\horstOpAppvalueOf{}{\horstFreeVarx},\horstOpAppmkConst{\horstADD{\horstParVaroffset}{\horstDIV{\horstOpAppmemOpTbw{\horstParVarop}{}}{8}}}{}},\horstOpAppishl{64}{\horstOpAppvalueOf{}{\horstFreeVarsize},\horstOpAppmkConst{16}{}}}}
    \horstPremise{\horstPredAppMState{\horstParVarfid,\horstParVarpc}{\horstFreeVarctx,\horstCONS{\horstFreeVarx}{\horstFreeVarst},\horstFreeVargt,\horstFreeVarlt,\horstConstructorAppMem{\horstFreeVari,\horstFreeVarv,\horstFreeVarsize},\horstFreeVartbl,\horstFreeVaratN,\horstFreeVargtN,\horstFreeVarmemN}}
    \horstPremise{\horstSimpleSumExp{AND}{\horstSelectorFunctionAppinterval{\horstParVari}{0,1}}{\horstOR{\horstNEQ{\horstOpAppiadd{64}{\horstOpAppvalueOf{}{\horstFreeVarx},\horstOpAppmkConst{\horstADD{\horstParVaroffset}{\horstParVari}}{}}}{\horstFreeVari}}{\horstEQ{\horstFreeVarv}{\horstACCESS{\horstFreeVarvs}{\horstParVari}}}}{{i}}}
    \horstPremise{\horstPredAppMState{\horstParVarfid,\horstParVarpc}{\horstFreeVarctx,\horstCONS{\horstFreeVarx}{\horstFreeVarst},\horstFreeVargt,\horstFreeVarlt,\horstConstructorAppMem{\horstOpAppiadd{64}{\horstOpAppvalueOf{}{\horstFreeVarx},\horstOpAppmkConst{\horstADD{\horstParVaroffset}{0}}{}},\horstACCESS{\horstFreeVarvs}{0},\horstFreeVarsize},\horstFreeVartbl,\horstFreeVaratN,\horstFreeVargtN,\horstFreeVarmemNN}}
    \horstPremise{\horstEQ{\horstFreeVaru}{\horstOpAppload{1}{\horstOpAppvaluesOf{1}{\horstFreeVarvs}}}}
    \horstPremise{\horstEQ{\horstFreeVarl}{\horstOpAppflub{1}{\horstOpApplabelsOf{1}{\horstFreeVarvs}}}}
    \horstPremise{\horstEQ{\horstFreeVarw}{\horstCOND{\horstEQ{\horstOpAppmemOpBw{\horstParVarop}{}}{\horstOpAppmemOpTbw{\horstParVarop}{}}}{\horstFreeVaru}{\horstCOND{\horstOpAppmemLoadSigned{\horstParVarop}{}}{\horstOpAppextends{\horstOpAppmemOpTbw{\horstParVarop}{},\horstOpAppmemOpBw{\horstParVarop}{}}{\horstFreeVaru}}{\horstOpAppextendu{\horstOpAppmemOpTbw{\horstParVarop}{},\horstOpAppmemOpBw{\horstParVarop}{}}{\horstFreeVaru}}}}}
    \horstConclusion{\horstPredAppMState{\horstParVarfid,\horstADD{\horstParVarpc}{1}}{\horstFreeVarctx,\horstCONS{\horstConstructorAppLVal{\horstFreeVarw,\horstOpAppflub{3}{\horstTUPINIT{\horstFreeVarl,\horstOpApplabelOf{}{\horstFreeVarx},\horstOpApplabelOfCtx{}{\horstFreeVarctx}}}}}{\horstFreeVarst},\horstFreeVargt,\horstFreeVarlt,\horstConstructorAppMem{\horstFreeVari,\horstFreeVarv,\horstFreeVarsize},\horstFreeVartbl,\horstFreeVaratN,\horstFreeVargtN,\horstFreeVarmemN}}
  \end{horstClause}
  \begin{horstClause}
    \horstFreeVar{atN}{\horstTypeHomInit{\horstTypeLValue}{\horstOpAppas{\horstParVarfid}{}}}
    \horstFreeVar{st}{\horstTypeHomInit{\horstTypeLValue}{\horstSUB{\horstOpAppss{\horstParVarfid,\horstParVarpc}{}}{1}}}
    \horstFreeVar{tbl}{\horstTypeTable}
    \horstFreeVar{memNN}{\horstTypeMemory}
    \horstFreeVar{memNI}{\horstTypeMemory}
    \horstFreeVar{lt}{\horstTypeHomInit{\horstTypeLValue}{\horstOpAppls{\horstParVarfid}{}}}
    \horstFreeVar{i}{\horstTypeValue}
    \horstFreeVar{gt}{\horstTypeHomInit{\horstTypeLValue}{\horstOpAppgs{}{}}}
    \horstFreeVar{l}{\horstTypeFlowLabel}
    \horstFreeVar{size}{\horstTypeLValue}
    \horstFreeVar{memN}{\horstTypeMemory}
    \horstFreeVar{u}{\horstTypeValue}
    \horstFreeVar{ctx}{\horstTypeContext}
    \horstFreeVar{v}{\horstTypeLValue}
    \horstFreeVar{w}{\horstTypeValue}
    \horstFreeVar{x}{\horstTypeLValue}
    \horstFreeVar{gtN}{\horstTypeHomInit{\horstTypeLValue}{\horstOpAppgs{}{}}}
    \horstFreeVar{vs}{\horstTypeHomInit{\horstTypeLValue}{2}}
    \horstPremise{\horstEQ{\horstOpAppmemOpTbw{\horstParVarop}{}}{16}}
    \horstPremise{\horstOpAppiltu{64}{\horstOpAppiadd{64}{\horstOpAppvalueOf{}{\horstFreeVarx},\horstOpAppmkConst{\horstADD{\horstParVaroffset}{\horstDIV{\horstOpAppmemOpTbw{\horstParVarop}{}}{8}}}{}},\horstOpAppishl{64}{\horstOpAppvalueOf{}{\horstFreeVarsize},\horstOpAppmkConst{16}{}}}}
    \horstPremise{\horstPredAppMState{\horstParVarfid,\horstParVarpc}{\horstFreeVarctx,\horstCONS{\horstFreeVarx}{\horstFreeVarst},\horstFreeVargt,\horstFreeVarlt,\horstConstructorAppMem{\horstFreeVari,\horstFreeVarv,\horstFreeVarsize},\horstFreeVartbl,\horstFreeVaratN,\horstFreeVargtN,\horstFreeVarmemN}}
    \horstPremise{\horstSimpleSumExp{AND}{\horstSelectorFunctionAppinterval{\horstParVari}{0,2}}{\horstOR{\horstNEQ{\horstOpAppiadd{64}{\horstOpAppvalueOf{}{\horstFreeVarx},\horstOpAppmkConst{\horstADD{\horstParVaroffset}{\horstParVari}}{}}}{\horstFreeVari}}{\horstEQ{\horstFreeVarv}{\horstACCESS{\horstFreeVarvs}{\horstParVari}}}}{{i}}}
    \horstPremise{\horstPredAppMState{\horstParVarfid,\horstParVarpc}{\horstFreeVarctx,\horstCONS{\horstFreeVarx}{\horstFreeVarst},\horstFreeVargt,\horstFreeVarlt,\horstConstructorAppMem{\horstOpAppiadd{64}{\horstOpAppvalueOf{}{\horstFreeVarx},\horstOpAppmkConst{\horstADD{\horstParVaroffset}{0}}{}},\horstACCESS{\horstFreeVarvs}{0},\horstFreeVarsize},\horstFreeVartbl,\horstFreeVaratN,\horstFreeVargtN,\horstFreeVarmemNN}}
    \horstPremise{\horstPredAppMState{\horstParVarfid,\horstParVarpc}{\horstFreeVarctx,\horstCONS{\horstFreeVarx}{\horstFreeVarst},\horstFreeVargt,\horstFreeVarlt,\horstConstructorAppMem{\horstOpAppiadd{64}{\horstOpAppvalueOf{}{\horstFreeVarx},\horstOpAppmkConst{\horstADD{\horstParVaroffset}{1}}{}},\horstACCESS{\horstFreeVarvs}{1},\horstFreeVarsize},\horstFreeVartbl,\horstFreeVaratN,\horstFreeVargtN,\horstFreeVarmemNI}}
    \horstPremise{\horstEQ{\horstFreeVaru}{\horstOpAppload{2}{\horstOpAppvaluesOf{2}{\horstFreeVarvs}}}}
    \horstPremise{\horstEQ{\horstFreeVarl}{\horstOpAppflub{2}{\horstOpApplabelsOf{2}{\horstFreeVarvs}}}}
    \horstPremise{\horstEQ{\horstFreeVarw}{\horstCOND{\horstEQ{\horstOpAppmemOpBw{\horstParVarop}{}}{\horstOpAppmemOpTbw{\horstParVarop}{}}}{\horstFreeVaru}{\horstCOND{\horstOpAppmemLoadSigned{\horstParVarop}{}}{\horstOpAppextends{\horstOpAppmemOpTbw{\horstParVarop}{},\horstOpAppmemOpBw{\horstParVarop}{}}{\horstFreeVaru}}{\horstOpAppextendu{\horstOpAppmemOpTbw{\horstParVarop}{},\horstOpAppmemOpBw{\horstParVarop}{}}{\horstFreeVaru}}}}}
    \horstConclusion{\horstPredAppMState{\horstParVarfid,\horstADD{\horstParVarpc}{1}}{\horstFreeVarctx,\horstCONS{\horstConstructorAppLVal{\horstFreeVarw,\horstOpAppflub{3}{\horstTUPINIT{\horstFreeVarl,\horstOpApplabelOf{}{\horstFreeVarx},\horstOpApplabelOfCtx{}{\horstFreeVarctx}}}}}{\horstFreeVarst},\horstFreeVargt,\horstFreeVarlt,\horstConstructorAppMem{\horstFreeVari,\horstFreeVarv,\horstFreeVarsize},\horstFreeVartbl,\horstFreeVaratN,\horstFreeVargtN,\horstFreeVarmemN}}
  \end{horstClause}
  \begin{horstClause}
    \horstFreeVar{atN}{\horstTypeHomInit{\horstTypeLValue}{\horstOpAppas{\horstParVarfid}{}}}
    \horstFreeVar{st}{\horstTypeHomInit{\horstTypeLValue}{\horstSUB{\horstOpAppss{\horstParVarfid,\horstParVarpc}{}}{1}}}
    \horstFreeVar{tbl}{\horstTypeTable}
    \horstFreeVar{memNN}{\horstTypeMemory}
    \horstFreeVar{memNI}{\horstTypeMemory}
    \horstFreeVar{lt}{\horstTypeHomInit{\horstTypeLValue}{\horstOpAppls{\horstParVarfid}{}}}
    \horstFreeVar{memNII}{\horstTypeMemory}
    \horstFreeVar{memNIII}{\horstTypeMemory}
    \horstFreeVar{i}{\horstTypeValue}
    \horstFreeVar{gt}{\horstTypeHomInit{\horstTypeLValue}{\horstOpAppgs{}{}}}
    \horstFreeVar{l}{\horstTypeFlowLabel}
    \horstFreeVar{size}{\horstTypeLValue}
    \horstFreeVar{memN}{\horstTypeMemory}
    \horstFreeVar{u}{\horstTypeValue}
    \horstFreeVar{ctx}{\horstTypeContext}
    \horstFreeVar{v}{\horstTypeLValue}
    \horstFreeVar{w}{\horstTypeValue}
    \horstFreeVar{x}{\horstTypeLValue}
    \horstFreeVar{gtN}{\horstTypeHomInit{\horstTypeLValue}{\horstOpAppgs{}{}}}
    \horstFreeVar{vs}{\horstTypeHomInit{\horstTypeLValue}{4}}
    \horstPremise{\horstEQ{\horstOpAppmemOpTbw{\horstParVarop}{}}{32}}
    \horstPremise{\horstOpAppiltu{64}{\horstOpAppiadd{64}{\horstOpAppvalueOf{}{\horstFreeVarx},\horstOpAppmkConst{\horstADD{\horstParVaroffset}{\horstDIV{\horstOpAppmemOpTbw{\horstParVarop}{}}{8}}}{}},\horstOpAppishl{64}{\horstOpAppvalueOf{}{\horstFreeVarsize},\horstOpAppmkConst{16}{}}}}
    \horstPremise{\horstPredAppMState{\horstParVarfid,\horstParVarpc}{\horstFreeVarctx,\horstCONS{\horstFreeVarx}{\horstFreeVarst},\horstFreeVargt,\horstFreeVarlt,\horstConstructorAppMem{\horstFreeVari,\horstFreeVarv,\horstFreeVarsize},\horstFreeVartbl,\horstFreeVaratN,\horstFreeVargtN,\horstFreeVarmemN}}
    \horstPremise{\horstSimpleSumExp{AND}{\horstSelectorFunctionAppinterval{\horstParVari}{0,4}}{\horstOR{\horstNEQ{\horstOpAppiadd{64}{\horstOpAppvalueOf{}{\horstFreeVarx},\horstOpAppmkConst{\horstADD{\horstParVaroffset}{\horstParVari}}{}}}{\horstFreeVari}}{\horstEQ{\horstFreeVarv}{\horstACCESS{\horstFreeVarvs}{\horstParVari}}}}{{i}}}
    \horstPremise{\horstPredAppMState{\horstParVarfid,\horstParVarpc}{\horstFreeVarctx,\horstCONS{\horstFreeVarx}{\horstFreeVarst},\horstFreeVargt,\horstFreeVarlt,\horstConstructorAppMem{\horstOpAppiadd{64}{\horstOpAppvalueOf{}{\horstFreeVarx},\horstOpAppmkConst{\horstADD{\horstParVaroffset}{0}}{}},\horstACCESS{\horstFreeVarvs}{0},\horstFreeVarsize},\horstFreeVartbl,\horstFreeVaratN,\horstFreeVargtN,\horstFreeVarmemNN}}
    \horstPremise{\horstPredAppMState{\horstParVarfid,\horstParVarpc}{\horstFreeVarctx,\horstCONS{\horstFreeVarx}{\horstFreeVarst},\horstFreeVargt,\horstFreeVarlt,\horstConstructorAppMem{\horstOpAppiadd{64}{\horstOpAppvalueOf{}{\horstFreeVarx},\horstOpAppmkConst{\horstADD{\horstParVaroffset}{1}}{}},\horstACCESS{\horstFreeVarvs}{1},\horstFreeVarsize},\horstFreeVartbl,\horstFreeVaratN,\horstFreeVargtN,\horstFreeVarmemNI}}
    \horstPremise{\horstPredAppMState{\horstParVarfid,\horstParVarpc}{\horstFreeVarctx,\horstCONS{\horstFreeVarx}{\horstFreeVarst},\horstFreeVargt,\horstFreeVarlt,\horstConstructorAppMem{\horstOpAppiadd{64}{\horstOpAppvalueOf{}{\horstFreeVarx},\horstOpAppmkConst{\horstADD{\horstParVaroffset}{2}}{}},\horstACCESS{\horstFreeVarvs}{2},\horstFreeVarsize},\horstFreeVartbl,\horstFreeVaratN,\horstFreeVargtN,\horstFreeVarmemNII}}
    \horstPremise{\horstPredAppMState{\horstParVarfid,\horstParVarpc}{\horstFreeVarctx,\horstCONS{\horstFreeVarx}{\horstFreeVarst},\horstFreeVargt,\horstFreeVarlt,\horstConstructorAppMem{\horstOpAppiadd{64}{\horstOpAppvalueOf{}{\horstFreeVarx},\horstOpAppmkConst{\horstADD{\horstParVaroffset}{3}}{}},\horstACCESS{\horstFreeVarvs}{3},\horstFreeVarsize},\horstFreeVartbl,\horstFreeVaratN,\horstFreeVargtN,\horstFreeVarmemNIII}}
    \horstPremise{\horstEQ{\horstFreeVaru}{\horstOpAppload{4}{\horstOpAppvaluesOf{4}{\horstFreeVarvs}}}}
    \horstPremise{\horstEQ{\horstFreeVarl}{\horstOpAppflub{4}{\horstOpApplabelsOf{4}{\horstFreeVarvs}}}}
    \horstPremise{\horstEQ{\horstFreeVarw}{\horstCOND{\horstEQ{\horstOpAppmemOpBw{\horstParVarop}{}}{\horstOpAppmemOpTbw{\horstParVarop}{}}}{\horstCOND{\horstEQ{\horstParVarop}{\horstConstFXXXIILOADMEM}}{\horstOpAppfreeValOrTop{}{}}{\horstFreeVaru}}{\horstCOND{\horstOpAppmemLoadSigned{\horstParVarop}{}}{\horstOpAppextends{\horstOpAppmemOpTbw{\horstParVarop}{},\horstOpAppmemOpBw{\horstParVarop}{}}{\horstFreeVaru}}{\horstOpAppextendu{\horstOpAppmemOpTbw{\horstParVarop}{},\horstOpAppmemOpBw{\horstParVarop}{}}{\horstFreeVaru}}}}}
    \horstConclusion{\horstPredAppMState{\horstParVarfid,\horstADD{\horstParVarpc}{1}}{\horstFreeVarctx,\horstCONS{\horstConstructorAppLVal{\horstFreeVarw,\horstOpAppflub{3}{\horstTUPINIT{\horstFreeVarl,\horstOpApplabelOf{}{\horstFreeVarx},\horstOpApplabelOfCtx{}{\horstFreeVarctx}}}}}{\horstFreeVarst},\horstFreeVargt,\horstFreeVarlt,\horstConstructorAppMem{\horstFreeVari,\horstFreeVarv,\horstFreeVarsize},\horstFreeVartbl,\horstFreeVaratN,\horstFreeVargtN,\horstFreeVarmemN}}
  \end{horstClause}
  \begin{horstClause}
    \horstFreeVar{atN}{\horstTypeHomInit{\horstTypeLValue}{\horstOpAppas{\horstParVarfid}{}}}
    \horstFreeVar{st}{\horstTypeHomInit{\horstTypeLValue}{\horstSUB{\horstOpAppss{\horstParVarfid,\horstParVarpc}{}}{1}}}
    \horstFreeVar{tbl}{\horstTypeTable}
    \horstFreeVar{memNN}{\horstTypeMemory}
    \horstFreeVar{memNI}{\horstTypeMemory}
    \horstFreeVar{lt}{\horstTypeHomInit{\horstTypeLValue}{\horstOpAppls{\horstParVarfid}{}}}
    \horstFreeVar{memNII}{\horstTypeMemory}
    \horstFreeVar{memNIII}{\horstTypeMemory}
    \horstFreeVar{memNIV}{\horstTypeMemory}
    \horstFreeVar{i}{\horstTypeValue}
    \horstFreeVar{memNV}{\horstTypeMemory}
    \horstFreeVar{memNVI}{\horstTypeMemory}
    \horstFreeVar{gt}{\horstTypeHomInit{\horstTypeLValue}{\horstOpAppgs{}{}}}
    \horstFreeVar{memNVII}{\horstTypeMemory}
    \horstFreeVar{l}{\horstTypeFlowLabel}
    \horstFreeVar{size}{\horstTypeLValue}
    \horstFreeVar{memN}{\horstTypeMemory}
    \horstFreeVar{u}{\horstTypeValue}
    \horstFreeVar{ctx}{\horstTypeContext}
    \horstFreeVar{v}{\horstTypeLValue}
    \horstFreeVar{w}{\horstTypeValue}
    \horstFreeVar{x}{\horstTypeLValue}
    \horstFreeVar{gtN}{\horstTypeHomInit{\horstTypeLValue}{\horstOpAppgs{}{}}}
    \horstFreeVar{vs}{\horstTypeHomInit{\horstTypeLValue}{8}}
    \horstPremise{\horstEQ{\horstOpAppmemOpTbw{\horstParVarop}{}}{64}}
    \horstPremise{\horstOpAppiltu{64}{\horstOpAppiadd{64}{\horstOpAppvalueOf{}{\horstFreeVarx},\horstOpAppmkConst{\horstADD{\horstParVaroffset}{\horstDIV{\horstOpAppmemOpTbw{\horstParVarop}{}}{8}}}{}},\horstOpAppishl{64}{\horstOpAppvalueOf{}{\horstFreeVarsize},\horstOpAppmkConst{16}{}}}}
    \horstPremise{\horstPredAppMState{\horstParVarfid,\horstParVarpc}{\horstFreeVarctx,\horstCONS{\horstFreeVarx}{\horstFreeVarst},\horstFreeVargt,\horstFreeVarlt,\horstConstructorAppMem{\horstFreeVari,\horstFreeVarv,\horstFreeVarsize},\horstFreeVartbl,\horstFreeVaratN,\horstFreeVargtN,\horstFreeVarmemN}}
    \horstPremise{\horstSimpleSumExp{AND}{\horstSelectorFunctionAppinterval{\horstParVari}{0,8}}{\horstOR{\horstNEQ{\horstOpAppiadd{64}{\horstOpAppvalueOf{}{\horstFreeVarx},\horstOpAppmkConst{\horstADD{\horstParVaroffset}{\horstParVari}}{}}}{\horstFreeVari}}{\horstEQ{\horstFreeVarv}{\horstACCESS{\horstFreeVarvs}{\horstParVari}}}}{{i}}}
    \horstPremise{\horstPredAppMState{\horstParVarfid,\horstParVarpc}{\horstFreeVarctx,\horstCONS{\horstFreeVarx}{\horstFreeVarst},\horstFreeVargt,\horstFreeVarlt,\horstConstructorAppMem{\horstOpAppiadd{64}{\horstOpAppvalueOf{}{\horstFreeVarx},\horstOpAppmkConst{\horstADD{\horstParVaroffset}{0}}{}},\horstACCESS{\horstFreeVarvs}{0},\horstFreeVarsize},\horstFreeVartbl,\horstFreeVaratN,\horstFreeVargtN,\horstFreeVarmemNN}}
    \horstPremise{\horstPredAppMState{\horstParVarfid,\horstParVarpc}{\horstFreeVarctx,\horstCONS{\horstFreeVarx}{\horstFreeVarst},\horstFreeVargt,\horstFreeVarlt,\horstConstructorAppMem{\horstOpAppiadd{64}{\horstOpAppvalueOf{}{\horstFreeVarx},\horstOpAppmkConst{\horstADD{\horstParVaroffset}{1}}{}},\horstACCESS{\horstFreeVarvs}{1},\horstFreeVarsize},\horstFreeVartbl,\horstFreeVaratN,\horstFreeVargtN,\horstFreeVarmemNI}}
    \horstPremise{\horstPredAppMState{\horstParVarfid,\horstParVarpc}{\horstFreeVarctx,\horstCONS{\horstFreeVarx}{\horstFreeVarst},\horstFreeVargt,\horstFreeVarlt,\horstConstructorAppMem{\horstOpAppiadd{64}{\horstOpAppvalueOf{}{\horstFreeVarx},\horstOpAppmkConst{\horstADD{\horstParVaroffset}{2}}{}},\horstACCESS{\horstFreeVarvs}{2},\horstFreeVarsize},\horstFreeVartbl,\horstFreeVaratN,\horstFreeVargtN,\horstFreeVarmemNII}}
    \horstPremise{\horstPredAppMState{\horstParVarfid,\horstParVarpc}{\horstFreeVarctx,\horstCONS{\horstFreeVarx}{\horstFreeVarst},\horstFreeVargt,\horstFreeVarlt,\horstConstructorAppMem{\horstOpAppiadd{64}{\horstOpAppvalueOf{}{\horstFreeVarx},\horstOpAppmkConst{\horstADD{\horstParVaroffset}{3}}{}},\horstACCESS{\horstFreeVarvs}{3},\horstFreeVarsize},\horstFreeVartbl,\horstFreeVaratN,\horstFreeVargtN,\horstFreeVarmemNIII}}
    \horstPremise{\horstPredAppMState{\horstParVarfid,\horstParVarpc}{\horstFreeVarctx,\horstCONS{\horstFreeVarx}{\horstFreeVarst},\horstFreeVargt,\horstFreeVarlt,\horstConstructorAppMem{\horstOpAppiadd{64}{\horstOpAppvalueOf{}{\horstFreeVarx},\horstOpAppmkConst{\horstADD{\horstParVaroffset}{4}}{}},\horstACCESS{\horstFreeVarvs}{4},\horstFreeVarsize},\horstFreeVartbl,\horstFreeVaratN,\horstFreeVargtN,\horstFreeVarmemNIV}}
    \horstPremise{\horstPredAppMState{\horstParVarfid,\horstParVarpc}{\horstFreeVarctx,\horstCONS{\horstFreeVarx}{\horstFreeVarst},\horstFreeVargt,\horstFreeVarlt,\horstConstructorAppMem{\horstOpAppiadd{64}{\horstOpAppvalueOf{}{\horstFreeVarx},\horstOpAppmkConst{\horstADD{\horstParVaroffset}{5}}{}},\horstACCESS{\horstFreeVarvs}{5},\horstFreeVarsize},\horstFreeVartbl,\horstFreeVaratN,\horstFreeVargtN,\horstFreeVarmemNV}}
    \horstPremise{\horstPredAppMState{\horstParVarfid,\horstParVarpc}{\horstFreeVarctx,\horstCONS{\horstFreeVarx}{\horstFreeVarst},\horstFreeVargt,\horstFreeVarlt,\horstConstructorAppMem{\horstOpAppiadd{64}{\horstOpAppvalueOf{}{\horstFreeVarx},\horstOpAppmkConst{\horstADD{\horstParVaroffset}{6}}{}},\horstACCESS{\horstFreeVarvs}{6},\horstFreeVarsize},\horstFreeVartbl,\horstFreeVaratN,\horstFreeVargtN,\horstFreeVarmemNVI}}
    \horstPremise{\horstPredAppMState{\horstParVarfid,\horstParVarpc}{\horstFreeVarctx,\horstCONS{\horstFreeVarx}{\horstFreeVarst},\horstFreeVargt,\horstFreeVarlt,\horstConstructorAppMem{\horstOpAppiadd{64}{\horstOpAppvalueOf{}{\horstFreeVarx},\horstOpAppmkConst{\horstADD{\horstParVaroffset}{7}}{}},\horstACCESS{\horstFreeVarvs}{7},\horstFreeVarsize},\horstFreeVartbl,\horstFreeVaratN,\horstFreeVargtN,\horstFreeVarmemNVII}}
    \horstPremise{\horstEQ{\horstFreeVaru}{\horstOpAppload{8}{\horstOpAppvaluesOf{8}{\horstFreeVarvs}}}}
    \horstPremise{\horstEQ{\horstFreeVarl}{\horstOpAppflub{8}{\horstOpApplabelsOf{8}{\horstFreeVarvs}}}}
    \horstPremise{\horstEQ{\horstFreeVarw}{\horstCOND{\horstEQ{\horstParVarop}{\horstConstFLXIVLOADMEM}}{\horstOpAppfreeValOrTop{}{}}{\horstFreeVaru}}}
    \horstConclusion{\horstPredAppMState{\horstParVarfid,\horstADD{\horstParVarpc}{1}}{\horstFreeVarctx,\horstCONS{\horstConstructorAppLVal{\horstFreeVarw,\horstOpAppflub{3}{\horstTUPINIT{\horstFreeVarl,\horstOpApplabelOf{}{\horstFreeVarx},\horstOpApplabelOfCtx{}{\horstFreeVarctx}}}}}{\horstFreeVarst},\horstFreeVargt,\horstFreeVarlt,\horstConstructorAppMem{\horstFreeVari,\horstFreeVarv,\horstFreeVarsize},\horstFreeVartbl,\horstFreeVaratN,\horstFreeVargtN,\horstFreeVarmemN}}
  \end{horstClause}
\end{horstRule}
\begin{horstRule}{testNoninterferenceImportedFunctionMemoryDataSat}
  \horstParVar{fid}{\horstTypeint}
  \horstParVar{cid}{\horstTypeint}
  \horstSelectorFunctionInvocation{\horstSelectorFunctionAppimportedFunctionIds{\horstParVarfid}{},\horstSelectorFunctionAppimportCallMemoryDataLeak{\horstParVarcid}{\horstParVarfid}}
  \begin{horstClause}
    \horstFreeVar{memN}{\horstTypeMemory}
    \horstFreeVar{atN}{\horstTypeHomInit{\horstTypeLValue}{\horstOpAppas{\horstParVarfid}{}}}
    \horstFreeVar{tbl}{\horstTypeTable}
    \horstFreeVar{ctx}{\horstTypeContext}
    \horstFreeVar{v}{\horstTypeLValue}
    \horstFreeVar{lt}{\horstTypeHomInit{\horstTypeLValue}{\horstOpAppls{\horstParVarfid}{}}}
    \horstFreeVar{gtN}{\horstTypeHomInit{\horstTypeLValue}{\horstOpAppgs{}{}}}
    \horstFreeVar{i}{\horstTypeValue}
    \horstFreeVar{gt}{\horstTypeHomInit{\horstTypeLValue}{\horstOpAppgs{}{}}}
    \horstFreeVar{size}{\horstTypeLValue}
    \horstPremise{\horstPredAppMState{\horstParVarfid,0}{\horstFreeVarctx,\horstTUPINIT{},\horstFreeVargt,\horstFreeVarlt,\horstConstructorAppMem{\horstFreeVari,\horstFreeVarv,\horstFreeVarsize},\horstFreeVartbl,\horstFreeVaratN,\horstFreeVargtN,\horstFreeVarmemN}}
    \horstPremise{\horstSimpleSumExp{OR}{\horstSelectorFunctionAppmemoryDataInLabelForImportedFunction{\horstParVarcd,\horstParVarid}{\horstParVarfid}}{\horstAND{\horstOpAppflowsTo{}{\horstOpAppmkLabel{}{\horstParVarcd,\horstParVarid},\horstOpAppperspectiveOfCtx{}{\horstFreeVarctx}}}{\horstEQ{\horstOpApplabelOf{}{\horstFreeVarv}}{\horstConstructorAppIllegal{}}}}{{cd}{id}}}
    \horstConclusion{\horstPredApptestNoninterferenceImportedFunctionMemoryDataSat{\horstParVarfid,\horstParVarcid}{}}
  \end{horstClause}
\end{horstRule}
\begin{horstRule}{cvtOpRule}
  \horstParVar{fid}{\horstTypeint}
  \horstParVar{op}{\horstTypeint}
  \horstParVar{pc}{\horstTypeint}
  \horstSelectorFunctionInvocation{\horstSelectorFunctionAppfunctionIds{\horstParVarfid}{},\horstSelectorFunctionAppcvtOps{\horstParVarop}{},\horstSelectorFunctionApppcsForFunctionIdAndOpcode{\horstParVarpc}{\horstParVarfid,\horstParVarop}}
  \begin{horstClause}
    \horstFreeVar{memN}{\horstTypeMemory}
    \horstFreeVar{atN}{\horstTypeHomInit{\horstTypeLValue}{\horstOpAppas{\horstParVarfid}{}}}
    \horstFreeVar{st}{\horstTypeHomInit{\horstTypeLValue}{\horstSUB{\horstOpAppss{\horstParVarfid,\horstParVarpc}{}}{1}}}
    \horstFreeVar{tbl}{\horstTypeTable}
    \horstFreeVar{ctx}{\horstTypeContext}
    \horstFreeVar{lt}{\horstTypeHomInit{\horstTypeLValue}{\horstOpAppls{\horstParVarfid}{}}}
    \horstFreeVar{mem}{\horstTypeMemory}
    \horstFreeVar{x}{\horstTypeLValue}
    \horstFreeVar{gtN}{\horstTypeHomInit{\horstTypeLValue}{\horstOpAppgs{}{}}}
    \horstFreeVar{gt}{\horstTypeHomInit{\horstTypeLValue}{\horstOpAppgs{}{}}}
    \horstPremise{\horstPredAppMState{\horstParVarfid,\horstParVarpc}{\horstFreeVarctx,\horstCONS{\horstFreeVarx}{\horstFreeVarst},\horstFreeVargt,\horstFreeVarlt,\horstFreeVarmem,\horstFreeVartbl,\horstFreeVaratN,\horstFreeVargtN,\horstFreeVarmemN}}
    \horstConclusion{\horstPredAppMState{\horstParVarfid,\horstADD{\horstParVarpc}{1}}{\horstFreeVarctx,\horstCONS{\horstOpAppraiseTo{}{\horstOpApplabelledCvtOp{\horstParVarop}{\horstFreeVarx},\horstOpApplabelOfCtx{}{\horstFreeVarctx}}}{\horstFreeVarst},\horstFreeVargt,\horstFreeVarlt,\horstFreeVarmem,\horstFreeVartbl,\horstFreeVaratN,\horstFreeVargtN,\horstFreeVarmemN}}
  \end{horstClause}
\end{horstRule}
\begin{horstRule}{testNoninterferenceMemoryAreaUnsat}
  \horstParVar{fid}{\horstTypeint}
  \horstParVar{cid}{\horstTypeint}
  \horstParVar{start}{\horstTypeint}
  \horstParVar{endInclusive}{\horstTypeint}
  \horstSelectorFunctionInvocation{\horstSelectorFunctionAppstartFunctionId{\horstParVarfid}{},\horstSelectorFunctionAppmemoryAreaSafe{\horstParVarcid,\horstParVarstart,\horstParVarendInclusive}{}}
  \begin{horstClause}
    \horstFreeVar{memN}{\horstTypeMemory}
    \horstFreeVar{atN}{\horstTypeHomInit{\horstTypeLValue}{\horstOpAppas{\horstParVarfid}{}}}
    \horstFreeVar{tbl}{\horstTypeTable}
    \horstFreeVar{rt}{\horstTypeHomInit{\horstTypeLValue}{\horstOpApprs{\horstParVarfid}{}}}
    \horstFreeVar{ctx}{\horstTypeContext}
    \horstFreeVar{mem}{\horstTypeMemory}
    \horstFreeVar{gtN}{\horstTypeHomInit{\horstTypeLValue}{\horstOpAppgs{}{}}}
    \horstFreeVar{gt}{\horstTypeHomInit{\horstTypeLValue}{\horstOpAppgs{}{}}}
    \horstPremise{\horstPredAppReturnCall{\horstParVarfid}{\horstFreeVarctx,\horstFreeVarrt,\horstFreeVargt,\horstFreeVarmem,\horstFreeVartbl,\horstFreeVaratN,\horstFreeVargtN,\horstFreeVarmemN}}
    \horstPremise{\horstOpAppabsle{}{\horstOpAppmkValue{}{\horstParVarstart},\horstOpAppindexOfMem{}{\horstFreeVarmem}}}
    \horstPremise{\horstOpAppabsle{}{\horstOpAppindexOfMem{}{\horstFreeVarmem},\horstOpAppmkValue{}{\horstParVarendInclusive}}}
    \horstPremise{\horstSimpleSumExp{OR}{\horstSelectorFunctionAppmemoryDataOutLabel{\horstParVarcd,\horstParVarid}{}}{\horstAND{\horstOpAppflowsTo{}{\horstOpAppmkLabel{}{\horstParVarcd,\horstParVarid},\horstOpAppperspectiveOfCtx{}{\horstFreeVarctx}}}{\horstEQ{\horstOpApplabelOf{}{\horstOpAppvalueOfMem{}{\horstFreeVarmem}}}{\horstConstructorAppIllegal{}}}}{{cd}{id}}}
    \horstConclusion{\horstPredApptestNoninterferenceMemoryAreaUnsat{\horstParVarfid,\horstParVarcid,\horstParVarstart,\horstParVarendInclusive}{}}
  \end{horstClause}
\end{horstRule}
\begin{horstRule}{functionExitRule}
  \horstParVar{fid}{\horstTypeint}
  \horstParVar{pc}{\horstTypeint}
  \horstSelectorFunctionInvocation{\horstSelectorFunctionAppfunctionIds{\horstParVarfid}{},\horstSelectorFunctionAppexitPointsForFunctionId{\horstParVarpc}{\horstParVarfid}}
  \begin{horstClause}
    \horstFreeVar{memN}{\horstTypeMemory}
    \horstFreeVar{atN}{\horstTypeHomInit{\horstTypeLValue}{\horstOpAppas{\horstParVarfid}{}}}
    \horstFreeVar{st}{\horstTypeHomInit{\horstTypeLValue}{\horstOpAppss{\horstParVarfid,\horstParVarpc}{}}}
    \horstFreeVar{tbl}{\horstTypeTable}
    \horstFreeVar{rt}{\horstTypeHomInit{\horstTypeLValue}{\horstOpApprs{\horstParVarfid}{}}}
    \horstFreeVar{p}{\horstTypeLabel}
    \horstFreeVar{lt}{\horstTypeHomInit{\horstTypeLValue}{\horstOpAppls{\horstParVarfid}{}}}
    \horstFreeVar{mem}{\horstTypeMemory}
    \horstFreeVar{gtN}{\horstTypeHomInit{\horstTypeLValue}{\horstOpAppgs{}{}}}
    \horstFreeVar{gt}{\horstTypeHomInit{\horstTypeLValue}{\horstOpAppgs{}{}}}
    \horstFreeVar{from}{\horstTypeint}
    \horstPremise{\horstPredAppMState{\horstParVarfid,\horstParVarpc}{\horstConstructorAppCtx{\horstFreeVarp,\horstConstructorAppLegal{},\horstFreeVarfrom},\horstFreeVarst,\horstFreeVargt,\horstFreeVarlt,\horstFreeVarmem,\horstFreeVartbl,\horstFreeVaratN,\horstFreeVargtN,\horstFreeVarmemN}}
    \horstConclusion{\horstPredAppReturnToJoin{\horstParVarfid}{\horstConstructorAppCtx{\horstFreeVarp,\horstConstructorAppLegal{},\horstParVarpc},\horstSLICE{\horstFreeVarst}{}{\horstOpApprs{\horstParVarfid}{}},\horstFreeVargt,\horstFreeVarmem,\horstFreeVartbl,\horstFreeVaratN,\horstFreeVargtN,\horstFreeVarmemN}}
  \end{horstClause}
  \begin{horstClause}
    \horstFreeVar{memN}{\horstTypeMemory}
    \horstFreeVar{atN}{\horstTypeHomInit{\horstTypeLValue}{\horstOpAppas{\horstParVarfid}{}}}
    \horstFreeVar{st}{\horstTypeHomInit{\horstTypeLValue}{\horstOpAppss{\horstParVarfid,\horstParVarpc}{}}}
    \horstFreeVar{tbl}{\horstTypeTable}
    \horstFreeVar{rt}{\horstTypeHomInit{\horstTypeLValue}{\horstOpApprs{\horstParVarfid}{}}}
    \horstFreeVar{p}{\horstTypeLabel}
    \horstFreeVar{lt}{\horstTypeHomInit{\horstTypeLValue}{\horstOpAppls{\horstParVarfid}{}}}
    \horstFreeVar{mem}{\horstTypeMemory}
    \horstFreeVar{gtN}{\horstTypeHomInit{\horstTypeLValue}{\horstOpAppgs{}{}}}
    \horstFreeVar{gt}{\horstTypeHomInit{\horstTypeLValue}{\horstOpAppgs{}{}}}
    \horstFreeVar{from}{\horstTypeint}
    \horstPremise{\horstPredAppMState{\horstParVarfid,\horstParVarpc}{\horstConstructorAppCtx{\horstFreeVarp,\horstConstructorAppIllegal{},\horstFreeVarfrom},\horstFreeVarst,\horstFreeVargt,\horstFreeVarlt,\horstFreeVarmem,\horstFreeVartbl,\horstFreeVaratN,\horstFreeVargtN,\horstFreeVarmemN}}
    \horstPremise{\horstLE{0}{\horstFreeVarfrom}}
    \horstConclusion{\horstPredAppReturnToJoin{\horstParVarfid}{\horstConstructorAppCtx{\horstFreeVarp,\horstConstructorAppIllegal{},\horstFreeVarfrom},\horstSLICE{\horstFreeVarst}{}{\horstOpApprs{\horstParVarfid}{}},\horstFreeVargt,\horstFreeVarmem,\horstFreeVartbl,\horstFreeVaratN,\horstFreeVargtN,\horstFreeVarmemN}}
  \end{horstClause}
  \begin{horstClause}
    \horstFreeVar{memN}{\horstTypeMemory}
    \horstFreeVar{atN}{\horstTypeHomInit{\horstTypeLValue}{\horstOpAppas{\horstParVarfid}{}}}
    \horstFreeVar{st}{\horstTypeHomInit{\horstTypeLValue}{\horstOpAppss{\horstParVarfid,\horstParVarpc}{}}}
    \horstFreeVar{tbl}{\horstTypeTable}
    \horstFreeVar{rt}{\horstTypeHomInit{\horstTypeLValue}{\horstOpApprs{\horstParVarfid}{}}}
    \horstFreeVar{p}{\horstTypeLabel}
    \horstFreeVar{lt}{\horstTypeHomInit{\horstTypeLValue}{\horstOpAppls{\horstParVarfid}{}}}
    \horstFreeVar{mem}{\horstTypeMemory}
    \horstFreeVar{gtN}{\horstTypeHomInit{\horstTypeLValue}{\horstOpAppgs{}{}}}
    \horstFreeVar{gt}{\horstTypeHomInit{\horstTypeLValue}{\horstOpAppgs{}{}}}
    \horstFreeVar{from}{\horstTypeint}
    \horstPremise{\horstPredAppMState{\horstParVarfid,\horstParVarpc}{\horstConstructorAppCtx{\horstFreeVarp,\horstConstructorAppIllegal{},\horstFreeVarfrom},\horstFreeVarst,\horstFreeVargt,\horstFreeVarlt,\horstFreeVarmem,\horstFreeVartbl,\horstFreeVaratN,\horstFreeVargtN,\horstFreeVarmemN}}
    \horstPremise{\horstLT{\horstFreeVarfrom}{0}}
    \horstConclusion{\horstPredAppReturn{\horstParVarfid}{\horstConstructorAppCtx{\horstFreeVarp,\horstConstructorAppIllegal{},\horstFreeVarfrom},\horstSLICE{\horstFreeVarst}{}{\horstOpApprs{\horstParVarfid}{}},\horstFreeVargt,\horstFreeVarmem,\horstFreeVartbl,\horstFreeVaratN,\horstFreeVargtN,\horstFreeVarmemN}}
  \end{horstClause}
  \begin{horstClause}
    \horstFreeVar{memN}{\horstTypeMemory}
    \horstFreeVar{atN}{\horstTypeHomInit{\horstTypeLValue}{\horstOpAppas{\horstParVarfid}{}}}
    \horstFreeVar{st}{\horstTypeHomInit{\horstTypeLValue}{\horstOpAppss{\horstParVarfid,\horstParVarpc}{}}}
    \horstFreeVar{tbl}{\horstTypeTable}
    \horstFreeVar{rt}{\horstTypeHomInit{\horstTypeLValue}{\horstOpApprs{\horstParVarfid}{}}}
    \horstFreeVar{p}{\horstTypeLabel}
    \horstFreeVar{lt}{\horstTypeHomInit{\horstTypeLValue}{\horstOpAppls{\horstParVarfid}{}}}
    \horstFreeVar{mem}{\horstTypeMemory}
    \horstFreeVar{gtN}{\horstTypeHomInit{\horstTypeLValue}{\horstOpAppgs{}{}}}
    \horstFreeVar{gt}{\horstTypeHomInit{\horstTypeLValue}{\horstOpAppgs{}{}}}
    \horstFreeVar{from}{\horstTypeint}
    \horstPremise{\horstPredAppMState{\horstParVarfid,\horstParVarpc}{\horstConstructorAppCtx{\horstFreeVarp,\horstConstructorAppIllegal{},\horstFreeVarfrom},\horstFreeVarst,\horstFreeVargt,\horstFreeVarlt,\horstFreeVarmem,\horstFreeVartbl,\horstFreeVaratN,\horstFreeVargtN,\horstFreeVarmemN}}
    \horstConclusion{\horstPredAppScopeExtend{\horstParVarfid}{\horstFreeVarfrom,\horstOpApppcmax{\horstParVarfid}{}}}
  \end{horstClause}
\end{horstRule}
\begin{horstRule}{growRule}
  \horstParVar{fid}{\horstTypeint}
  \horstParVar{pc}{\horstTypeint}
  \horstSelectorFunctionInvocation{\horstSelectorFunctionAppfunctionIds{\horstParVarfid}{},\horstSelectorFunctionApppcsForFunctionIdAndOpcode{\horstParVarpc}{\horstParVarfid,\horstConstMEMORYGROW}}
  \begin{horstClause}
    \horstFreeVar{atN}{\horstTypeHomInit{\horstTypeLValue}{\horstOpAppas{\horstParVarfid}{}}}
    \horstFreeVar{st}{\horstTypeHomInit{\horstTypeLValue}{\horstSUB{\horstOpAppss{\horstParVarfid,\horstParVarpc}{}}{1}}}
    \horstFreeVar{tbl}{\horstTypeTable}
    \horstFreeVar{lt}{\horstTypeHomInit{\horstTypeLValue}{\horstOpAppls{\horstParVarfid}{}}}
    \horstFreeVar{i}{\horstTypeValue}
    \horstFreeVar{gt}{\horstTypeHomInit{\horstTypeLValue}{\horstOpAppgs{}{}}}
    \horstFreeVar{size}{\horstTypeLValue}
    \horstFreeVar{memN}{\horstTypeMemory}
    \horstFreeVar{nsize}{\horstTypeLValue}
    \horstFreeVar{ctx}{\horstTypeContext}
    \horstFreeVar{v}{\horstTypeLValue}
    \horstFreeVar{x}{\horstTypeLValue}
    \horstFreeVar{gtN}{\horstTypeHomInit{\horstTypeLValue}{\horstOpAppgs{}{}}}
    \horstFreeVar{ret}{\horstTypeLValue}
    \horstPremise{\horstPredAppMState{\horstParVarfid,\horstParVarpc}{\horstFreeVarctx,\horstCONS{\horstFreeVarx}{\horstFreeVarst},\horstFreeVargt,\horstFreeVarlt,\horstConstructorAppMem{\horstFreeVari,\horstFreeVarv,\horstFreeVarsize},\horstFreeVartbl,\horstFreeVaratN,\horstFreeVargtN,\horstFreeVarmemN}}
    \horstPremise{\horstOpAppiltu{64}{\horstFreeVari,\horstOpAppishl{64}{\horstOpAppvalueOf{}{\horstFreeVarsize},\horstOpAppmkConst{16}{}}}}
    \horstPremise{\horstEQ{\horstOpAppvalueOf{}{\horstFreeVarret}}{\horstCOND{\horstOpAppgrowOK{}{\horstOpAppbase{}{\horstOpAppvalueOf{}{\horstFreeVarx}},\horstOpAppbase{}{\horstOpAppvalueOf{}{\horstFreeVarsize}},\horstOpAppbase{}{\horstOpAppmkConst{\horstOpAppmms{}{}}{}}}}{\horstOpAppvalueOf{}{\horstFreeVarsize}}{\horstOpAppmkConst{-1}{}}}}
    \horstPremise{\horstEQ{\horstOpApplabelOf{}{\horstFreeVarret}}{\horstOpAppflub{3}{\horstTUPINIT{\horstOpApplabelOfCtx{}{\horstFreeVarctx},\horstOpApplabelOf{}{\horstFreeVarx},\horstOpApplabelOf{}{\horstFreeVarsize}}}}}
    \horstPremise{\horstEQ{\horstOpAppvalueOf{}{\horstFreeVarnsize}}{\horstOpAppgrow{}{\horstOpAppvalueOf{}{\horstFreeVarx},\horstOpAppvalueOf{}{\horstFreeVarsize},\horstOpAppmkConst{\horstOpAppmms{}{}}{}}}}
    \horstPremise{\horstEQ{\horstOpApplabelOf{}{\horstFreeVarnsize}}{\horstOpAppflub{3}{\horstTUPINIT{\horstOpApplabelOfCtx{}{\horstFreeVarctx},\horstOpApplabelOf{}{\horstFreeVarx},\horstOpApplabelOf{}{\horstFreeVarsize}}}}}
    \horstConclusion{\horstPredAppMState{\horstParVarfid,\horstADD{\horstParVarpc}{1}}{\horstFreeVarctx,\horstCONS{\horstFreeVarret}{\horstFreeVarst},\horstFreeVargt,\horstFreeVarlt,\horstConstructorAppMem{\horstFreeVari,\horstFreeVarv,\horstFreeVarnsize},\horstFreeVartbl,\horstFreeVaratN,\horstFreeVargtN,\horstFreeVarmemN}}
  \end{horstClause}
  \begin{horstClause}
    \horstFreeVar{atN}{\horstTypeHomInit{\horstTypeLValue}{\horstOpAppas{\horstParVarfid}{}}}
    \horstFreeVar{st}{\horstTypeHomInit{\horstTypeLValue}{\horstSUB{\horstOpAppss{\horstParVarfid,\horstParVarpc}{}}{1}}}
    \horstFreeVar{tbl}{\horstTypeTable}
    \horstFreeVar{lt}{\horstTypeHomInit{\horstTypeLValue}{\horstOpAppls{\horstParVarfid}{}}}
    \horstFreeVar{i}{\horstTypeValue}
    \horstFreeVar{gt}{\horstTypeHomInit{\horstTypeLValue}{\horstOpAppgs{}{}}}
    \horstFreeVar{size}{\horstTypeLValue}
    \horstFreeVar{memN}{\horstTypeMemory}
    \horstFreeVar{nsize}{\horstTypeLValue}
    \horstFreeVar{ctx}{\horstTypeContext}
    \horstFreeVar{v}{\horstTypeLValue}
    \horstFreeVar{x}{\horstTypeLValue}
    \horstFreeVar{gtN}{\horstTypeHomInit{\horstTypeLValue}{\horstOpAppgs{}{}}}
    \horstFreeVar{ret}{\horstTypeLValue}
    \horstPremise{\horstPredAppMState{\horstParVarfid,\horstParVarpc}{\horstFreeVarctx,\horstCONS{\horstFreeVarx}{\horstFreeVarst},\horstFreeVargt,\horstFreeVarlt,\horstConstructorAppMem{\horstFreeVari,\horstFreeVarv,\horstFreeVarsize},\horstFreeVartbl,\horstFreeVaratN,\horstFreeVargtN,\horstFreeVarmemN}}
    \horstPremise{\horstOpAppigeu{64}{\horstFreeVari,\horstOpAppishl{64}{\horstOpAppvalueOf{}{\horstFreeVarsize},\horstOpAppmkConst{16}{}}}}
    \horstPremise{\horstEQ{\horstOpAppvalueOf{}{\horstFreeVarret}}{\horstCOND{\horstOpAppgrowOK{}{\horstOpAppbase{}{\horstOpAppvalueOf{}{\horstFreeVarx}},\horstOpAppbase{}{\horstOpAppvalueOf{}{\horstFreeVarsize}},\horstOpAppbase{}{\horstOpAppmkConst{\horstOpAppmms{}{}}{}}}}{\horstOpAppvalueOf{}{\horstFreeVarsize}}{\horstOpAppmkConst{-1}{}}}}
    \horstPremise{\horstEQ{\horstOpApplabelOf{}{\horstFreeVarret}}{\horstOpAppflub{3}{\horstTUPINIT{\horstOpApplabelOfCtx{}{\horstFreeVarctx},\horstOpApplabelOf{}{\horstFreeVarx},\horstOpApplabelOf{}{\horstFreeVarsize}}}}}
    \horstPremise{\horstEQ{\horstOpAppvalueOf{}{\horstFreeVarnsize}}{\horstOpAppgrow{}{\horstOpAppvalueOf{}{\horstFreeVarx},\horstOpAppvalueOf{}{\horstFreeVarsize},\horstOpAppmkConst{\horstOpAppmms{}{}}{}}}}
    \horstPremise{\horstEQ{\horstOpApplabelOf{}{\horstFreeVarnsize}}{\horstOpAppflub{3}{\horstTUPINIT{\horstOpApplabelOfCtx{}{\horstFreeVarctx},\horstOpApplabelOf{}{\horstFreeVarx},\horstOpApplabelOf{}{\horstFreeVarsize}}}}}
    \horstConclusion{\horstPredAppMState{\horstParVarfid,\horstADD{\horstParVarpc}{1}}{\horstFreeVarctx,\horstCONS{\horstFreeVarret}{\horstFreeVarst},\horstFreeVargt,\horstFreeVarlt,\horstConstructorAppMem{\horstFreeVari,\horstOpAppraiseTo{}{\horstFreeVarv,\horstOpApplabelOf{}{\horstFreeVarnsize}},\horstFreeVarnsize},\horstFreeVartbl,\horstFreeVaratN,\horstFreeVargtN,\horstFreeVarmemN}}
  \end{horstClause}
\end{horstRule}
\begin{horstRule}{testNoninterferenceImportedFunctionContextSat}
  \horstParVar{fid}{\horstTypeint}
  \horstParVar{cid}{\horstTypeint}
  \horstSelectorFunctionInvocation{\horstSelectorFunctionAppimportedFunctionIds{\horstParVarfid}{},\horstSelectorFunctionAppimportCallContextLeak{\horstParVarcid}{\horstParVarfid}}
  \begin{horstClause}
    \horstFreeVar{memN}{\horstTypeMemory}
    \horstFreeVar{atN}{\horstTypeHomInit{\horstTypeLValue}{\horstOpAppas{\horstParVarfid}{}}}
    \horstFreeVar{tbl}{\horstTypeTable}
    \horstFreeVar{ctx}{\horstTypeContext}
    \horstFreeVar{lt}{\horstTypeHomInit{\horstTypeLValue}{\horstOpAppls{\horstParVarfid}{}}}
    \horstFreeVar{mem}{\horstTypeMemory}
    \horstFreeVar{gtN}{\horstTypeHomInit{\horstTypeLValue}{\horstOpAppgs{}{}}}
    \horstFreeVar{gt}{\horstTypeHomInit{\horstTypeLValue}{\horstOpAppgs{}{}}}
    \horstPremise{\horstPredAppMState{\horstParVarfid,0}{\horstFreeVarctx,\horstTUPINIT{},\horstFreeVargt,\horstFreeVarlt,\horstFreeVarmem,\horstFreeVartbl,\horstFreeVaratN,\horstFreeVargtN,\horstFreeVarmemN}}
    \horstPremise{\horstSimpleSumExp{OR}{\horstSelectorFunctionAppcontextLabelForImportedFunction{\horstParVarcc,\horstParVaric}{\horstParVarfid}}{\horstAND{\horstOpAppflowsTo{}{\horstOpAppmkLabel{}{\horstParVarcc,\horstParVaric},\horstOpAppperspectiveOfCtx{}{\horstFreeVarctx}}}{\horstEQ{\horstOpApplabelOfCtx{}{\horstFreeVarctx}}{\horstConstructorAppIllegal{}}}}{{cc}{ic}}}
    \horstConclusion{\horstPredApptestNoninterferenceImportedFunctionContextSat{\horstParVarfid,\horstParVarcid}{}}
  \end{horstClause}
\end{horstRule}
\begin{horstRule}{testNoninterferenceImportedFunctionMemoryDataUnsat}
  \horstParVar{fid}{\horstTypeint}
  \horstParVar{cid}{\horstTypeint}
  \horstSelectorFunctionInvocation{\horstSelectorFunctionAppimportedFunctionIds{\horstParVarfid}{},\horstSelectorFunctionAppimportCallMemoryDataSafe{\horstParVarcid}{\horstParVarfid}}
  \begin{horstClause}
    \horstFreeVar{memN}{\horstTypeMemory}
    \horstFreeVar{atN}{\horstTypeHomInit{\horstTypeLValue}{\horstOpAppas{\horstParVarfid}{}}}
    \horstFreeVar{tbl}{\horstTypeTable}
    \horstFreeVar{ctx}{\horstTypeContext}
    \horstFreeVar{v}{\horstTypeLValue}
    \horstFreeVar{lt}{\horstTypeHomInit{\horstTypeLValue}{\horstOpAppls{\horstParVarfid}{}}}
    \horstFreeVar{gtN}{\horstTypeHomInit{\horstTypeLValue}{\horstOpAppgs{}{}}}
    \horstFreeVar{i}{\horstTypeValue}
    \horstFreeVar{gt}{\horstTypeHomInit{\horstTypeLValue}{\horstOpAppgs{}{}}}
    \horstFreeVar{size}{\horstTypeLValue}
    \horstPremise{\horstPredAppMState{\horstParVarfid,0}{\horstFreeVarctx,\horstTUPINIT{},\horstFreeVargt,\horstFreeVarlt,\horstConstructorAppMem{\horstFreeVari,\horstFreeVarv,\horstFreeVarsize},\horstFreeVartbl,\horstFreeVaratN,\horstFreeVargtN,\horstFreeVarmemN}}
    \horstPremise{\horstSimpleSumExp{OR}{\horstSelectorFunctionAppmemoryDataInLabelForImportedFunction{\horstParVarcd,\horstParVarid}{\horstParVarfid}}{\horstAND{\horstOpAppflowsTo{}{\horstOpAppmkLabel{}{\horstParVarcd,\horstParVarid},\horstOpAppperspectiveOfCtx{}{\horstFreeVarctx}}}{\horstEQ{\horstOpApplabelOf{}{\horstFreeVarv}}{\horstConstructorAppIllegal{}}}}{{cd}{id}}}
    \horstConclusion{\horstPredApptestNoninterferenceImportedFunctionMemoryDataUnsat{\horstParVarfid,\horstParVarcid}{}}
  \end{horstClause}
\end{horstRule}
\begin{horstRule}{callIndirectHavokRule}
  \horstParVar{fid}{\horstTypeint}
  \horstParVar{pc}{\horstTypeint}
  \horstParVar{cid}{\horstTypeint}
  \horstSelectorFunctionInvocation{\horstSelectorFunctionAppfunctionIds{\horstParVarfid}{},\horstSelectorFunctionApppcsForFunctionIdAndOpcode{\horstParVarpc}{\horstParVarfid,\horstConstCALLINDIRECT},\horstSelectorFunctionApppossibleHavokCallTargets{\horstParVarcid}{\horstParVarfid,\horstParVarpc}}
  \begin{horstClause}
    \horstFreeVar{atN}{\horstTypeHomInit{\horstTypeLValue}{\horstOpAppas{\horstParVarfid}{}}}
    \horstFreeVar{st}{\horstTypeHomInit{\horstTypeLValue}{\horstSUB{\horstOpAppss{\horstParVarfid,\horstParVarpc}{}}{1}}}
    \horstFreeVar{ngt}{\horstTypeHomInit{\horstTypeLValue}{\horstOpAppgs{}{}}}
    \horstFreeVar{lt}{\horstTypeHomInit{\horstTypeLValue}{\horstOpAppls{\horstParVarfid}{}}}
    \horstFreeVar{cl}{\horstTypeFlowLabel}
    \horstFreeVar{mem}{\horstTypeMemory}
    \horstFreeVar{gt}{\horstTypeHomInit{\horstTypeLValue}{\horstOpAppgs{}{}}}
    \horstFreeVar{nmem}{\horstTypeMemory}
    \horstFreeVar{memN}{\horstTypeMemory}
    \horstFreeVar{p}{\horstTypeLabel}
    \horstFreeVar{at}{\horstTypeHomInit{\horstTypeLValue}{\horstOpAppas{\horstParVarcid}{}}}
    \horstFreeVar{ctx}{\horstTypeContext}
    \horstFreeVar{x}{\horstTypeLValue}
    \horstFreeVar{gtN}{\horstTypeHomInit{\horstTypeLValue}{\horstOpAppgs{}{}}}
    \horstFreeVar{ltbl}{\horstTypeFlowLabel}
    \horstPremise{\horstPredAppMState{\horstParVarfid,\horstParVarpc}{\horstFreeVarctx,\horstCONS{\horstFreeVarx}{\horstFreeVarst},\horstFreeVargt,\horstFreeVarlt,\horstFreeVarmem,\horstConstructorAppTbl{\horstConstructorAppTblImprecise{},\horstFreeVarltbl},\horstFreeVaratN,\horstFreeVargtN,\horstFreeVarmemN}}
    \horstPremise{\horstOpAppoverApproximateCallArguments{\horstParVarcid}{\horstOpAppreverse{\horstOpAppas{\horstParVarcid}{}}{\horstSLICE{\horstFreeVarst}{}{\horstOpAppas{\horstParVarcid}{}}},\horstFreeVarat}}
    \horstPremise{\horstOpAppoverApproximateCallGlobals{\horstParVarcid}{\horstFreeVargt,\horstFreeVarngt}}
    \horstPremise{\horstOpAppoverApproximateCallMemory{\horstParVarcid}{\horstFreeVarmem,\horstFreeVarnmem}}
    \horstPremise{\horstEQ{\horstFreeVarp}{\horstOpAppperspectiveOfCtx{}{\horstFreeVarctx}}}
    \horstPremise{\horstEQ{\horstFreeVarcl}{\horstOpAppflub{3}{\horstTUPINIT{\horstOpApplabelOfCtx{}{\horstFreeVarctx},\horstOpApplabelOf{}{\horstFreeVarx},\horstFreeVarltbl}}}}
    \horstConclusion{\horstPredAppMState{\horstParVarcid,0}{\horstOpAppmkCtx{}{\horstFreeVarp,\horstFreeVarcl},\horstTUPINIT{},\horstFreeVarngt,\horstCONCAT{\horstFreeVarat}{\horstHOMINIT{\horstOpAppmkLConst{0}{}}{\horstSUB{\horstOpAppls{\horstParVarcid}{}}{\horstOpAppas{\horstParVarcid}{}}}},\horstFreeVarnmem,\horstConstructorAppTbl{\horstConstructorAppTblImprecise{},\horstFreeVarltbl},\horstFreeVarat,\horstFreeVarngt,\horstFreeVarnmem}}
  \end{horstClause}
  \begin{horstClause}
    \horstFreeVar{atN}{\horstTypeHomInit{\horstTypeLValue}{\horstOpAppas{\horstParVarfid}{}}}
    \horstFreeVar{st}{\horstTypeHomInit{\horstTypeLValue}{\horstSUB{\horstOpAppss{\horstParVarfid,\horstParVarpc}{}}{1}}}
    \horstFreeVar{rt}{\horstTypeHomInit{\horstTypeLValue}{\horstOpApprs{\horstParVarcid}{}}}
    \horstFreeVar{ngt}{\horstTypeHomInit{\horstTypeLValue}{\horstOpAppgs{}{}}}
    \horstFreeVar{rtbl}{\horstTypeTable}
    \horstFreeVar{lt}{\horstTypeHomInit{\horstTypeLValue}{\horstOpAppls{\horstParVarfid}{}}}
    \horstFreeVar{cl}{\horstTypeFlowLabel}
    \horstFreeVar{mem}{\horstTypeMemory}
    \horstFreeVar{rfrom}{\horstTypeint}
    \horstFreeVar{gt}{\horstTypeHomInit{\horstTypeLValue}{\horstOpAppgs{}{}}}
    \horstFreeVar{nmem}{\horstTypeMemory}
    \horstFreeVar{rmem}{\horstTypeMemory}
    \horstFreeVar{memN}{\horstTypeMemory}
    \horstFreeVar{p}{\horstTypeLabel}
    \horstFreeVar{at}{\horstTypeHomInit{\horstTypeLValue}{\horstOpAppas{\horstParVarcid}{}}}
    \horstFreeVar{ctx}{\horstTypeContext}
    \horstFreeVar{x}{\horstTypeLValue}
    \horstFreeVar{gtN}{\horstTypeHomInit{\horstTypeLValue}{\horstOpAppgs{}{}}}
    \horstFreeVar{ltbl}{\horstTypeFlowLabel}
    \horstFreeVar{rgt}{\horstTypeHomInit{\horstTypeLValue}{\horstOpAppgs{}{}}}
    \horstPremise{\horstPredAppMState{\horstParVarfid,\horstParVarpc}{\horstFreeVarctx,\horstCONS{\horstFreeVarx}{\horstFreeVarst},\horstFreeVargt,\horstFreeVarlt,\horstFreeVarmem,\horstConstructorAppTbl{\horstConstructorAppTblImprecise{},\horstFreeVarltbl},\horstFreeVaratN,\horstFreeVargtN,\horstFreeVarmemN}}
    \horstPremise{\horstEQ{\horstOpApplabelOfCtx{}{\horstFreeVarctx}}{\horstConstructorAppLegal{}}}
    \horstPremise{\horstEQ{\horstOpAppflub{2}{\horstTUPINIT{\horstOpApplabelOf{}{\horstFreeVarx},\horstFreeVarltbl}}}{\horstConstructorAppIllegal{}}}
    \horstPremise{\horstOpAppoverApproximateCallArguments{\horstParVarcid}{\horstOpAppreverse{\horstOpAppas{\horstParVarcid}{}}{\horstSLICE{\horstFreeVarst}{}{\horstOpAppas{\horstParVarcid}{}}},\horstFreeVarat}}
    \horstPremise{\horstOpAppoverApproximateCallGlobals{\horstParVarcid}{\horstFreeVargt,\horstFreeVarngt}}
    \horstPremise{\horstOpAppoverApproximateCallMemory{\horstParVarcid}{\horstFreeVarmem,\horstFreeVarnmem}}
    \horstPremise{\horstEQ{\horstFreeVarp}{\horstOpAppperspectiveOfCtx{}{\horstFreeVarctx}}}
    \horstPremise{\horstPredAppReturn{\horstParVarcid}{\horstConstructorAppCtx{\horstFreeVarp,\horstConstructorAppIllegal{},\horstFreeVarrfrom},\horstFreeVarrt,\horstFreeVarrgt,\horstFreeVarrmem,\horstFreeVarrtbl,\horstFreeVarat,\horstFreeVarngt,\horstFreeVarnmem}}
    \horstPremise{\horstLT{\horstFreeVarrfrom}{0}}
    \horstConclusion{\horstPredAppMStateToJoin{\horstParVarfid,\horstADD{\horstParVarpc}{1}}{\horstFreeVarctx,\horstCONCAT{\horstFreeVarrt}{\horstSLICE{\horstFreeVarst}{\horstOpAppas{\horstParVarcid}{}}{}},\horstFreeVarrgt,\horstFreeVarlt,\horstFreeVarrmem,\horstFreeVarrtbl,\horstFreeVaratN,\horstFreeVargtN,\horstFreeVarmemN}}
  \end{horstClause}
  \begin{horstClause}
    \horstFreeVar{atN}{\horstTypeHomInit{\horstTypeLValue}{\horstOpAppas{\horstParVarfid}{}}}
    \horstFreeVar{st}{\horstTypeHomInit{\horstTypeLValue}{\horstSUB{\horstOpAppss{\horstParVarfid,\horstParVarpc}{}}{1}}}
    \horstFreeVar{rt}{\horstTypeHomInit{\horstTypeLValue}{\horstOpApprs{\horstParVarcid}{}}}
    \horstFreeVar{ngt}{\horstTypeHomInit{\horstTypeLValue}{\horstOpAppgs{}{}}}
    \horstFreeVar{rtbl}{\horstTypeTable}
    \horstFreeVar{lt}{\horstTypeHomInit{\horstTypeLValue}{\horstOpAppls{\horstParVarfid}{}}}
    \horstFreeVar{cl}{\horstTypeFlowLabel}
    \horstFreeVar{mem}{\horstTypeMemory}
    \horstFreeVar{rfrom}{\horstTypeint}
    \horstFreeVar{gt}{\horstTypeHomInit{\horstTypeLValue}{\horstOpAppgs{}{}}}
    \horstFreeVar{nmem}{\horstTypeMemory}
    \horstFreeVar{rmem}{\horstTypeMemory}
    \horstFreeVar{memN}{\horstTypeMemory}
    \horstFreeVar{p}{\horstTypeLabel}
    \horstFreeVar{at}{\horstTypeHomInit{\horstTypeLValue}{\horstOpAppas{\horstParVarcid}{}}}
    \horstFreeVar{ctx}{\horstTypeContext}
    \horstFreeVar{rl}{\horstTypeFlowLabel}
    \horstFreeVar{x}{\horstTypeLValue}
    \horstFreeVar{gtN}{\horstTypeHomInit{\horstTypeLValue}{\horstOpAppgs{}{}}}
    \horstFreeVar{ltbl}{\horstTypeFlowLabel}
    \horstFreeVar{rgt}{\horstTypeHomInit{\horstTypeLValue}{\horstOpAppgs{}{}}}
    \horstPremise{\horstPredAppMState{\horstParVarfid,\horstParVarpc}{\horstFreeVarctx,\horstCONS{\horstFreeVarx}{\horstFreeVarst},\horstFreeVargt,\horstFreeVarlt,\horstFreeVarmem,\horstConstructorAppTbl{\horstConstructorAppTblImprecise{},\horstFreeVarltbl},\horstFreeVaratN,\horstFreeVargtN,\horstFreeVarmemN}}
    \horstPremise{\horstOpAppoverApproximateCallArguments{\horstParVarcid}{\horstOpAppreverse{\horstOpAppas{\horstParVarcid}{}}{\horstSLICE{\horstFreeVarst}{}{\horstOpAppas{\horstParVarcid}{}}},\horstFreeVarat}}
    \horstPremise{\horstOpAppoverApproximateCallGlobals{\horstParVarcid}{\horstFreeVargt,\horstFreeVarngt}}
    \horstPremise{\horstOpAppoverApproximateCallMemory{\horstParVarcid}{\horstFreeVarmem,\horstFreeVarnmem}}
    \horstPremise{\horstEQ{\horstFreeVarp}{\horstOpAppperspectiveOfCtx{}{\horstFreeVarctx}}}
    \horstPremise{\horstEQ{\horstOpApplabelOfCtx{}{\horstFreeVarctx}}{\horstOpAppflub{2}{\horstTUPINIT{\horstOpApplabelOf{}{\horstFreeVarx},\horstFreeVarltbl}}}}
    \horstPremise{\horstEQ{\horstFreeVarcl}{\horstOpApplabelOfCtx{}{\horstFreeVarctx}}}
    \horstPremise{\horstEQ{\horstFreeVarcl}{\horstOpAppflub{2}{\horstTUPINIT{\horstOpApplabelOf{}{\horstFreeVarx},\horstFreeVarltbl}}}}
    \horstPremise{\horstPredAppReturn{\horstParVarcid}{\horstConstructorAppCtx{\horstFreeVarp,\horstFreeVarrl,\horstFreeVarrfrom},\horstFreeVarrt,\horstFreeVarrgt,\horstFreeVarrmem,\horstFreeVarrtbl,\horstFreeVarat,\horstFreeVarngt,\horstFreeVarnmem}}
    \horstPremise{\horstOR{\horstAND{\horstEQ{\horstFreeVarcl}{\horstConstructorAppLegal{}}}{\horstLE{0}{\horstFreeVarrfrom}}}{\horstAND{\horstEQ{\horstFreeVarcl}{\horstConstructorAppIllegal{}}}{\horstLT{\horstFreeVarrfrom}{0}}}}
    \horstConclusion{\horstPredAppMState{\horstParVarfid,\horstADD{\horstParVarpc}{1}}{\horstFreeVarctx,\horstCONCAT{\horstFreeVarrt}{\horstSLICE{\horstFreeVarst}{\horstOpAppas{\horstParVarcid}{}}{}},\horstFreeVarrgt,\horstFreeVarlt,\horstFreeVarrmem,\horstFreeVarrtbl,\horstFreeVaratN,\horstFreeVargtN,\horstFreeVarmemN}}
  \end{horstClause}
\end{horstRule}
\begin{horstRule}{testNoninterferenceMemorySizeSat}
  \horstParVar{fid}{\horstTypeint}
  \horstParVar{cid}{\horstTypeint}
  \horstSelectorFunctionInvocation{\horstSelectorFunctionAppstartFunctionId{\horstParVarfid}{},\horstSelectorFunctionAppmemorySizeLeak{\horstParVarcid}{}}
  \begin{horstClause}
    \horstFreeVar{memN}{\horstTypeMemory}
    \horstFreeVar{atN}{\horstTypeHomInit{\horstTypeLValue}{\horstOpAppas{\horstParVarfid}{}}}
    \horstFreeVar{tbl}{\horstTypeTable}
    \horstFreeVar{rt}{\horstTypeHomInit{\horstTypeLValue}{\horstOpApprs{\horstParVarfid}{}}}
    \horstFreeVar{ctx}{\horstTypeContext}
    \horstFreeVar{mem}{\horstTypeMemory}
    \horstFreeVar{gtN}{\horstTypeHomInit{\horstTypeLValue}{\horstOpAppgs{}{}}}
    \horstFreeVar{gt}{\horstTypeHomInit{\horstTypeLValue}{\horstOpAppgs{}{}}}
    \horstPremise{\horstPredAppReturnCall{\horstParVarfid}{\horstFreeVarctx,\horstFreeVarrt,\horstFreeVargt,\horstFreeVarmem,\horstFreeVartbl,\horstFreeVaratN,\horstFreeVargtN,\horstFreeVarmemN}}
    \horstPremise{\horstSimpleSumExp{OR}{\horstSelectorFunctionAppmemorySizeOutLabel{\horstParVarcs,\horstParVaris}{}}{\horstAND{\horstOpAppflowsTo{}{\horstOpAppmkLabel{}{\horstParVarcs,\horstParVaris},\horstOpAppperspectiveOfCtx{}{\horstFreeVarctx}}}{\horstEQ{\horstOpApplabelOf{}{\horstOpAppsizeOfMem{}{\horstFreeVarmem}}}{\horstConstructorAppIllegal{}}}}{{cs}{is}}}
    \horstConclusion{\horstPredApptestNoninterferenceMemorySizeSat{\horstParVarfid,\horstParVarcid}{}}
  \end{horstClause}
\end{horstRule}
\begin{horstRule}{testNoninterferenceImportedFunctionParamsSat}
  \horstParVar{fid}{\horstTypeint}
  \horstParVar{cid}{\horstTypeint}
  \horstSelectorFunctionInvocation{\horstSelectorFunctionAppimportedFunctionIds{\horstParVarfid}{},\horstSelectorFunctionAppimportCallArgumentLeak{\horstParVarcid}{\horstParVarfid}}
  \begin{horstClause}
    \horstFreeVar{memN}{\horstTypeMemory}
    \horstFreeVar{atN}{\horstTypeHomInit{\horstTypeLValue}{\horstOpAppas{\horstParVarfid}{}}}
    \horstFreeVar{tbl}{\horstTypeTable}
    \horstFreeVar{ctx}{\horstTypeContext}
    \horstFreeVar{lt}{\horstTypeHomInit{\horstTypeLValue}{\horstOpAppls{\horstParVarfid}{}}}
    \horstFreeVar{mem}{\horstTypeMemory}
    \horstFreeVar{gtN}{\horstTypeHomInit{\horstTypeLValue}{\horstOpAppgs{}{}}}
    \horstFreeVar{gt}{\horstTypeHomInit{\horstTypeLValue}{\horstOpAppgs{}{}}}
    \horstPremise{\horstPredAppMState{\horstParVarfid,0}{\horstFreeVarctx,\horstTUPINIT{},\horstFreeVargt,\horstFreeVarlt,\horstFreeVarmem,\horstFreeVartbl,\horstFreeVaratN,\horstFreeVargtN,\horstFreeVarmemN}}
    \horstPremise{\horstSimpleSumExp{OR}{\horstSelectorFunctionAppinterval{\horstParVaridx}{0,\horstOpAppas{\horstParVarfid}{}},\horstSelectorFunctionAppargumentLabelForImportedFunctionAndPosition{\horstParVarca,\horstParVaria}{\horstParVarfid,\horstParVaridx}}{\horstAND{\horstOpAppflowsTo{}{\horstOpAppmkLabel{}{\horstParVarca,\horstParVaria},\horstOpAppperspectiveOfCtx{}{\horstFreeVarctx}}}{\horstEQ{\horstOpApplabelOf{}{\horstACCESS{\horstFreeVaratN}{\horstParVaridx}}}{\horstConstructorAppIllegal{}}}}{{idx}{ca}{ia}}}
    \horstConclusion{\horstPredApptestNoninterferenceImportedFunctionParamsSat{\horstParVarfid,\horstParVarcid}{}}
  \end{horstClause}
\end{horstRule}
\begin{horstRule}{testNoninterferenceResultSat}
  \horstParVar{fid}{\horstTypeint}
  \horstParVar{cid}{\horstTypeint}
  \horstSelectorFunctionInvocation{\horstSelectorFunctionAppstartFunctionId{\horstParVarfid}{},\horstSelectorFunctionAppresultLeak{\horstParVarcid}{}}
  \begin{horstClause}
    \horstFreeVar{memN}{\horstTypeMemory}
    \horstFreeVar{atN}{\horstTypeHomInit{\horstTypeLValue}{\horstOpAppas{\horstParVarfid}{}}}
    \horstFreeVar{tbl}{\horstTypeTable}
    \horstFreeVar{rt}{\horstTypeHomInit{\horstTypeLValue}{\horstOpApprs{\horstParVarfid}{}}}
    \horstFreeVar{ctx}{\horstTypeContext}
    \horstFreeVar{mem}{\horstTypeMemory}
    \horstFreeVar{gtN}{\horstTypeHomInit{\horstTypeLValue}{\horstOpAppgs{}{}}}
    \horstFreeVar{gt}{\horstTypeHomInit{\horstTypeLValue}{\horstOpAppgs{}{}}}
    \horstPremise{\horstPredAppReturnCall{\horstParVarfid}{\horstFreeVarctx,\horstFreeVarrt,\horstFreeVargt,\horstFreeVarmem,\horstFreeVartbl,\horstFreeVaratN,\horstFreeVargtN,\horstFreeVarmemN}}
    \horstPremise{\horstSimpleSumExp{OR}{\horstSelectorFunctionAppinterval{\horstParVaridx}{0,\horstOpApprs{\horstParVarfid}{}},\horstSelectorFunctionAppresultLabelForPosition{\horstParVarcr,\horstParVarir}{\horstParVaridx}}{\horstAND{\horstOpAppflowsTo{}{\horstOpAppmkLabel{}{\horstParVarcr,\horstParVarir},\horstOpAppperspectiveOfCtx{}{\horstFreeVarctx}}}{\horstEQ{\horstOpApplabelOf{}{\horstACCESS{\horstFreeVarrt}{\horstParVaridx}}}{\horstConstructorAppIllegal{}}}}{{idx}{cr}{ir}}}
    \horstConclusion{\horstPredApptestNoninterferenceResultSat{\horstParVarfid,\horstParVarcid}{}}
  \end{horstClause}
\end{horstRule}
\begin{horstRule}{trappingCvtOpRule}
  \horstParVar{fid}{\horstTypeint}
  \horstParVar{op}{\horstTypeint}
  \horstParVar{pc}{\horstTypeint}
  \horstSelectorFunctionInvocation{\horstSelectorFunctionAppfunctionIds{\horstParVarfid}{},\horstSelectorFunctionApptrappingCvtOps{\horstParVarop}{},\horstSelectorFunctionApppcsForFunctionIdAndOpcode{\horstParVarpc}{\horstParVarfid,\horstParVarop}}
  \begin{horstClause}
    \horstFreeVar{memN}{\horstTypeMemory}
    \horstFreeVar{atN}{\horstTypeHomInit{\horstTypeLValue}{\horstOpAppas{\horstParVarfid}{}}}
    \horstFreeVar{st}{\horstTypeHomInit{\horstTypeLValue}{\horstSUB{\horstOpAppss{\horstParVarfid,\horstParVarpc}{}}{1}}}
    \horstFreeVar{tbl}{\horstTypeTable}
    \horstFreeVar{ctx}{\horstTypeContext}
    \horstFreeVar{lt}{\horstTypeHomInit{\horstTypeLValue}{\horstOpAppls{\horstParVarfid}{}}}
    \horstFreeVar{mem}{\horstTypeMemory}
    \horstFreeVar{x}{\horstTypeLValue}
    \horstFreeVar{gtN}{\horstTypeHomInit{\horstTypeLValue}{\horstOpAppgs{}{}}}
    \horstFreeVar{gt}{\horstTypeHomInit{\horstTypeLValue}{\horstOpAppgs{}{}}}
    \horstPremise{\horstPredAppMState{\horstParVarfid,\horstParVarpc}{\horstFreeVarctx,\horstCONS{\horstFreeVarx}{\horstFreeVarst},\horstFreeVargt,\horstFreeVarlt,\horstFreeVarmem,\horstFreeVartbl,\horstFreeVaratN,\horstFreeVargtN,\horstFreeVarmemN}}
    \horstConclusion{\horstPredAppMState{\horstParVarfid,\horstADD{\horstParVarpc}{1}}{\horstFreeVarctx,\horstCONS{\horstOpAppraiseTo{}{\horstOpApplabelledCvtOp{\horstParVarop}{\horstFreeVarx},\horstOpApplabelOfCtx{}{\horstFreeVarctx}}}{\horstFreeVarst},\horstFreeVargt,\horstFreeVarlt,\horstFreeVarmem,\horstFreeVartbl,\horstFreeVaratN,\horstFreeVargtN,\horstFreeVarmemN}}
  \end{horstClause}
\end{horstRule}
\begin{horstRule}{testNoninterferenceTableUnsat}
  \horstParVar{fid}{\horstTypeint}
  \horstParVar{cid}{\horstTypeint}
  \horstSelectorFunctionInvocation{\horstSelectorFunctionAppstartFunctionId{\horstParVarfid}{},\horstSelectorFunctionApptableSafe{\horstParVarcid}{}}
  \begin{horstClause}
    \horstFreeVar{memN}{\horstTypeMemory}
    \horstFreeVar{atN}{\horstTypeHomInit{\horstTypeLValue}{\horstOpAppas{\horstParVarfid}{}}}
    \horstFreeVar{ptbl}{\horstTypeTablePrecision}
    \horstFreeVar{ctx}{\horstTypeContext}
    \horstFreeVar{lt}{\horstTypeHomInit{\horstTypeLValue}{\horstOpAppls{\horstParVarfid}{}}}
    \horstFreeVar{mem}{\horstTypeMemory}
    \horstFreeVar{gtN}{\horstTypeHomInit{\horstTypeLValue}{\horstOpAppgs{}{}}}
    \horstFreeVar{gt}{\horstTypeHomInit{\horstTypeLValue}{\horstOpAppgs{}{}}}
    \horstPremise{\horstPredAppMState{\horstParVarfid,0}{\horstFreeVarctx,\horstTUPINIT{},\horstFreeVargt,\horstFreeVarlt,\horstFreeVarmem,\horstConstructorAppTbl{\horstFreeVarptbl,\horstConstructorAppIllegal{}},\horstFreeVaratN,\horstFreeVargtN,\horstFreeVarmemN}}
    \horstPremise{\horstSimpleSumExp{OR}{\horstSelectorFunctionApptableOutLabel{\horstParVarct,\horstParVarit}{}}{\horstOpAppflowsTo{}{\horstOpAppmkLabel{}{\horstParVarct,\horstParVarit},\horstOpAppperspectiveOfCtx{}{\horstFreeVarctx}}}{{ct}{it}}}
    \horstConclusion{\horstPredApptestNoninterferenceTableUnsat{\horstParVarfid,\horstParVarcid}{}}
  \end{horstClause}
\end{horstRule}
\begin{horstRule}{trappingBinOpRule}
  \horstParVar{fid}{\horstTypeint}
  \horstParVar{op}{\horstTypeint}
  \horstParVar{pc}{\horstTypeint}
  \horstSelectorFunctionInvocation{\horstSelectorFunctionAppfunctionIds{\horstParVarfid}{},\horstSelectorFunctionApptrappingBinOps{\horstParVarop}{},\horstSelectorFunctionApppcsForFunctionIdAndOpcode{\horstParVarpc}{\horstParVarfid,\horstParVarop}}
  \begin{horstClause}
    \horstFreeVar{memN}{\horstTypeMemory}
    \horstFreeVar{atN}{\horstTypeHomInit{\horstTypeLValue}{\horstOpAppas{\horstParVarfid}{}}}
    \horstFreeVar{st}{\horstTypeHomInit{\horstTypeLValue}{\horstSUB{\horstOpAppss{\horstParVarfid,\horstParVarpc}{}}{2}}}
    \horstFreeVar{tbl}{\horstTypeTable}
    \horstFreeVar{ctx}{\horstTypeContext}
    \horstFreeVar{lt}{\horstTypeHomInit{\horstTypeLValue}{\horstOpAppls{\horstParVarfid}{}}}
    \horstFreeVar{mem}{\horstTypeMemory}
    \horstFreeVar{x}{\horstTypeLValue}
    \horstFreeVar{gtN}{\horstTypeHomInit{\horstTypeLValue}{\horstOpAppgs{}{}}}
    \horstFreeVar{y}{\horstTypeLValue}
    \horstFreeVar{gt}{\horstTypeHomInit{\horstTypeLValue}{\horstOpAppgs{}{}}}
    \horstFreeVar{res}{\horstTypeMaybeValue}
    \horstPremise{\horstPredAppMState{\horstParVarfid,\horstParVarpc}{\horstFreeVarctx,\horstCONS{\horstFreeVarx}{\horstCONS{\horstFreeVary}{\horstFreeVarst}},\horstFreeVargt,\horstFreeVarlt,\horstFreeVarmem,\horstFreeVartbl,\horstFreeVaratN,\horstFreeVargtN,\horstFreeVarmemN}}
    \horstPremise{\horstEQ{\horstFreeVarres}{\horstOpApptrappingBinOp{\horstParVarop}{\horstOpAppvalueOf{}{\horstFreeVary},\horstOpAppvalueOf{}{\horstFreeVarx}}}}
    \horstPremise{\horstOpAppisJustV{}{\horstFreeVarres}}
    \horstConclusion{\horstPredAppMState{\horstParVarfid,\horstADD{\horstParVarpc}{1}}{\horstFreeVarctx,\horstCONS{\horstConstructorAppLVal{\horstOpAppfromJustV{}{\horstFreeVarres},\horstOpAppflub{3}{\horstTUPINIT{\horstOpApplabelOf{}{\horstFreeVary},\horstOpApplabelOf{}{\horstFreeVarx},\horstOpApplabelOfCtx{}{\horstFreeVarctx}}}}}{\horstFreeVarst},\horstFreeVargt,\horstFreeVarlt,\horstFreeVarmem,\horstFreeVartbl,\horstFreeVaratN,\horstFreeVargtN,\horstFreeVarmemN}}
  \end{horstClause}
\end{horstRule}
\begin{horstRule}{storeRule}
  \horstParVar{fid}{\horstTypeint}
  \horstParVar{op}{\horstTypeint}
  \horstParVar{pc}{\horstTypeint}
  \horstParVar{offset}{\horstTypeint}
  \horstSelectorFunctionInvocation{\horstSelectorFunctionAppfunctionIds{\horstParVarfid}{},\horstSelectorFunctionAppstoreOps{\horstParVarop}{},\horstSelectorFunctionApppcsForFunctionIdAndOpcode{\horstParVarpc}{\horstParVarfid,\horstParVarop},\horstSelectorFunctionAppmemoryOffsetForFunctionIdAndPc{\horstParVaroffset}{\horstParVarfid,\horstParVarpc}}
  \begin{horstClause}
    \horstFreeVar{atN}{\horstTypeHomInit{\horstTypeLValue}{\horstOpAppas{\horstParVarfid}{}}}
    \horstFreeVar{st}{\horstTypeHomInit{\horstTypeLValue}{\horstSUB{\horstOpAppss{\horstParVarfid,\horstParVarpc}{}}{2}}}
    \horstFreeVar{tbl}{\horstTypeTable}
    \horstFreeVar{d}{\horstTypeValue}
    \horstFreeVar{nv}{\horstTypeLValue}
    \horstFreeVar{lt}{\horstTypeHomInit{\horstTypeLValue}{\horstOpAppls{\horstParVarfid}{}}}
    \horstFreeVar{i}{\horstTypeValue}
    \horstFreeVar{gt}{\horstTypeHomInit{\horstTypeLValue}{\horstOpAppgs{}{}}}
    \horstFreeVar{size}{\horstTypeLValue}
    \horstFreeVar{memN}{\horstTypeMemory}
    \horstFreeVar{u}{\horstTypeValue}
    \horstFreeVar{ctx}{\horstTypeContext}
    \horstFreeVar{v}{\horstTypeLValue}
    \horstFreeVar{w}{\horstTypeLValue}
    \horstFreeVar{x}{\horstTypeLValue}
    \horstFreeVar{gtN}{\horstTypeHomInit{\horstTypeLValue}{\horstOpAppgs{}{}}}
    \horstFreeVar{y}{\horstTypeLValue}
    \horstPremise{\horstPredAppMState{\horstParVarfid,\horstParVarpc}{\horstFreeVarctx,\horstCONS{\horstFreeVarx}{\horstCONS{\horstFreeVary}{\horstFreeVarst}},\horstFreeVargt,\horstFreeVarlt,\horstConstructorAppMem{\horstFreeVari,\horstFreeVarv,\horstFreeVarsize},\horstFreeVartbl,\horstFreeVaratN,\horstFreeVargtN,\horstFreeVarmemN}}
    \horstPremise{\horstOpAppiltu{64}{\horstOpAppiadd{64}{\horstOpAppvalueOf{}{\horstFreeVary},\horstOpAppmkConst{\horstADD{\horstParVaroffset}{\horstDIV{\horstOpAppmemOpTbw{\horstParVarop}{}}{8}}}{}},\horstOpAppishl{64}{\horstOpAppvalueOf{}{\horstFreeVarsize},\horstOpAppmkConst{16}{}}}}
    \horstPremise{\horstEQ{\horstFreeVard}{\horstOpAppisub{64}{\horstFreeVari,\horstOpAppiadd{64}{\horstOpAppvalueOf{}{\horstFreeVary},\horstOpAppmkConst{\horstParVaroffset}{}}}}}
    \horstPremise{\horstEQ{\horstFreeVaru}{\horstCOND{\horstEQ{\horstOpAppmemOpBw{\horstParVarop}{}}{\horstOpAppmemOpTbw{\horstParVarop}{}}}{\horstCOND{\horstOR{\horstEQ{\horstParVarop}{\horstConstFXXXIISTOREMEM}}{\horstEQ{\horstParVarop}{\horstConstFLXIVSTOREMEM}}}{\horstOpAppfreeValOrTop{}{}}{\horstOpAppvalueOf{}{\horstFreeVarx}}}{\horstOpAppwrap{\horstOpAppmemOpBw{\horstParVarop}{},\horstOpAppmemOpTbw{\horstParVarop}{}}{\horstOpAppvalueOf{}{\horstFreeVarx}}}}}
    \horstPremise{\horstEQ{\horstOpAppvalueOf{}{\horstFreeVarw}}{\horstOpAppiand{64}{\horstOpAppilshr{64}{\horstFreeVaru,\horstOpAppishl{64}{\horstFreeVard,\horstOpAppmkConst{3}{}}},\horstOpAppmkConst{\horstSUB{\horstOpApppow{8}{2}}{1}}{}}}}
    \horstPremise{\horstEQ{\horstOpApplabelOf{}{\horstFreeVarw}}{\horstOpAppflub{4}{\horstTUPINIT{\horstOpApplabelOf{}{\horstFreeVarx},\horstOpApplabelOf{}{\horstFreeVary},\horstOpApplabelOf{}{\horstFreeVarsize},\horstOpApplabelOfCtx{}{\horstFreeVarctx}}}}}
    \horstPremise{\horstEQ{\horstFreeVarnv}{\horstCOND{\horstOpAppiltu{64}{\horstFreeVard,\horstOpAppmkConst{\horstDIV{\horstOpAppmemOpTbw{\horstParVarop}{}}{8}}{}}}{\horstFreeVarw}{\horstFreeVarv}}}
    \horstConclusion{\horstPredAppMState{\horstParVarfid,\horstADD{\horstParVarpc}{1}}{\horstFreeVarctx,\horstFreeVarst,\horstFreeVargt,\horstFreeVarlt,\horstConstructorAppMem{\horstFreeVari,\horstFreeVarnv,\horstFreeVarsize},\horstFreeVartbl,\horstFreeVaratN,\horstFreeVargtN,\horstFreeVarmemN}}
  \end{horstClause}
\end{horstRule}
\begin{horstRule}{brIfRule}
  \horstParVar{fid}{\horstTypeint}
  \horstParVar{pc}{\horstTypeint}
  \horstParVar{br}{\horstTypeint}
  \horstParVar{n}{\horstTypeint}
  \horstSelectorFunctionInvocation{\horstSelectorFunctionAppfunctionIds{\horstParVarfid}{},\horstSelectorFunctionApppcsForFunctionIdAndOpcode{\horstParVarpc}{\horstParVarfid,\horstConstBRIF},\horstSelectorFunctionAppbreakDestinations{\horstParVarbr}{\horstParVarfid,\horstParVarpc},\horstSelectorFunctionAppgetAmountOfReturnValuesInBlock{\horstParVarn}{\horstParVarfid,\horstParVarpc}}
  \begin{horstClause}
    \horstFreeVar{memN}{\horstTypeMemory}
    \horstFreeVar{atN}{\horstTypeHomInit{\horstTypeLValue}{\horstOpAppas{\horstParVarfid}{}}}
    \horstFreeVar{st}{\horstTypeHomInit{\horstTypeLValue}{\horstSUB{\horstOpAppss{\horstParVarfid,\horstParVarpc}{}}{1}}}
    \horstFreeVar{tbl}{\horstTypeTable}
    \horstFreeVar{ctx}{\horstTypeContext}
    \horstFreeVar{lt}{\horstTypeHomInit{\horstTypeLValue}{\horstOpAppls{\horstParVarfid}{}}}
    \horstFreeVar{mem}{\horstTypeMemory}
    \horstFreeVar{x}{\horstTypeLValue}
    \horstFreeVar{gtN}{\horstTypeHomInit{\horstTypeLValue}{\horstOpAppgs{}{}}}
    \horstFreeVar{gt}{\horstTypeHomInit{\horstTypeLValue}{\horstOpAppgs{}{}}}
    \horstPremise{\horstPredAppMState{\horstParVarfid,\horstParVarpc}{\horstFreeVarctx,\horstCONS{\horstFreeVarx}{\horstFreeVarst},\horstFreeVargt,\horstFreeVarlt,\horstFreeVarmem,\horstFreeVartbl,\horstFreeVaratN,\horstFreeVargtN,\horstFreeVarmemN}}
    \horstPremise{\horstOpAppabsneq{}{\horstOpAppvalueOf{}{\horstFreeVarx},\horstOpAppmkConst{0}{}}}
    \horstConclusion{\horstPredAppMState{\horstParVarfid,\horstParVarbr}{\horstOpAppraiseCtxTo{\horstParVarpc}{\horstFreeVarctx,\horstOpApplabelOf{}{\horstFreeVarx}},\horstOpAppdrop{\horstSUB{\horstOpAppss{\horstParVarfid,\horstParVarpc}{}}{1},\horstParVarn,\horstSUB{\horstSUB{\horstOpAppss{\horstParVarfid,\horstParVarpc}{}}{1}}{\horstOpAppss{\horstParVarfid,\horstParVarbr}{}}}{\horstFreeVarst},\horstFreeVargt,\horstFreeVarlt,\horstFreeVarmem,\horstFreeVartbl,\horstFreeVaratN,\horstFreeVargtN,\horstFreeVarmemN}}
  \end{horstClause}
  \begin{horstClause}
    \horstFreeVar{memN}{\horstTypeMemory}
    \horstFreeVar{atN}{\horstTypeHomInit{\horstTypeLValue}{\horstOpAppas{\horstParVarfid}{}}}
    \horstFreeVar{st}{\horstTypeHomInit{\horstTypeLValue}{\horstSUB{\horstOpAppss{\horstParVarfid,\horstParVarpc}{}}{1}}}
    \horstFreeVar{p}{\horstTypeLabel}
    \horstFreeVar{tbl}{\horstTypeTable}
    \horstFreeVar{ctx}{\horstTypeContext}
    \horstFreeVar{lt}{\horstTypeHomInit{\horstTypeLValue}{\horstOpAppls{\horstParVarfid}{}}}
    \horstFreeVar{mem}{\horstTypeMemory}
    \horstFreeVar{x}{\horstTypeLValue}
    \horstFreeVar{gtN}{\horstTypeHomInit{\horstTypeLValue}{\horstOpAppgs{}{}}}
    \horstFreeVar{gt}{\horstTypeHomInit{\horstTypeLValue}{\horstOpAppgs{}{}}}
    \horstFreeVar{l}{\horstTypeFlowLabel}
    \horstFreeVar{from}{\horstTypeint}
    \horstPremise{\horstPredAppMState{\horstParVarfid,\horstParVarpc}{\horstFreeVarctx,\horstCONS{\horstFreeVarx}{\horstFreeVarst},\horstFreeVargt,\horstFreeVarlt,\horstFreeVarmem,\horstFreeVartbl,\horstFreeVaratN,\horstFreeVargtN,\horstFreeVarmemN}}
    \horstPremise{\horstOpAppabsneq{}{\horstOpAppvalueOf{}{\horstFreeVarx},\horstOpAppmkConst{0}{}}}
    \horstPremise{\horstEQ{\horstOpAppraiseCtxTo{\horstParVarpc}{\horstFreeVarctx,\horstOpApplabelOf{}{\horstFreeVarx}}}{\horstConstructorAppCtx{\horstFreeVarp,\horstConstructorAppIllegal{},\horstFreeVarfrom}}}
    \horstConclusion{\horstPredAppScopeExtend{\horstParVarfid}{\horstFreeVarfrom,\horstOpAppmax{}{\horstParVarpc,\horstParVarbr}}}
  \end{horstClause}
  \begin{horstClause}
    \horstFreeVar{memN}{\horstTypeMemory}
    \horstFreeVar{atN}{\horstTypeHomInit{\horstTypeLValue}{\horstOpAppas{\horstParVarfid}{}}}
    \horstFreeVar{st}{\horstTypeHomInit{\horstTypeLValue}{\horstSUB{\horstOpAppss{\horstParVarfid,\horstParVarpc}{}}{1}}}
    \horstFreeVar{tbl}{\horstTypeTable}
    \horstFreeVar{ctx}{\horstTypeContext}
    \horstFreeVar{lt}{\horstTypeHomInit{\horstTypeLValue}{\horstOpAppls{\horstParVarfid}{}}}
    \horstFreeVar{mem}{\horstTypeMemory}
    \horstFreeVar{x}{\horstTypeLValue}
    \horstFreeVar{gtN}{\horstTypeHomInit{\horstTypeLValue}{\horstOpAppgs{}{}}}
    \horstFreeVar{gt}{\horstTypeHomInit{\horstTypeLValue}{\horstOpAppgs{}{}}}
    \horstPremise{\horstLE{\horstParVarpc}{\horstParVarbr}}
    \horstPremise{\horstPredAppMState{\horstParVarfid,\horstParVarpc}{\horstFreeVarctx,\horstCONS{\horstFreeVarx}{\horstFreeVarst},\horstFreeVargt,\horstFreeVarlt,\horstFreeVarmem,\horstFreeVartbl,\horstFreeVaratN,\horstFreeVargtN,\horstFreeVarmemN}}
    \horstPremise{\horstOpAppabseq{}{\horstOpAppvalueOf{}{\horstFreeVarx},\horstOpAppmkConst{0}{}}}
    \horstConclusion{\horstPredAppMState{\horstParVarfid,\horstADD{\horstParVarpc}{1}}{\horstOpAppraiseCtxTo{\horstParVarpc}{\horstFreeVarctx,\horstOpApplabelOf{}{\horstFreeVarx}},\horstFreeVarst,\horstFreeVargt,\horstFreeVarlt,\horstFreeVarmem,\horstFreeVartbl,\horstFreeVaratN,\horstFreeVargtN,\horstFreeVarmemN}}
  \end{horstClause}
  \begin{horstClause}
    \horstFreeVar{memN}{\horstTypeMemory}
    \horstFreeVar{atN}{\horstTypeHomInit{\horstTypeLValue}{\horstOpAppas{\horstParVarfid}{}}}
    \horstFreeVar{st}{\horstTypeHomInit{\horstTypeLValue}{\horstSUB{\horstOpAppss{\horstParVarfid,\horstParVarpc}{}}{1}}}
    \horstFreeVar{tbl}{\horstTypeTable}
    \horstFreeVar{ctx}{\horstTypeContext}
    \horstFreeVar{lt}{\horstTypeHomInit{\horstTypeLValue}{\horstOpAppls{\horstParVarfid}{}}}
    \horstFreeVar{mem}{\horstTypeMemory}
    \horstFreeVar{x}{\horstTypeLValue}
    \horstFreeVar{gtN}{\horstTypeHomInit{\horstTypeLValue}{\horstOpAppgs{}{}}}
    \horstFreeVar{gt}{\horstTypeHomInit{\horstTypeLValue}{\horstOpAppgs{}{}}}
    \horstPremise{\horstLT{\horstParVarbr}{\horstParVarpc}}
    \horstPremise{\horstPredAppMState{\horstParVarfid,\horstParVarpc}{\horstFreeVarctx,\horstCONS{\horstFreeVarx}{\horstFreeVarst},\horstFreeVargt,\horstFreeVarlt,\horstFreeVarmem,\horstFreeVartbl,\horstFreeVaratN,\horstFreeVargtN,\horstFreeVarmemN}}
    \horstPremise{\horstEQ{\horstOpApplabelOfCtx{}{\horstFreeVarctx}}{\horstConstructorAppLegal{}}}
    \horstPremise{\horstEQ{\horstOpApplabelOf{}{\horstFreeVarx}}{\horstConstructorAppLegal{}}}
    \horstPremise{\horstOpAppabseq{}{\horstOpAppvalueOf{}{\horstFreeVarx},\horstOpAppmkConst{0}{}}}
    \horstConclusion{\horstPredAppMState{\horstParVarfid,\horstADD{\horstParVarpc}{1}}{\horstFreeVarctx,\horstFreeVarst,\horstFreeVargt,\horstFreeVarlt,\horstFreeVarmem,\horstFreeVartbl,\horstFreeVaratN,\horstFreeVargtN,\horstFreeVarmemN}}
  \end{horstClause}
  \begin{horstClause}
    \horstFreeVar{memN}{\horstTypeMemory}
    \horstFreeVar{atN}{\horstTypeHomInit{\horstTypeLValue}{\horstOpAppas{\horstParVarfid}{}}}
    \horstFreeVar{st}{\horstTypeHomInit{\horstTypeLValue}{\horstSUB{\horstOpAppss{\horstParVarfid,\horstParVarpc}{}}{1}}}
    \horstFreeVar{p}{\horstTypeLabel}
    \horstFreeVar{tbl}{\horstTypeTable}
    \horstFreeVar{ctx}{\horstTypeContext}
    \horstFreeVar{lt}{\horstTypeHomInit{\horstTypeLValue}{\horstOpAppls{\horstParVarfid}{}}}
    \horstFreeVar{mem}{\horstTypeMemory}
    \horstFreeVar{x}{\horstTypeLValue}
    \horstFreeVar{gtN}{\horstTypeHomInit{\horstTypeLValue}{\horstOpAppgs{}{}}}
    \horstFreeVar{gt}{\horstTypeHomInit{\horstTypeLValue}{\horstOpAppgs{}{}}}
    \horstPremise{\horstLT{\horstParVarbr}{\horstParVarpc}}
    \horstPremise{\horstPredAppMState{\horstParVarfid,\horstParVarpc}{\horstFreeVarctx,\horstCONS{\horstFreeVarx}{\horstFreeVarst},\horstFreeVargt,\horstFreeVarlt,\horstFreeVarmem,\horstFreeVartbl,\horstFreeVaratN,\horstFreeVargtN,\horstFreeVarmemN}}
    \horstPremise{\horstEQ{\horstOpAppraiseCtxTo{\horstParVarpc}{\horstFreeVarctx,\horstOpApplabelOf{}{\horstFreeVarx}}}{\horstConstructorAppCtx{\horstFreeVarp,\horstConstructorAppIllegal{},\horstParVarpc}}}
    \horstPremise{\horstOpAppabseq{}{\horstOpAppvalueOf{}{\horstFreeVarx},\horstOpAppmkConst{0}{}}}
    \horstConclusion{\horstPredAppMStateToJoin{\horstParVarfid,\horstADD{\horstParVarpc}{1}}{\horstConstructorAppCtx{\horstFreeVarp,\horstConstructorAppIllegal{},\horstParVarpc},\horstFreeVarst,\horstFreeVargt,\horstFreeVarlt,\horstFreeVarmem,\horstFreeVartbl,\horstFreeVaratN,\horstFreeVargtN,\horstFreeVarmemN}}
  \end{horstClause}
  \begin{horstClause}
    \horstFreeVar{memN}{\horstTypeMemory}
    \horstFreeVar{atN}{\horstTypeHomInit{\horstTypeLValue}{\horstOpAppas{\horstParVarfid}{}}}
    \horstFreeVar{st}{\horstTypeHomInit{\horstTypeLValue}{\horstSUB{\horstOpAppss{\horstParVarfid,\horstParVarpc}{}}{1}}}
    \horstFreeVar{p}{\horstTypeLabel}
    \horstFreeVar{tbl}{\horstTypeTable}
    \horstFreeVar{ctx}{\horstTypeContext}
    \horstFreeVar{lt}{\horstTypeHomInit{\horstTypeLValue}{\horstOpAppls{\horstParVarfid}{}}}
    \horstFreeVar{mem}{\horstTypeMemory}
    \horstFreeVar{x}{\horstTypeLValue}
    \horstFreeVar{to}{\horstTypeint}
    \horstFreeVar{gtN}{\horstTypeHomInit{\horstTypeLValue}{\horstOpAppgs{}{}}}
    \horstFreeVar{gt}{\horstTypeHomInit{\horstTypeLValue}{\horstOpAppgs{}{}}}
    \horstFreeVar{from}{\horstTypeint}
    \horstPremise{\horstLT{\horstParVarbr}{\horstParVarpc}}
    \horstPremise{\horstPredAppMState{\horstParVarfid,\horstParVarpc}{\horstFreeVarctx,\horstCONS{\horstFreeVarx}{\horstFreeVarst},\horstFreeVargt,\horstFreeVarlt,\horstFreeVarmem,\horstFreeVartbl,\horstFreeVaratN,\horstFreeVargtN,\horstFreeVarmemN}}
    \horstPremise{\horstOpAppabseq{}{\horstOpAppvalueOf{}{\horstFreeVarx},\horstOpAppmkConst{0}{}}}
    \horstPremise{\horstEQ{\horstOpAppraiseCtxTo{\horstParVarpc}{\horstFreeVarctx,\horstOpApplabelOf{}{\horstFreeVarx}}}{\horstConstructorAppCtx{\horstFreeVarp,\horstConstructorAppIllegal{},\horstFreeVarfrom}}}
    \horstPremise{\horstPredAppScopeExtend{\horstParVarfid}{\horstFreeVarfrom,\horstFreeVarto}}
    \horstPremise{\horstLT{\horstADD{\horstParVarpc}{1}}{\horstFreeVarto}}
    \horstConclusion{\horstPredAppMState{\horstParVarfid,\horstADD{\horstParVarpc}{1}}{\horstConstructorAppCtx{\horstFreeVarp,\horstConstructorAppIllegal{},\horstFreeVarfrom},\horstFreeVarst,\horstFreeVargt,\horstFreeVarlt,\horstFreeVarmem,\horstFreeVartbl,\horstFreeVaratN,\horstFreeVargtN,\horstFreeVarmemN}}
  \end{horstClause}
\end{horstRule}
\begin{horstRule}{testNoninterferenceImportedFunctionTableSat}
  \horstParVar{fid}{\horstTypeint}
  \horstParVar{cid}{\horstTypeint}
  \horstSelectorFunctionInvocation{\horstSelectorFunctionAppimportedFunctionIds{\horstParVarfid}{},\horstSelectorFunctionAppimportCallTableLeak{\horstParVarcid}{\horstParVarfid}}
  \begin{horstClause}
    \horstFreeVar{memN}{\horstTypeMemory}
    \horstFreeVar{atN}{\horstTypeHomInit{\horstTypeLValue}{\horstOpAppas{\horstParVarfid}{}}}
    \horstFreeVar{ptbl}{\horstTypeTablePrecision}
    \horstFreeVar{ctx}{\horstTypeContext}
    \horstFreeVar{lt}{\horstTypeHomInit{\horstTypeLValue}{\horstOpAppls{\horstParVarfid}{}}}
    \horstFreeVar{mem}{\horstTypeMemory}
    \horstFreeVar{gtN}{\horstTypeHomInit{\horstTypeLValue}{\horstOpAppgs{}{}}}
    \horstFreeVar{gt}{\horstTypeHomInit{\horstTypeLValue}{\horstOpAppgs{}{}}}
    \horstPremise{\horstPredAppMState{\horstParVarfid,0}{\horstFreeVarctx,\horstTUPINIT{},\horstFreeVargt,\horstFreeVarlt,\horstFreeVarmem,\horstConstructorAppTbl{\horstFreeVarptbl,\horstConstructorAppIllegal{}},\horstFreeVaratN,\horstFreeVargtN,\horstFreeVarmemN}}
    \horstPremise{\horstSimpleSumExp{OR}{\horstSelectorFunctionApptableInLabelForImportedFunction{\horstParVarct,\horstParVarit}{\horstParVarfid}}{\horstOpAppflowsTo{}{\horstOpAppmkLabel{}{\horstParVarct,\horstParVarit},\horstOpAppperspectiveOfCtx{}{\horstFreeVarctx}}}{{ct}{it}}}
    \horstConclusion{\horstPredApptestNoninterferenceImportedFunctionTableSat{\horstParVarfid,\horstParVarcid}{}}
  \end{horstClause}
\end{horstRule}
\begin{horstRule}{testNoninterferenceMemoryDataSat}
  \horstParVar{fid}{\horstTypeint}
  \horstParVar{cid}{\horstTypeint}
  \horstSelectorFunctionInvocation{\horstSelectorFunctionAppstartFunctionId{\horstParVarfid}{},\horstSelectorFunctionAppmemoryDataLeak{\horstParVarcid}{}}
  \begin{horstClause}
    \horstFreeVar{memN}{\horstTypeMemory}
    \horstFreeVar{atN}{\horstTypeHomInit{\horstTypeLValue}{\horstOpAppas{\horstParVarfid}{}}}
    \horstFreeVar{tbl}{\horstTypeTable}
    \horstFreeVar{rt}{\horstTypeHomInit{\horstTypeLValue}{\horstOpApprs{\horstParVarfid}{}}}
    \horstFreeVar{ctx}{\horstTypeContext}
    \horstFreeVar{mem}{\horstTypeMemory}
    \horstFreeVar{gtN}{\horstTypeHomInit{\horstTypeLValue}{\horstOpAppgs{}{}}}
    \horstFreeVar{gt}{\horstTypeHomInit{\horstTypeLValue}{\horstOpAppgs{}{}}}
    \horstPremise{\horstPredAppReturnCall{\horstParVarfid}{\horstFreeVarctx,\horstFreeVarrt,\horstFreeVargt,\horstFreeVarmem,\horstFreeVartbl,\horstFreeVaratN,\horstFreeVargtN,\horstFreeVarmemN}}
    \horstPremise{\horstSimpleSumExp{OR}{\horstSelectorFunctionAppmemoryDataOutLabel{\horstParVarcd,\horstParVarid}{}}{\horstAND{\horstOpAppflowsTo{}{\horstOpAppmkLabel{}{\horstParVarcd,\horstParVarid},\horstOpAppperspectiveOfCtx{}{\horstFreeVarctx}}}{\horstEQ{\horstOpApplabelOf{}{\horstOpAppvalueOfMem{}{\horstFreeVarmem}}}{\horstConstructorAppIllegal{}}}}{{cd}{id}}}
    \horstConclusion{\horstPredApptestNoninterferenceMemoryDataSat{\horstParVarfid,\horstParVarcid}{}}
  \end{horstClause}
\end{horstRule}
\begin{horstRule}{testNoninterferenceMemorAreaSat}
  \horstParVar{fid}{\horstTypeint}
  \horstParVar{cid}{\horstTypeint}
  \horstParVar{start}{\horstTypeint}
  \horstParVar{endInclusive}{\horstTypeint}
  \horstSelectorFunctionInvocation{\horstSelectorFunctionAppstartFunctionId{\horstParVarfid}{},\horstSelectorFunctionAppmemoryAreaLeak{\horstParVarcid,\horstParVarstart,\horstParVarendInclusive}{}}
  \begin{horstClause}
    \horstFreeVar{memN}{\horstTypeMemory}
    \horstFreeVar{atN}{\horstTypeHomInit{\horstTypeLValue}{\horstOpAppas{\horstParVarfid}{}}}
    \horstFreeVar{tbl}{\horstTypeTable}
    \horstFreeVar{rt}{\horstTypeHomInit{\horstTypeLValue}{\horstOpApprs{\horstParVarfid}{}}}
    \horstFreeVar{ctx}{\horstTypeContext}
    \horstFreeVar{mem}{\horstTypeMemory}
    \horstFreeVar{gtN}{\horstTypeHomInit{\horstTypeLValue}{\horstOpAppgs{}{}}}
    \horstFreeVar{gt}{\horstTypeHomInit{\horstTypeLValue}{\horstOpAppgs{}{}}}
    \horstPremise{\horstPredAppReturnCall{\horstParVarfid}{\horstFreeVarctx,\horstFreeVarrt,\horstFreeVargt,\horstFreeVarmem,\horstFreeVartbl,\horstFreeVaratN,\horstFreeVargtN,\horstFreeVarmemN}}
    \horstPremise{\horstOpAppabsle{}{\horstOpAppmkValue{}{\horstParVarstart},\horstOpAppindexOfMem{}{\horstFreeVarmem}}}
    \horstPremise{\horstOpAppabsle{}{\horstOpAppindexOfMem{}{\horstFreeVarmem},\horstOpAppmkValue{}{\horstParVarendInclusive}}}
    \horstPremise{\horstSimpleSumExp{OR}{\horstSelectorFunctionAppmemoryDataOutLabel{\horstParVarcd,\horstParVarid}{}}{\horstAND{\horstOpAppflowsTo{}{\horstOpAppmkLabel{}{\horstParVarcd,\horstParVarid},\horstOpAppperspectiveOfCtx{}{\horstFreeVarctx}}}{\horstEQ{\horstOpApplabelOf{}{\horstOpAppvalueOfMem{}{\horstFreeVarmem}}}{\horstConstructorAppIllegal{}}}}{{cd}{id}}}
    \horstConclusion{\horstPredApptestNoninterferenceMemorAreaSat{\horstParVarfid,\horstParVarcid,\horstParVarstart,\horstParVarendInclusive}{}}
  \end{horstClause}
\end{horstRule}
\begin{horstRule}{testNoninterferenceMemorySizeUnsat}
  \horstParVar{fid}{\horstTypeint}
  \horstParVar{cid}{\horstTypeint}
  \horstSelectorFunctionInvocation{\horstSelectorFunctionAppstartFunctionId{\horstParVarfid}{},\horstSelectorFunctionAppmemorySizeSafe{\horstParVarcid}{}}
  \begin{horstClause}
    \horstFreeVar{memN}{\horstTypeMemory}
    \horstFreeVar{atN}{\horstTypeHomInit{\horstTypeLValue}{\horstOpAppas{\horstParVarfid}{}}}
    \horstFreeVar{tbl}{\horstTypeTable}
    \horstFreeVar{rt}{\horstTypeHomInit{\horstTypeLValue}{\horstOpApprs{\horstParVarfid}{}}}
    \horstFreeVar{ctx}{\horstTypeContext}
    \horstFreeVar{mem}{\horstTypeMemory}
    \horstFreeVar{gtN}{\horstTypeHomInit{\horstTypeLValue}{\horstOpAppgs{}{}}}
    \horstFreeVar{gt}{\horstTypeHomInit{\horstTypeLValue}{\horstOpAppgs{}{}}}
    \horstPremise{\horstPredAppReturnCall{\horstParVarfid}{\horstFreeVarctx,\horstFreeVarrt,\horstFreeVargt,\horstFreeVarmem,\horstFreeVartbl,\horstFreeVaratN,\horstFreeVargtN,\horstFreeVarmemN}}
    \horstPremise{\horstSimpleSumExp{OR}{\horstSelectorFunctionAppmemorySizeOutLabel{\horstParVarcs,\horstParVaris}{}}{\horstAND{\horstOpAppflowsTo{}{\horstOpAppmkLabel{}{\horstParVarcs,\horstParVaris},\horstOpAppperspectiveOfCtx{}{\horstFreeVarctx}}}{\horstEQ{\horstOpApplabelOf{}{\horstOpAppsizeOfMem{}{\horstFreeVarmem}}}{\horstConstructorAppIllegal{}}}}{{cs}{is}}}
    \horstConclusion{\horstPredApptestNoninterferenceMemorySizeUnsat{\horstParVarfid,\horstParVarcid}{}}
  \end{horstClause}
\end{horstRule}
\begin{horstRule}{testNoninterferenceImportedFunctionContextUnsat}
  \horstParVar{fid}{\horstTypeint}
  \horstParVar{cid}{\horstTypeint}
  \horstSelectorFunctionInvocation{\horstSelectorFunctionAppimportedFunctionIds{\horstParVarfid}{},\horstSelectorFunctionAppimportCallContextSafe{\horstParVarcid}{\horstParVarfid}}
  \begin{horstClause}
    \horstFreeVar{memN}{\horstTypeMemory}
    \horstFreeVar{atN}{\horstTypeHomInit{\horstTypeLValue}{\horstOpAppas{\horstParVarfid}{}}}
    \horstFreeVar{tbl}{\horstTypeTable}
    \horstFreeVar{ctx}{\horstTypeContext}
    \horstFreeVar{lt}{\horstTypeHomInit{\horstTypeLValue}{\horstOpAppls{\horstParVarfid}{}}}
    \horstFreeVar{mem}{\horstTypeMemory}
    \horstFreeVar{gtN}{\horstTypeHomInit{\horstTypeLValue}{\horstOpAppgs{}{}}}
    \horstFreeVar{gt}{\horstTypeHomInit{\horstTypeLValue}{\horstOpAppgs{}{}}}
    \horstPremise{\horstPredAppMState{\horstParVarfid,0}{\horstFreeVarctx,\horstTUPINIT{},\horstFreeVargt,\horstFreeVarlt,\horstFreeVarmem,\horstFreeVartbl,\horstFreeVaratN,\horstFreeVargtN,\horstFreeVarmemN}}
    \horstPremise{\horstSimpleSumExp{OR}{\horstSelectorFunctionAppcontextLabelForImportedFunction{\horstParVarcc,\horstParVaric}{\horstParVarfid}}{\horstAND{\horstOpAppflowsTo{}{\horstOpAppmkLabel{}{\horstParVarcc,\horstParVaric},\horstOpAppperspectiveOfCtx{}{\horstFreeVarctx}}}{\horstEQ{\horstOpApplabelOfCtx{}{\horstFreeVarctx}}{\horstConstructorAppIllegal{}}}}{{cc}{ic}}}
    \horstConclusion{\horstPredApptestNoninterferenceImportedFunctionContextUnsat{\horstParVarfid,\horstParVarcid}{}}
  \end{horstClause}
\end{horstRule}
\begin{horstRule}{localGetRule}
  \horstParVar{fid}{\horstTypeint}
  \horstParVar{pc}{\horstTypeint}
  \horstParVar{idx}{\horstTypeint}
  \horstSelectorFunctionInvocation{\horstSelectorFunctionAppfunctionIds{\horstParVarfid}{},\horstSelectorFunctionApppcsForFunctionIdAndOpcode{\horstParVarpc}{\horstParVarfid,\horstConstLOCALGET},\horstSelectorFunctionAppimmediateForFunctionIdAndPc{\horstParVaridx}{\horstParVarfid,\horstParVarpc}}
  \begin{horstClause}
    \horstFreeVar{memN}{\horstTypeMemory}
    \horstFreeVar{atN}{\horstTypeHomInit{\horstTypeLValue}{\horstOpAppas{\horstParVarfid}{}}}
    \horstFreeVar{st}{\horstTypeHomInit{\horstTypeLValue}{\horstOpAppss{\horstParVarfid,\horstParVarpc}{}}}
    \horstFreeVar{tbl}{\horstTypeTable}
    \horstFreeVar{ctx}{\horstTypeContext}
    \horstFreeVar{lt}{\horstTypeHomInit{\horstTypeLValue}{\horstOpAppls{\horstParVarfid}{}}}
    \horstFreeVar{mem}{\horstTypeMemory}
    \horstFreeVar{gtN}{\horstTypeHomInit{\horstTypeLValue}{\horstOpAppgs{}{}}}
    \horstFreeVar{gt}{\horstTypeHomInit{\horstTypeLValue}{\horstOpAppgs{}{}}}
    \horstPremise{\horstPredAppMState{\horstParVarfid,\horstParVarpc}{\horstFreeVarctx,\horstFreeVarst,\horstFreeVargt,\horstFreeVarlt,\horstFreeVarmem,\horstFreeVartbl,\horstFreeVaratN,\horstFreeVargtN,\horstFreeVarmemN}}
    \horstConclusion{\horstPredAppMState{\horstParVarfid,\horstADD{\horstParVarpc}{1}}{\horstFreeVarctx,\horstCONS{\horstOpAppraiseTo{}{\horstACCESS{\horstFreeVarlt}{\horstParVaridx},\horstOpApplabelOfCtx{}{\horstFreeVarctx}}}{\horstFreeVarst},\horstFreeVargt,\horstFreeVarlt,\horstFreeVarmem,\horstFreeVartbl,\horstFreeVaratN,\horstFreeVargtN,\horstFreeVarmemN}}
  \end{horstClause}
\end{horstRule}
\begin{horstRule}{testNoninterferenceGlobalSat}
  \horstParVar{fid}{\horstTypeint}
  \horstParVar{cid}{\horstTypeint}
  \horstSelectorFunctionInvocation{\horstSelectorFunctionAppstartFunctionId{\horstParVarfid}{},\horstSelectorFunctionAppglobalLeak{\horstParVarcid}{}}
  \begin{horstClause}
    \horstFreeVar{memN}{\horstTypeMemory}
    \horstFreeVar{atN}{\horstTypeHomInit{\horstTypeLValue}{\horstOpAppas{\horstParVarfid}{}}}
    \horstFreeVar{tbl}{\horstTypeTable}
    \horstFreeVar{rt}{\horstTypeHomInit{\horstTypeLValue}{\horstOpApprs{\horstParVarfid}{}}}
    \horstFreeVar{ctx}{\horstTypeContext}
    \horstFreeVar{mem}{\horstTypeMemory}
    \horstFreeVar{gtN}{\horstTypeHomInit{\horstTypeLValue}{\horstOpAppgs{}{}}}
    \horstFreeVar{gt}{\horstTypeHomInit{\horstTypeLValue}{\horstOpAppgs{}{}}}
    \horstPremise{\horstPredAppReturnCall{\horstParVarfid}{\horstFreeVarctx,\horstFreeVarrt,\horstFreeVargt,\horstFreeVarmem,\horstFreeVartbl,\horstFreeVaratN,\horstFreeVargtN,\horstFreeVarmemN}}
    \horstPremise{\horstSimpleSumExp{OR}{\horstSelectorFunctionAppinterval{\horstParVaridx}{0,\horstOpAppgs{}{}},\horstSelectorFunctionAppglobalOutLabelForPosition{\horstParVarcg,\horstParVarig}{\horstParVaridx}}{\horstAND{\horstOpAppflowsTo{}{\horstOpAppmkLabel{}{\horstParVarcg,\horstParVarig},\horstOpAppperspectiveOfCtx{}{\horstFreeVarctx}}}{\horstEQ{\horstOpApplabelOf{}{\horstACCESS{\horstFreeVargt}{\horstParVaridx}}}{\horstConstructorAppIllegal{}}}}{{idx}{cg}{ig}}}
    \horstConclusion{\horstPredApptestNoninterferenceGlobalSat{\horstParVarfid,\horstParVarcid}{}}
  \end{horstClause}
\end{horstRule}
\begin{horstRule}{testNoninterferenceImportedFunctionTableUnsat}
  \horstParVar{fid}{\horstTypeint}
  \horstParVar{cid}{\horstTypeint}
  \horstSelectorFunctionInvocation{\horstSelectorFunctionAppimportedFunctionIds{\horstParVarfid}{},\horstSelectorFunctionAppimportCallTableSafe{\horstParVarcid}{\horstParVarfid}}
  \begin{horstClause}
    \horstFreeVar{memN}{\horstTypeMemory}
    \horstFreeVar{atN}{\horstTypeHomInit{\horstTypeLValue}{\horstOpAppas{\horstParVarfid}{}}}
    \horstFreeVar{ptbl}{\horstTypeTablePrecision}
    \horstFreeVar{ctx}{\horstTypeContext}
    \horstFreeVar{lt}{\horstTypeHomInit{\horstTypeLValue}{\horstOpAppls{\horstParVarfid}{}}}
    \horstFreeVar{mem}{\horstTypeMemory}
    \horstFreeVar{gtN}{\horstTypeHomInit{\horstTypeLValue}{\horstOpAppgs{}{}}}
    \horstFreeVar{gt}{\horstTypeHomInit{\horstTypeLValue}{\horstOpAppgs{}{}}}
    \horstPremise{\horstPredAppMState{\horstParVarfid,0}{\horstFreeVarctx,\horstTUPINIT{},\horstFreeVargt,\horstFreeVarlt,\horstFreeVarmem,\horstConstructorAppTbl{\horstFreeVarptbl,\horstConstructorAppIllegal{}},\horstFreeVaratN,\horstFreeVargtN,\horstFreeVarmemN}}
    \horstPremise{\horstSimpleSumExp{OR}{\horstSelectorFunctionApptableInLabelForImportedFunction{\horstParVarct,\horstParVarit}{\horstParVarfid}}{\horstOpAppflowsTo{}{\horstOpAppmkLabel{}{\horstParVarct,\horstParVarit},\horstOpAppperspectiveOfCtx{}{\horstFreeVarctx}}}{{ct}{it}}}
    \horstConclusion{\horstPredApptestNoninterferenceImportedFunctionTableUnsat{\horstParVarfid,\horstParVarcid}{}}
  \end{horstClause}
\end{horstRule}
\begin{horstRule}{constRule}
  \horstParVar{fid}{\horstTypeint}
  \horstParVar{pc}{\horstTypeint}
  \horstParVar{v}{\horstTypeint}
  \horstParVar{top}{\horstTypebool}
  \horstSelectorFunctionInvocation{\horstSelectorFunctionAppfunctionIds{\horstParVarfid}{},\horstSelectorFunctionApppcsAndValueAndTopOfConstsForFunctionId{\horstParVarpc,\horstParVarv,\horstParVartop}{\horstParVarfid}}
  \begin{horstClause}
    \horstFreeVar{memN}{\horstTypeMemory}
    \horstFreeVar{atN}{\horstTypeHomInit{\horstTypeLValue}{\horstOpAppas{\horstParVarfid}{}}}
    \horstFreeVar{st}{\horstTypeHomInit{\horstTypeLValue}{\horstOpAppss{\horstParVarfid,\horstParVarpc}{}}}
    \horstFreeVar{tbl}{\horstTypeTable}
    \horstFreeVar{ctx}{\horstTypeContext}
    \horstFreeVar{lt}{\horstTypeHomInit{\horstTypeLValue}{\horstOpAppls{\horstParVarfid}{}}}
    \horstFreeVar{mem}{\horstTypeMemory}
    \horstFreeVar{gtN}{\horstTypeHomInit{\horstTypeLValue}{\horstOpAppgs{}{}}}
    \horstFreeVar{gt}{\horstTypeHomInit{\horstTypeLValue}{\horstOpAppgs{}{}}}
    \horstPremise{\horstPredAppMState{\horstParVarfid,\horstParVarpc}{\horstFreeVarctx,\horstFreeVarst,\horstFreeVargt,\horstFreeVarlt,\horstFreeVarmem,\horstFreeVartbl,\horstFreeVaratN,\horstFreeVargtN,\horstFreeVarmemN}}
    \horstConclusion{\horstPredAppMState{\horstParVarfid,\horstADD{\horstParVarpc}{1}}{\horstFreeVarctx,\horstCONS{\horstConstructorAppLVal{\horstOpAppval{\horstParVartop,\horstParVarv}{},\horstOpApplabelOfCtx{}{\horstFreeVarctx}}}{\horstFreeVarst},\horstFreeVargt,\horstFreeVarlt,\horstFreeVarmem,\horstFreeVartbl,\horstFreeVaratN,\horstFreeVargtN,\horstFreeVarmemN}}
  \end{horstClause}
\end{horstRule}
\begin{horstRule}{importedFunctionRule}
  \horstParVar{fid}{\horstTypeint}
  \horstSelectorFunctionInvocation{\horstSelectorFunctionAppimportedFunctionIds{\horstParVarfid}{}}
  \begin{horstClause}
    \horstFreeVar{atN}{\horstTypeHomInit{\horstTypeLValue}{\horstOpAppas{\horstParVarfid}{}}}
    \horstFreeVar{rt}{\horstTypeHomInit{\horstTypeLValue}{\horstOpApprs{\horstParVarfid}{}}}
    \horstFreeVar{tbl}{\horstTypeTable}
    \horstFreeVar{rtbl}{\horstTypeTable}
    \horstFreeVar{nv}{\horstTypeLValue}
    \horstFreeVar{lt}{\horstTypeHomInit{\horstTypeLValue}{\horstOpAppls{\horstParVarfid}{}}}
    \horstFreeVar{rfrom}{\horstTypeint}
    \horstFreeVar{i}{\horstTypeValue}
    \horstFreeVar{gt}{\horstTypeHomInit{\horstTypeLValue}{\horstOpAppgs{}{}}}
    \horstFreeVar{size}{\horstTypeLValue}
    \horstFreeVar{l}{\horstTypeFlowLabel}
    \horstFreeVar{from}{\horstTypeint}
    \horstFreeVar{memN}{\horstTypeMemory}
    \horstFreeVar{p}{\horstTypeLabel}
    \horstFreeVar{nsize}{\horstTypeLValue}
    \horstFreeVar{v}{\horstTypeLValue}
    \horstFreeVar{rl}{\horstTypeFlowLabel}
    \horstFreeVar{gtN}{\horstTypeHomInit{\horstTypeLValue}{\horstOpAppgs{}{}}}
    \horstFreeVar{rgt}{\horstTypeHomInit{\horstTypeLValue}{\horstOpAppgs{}{}}}
    \horstPremise{\horstPredAppMState{\horstParVarfid,0}{\horstConstructorAppCtx{\horstFreeVarp,\horstFreeVarl,\horstFreeVarfrom},\horstTUPINIT{},\horstFreeVargt,\horstFreeVarlt,\horstConstructorAppMem{\horstFreeVari,\horstFreeVarv,\horstFreeVarsize},\horstFreeVartbl,\horstFreeVaratN,\horstFreeVargtN,\horstFreeVarmemN}}
    \horstPremise{\horstSimpleSumExp{AND}{\horstSelectorFunctionAppinterval{\horstParVari}{0,\horstOpApprs{\horstParVarfid}{}},\horstSelectorFunctionAppbitwidthForResult{\horstParVarbw}{\horstParVarfid,\horstParVari},\horstSelectorFunctionAppresultLabelForImportedFunctionAndPosition{\horstParVarcr,\horstParVarir}{\horstParVarfid,\horstParVari}}{\horstAND{\horstOpAppisInRange{\horstParVarbw}{\horstOpAppvalueOf{}{\horstACCESS{\horstFreeVarrt}{\horstParVari}}}}{\horstEQ{\horstOpApplabelOf{}{\horstACCESS{\horstFreeVarrt}{\horstParVari}}}{\horstOpAppflub{2}{\horstTUPINIT{\horstFreeVarl,\horstOpAppmkFlowLabel{}{\horstOpAppmkLabel{}{\horstParVarcr,\horstParVarir},\horstFreeVarp}}}}}}{{i}{bw}{cr}{ir}}}
    \horstPremise{\horstSimpleSumExp{AND}{\horstSelectorFunctionAppinterval{\horstParVari}{0,\horstOpAppgs{}{}},\horstSelectorFunctionAppbitwidthForGlobal{\horstParVarbw}{\horstParVari},\horstSelectorFunctionAppglobalOutLabelForImportedFunctionAndPosition{\horstParVarpass,\horstParVarcg,\horstParVarig}{\horstParVarfid,\horstParVari}}{\horstCOND{\horstParVarpass}{\horstEQ{\horstACCESS{\horstFreeVargt}{\horstParVari}}{\horstACCESS{\horstFreeVarrgt}{\horstParVari}}}{\horstAND{\horstOpAppisInRange{\horstParVarbw}{\horstOpAppvalueOf{}{\horstACCESS{\horstFreeVarrgt}{\horstParVari}}}}{\horstEQ{\horstOpApplabelOf{}{\horstACCESS{\horstFreeVarrgt}{\horstParVari}}}{\horstOpAppflub{2}{\horstTUPINIT{\horstFreeVarl,\horstOpAppmkFlowLabel{}{\horstOpAppmkLabel{}{\horstParVarcg,\horstParVarig},\horstFreeVarp}}}}}}}{{i}{bw}{pass}{cg}{ig}}}
    \horstPremise{\horstSimpleSumExp{AND}{\horstSelectorFunctionAppmemoryDataOutLabelForImportedFunction{\horstParVarpass,\horstParVarcd,\horstParVarid}{\horstParVarfid}}{\horstCOND{\horstParVarpass}{\horstEQ{\horstFreeVarv}{\horstFreeVarnv}}{\horstAND{\horstOpAppisInRange{8}{\horstOpAppvalueOf{}{\horstFreeVarnv}}}{\horstEQ{\horstOpApplabelOf{}{\horstFreeVarnv}}{\horstOpAppflub{2}{\horstTUPINIT{\horstFreeVarl,\horstOpAppmkFlowLabel{}{\horstOpAppmkLabel{}{\horstParVarcd,\horstParVarid},\horstFreeVarp}}}}}}}{{pass}{cd}{id}}}
    \horstPremise{\horstSimpleSumExp{AND}{\horstSelectorFunctionAppmemorySizeOutLabelForImportedFunction{\horstParVarpass,\horstParVarcs,\horstParVaris}{\horstParVarfid}}{\horstCOND{\horstParVarpass}{\horstEQ{\horstFreeVarsize}{\horstFreeVarnsize}}{\horstAND{\horstAND{\horstOpAppileu{64}{\horstOpAppvalueOf{}{\horstFreeVarsize},\horstOpAppvalueOf{}{\horstFreeVarnsize}}}{\horstOpAppileu{64}{\horstOpAppvalueOf{}{\horstFreeVarnsize},\horstOpAppmkConst{\horstMUL{\horstOpAppmms{}{}}{\horstOpApppow{16}{2}}}{}}}}{\horstEQ{\horstOpApplabelOf{}{\horstFreeVarnsize}}{\horstOpAppflub{2}{\horstTUPINIT{\horstFreeVarl,\horstOpAppmkFlowLabel{}{\horstOpAppmkLabel{}{\horstParVarcs,\horstParVaris},\horstFreeVarp}}}}}}}{{pass}{cs}{is}}}
    \horstPremise{\horstSimpleSumExp{AND}{\horstSelectorFunctionApptableOutLabelForImportedFunction{\horstParVarpass,\horstParVarct,\horstParVarit}{\horstParVarfid}}{\horstCOND{\horstParVarpass}{\horstEQ{\horstFreeVartbl}{\horstFreeVarrtbl}}{\horstEQ{\horstFreeVarrtbl}{\horstConstructorAppTbl{\horstConstructorAppTblImprecise{},\horstOpAppflub{2}{\horstTUPINIT{\horstFreeVarl,\horstOpAppmkFlowLabel{}{\horstOpAppmkLabel{}{\horstParVarct,\horstParVarit},\horstFreeVarp}}}}}}}{{pass}{ct}{it}}}
    \horstPremise{\horstOR{\horstEQ{\horstFreeVarl}{\horstConstructorAppLegal{}}}{\horstAND{\horstEQ{\horstFreeVarrl}{\horstConstructorAppIllegal{}}}{\horstEQ{\horstFreeVarrfrom}{-1}}}}
    \horstConclusion{\horstPredAppReturn{\horstParVarfid}{\horstConstructorAppCtx{\horstFreeVarp,\horstFreeVarrl,\horstFreeVarrfrom},\horstFreeVarrt,\horstFreeVarrgt,\horstConstructorAppMem{\horstFreeVari,\horstFreeVarnv,\horstFreeVarnsize},\horstFreeVarrtbl,\horstFreeVaratN,\horstFreeVargtN,\horstFreeVarmemN}}
  \end{horstClause}
\end{horstRule}
\begin{horstRule}{testNoninterferenceImportedFunctionMemorySizeUnsat}
  \horstParVar{fid}{\horstTypeint}
  \horstParVar{cid}{\horstTypeint}
  \horstSelectorFunctionInvocation{\horstSelectorFunctionAppimportedFunctionIds{\horstParVarfid}{},\horstSelectorFunctionAppimportCallMemorySizeSafe{\horstParVarcid}{\horstParVarfid}}
  \begin{horstClause}
    \horstFreeVar{memN}{\horstTypeMemory}
    \horstFreeVar{atN}{\horstTypeHomInit{\horstTypeLValue}{\horstOpAppas{\horstParVarfid}{}}}
    \horstFreeVar{tbl}{\horstTypeTable}
    \horstFreeVar{ctx}{\horstTypeContext}
    \horstFreeVar{v}{\horstTypeLValue}
    \horstFreeVar{lt}{\horstTypeHomInit{\horstTypeLValue}{\horstOpAppls{\horstParVarfid}{}}}
    \horstFreeVar{gtN}{\horstTypeHomInit{\horstTypeLValue}{\horstOpAppgs{}{}}}
    \horstFreeVar{i}{\horstTypeValue}
    \horstFreeVar{gt}{\horstTypeHomInit{\horstTypeLValue}{\horstOpAppgs{}{}}}
    \horstFreeVar{size}{\horstTypeLValue}
    \horstPremise{\horstPredAppMState{\horstParVarfid,0}{\horstFreeVarctx,\horstTUPINIT{},\horstFreeVargt,\horstFreeVarlt,\horstConstructorAppMem{\horstFreeVari,\horstFreeVarv,\horstFreeVarsize},\horstFreeVartbl,\horstFreeVaratN,\horstFreeVargtN,\horstFreeVarmemN}}
    \horstPremise{\horstSimpleSumExp{OR}{\horstSelectorFunctionAppmemoryDataInLabelForImportedFunction{\horstParVarcs,\horstParVaris}{\horstParVarfid}}{\horstAND{\horstOpAppflowsTo{}{\horstOpAppmkLabel{}{\horstParVarcs,\horstParVaris},\horstOpAppperspectiveOfCtx{}{\horstFreeVarctx}}}{\horstEQ{\horstOpApplabelOf{}{\horstFreeVarsize}}{\horstConstructorAppIllegal{}}}}{{cs}{is}}}
    \horstConclusion{\horstPredApptestNoninterferenceImportedFunctionMemorySizeUnsat{\horstParVarfid,\horstParVarcid}{}}
  \end{horstClause}
\end{horstRule}
\begin{horstRule}{joinRule}
  \horstParVar{fid}{\horstTypeint}
  \horstParVar{pc}{\horstTypeint}
  \horstSelectorFunctionInvocation{\horstSelectorFunctionAppfunctionIds{\horstParVarfid}{},\horstSelectorFunctionAppjoinsForFunctionId{\horstParVarpc}{\horstParVarfid}}
  \begin{horstClause}
    \horstFreeVar{gtNI}{\horstTypeHomInit{\horstTypeLValue}{\horstOpAppgs{}{}}}
    \horstFreeVar{gtNII}{\horstTypeHomInit{\horstTypeLValue}{\horstOpAppgs{}{}}}
    \horstFreeVar{memNI}{\horstTypeMemory}
    \horstFreeVar{memNII}{\horstTypeMemory}
    \horstFreeVar{stII}{\horstTypeHomInit{\horstTypeLValue}{\horstOpAppss{\horstParVarfid,\horstParVarpc}{}}}
    \horstFreeVar{stIII}{\horstTypeHomInit{\horstTypeLValue}{\horstOpAppss{\horstParVarfid,\horstParVarpc}{}}}
    \horstFreeVar{stI}{\horstTypeHomInit{\horstTypeLValue}{\horstOpAppss{\horstParVarfid,\horstParVarpc}{}}}
    \horstFreeVar{atNII}{\horstTypeHomInit{\horstTypeLValue}{\horstOpAppas{\horstParVarfid}{}}}
    \horstFreeVar{memI}{\horstTypeMemory}
    \horstFreeVar{p}{\horstTypeLabel}
    \horstFreeVar{ltIII}{\horstTypeHomInit{\horstTypeLValue}{\horstOpAppls{\horstParVarfid}{}}}
    \horstFreeVar{ltI}{\horstTypeHomInit{\horstTypeLValue}{\horstOpAppls{\horstParVarfid}{}}}
    \horstFreeVar{ltII}{\horstTypeHomInit{\horstTypeLValue}{\horstOpAppls{\horstParVarfid}{}}}
    \horstFreeVar{memII}{\horstTypeMemory}
    \horstFreeVar{atNI}{\horstTypeHomInit{\horstTypeLValue}{\horstOpAppas{\horstParVarfid}{}}}
    \horstFreeVar{ctx}{\horstTypeContext}
    \horstFreeVar{memIII}{\horstTypeMemory}
    \horstFreeVar{gtII}{\horstTypeHomInit{\horstTypeLValue}{\horstOpAppgs{}{}}}
    \horstFreeVar{gtIII}{\horstTypeHomInit{\horstTypeLValue}{\horstOpAppgs{}{}}}
    \horstFreeVar{gtI}{\horstTypeHomInit{\horstTypeLValue}{\horstOpAppgs{}{}}}
    \horstFreeVar{tblI}{\horstTypeTable}
    \horstFreeVar{tblII}{\horstTypeTable}
    \horstPremise{\horstPredAppMStateToJoin{\horstParVarfid,\horstParVarpc}{\horstFreeVarctx,\horstFreeVarstI,\horstFreeVargtI,\horstFreeVarltI,\horstFreeVarmemI,\horstFreeVartblI,\horstFreeVaratNI,\horstFreeVargtNI,\horstFreeVarmemNI}}
    \horstPremise{\horstPredAppMStateToJoin{\horstParVarfid,\horstParVarpc}{\horstFreeVarctx,\horstFreeVarstII,\horstFreeVargtII,\horstFreeVarltII,\horstFreeVarmemII,\horstFreeVartblII,\horstFreeVaratNII,\horstFreeVargtNII,\horstFreeVarmemNII}}
    \horstPremise{\horstOpApplowEq{\horstOpAppas{\horstParVarfid}{}}{\horstFreeVaratNI,\horstFreeVaratNII}}
    \horstPremise{\horstOpApplowEq{\horstOpAppgs{}{}}{\horstFreeVargtNI,\horstFreeVargtNII}}
    \horstPremise{\horstOpApplowEqMem{}{\horstFreeVarmemNI,\horstFreeVarmemNII}}
    \horstPremise{\horstEQ{\horstOpAppperspectiveOfCtx{}{\horstFreeVarctx}}{\horstFreeVarp}}
    \horstPremise{\horstOpAppjoinTuples{\horstOpAppss{\horstParVarfid,\horstParVarpc}{}}{\horstFreeVarstI,\horstFreeVarstII,\horstFreeVarstIII}}
    \horstPremise{\horstOpAppjoinTuples{\horstOpAppgs{}{}}{\horstFreeVargtI,\horstFreeVargtII,\horstFreeVargtIII}}
    \horstPremise{\horstOpAppjoinTuples{\horstOpAppls{\horstParVarfid}{}}{\horstFreeVarltI,\horstFreeVarltII,\horstFreeVarltIII}}
    \horstPremise{\horstOpAppjoinMem{}{\horstFreeVarmemI,\horstFreeVarmemII,\horstFreeVarmemIII}}
    \horstConclusion{\horstPredAppMState{\horstParVarfid,\horstParVarpc}{\horstConstructorAppCtx{\horstFreeVarp,\horstConstructorAppLegal{},-1},\horstFreeVarstIII,\horstFreeVargtIII,\horstFreeVarltIII,\horstFreeVarmemIII,\horstFreeVartblI,\horstFreeVaratNI,\horstFreeVargtNI,\horstFreeVarmemNI}}
  \end{horstClause}
\end{horstRule}
\begin{horstRule}{dropRule}
  \horstParVar{fid}{\horstTypeint}
  \horstParVar{pc}{\horstTypeint}
  \horstSelectorFunctionInvocation{\horstSelectorFunctionAppfunctionIds{\horstParVarfid}{},\horstSelectorFunctionApppcsForFunctionIdAndOpcode{\horstParVarpc}{\horstParVarfid,\horstConstDROP}}
  \begin{horstClause}
    \horstFreeVar{memN}{\horstTypeMemory}
    \horstFreeVar{atN}{\horstTypeHomInit{\horstTypeLValue}{\horstOpAppas{\horstParVarfid}{}}}
    \horstFreeVar{st}{\horstTypeHomInit{\horstTypeLValue}{\horstSUB{\horstOpAppss{\horstParVarfid,\horstParVarpc}{}}{1}}}
    \horstFreeVar{tbl}{\horstTypeTable}
    \horstFreeVar{ctx}{\horstTypeContext}
    \horstFreeVar{lt}{\horstTypeHomInit{\horstTypeLValue}{\horstOpAppls{\horstParVarfid}{}}}
    \horstFreeVar{mem}{\horstTypeMemory}
    \horstFreeVar{x}{\horstTypeLValue}
    \horstFreeVar{gtN}{\horstTypeHomInit{\horstTypeLValue}{\horstOpAppgs{}{}}}
    \horstFreeVar{gt}{\horstTypeHomInit{\horstTypeLValue}{\horstOpAppgs{}{}}}
    \horstPremise{\horstPredAppMState{\horstParVarfid,\horstParVarpc}{\horstFreeVarctx,\horstCONS{\horstFreeVarx}{\horstFreeVarst},\horstFreeVargt,\horstFreeVarlt,\horstFreeVarmem,\horstFreeVartbl,\horstFreeVaratN,\horstFreeVargtN,\horstFreeVarmemN}}
    \horstConclusion{\horstPredAppMState{\horstParVarfid,\horstADD{\horstParVarpc}{1}}{\horstFreeVarctx,\horstFreeVarst,\horstFreeVargt,\horstFreeVarlt,\horstFreeVarmem,\horstFreeVartbl,\horstFreeVaratN,\horstFreeVargtN,\horstFreeVarmemN}}
  \end{horstClause}
\end{horstRule}
\begin{horstRule}{globalSetRule}
  \horstParVar{fid}{\horstTypeint}
  \horstParVar{pc}{\horstTypeint}
  \horstParVar{idx}{\horstTypeint}
  \horstSelectorFunctionInvocation{\horstSelectorFunctionAppfunctionIds{\horstParVarfid}{},\horstSelectorFunctionApppcsForFunctionIdAndOpcode{\horstParVarpc}{\horstParVarfid,\horstConstGLOBALSET},\horstSelectorFunctionAppimmediateForFunctionIdAndPc{\horstParVaridx}{\horstParVarfid,\horstParVarpc}}
  \begin{horstClause}
    \horstFreeVar{memN}{\horstTypeMemory}
    \horstFreeVar{atN}{\horstTypeHomInit{\horstTypeLValue}{\horstOpAppas{\horstParVarfid}{}}}
    \horstFreeVar{st}{\horstTypeHomInit{\horstTypeLValue}{\horstSUB{\horstOpAppss{\horstParVarfid,\horstParVarpc}{}}{1}}}
    \horstFreeVar{tbl}{\horstTypeTable}
    \horstFreeVar{ctx}{\horstTypeContext}
    \horstFreeVar{lt}{\horstTypeHomInit{\horstTypeLValue}{\horstOpAppls{\horstParVarfid}{}}}
    \horstFreeVar{mem}{\horstTypeMemory}
    \horstFreeVar{x}{\horstTypeLValue}
    \horstFreeVar{gtN}{\horstTypeHomInit{\horstTypeLValue}{\horstOpAppgs{}{}}}
    \horstFreeVar{gt}{\horstTypeHomInit{\horstTypeLValue}{\horstOpAppgs{}{}}}
    \horstPremise{\horstPredAppMState{\horstParVarfid,\horstParVarpc}{\horstFreeVarctx,\horstCONS{\horstFreeVarx}{\horstFreeVarst},\horstFreeVargt,\horstFreeVarlt,\horstFreeVarmem,\horstFreeVartbl,\horstFreeVaratN,\horstFreeVargtN,\horstFreeVarmemN}}
    \horstConclusion{\horstPredAppMState{\horstParVarfid,\horstADD{\horstParVarpc}{1}}{\horstFreeVarctx,\horstFreeVarst,\horstOpAppset{\horstOpAppgs{}{},\horstParVaridx}{\horstOpAppraiseTo{}{\horstFreeVarx,\horstOpApplabelOfCtx{}{\horstFreeVarctx}},\horstFreeVargt},\horstFreeVarlt,\horstFreeVarmem,\horstFreeVartbl,\horstFreeVaratN,\horstFreeVargtN,\horstFreeVarmemN}}
  \end{horstClause}
\end{horstRule}
\begin{horstRule}{nopRule}
  \horstParVar{fid}{\horstTypeint}
  \horstParVar{pc}{\horstTypeint}
  \horstSelectorFunctionInvocation{\horstSelectorFunctionAppfunctionIds{\horstParVarfid}{},\horstSelectorFunctionApppcsForFunctionIdAndOpcode{\horstParVarpc}{\horstParVarfid,\horstConstNOP}}
  \begin{horstClause}
    \horstFreeVar{memN}{\horstTypeMemory}
    \horstFreeVar{atN}{\horstTypeHomInit{\horstTypeLValue}{\horstOpAppas{\horstParVarfid}{}}}
    \horstFreeVar{st}{\horstTypeHomInit{\horstTypeLValue}{\horstOpAppss{\horstParVarfid,\horstParVarpc}{}}}
    \horstFreeVar{tbl}{\horstTypeTable}
    \horstFreeVar{ctx}{\horstTypeContext}
    \horstFreeVar{lt}{\horstTypeHomInit{\horstTypeLValue}{\horstOpAppls{\horstParVarfid}{}}}
    \horstFreeVar{mem}{\horstTypeMemory}
    \horstFreeVar{gtN}{\horstTypeHomInit{\horstTypeLValue}{\horstOpAppgs{}{}}}
    \horstFreeVar{gt}{\horstTypeHomInit{\horstTypeLValue}{\horstOpAppgs{}{}}}
    \horstPremise{\horstPredAppMState{\horstParVarfid,\horstParVarpc}{\horstFreeVarctx,\horstFreeVarst,\horstFreeVargt,\horstFreeVarlt,\horstFreeVarmem,\horstFreeVartbl,\horstFreeVaratN,\horstFreeVargtN,\horstFreeVarmemN}}
    \horstConclusion{\horstPredAppMState{\horstParVarfid,\horstADD{\horstParVarpc}{1}}{\horstFreeVarctx,\horstFreeVarst,\horstFreeVargt,\horstFreeVarlt,\horstFreeVarmem,\horstFreeVartbl,\horstFreeVaratN,\horstFreeVargtN,\horstFreeVarmemN}}
  \end{horstClause}
\end{horstRule}
\begin{horstRule}{brRule}
  \horstParVar{fid}{\horstTypeint}
  \horstParVar{pc}{\horstTypeint}
  \horstParVar{br}{\horstTypeint}
  \horstParVar{n}{\horstTypeint}
  \horstSelectorFunctionInvocation{\horstSelectorFunctionAppfunctionIds{\horstParVarfid}{},\horstSelectorFunctionApppcsForFunctionIdAndOpcode{\horstParVarpc}{\horstParVarfid,\horstConstBR},\horstSelectorFunctionAppbreakDestinations{\horstParVarbr}{\horstParVarfid,\horstParVarpc},\horstSelectorFunctionAppgetAmountOfReturnValuesInBlock{\horstParVarn}{\horstParVarfid,\horstParVarpc}}
  \begin{horstClause}
    \horstFreeVar{memN}{\horstTypeMemory}
    \horstFreeVar{atN}{\horstTypeHomInit{\horstTypeLValue}{\horstOpAppas{\horstParVarfid}{}}}
    \horstFreeVar{st}{\horstTypeHomInit{\horstTypeLValue}{\horstOpAppss{\horstParVarfid,\horstParVarpc}{}}}
    \horstFreeVar{tbl}{\horstTypeTable}
    \horstFreeVar{ctx}{\horstTypeContext}
    \horstFreeVar{lt}{\horstTypeHomInit{\horstTypeLValue}{\horstOpAppls{\horstParVarfid}{}}}
    \horstFreeVar{mem}{\horstTypeMemory}
    \horstFreeVar{gtN}{\horstTypeHomInit{\horstTypeLValue}{\horstOpAppgs{}{}}}
    \horstFreeVar{gt}{\horstTypeHomInit{\horstTypeLValue}{\horstOpAppgs{}{}}}
    \horstPremise{\horstPredAppMState{\horstParVarfid,\horstParVarpc}{\horstFreeVarctx,\horstFreeVarst,\horstFreeVargt,\horstFreeVarlt,\horstFreeVarmem,\horstFreeVartbl,\horstFreeVaratN,\horstFreeVargtN,\horstFreeVarmemN}}
    \horstConclusion{\horstPredAppMState{\horstParVarfid,\horstParVarbr}{\horstFreeVarctx,\horstOpAppdrop{\horstOpAppss{\horstParVarfid,\horstParVarpc}{},\horstParVarn,\horstSUB{\horstOpAppss{\horstParVarfid,\horstParVarpc}{}}{\horstOpAppss{\horstParVarfid,\horstParVarbr}{}}}{\horstFreeVarst},\horstFreeVargt,\horstFreeVarlt,\horstFreeVarmem,\horstFreeVartbl,\horstFreeVaratN,\horstFreeVargtN,\horstFreeVarmemN}}
  \end{horstClause}
  \begin{horstClause}
    \horstFreeVar{memN}{\horstTypeMemory}
    \horstFreeVar{atN}{\horstTypeHomInit{\horstTypeLValue}{\horstOpAppas{\horstParVarfid}{}}}
    \horstFreeVar{st}{\horstTypeHomInit{\horstTypeLValue}{\horstOpAppss{\horstParVarfid,\horstParVarpc}{}}}
    \horstFreeVar{tbl}{\horstTypeTable}
    \horstFreeVar{p}{\horstTypeLabel}
    \horstFreeVar{lt}{\horstTypeHomInit{\horstTypeLValue}{\horstOpAppls{\horstParVarfid}{}}}
    \horstFreeVar{mem}{\horstTypeMemory}
    \horstFreeVar{gtN}{\horstTypeHomInit{\horstTypeLValue}{\horstOpAppgs{}{}}}
    \horstFreeVar{gt}{\horstTypeHomInit{\horstTypeLValue}{\horstOpAppgs{}{}}}
    \horstFreeVar{from}{\horstTypeint}
    \horstPremise{\horstPredAppMState{\horstParVarfid,\horstParVarpc}{\horstConstructorAppCtx{\horstFreeVarp,\horstConstructorAppIllegal{},\horstFreeVarfrom},\horstFreeVarst,\horstFreeVargt,\horstFreeVarlt,\horstFreeVarmem,\horstFreeVartbl,\horstFreeVaratN,\horstFreeVargtN,\horstFreeVarmemN}}
    \horstConclusion{\horstPredAppScopeExtend{\horstParVarfid}{\horstFreeVarfrom,\horstParVarbr}}
  \end{horstClause}
\end{horstRule}
\begin{horstRule}{endRule}
  \horstParVar{fid}{\horstTypeint}
  \horstParVar{pc}{\horstTypeint}
  \horstParVar{next}{\horstTypeint}
  \horstSelectorFunctionInvocation{\horstSelectorFunctionAppfunctionIds{\horstParVarfid}{},\horstSelectorFunctionAppendsForFunctionId{\horstParVarpc,\horstParVarnext}{\horstParVarfid}}
  \begin{horstClause}
    \horstFreeVar{memN}{\horstTypeMemory}
    \horstFreeVar{atN}{\horstTypeHomInit{\horstTypeLValue}{\horstOpAppas{\horstParVarfid}{}}}
    \horstFreeVar{st}{\horstTypeHomInit{\horstTypeLValue}{\horstOpAppss{\horstParVarfid,\horstParVarpc}{}}}
    \horstFreeVar{tbl}{\horstTypeTable}
    \horstFreeVar{ctx}{\horstTypeContext}
    \horstFreeVar{lt}{\horstTypeHomInit{\horstTypeLValue}{\horstOpAppls{\horstParVarfid}{}}}
    \horstFreeVar{mem}{\horstTypeMemory}
    \horstFreeVar{gtN}{\horstTypeHomInit{\horstTypeLValue}{\horstOpAppgs{}{}}}
    \horstFreeVar{gt}{\horstTypeHomInit{\horstTypeLValue}{\horstOpAppgs{}{}}}
    \horstPremise{\horstPredAppMState{\horstParVarfid,\horstParVarpc}{\horstFreeVarctx,\horstFreeVarst,\horstFreeVargt,\horstFreeVarlt,\horstFreeVarmem,\horstFreeVartbl,\horstFreeVaratN,\horstFreeVargtN,\horstFreeVarmemN}}
    \horstPremise{\horstEQ{\horstOpApplabelOfCtx{}{\horstFreeVarctx}}{\horstConstructorAppLegal{}}}
    \horstConclusion{\horstPredAppMState{\horstParVarfid,\horstParVarnext}{\horstFreeVarctx,\horstFreeVarst,\horstFreeVargt,\horstFreeVarlt,\horstFreeVarmem,\horstFreeVartbl,\horstFreeVaratN,\horstFreeVargtN,\horstFreeVarmemN}}
  \end{horstClause}
  \begin{horstClause}
    \horstFreeVar{memN}{\horstTypeMemory}
    \horstFreeVar{atN}{\horstTypeHomInit{\horstTypeLValue}{\horstOpAppas{\horstParVarfid}{}}}
    \horstFreeVar{st}{\horstTypeHomInit{\horstTypeLValue}{\horstOpAppss{\horstParVarfid,\horstParVarpc}{}}}
    \horstFreeVar{tbl}{\horstTypeTable}
    \horstFreeVar{ctx}{\horstTypeContext}
    \horstFreeVar{lt}{\horstTypeHomInit{\horstTypeLValue}{\horstOpAppls{\horstParVarfid}{}}}
    \horstFreeVar{mem}{\horstTypeMemory}
    \horstFreeVar{gtN}{\horstTypeHomInit{\horstTypeLValue}{\horstOpAppgs{}{}}}
    \horstFreeVar{gt}{\horstTypeHomInit{\horstTypeLValue}{\horstOpAppgs{}{}}}
    \horstPremise{\horstPredAppMState{\horstParVarfid,\horstParVarpc}{\horstFreeVarctx,\horstFreeVarst,\horstFreeVargt,\horstFreeVarlt,\horstFreeVarmem,\horstFreeVartbl,\horstFreeVaratN,\horstFreeVargtN,\horstFreeVarmemN}}
    \horstPremise{\horstEQ{\horstOpApplabelOfCtx{}{\horstFreeVarctx}}{\horstConstructorAppIllegal{}}}
    \horstConclusion{\horstPredAppMStateToJoin{\horstParVarfid,\horstParVarnext}{\horstFreeVarctx,\horstFreeVarst,\horstFreeVargt,\horstFreeVarlt,\horstFreeVarmem,\horstFreeVartbl,\horstFreeVaratN,\horstFreeVargtN,\horstFreeVarmemN}}
  \end{horstClause}
  \begin{horstClause}
    \horstFreeVar{memN}{\horstTypeMemory}
    \horstFreeVar{atN}{\horstTypeHomInit{\horstTypeLValue}{\horstOpAppas{\horstParVarfid}{}}}
    \horstFreeVar{st}{\horstTypeHomInit{\horstTypeLValue}{\horstOpAppss{\horstParVarfid,\horstParVarpc}{}}}
    \horstFreeVar{tbl}{\horstTypeTable}
    \horstFreeVar{p}{\horstTypeLabel}
    \horstFreeVar{lt}{\horstTypeHomInit{\horstTypeLValue}{\horstOpAppls{\horstParVarfid}{}}}
    \horstFreeVar{mem}{\horstTypeMemory}
    \horstFreeVar{gtN}{\horstTypeHomInit{\horstTypeLValue}{\horstOpAppgs{}{}}}
    \horstFreeVar{to}{\horstTypeint}
    \horstFreeVar{gt}{\horstTypeHomInit{\horstTypeLValue}{\horstOpAppgs{}{}}}
    \horstFreeVar{from}{\horstTypeint}
    \horstPremise{\horstPredAppMState{\horstParVarfid,\horstParVarpc}{\horstConstructorAppCtx{\horstFreeVarp,\horstConstructorAppIllegal{},\horstFreeVarfrom},\horstFreeVarst,\horstFreeVargt,\horstFreeVarlt,\horstFreeVarmem,\horstFreeVartbl,\horstFreeVaratN,\horstFreeVargtN,\horstFreeVarmemN}}
    \horstPremise{\horstPredAppScopeExtend{\horstParVarfid}{\horstFreeVarfrom,\horstFreeVarto}}
    \horstPremise{\horstLT{\horstParVarnext}{\horstFreeVarto}}
    \horstConclusion{\horstPredAppMState{\horstParVarfid,\horstParVarnext}{\horstConstructorAppCtx{\horstFreeVarp,\horstConstructorAppIllegal{},\horstFreeVarfrom},\horstFreeVarst,\horstFreeVargt,\horstFreeVarlt,\horstFreeVarmem,\horstFreeVartbl,\horstFreeVaratN,\horstFreeVargtN,\horstFreeVarmemN}}
  \end{horstClause}
\end{horstRule}
\begin{horstRule}{testNoninterferenceImportedFunctionParamsUnsat}
  \horstParVar{fid}{\horstTypeint}
  \horstParVar{cid}{\horstTypeint}
  \horstSelectorFunctionInvocation{\horstSelectorFunctionAppimportedFunctionIds{\horstParVarfid}{},\horstSelectorFunctionAppimportCallArgumentSafe{\horstParVarcid}{\horstParVarfid}}
  \begin{horstClause}
    \horstFreeVar{memN}{\horstTypeMemory}
    \horstFreeVar{atN}{\horstTypeHomInit{\horstTypeLValue}{\horstOpAppas{\horstParVarfid}{}}}
    \horstFreeVar{tbl}{\horstTypeTable}
    \horstFreeVar{ctx}{\horstTypeContext}
    \horstFreeVar{lt}{\horstTypeHomInit{\horstTypeLValue}{\horstOpAppls{\horstParVarfid}{}}}
    \horstFreeVar{mem}{\horstTypeMemory}
    \horstFreeVar{gtN}{\horstTypeHomInit{\horstTypeLValue}{\horstOpAppgs{}{}}}
    \horstFreeVar{gt}{\horstTypeHomInit{\horstTypeLValue}{\horstOpAppgs{}{}}}
    \horstPremise{\horstPredAppMState{\horstParVarfid,0}{\horstFreeVarctx,\horstTUPINIT{},\horstFreeVargt,\horstFreeVarlt,\horstFreeVarmem,\horstFreeVartbl,\horstFreeVaratN,\horstFreeVargtN,\horstFreeVarmemN}}
    \horstPremise{\horstSimpleSumExp{OR}{\horstSelectorFunctionAppinterval{\horstParVaridx}{0,\horstOpAppas{\horstParVarfid}{}},\horstSelectorFunctionAppargumentLabelForImportedFunctionAndPosition{\horstParVarca,\horstParVaria}{\horstParVarfid,\horstParVaridx}}{\horstAND{\horstOpAppflowsTo{}{\horstOpAppmkLabel{}{\horstParVarca,\horstParVaria},\horstOpAppperspectiveOfCtx{}{\horstFreeVarctx}}}{\horstEQ{\horstOpApplabelOf{}{\horstACCESS{\horstFreeVaratN}{\horstParVaridx}}}{\horstConstructorAppIllegal{}}}}{{idx}{ca}{ia}}}
    \horstConclusion{\horstPredApptestNoninterferenceImportedFunctionParamsUnsat{\horstParVarfid,\horstParVarcid}{}}
  \end{horstClause}
\end{horstRule}
\begin{horstRule}{ifThenElseRule}
  \horstParVar{fid}{\horstTypeint}
  \horstParVar{pc}{\horstTypeint}
  \horstParVar{else}{\horstTypeint}
  \horstParVar{end}{\horstTypeint}
  \horstSelectorFunctionInvocation{\horstSelectorFunctionAppfunctionIds{\horstParVarfid}{},\horstSelectorFunctionAppifs{\horstParVarpc}{\horstParVarfid},\horstSelectorFunctionAppelseForIf{\horstParVarelse}{\horstParVarfid,\horstParVarpc},\horstSelectorFunctionAppendForIf{\horstParVarend}{\horstParVarfid,\horstParVarpc}}
  \begin{horstClause}
    \horstFreeVar{memN}{\horstTypeMemory}
    \horstFreeVar{atN}{\horstTypeHomInit{\horstTypeLValue}{\horstOpAppas{\horstParVarfid}{}}}
    \horstFreeVar{st}{\horstTypeHomInit{\horstTypeLValue}{\horstSUB{\horstOpAppss{\horstParVarfid,\horstParVarpc}{}}{1}}}
    \horstFreeVar{tbl}{\horstTypeTable}
    \horstFreeVar{ctx}{\horstTypeContext}
    \horstFreeVar{lt}{\horstTypeHomInit{\horstTypeLValue}{\horstOpAppls{\horstParVarfid}{}}}
    \horstFreeVar{mem}{\horstTypeMemory}
    \horstFreeVar{x}{\horstTypeLValue}
    \horstFreeVar{gtN}{\horstTypeHomInit{\horstTypeLValue}{\horstOpAppgs{}{}}}
    \horstFreeVar{gt}{\horstTypeHomInit{\horstTypeLValue}{\horstOpAppgs{}{}}}
    \horstPremise{\horstPredAppMState{\horstParVarfid,\horstParVarpc}{\horstFreeVarctx,\horstCONS{\horstFreeVarx}{\horstFreeVarst},\horstFreeVargt,\horstFreeVarlt,\horstFreeVarmem,\horstFreeVartbl,\horstFreeVaratN,\horstFreeVargtN,\horstFreeVarmemN}}
    \horstPremise{\horstOpAppabsneq{}{\horstOpAppvalueOf{}{\horstFreeVarx},\horstOpAppmkConst{0}{}}}
    \horstConclusion{\horstPredAppMState{\horstParVarfid,\horstADD{\horstParVarpc}{1}}{\horstOpAppraiseCtxTo{\horstParVarpc}{\horstFreeVarctx,\horstOpApplabelOf{}{\horstFreeVarx}},\horstFreeVarst,\horstFreeVargt,\horstFreeVarlt,\horstFreeVarmem,\horstFreeVartbl,\horstFreeVaratN,\horstFreeVargtN,\horstFreeVarmemN}}
  \end{horstClause}
  \begin{horstClause}
    \horstFreeVar{memN}{\horstTypeMemory}
    \horstFreeVar{atN}{\horstTypeHomInit{\horstTypeLValue}{\horstOpAppas{\horstParVarfid}{}}}
    \horstFreeVar{st}{\horstTypeHomInit{\horstTypeLValue}{\horstSUB{\horstOpAppss{\horstParVarfid,\horstParVarpc}{}}{1}}}
    \horstFreeVar{tbl}{\horstTypeTable}
    \horstFreeVar{ctx}{\horstTypeContext}
    \horstFreeVar{lt}{\horstTypeHomInit{\horstTypeLValue}{\horstOpAppls{\horstParVarfid}{}}}
    \horstFreeVar{mem}{\horstTypeMemory}
    \horstFreeVar{x}{\horstTypeLValue}
    \horstFreeVar{gtN}{\horstTypeHomInit{\horstTypeLValue}{\horstOpAppgs{}{}}}
    \horstFreeVar{gt}{\horstTypeHomInit{\horstTypeLValue}{\horstOpAppgs{}{}}}
    \horstPremise{\horstPredAppMState{\horstParVarfid,\horstParVarpc}{\horstFreeVarctx,\horstCONS{\horstFreeVarx}{\horstFreeVarst},\horstFreeVargt,\horstFreeVarlt,\horstFreeVarmem,\horstFreeVartbl,\horstFreeVaratN,\horstFreeVargtN,\horstFreeVarmemN}}
    \horstPremise{\horstOpAppabseq{}{\horstOpAppvalueOf{}{\horstFreeVarx},\horstOpAppmkConst{0}{}}}
    \horstConclusion{\horstPredAppMState{\horstParVarfid,\horstParVarelse}{\horstOpAppraiseCtxTo{\horstParVarpc}{\horstFreeVarctx,\horstOpApplabelOf{}{\horstFreeVarx}},\horstFreeVarst,\horstFreeVargt,\horstFreeVarlt,\horstFreeVarmem,\horstFreeVartbl,\horstFreeVaratN,\horstFreeVargtN,\horstFreeVarmemN}}
  \end{horstClause}
  \begin{horstClause}
    \horstFreeVar{memN}{\horstTypeMemory}
    \horstFreeVar{atN}{\horstTypeHomInit{\horstTypeLValue}{\horstOpAppas{\horstParVarfid}{}}}
    \horstFreeVar{st}{\horstTypeHomInit{\horstTypeLValue}{\horstSUB{\horstOpAppss{\horstParVarfid,\horstParVarpc}{}}{1}}}
    \horstFreeVar{p}{\horstTypeLabel}
    \horstFreeVar{tbl}{\horstTypeTable}
    \horstFreeVar{ctx}{\horstTypeContext}
    \horstFreeVar{lt}{\horstTypeHomInit{\horstTypeLValue}{\horstOpAppls{\horstParVarfid}{}}}
    \horstFreeVar{mem}{\horstTypeMemory}
    \horstFreeVar{x}{\horstTypeLValue}
    \horstFreeVar{gtN}{\horstTypeHomInit{\horstTypeLValue}{\horstOpAppgs{}{}}}
    \horstFreeVar{gt}{\horstTypeHomInit{\horstTypeLValue}{\horstOpAppgs{}{}}}
    \horstFreeVar{from}{\horstTypeint}
    \horstPremise{\horstPredAppMState{\horstParVarfid,\horstParVarpc}{\horstFreeVarctx,\horstCONS{\horstFreeVarx}{\horstFreeVarst},\horstFreeVargt,\horstFreeVarlt,\horstFreeVarmem,\horstFreeVartbl,\horstFreeVaratN,\horstFreeVargtN,\horstFreeVarmemN}}
    \horstPremise{\horstEQ{\horstOpAppraiseCtxTo{\horstParVarpc}{\horstFreeVarctx,\horstOpApplabelOf{}{\horstFreeVarx}}}{\horstConstructorAppCtx{\horstFreeVarp,\horstConstructorAppIllegal{},\horstFreeVarfrom}}}
    \horstConclusion{\horstPredAppScopeExtend{\horstParVarfid}{\horstFreeVarfrom,\horstParVarend}}
  \end{horstClause}
\end{horstRule}
\begin{horstRule}{loopRule}
  \horstParVar{fid}{\horstTypeint}
  \horstParVar{pc}{\horstTypeint}
  \horstSelectorFunctionInvocation{\horstSelectorFunctionAppfunctionIds{\horstParVarfid}{},\horstSelectorFunctionApploopsForFunctionId{\horstParVarpc}{\horstParVarfid}}
  \begin{horstClause}
    \horstFreeVar{memN}{\horstTypeMemory}
    \horstFreeVar{atN}{\horstTypeHomInit{\horstTypeLValue}{\horstOpAppas{\horstParVarfid}{}}}
    \horstFreeVar{st}{\horstTypeHomInit{\horstTypeLValue}{\horstOpAppss{\horstParVarfid,\horstParVarpc}{}}}
    \horstFreeVar{tbl}{\horstTypeTable}
    \horstFreeVar{ngt}{\horstTypeHomInit{\horstTypeLValue}{\horstOpAppgs{}{}}}
    \horstFreeVar{ctx}{\horstTypeContext}
    \horstFreeVar{lt}{\horstTypeHomInit{\horstTypeLValue}{\horstOpAppls{\horstParVarfid}{}}}
    \horstFreeVar{mem}{\horstTypeMemory}
    \horstFreeVar{nlt}{\horstTypeHomInit{\horstTypeLValue}{\horstOpAppls{\horstParVarfid}{}}}
    \horstFreeVar{gtN}{\horstTypeHomInit{\horstTypeLValue}{\horstOpAppgs{}{}}}
    \horstFreeVar{gt}{\horstTypeHomInit{\horstTypeLValue}{\horstOpAppgs{}{}}}
    \horstFreeVar{nmem}{\horstTypeMemory}
    \horstPremise{\horstPredAppMState{\horstParVarfid,\horstParVarpc}{\horstFreeVarctx,\horstFreeVarst,\horstFreeVargt,\horstFreeVarlt,\horstFreeVarmem,\horstFreeVartbl,\horstFreeVaratN,\horstFreeVargtN,\horstFreeVarmemN}}
    \horstPremise{\horstOpAppoverApproximateLoopGlobals{\horstParVarfid,\horstParVarpc}{\horstFreeVargt,\horstFreeVarngt}}
    \horstPremise{\horstOpAppoverApproximateLoopLocals{\horstParVarfid,\horstParVarpc}{\horstFreeVarlt,\horstFreeVarnlt}}
    \horstPremise{\horstOpAppoverApproximateLoopMemory{\horstParVarfid,\horstParVarpc}{\horstFreeVarmem,\horstFreeVarnmem}}
    \horstConclusion{\horstPredAppMState{\horstParVarfid,\horstADD{\horstParVarpc}{1}}{\horstFreeVarctx,\horstFreeVarst,\horstFreeVarngt,\horstFreeVarnlt,\horstFreeVarnmem,\horstFreeVartbl,\horstFreeVaratN,\horstFreeVargtN,\horstFreeVarmemN}}
  \end{horstClause}
\end{horstRule}
\begin{horstRule}{localTeeRule}
  \horstParVar{fid}{\horstTypeint}
  \horstParVar{pc}{\horstTypeint}
  \horstParVar{idx}{\horstTypeint}
  \horstSelectorFunctionInvocation{\horstSelectorFunctionAppfunctionIds{\horstParVarfid}{},\horstSelectorFunctionApppcsForFunctionIdAndOpcode{\horstParVarpc}{\horstParVarfid,\horstConstLOCALTEE},\horstSelectorFunctionAppimmediateForFunctionIdAndPc{\horstParVaridx}{\horstParVarfid,\horstParVarpc}}
  \begin{horstClause}
    \horstFreeVar{memN}{\horstTypeMemory}
    \horstFreeVar{atN}{\horstTypeHomInit{\horstTypeLValue}{\horstOpAppas{\horstParVarfid}{}}}
    \horstFreeVar{st}{\horstTypeHomInit{\horstTypeLValue}{\horstSUB{\horstOpAppss{\horstParVarfid,\horstParVarpc}{}}{1}}}
    \horstFreeVar{tbl}{\horstTypeTable}
    \horstFreeVar{ctx}{\horstTypeContext}
    \horstFreeVar{lt}{\horstTypeHomInit{\horstTypeLValue}{\horstOpAppls{\horstParVarfid}{}}}
    \horstFreeVar{mem}{\horstTypeMemory}
    \horstFreeVar{x}{\horstTypeLValue}
    \horstFreeVar{gtN}{\horstTypeHomInit{\horstTypeLValue}{\horstOpAppgs{}{}}}
    \horstFreeVar{gt}{\horstTypeHomInit{\horstTypeLValue}{\horstOpAppgs{}{}}}
    \horstPremise{\horstPredAppMState{\horstParVarfid,\horstParVarpc}{\horstFreeVarctx,\horstCONS{\horstFreeVarx}{\horstFreeVarst},\horstFreeVargt,\horstFreeVarlt,\horstFreeVarmem,\horstFreeVartbl,\horstFreeVaratN,\horstFreeVargtN,\horstFreeVarmemN}}
    \horstConclusion{\horstPredAppMState{\horstParVarfid,\horstADD{\horstParVarpc}{1}}{\horstFreeVarctx,\horstCONS{\horstFreeVarx}{\horstFreeVarst},\horstFreeVargt,\horstOpAppset{\horstOpAppls{\horstParVarfid}{},\horstParVaridx}{\horstOpAppraiseTo{}{\horstFreeVarx,\horstOpApplabelOfCtx{}{\horstFreeVarctx}},\horstFreeVarlt},\horstFreeVarmem,\horstFreeVartbl,\horstFreeVaratN,\horstFreeVargtN,\horstFreeVarmemN}}
  \end{horstClause}
\end{horstRule}
\begin{horstRule}{callRule}
  \horstParVar{fid}{\horstTypeint}
  \horstParVar{pc}{\horstTypeint}
  \horstParVar{cid}{\horstTypeint}
  \horstSelectorFunctionInvocation{\horstSelectorFunctionAppfunctionIds{\horstParVarfid}{},\horstSelectorFunctionApppcsForFunctionIdAndOpcode{\horstParVarpc}{\horstParVarfid,\horstConstCALLFUNCTION},\horstSelectorFunctionAppimmediateForFunctionIdAndPc{\horstParVarcid}{\horstParVarfid,\horstParVarpc}}
  \begin{horstClause}
    \horstFreeVar{memN}{\horstTypeMemory}
    \horstFreeVar{atN}{\horstTypeHomInit{\horstTypeLValue}{\horstOpAppas{\horstParVarfid}{}}}
    \horstFreeVar{st}{\horstTypeHomInit{\horstTypeLValue}{\horstOpAppss{\horstParVarfid,\horstParVarpc}{}}}
    \horstFreeVar{tbl}{\horstTypeTable}
    \horstFreeVar{ngt}{\horstTypeHomInit{\horstTypeLValue}{\horstOpAppgs{}{}}}
    \horstFreeVar{at}{\horstTypeHomInit{\horstTypeLValue}{\horstOpAppas{\horstParVarcid}{}}}
    \horstFreeVar{ctx}{\horstTypeContext}
    \horstFreeVar{lt}{\horstTypeHomInit{\horstTypeLValue}{\horstOpAppls{\horstParVarfid}{}}}
    \horstFreeVar{mem}{\horstTypeMemory}
    \horstFreeVar{gtN}{\horstTypeHomInit{\horstTypeLValue}{\horstOpAppgs{}{}}}
    \horstFreeVar{gt}{\horstTypeHomInit{\horstTypeLValue}{\horstOpAppgs{}{}}}
    \horstFreeVar{nmem}{\horstTypeMemory}
    \horstPremise{\horstPredAppMState{\horstParVarfid,\horstParVarpc}{\horstFreeVarctx,\horstFreeVarst,\horstFreeVargt,\horstFreeVarlt,\horstFreeVarmem,\horstFreeVartbl,\horstFreeVaratN,\horstFreeVargtN,\horstFreeVarmemN}}
    \horstPremise{\horstOpAppoverApproximateCallArguments{\horstParVarcid}{\horstOpAppreverse{\horstOpAppas{\horstParVarcid}{}}{\horstSLICE{\horstFreeVarst}{}{\horstOpAppas{\horstParVarcid}{}}},\horstFreeVarat}}
    \horstPremise{\horstOpAppoverApproximateCallGlobals{\horstParVarcid}{\horstFreeVargt,\horstFreeVarngt}}
    \horstPremise{\horstOpAppoverApproximateCallMemory{\horstParVarcid}{\horstFreeVarmem,\horstFreeVarnmem}}
    \horstConclusion{\horstPredAppMState{\horstParVarcid,0}{\horstOpAppmkCtx{}{\horstOpAppperspectiveOfCtx{}{\horstFreeVarctx},\horstOpApplabelOfCtx{}{\horstFreeVarctx}},\horstTUPINIT{},\horstFreeVarngt,\horstCONCAT{\horstFreeVarat}{\horstHOMINIT{\horstOpAppmkLConst{0}{}}{\horstSUB{\horstOpAppls{\horstParVarcid}{}}{\horstOpAppas{\horstParVarcid}{}}}},\horstFreeVarnmem,\horstFreeVartbl,\horstFreeVarat,\horstFreeVarngt,\horstFreeVarnmem}}
  \end{horstClause}
  \begin{horstClause}
    \horstFreeVar{atN}{\horstTypeHomInit{\horstTypeLValue}{\horstOpAppas{\horstParVarfid}{}}}
    \horstFreeVar{st}{\horstTypeHomInit{\horstTypeLValue}{\horstOpAppss{\horstParVarfid,\horstParVarpc}{}}}
    \horstFreeVar{rt}{\horstTypeHomInit{\horstTypeLValue}{\horstOpApprs{\horstParVarcid}{}}}
    \horstFreeVar{tbl}{\horstTypeTable}
    \horstFreeVar{ngt}{\horstTypeHomInit{\horstTypeLValue}{\horstOpAppgs{}{}}}
    \horstFreeVar{rtbl}{\horstTypeTable}
    \horstFreeVar{lt}{\horstTypeHomInit{\horstTypeLValue}{\horstOpAppls{\horstParVarfid}{}}}
    \horstFreeVar{cl}{\horstTypeFlowLabel}
    \horstFreeVar{mem}{\horstTypeMemory}
    \horstFreeVar{rfrom}{\horstTypeint}
    \horstFreeVar{gt}{\horstTypeHomInit{\horstTypeLValue}{\horstOpAppgs{}{}}}
    \horstFreeVar{nmem}{\horstTypeMemory}
    \horstFreeVar{rmem}{\horstTypeMemory}
    \horstFreeVar{memN}{\horstTypeMemory}
    \horstFreeVar{p}{\horstTypeLabel}
    \horstFreeVar{at}{\horstTypeHomInit{\horstTypeLValue}{\horstOpAppas{\horstParVarcid}{}}}
    \horstFreeVar{ctx}{\horstTypeContext}
    \horstFreeVar{rl}{\horstTypeFlowLabel}
    \horstFreeVar{gtN}{\horstTypeHomInit{\horstTypeLValue}{\horstOpAppgs{}{}}}
    \horstFreeVar{rgt}{\horstTypeHomInit{\horstTypeLValue}{\horstOpAppgs{}{}}}
    \horstPremise{\horstPredAppMState{\horstParVarfid,\horstParVarpc}{\horstFreeVarctx,\horstFreeVarst,\horstFreeVargt,\horstFreeVarlt,\horstFreeVarmem,\horstFreeVartbl,\horstFreeVaratN,\horstFreeVargtN,\horstFreeVarmemN}}
    \horstPremise{\horstOpAppoverApproximateCallArguments{\horstParVarcid}{\horstOpAppreverse{\horstOpAppas{\horstParVarcid}{}}{\horstSLICE{\horstFreeVarst}{}{\horstOpAppas{\horstParVarcid}{}}},\horstFreeVarat}}
    \horstPremise{\horstOpAppoverApproximateCallGlobals{\horstParVarcid}{\horstFreeVargt,\horstFreeVarngt}}
    \horstPremise{\horstOpAppoverApproximateCallMemory{\horstParVarcid}{\horstFreeVarmem,\horstFreeVarnmem}}
    \horstPremise{\horstEQ{\horstFreeVarp}{\horstOpAppperspectiveOfCtx{}{\horstFreeVarctx}}}
    \horstPremise{\horstEQ{\horstFreeVarcl}{\horstOpApplabelOfCtx{}{\horstFreeVarctx}}}
    \horstPremise{\horstPredAppReturn{\horstParVarcid}{\horstConstructorAppCtx{\horstFreeVarp,\horstFreeVarrl,\horstFreeVarrfrom},\horstFreeVarrt,\horstFreeVarrgt,\horstFreeVarrmem,\horstFreeVarrtbl,\horstFreeVarat,\horstFreeVarngt,\horstFreeVarnmem}}
    \horstPremise{\horstOR{\horstAND{\horstEQ{\horstFreeVarcl}{\horstConstructorAppLegal{}}}{\horstLE{0}{\horstFreeVarrfrom}}}{\horstAND{\horstEQ{\horstFreeVarcl}{\horstConstructorAppIllegal{}}}{\horstLT{\horstFreeVarrfrom}{0}}}}
    \horstConclusion{\horstPredAppMState{\horstParVarfid,\horstADD{\horstParVarpc}{1}}{\horstFreeVarctx,\horstCONCAT{\horstFreeVarrt}{\horstSLICE{\horstFreeVarst}{\horstOpAppas{\horstParVarcid}{}}{}},\horstFreeVarrgt,\horstFreeVarlt,\horstFreeVarrmem,\horstFreeVarrtbl,\horstFreeVaratN,\horstFreeVargtN,\horstFreeVarmemN}}
  \end{horstClause}
\end{horstRule}
\begin{horstRule}{callIndirectRule}
  \horstParVar{fid}{\horstTypeint}
  \horstParVar{pc}{\horstTypeint}
  \horstParVar{idx}{\horstTypeint}
  \horstParVar{cid}{\horstTypeint}
  \horstSelectorFunctionInvocation{\horstSelectorFunctionAppfunctionIds{\horstParVarfid}{},\horstSelectorFunctionApppcsForFunctionIdAndOpcode{\horstParVarpc}{\horstParVarfid,\horstConstCALLINDIRECT},\horstSelectorFunctionApppossibleCallTargets{\horstParVaridx,\horstParVarcid}{\horstParVarfid,\horstParVarpc}}
  \begin{horstClause}
    \horstFreeVar{atN}{\horstTypeHomInit{\horstTypeLValue}{\horstOpAppas{\horstParVarfid}{}}}
    \horstFreeVar{st}{\horstTypeHomInit{\horstTypeLValue}{\horstSUB{\horstOpAppss{\horstParVarfid,\horstParVarpc}{}}{1}}}
    \horstFreeVar{ngt}{\horstTypeHomInit{\horstTypeLValue}{\horstOpAppgs{}{}}}
    \horstFreeVar{lt}{\horstTypeHomInit{\horstTypeLValue}{\horstOpAppls{\horstParVarfid}{}}}
    \horstFreeVar{cl}{\horstTypeFlowLabel}
    \horstFreeVar{mem}{\horstTypeMemory}
    \horstFreeVar{gt}{\horstTypeHomInit{\horstTypeLValue}{\horstOpAppgs{}{}}}
    \horstFreeVar{nmem}{\horstTypeMemory}
    \horstFreeVar{memN}{\horstTypeMemory}
    \horstFreeVar{p}{\horstTypeLabel}
    \horstFreeVar{at}{\horstTypeHomInit{\horstTypeLValue}{\horstOpAppas{\horstParVarcid}{}}}
    \horstFreeVar{ctx}{\horstTypeContext}
    \horstFreeVar{x}{\horstTypeLValue}
    \horstFreeVar{gtN}{\horstTypeHomInit{\horstTypeLValue}{\horstOpAppgs{}{}}}
    \horstFreeVar{ltbl}{\horstTypeFlowLabel}
    \horstPremise{\horstPredAppMState{\horstParVarfid,\horstParVarpc}{\horstFreeVarctx,\horstCONS{\horstFreeVarx}{\horstFreeVarst},\horstFreeVargt,\horstFreeVarlt,\horstFreeVarmem,\horstConstructorAppTbl{\horstConstructorAppTblPrecise{},\horstFreeVarltbl},\horstFreeVaratN,\horstFreeVargtN,\horstFreeVarmemN}}
    \horstPremise{\horstOpAppabseq{}{\horstOpAppvalueOf{}{\horstFreeVarx},\horstOpAppmkConst{\horstParVaridx}{}}}
    \horstPremise{\horstOpAppoverApproximateCallArguments{\horstParVarcid}{\horstOpAppreverse{\horstOpAppas{\horstParVarcid}{}}{\horstSLICE{\horstFreeVarst}{}{\horstOpAppas{\horstParVarcid}{}}},\horstFreeVarat}}
    \horstPremise{\horstOpAppoverApproximateCallGlobals{\horstParVarcid}{\horstFreeVargt,\horstFreeVarngt}}
    \horstPremise{\horstOpAppoverApproximateCallMemory{\horstParVarcid}{\horstFreeVarmem,\horstFreeVarnmem}}
    \horstPremise{\horstEQ{\horstFreeVarp}{\horstOpAppperspectiveOfCtx{}{\horstFreeVarctx}}}
    \horstPremise{\horstEQ{\horstFreeVarcl}{\horstOpAppflub{3}{\horstTUPINIT{\horstOpApplabelOfCtx{}{\horstFreeVarctx},\horstOpApplabelOf{}{\horstFreeVarx},\horstFreeVarltbl}}}}
    \horstConclusion{\horstPredAppMState{\horstParVarcid,0}{\horstOpAppmkCtx{}{\horstFreeVarp,\horstFreeVarcl},\horstTUPINIT{},\horstFreeVarngt,\horstCONCAT{\horstFreeVarat}{\horstHOMINIT{\horstOpAppmkLConst{0}{}}{\horstSUB{\horstOpAppls{\horstParVarcid}{}}{\horstOpAppas{\horstParVarcid}{}}}},\horstFreeVarnmem,\horstConstructorAppTbl{\horstConstructorAppTblPrecise{},\horstFreeVarltbl},\horstFreeVarat,\horstFreeVarngt,\horstFreeVarnmem}}
  \end{horstClause}
  \begin{horstClause}
    \horstFreeVar{atN}{\horstTypeHomInit{\horstTypeLValue}{\horstOpAppas{\horstParVarfid}{}}}
    \horstFreeVar{st}{\horstTypeHomInit{\horstTypeLValue}{\horstSUB{\horstOpAppss{\horstParVarfid,\horstParVarpc}{}}{1}}}
    \horstFreeVar{rt}{\horstTypeHomInit{\horstTypeLValue}{\horstOpApprs{\horstParVarcid}{}}}
    \horstFreeVar{ngt}{\horstTypeHomInit{\horstTypeLValue}{\horstOpAppgs{}{}}}
    \horstFreeVar{rtbl}{\horstTypeTable}
    \horstFreeVar{lt}{\horstTypeHomInit{\horstTypeLValue}{\horstOpAppls{\horstParVarfid}{}}}
    \horstFreeVar{cl}{\horstTypeFlowLabel}
    \horstFreeVar{mem}{\horstTypeMemory}
    \horstFreeVar{rfrom}{\horstTypeint}
    \horstFreeVar{gt}{\horstTypeHomInit{\horstTypeLValue}{\horstOpAppgs{}{}}}
    \horstFreeVar{nmem}{\horstTypeMemory}
    \horstFreeVar{rmem}{\horstTypeMemory}
    \horstFreeVar{memN}{\horstTypeMemory}
    \horstFreeVar{p}{\horstTypeLabel}
    \horstFreeVar{at}{\horstTypeHomInit{\horstTypeLValue}{\horstOpAppas{\horstParVarcid}{}}}
    \horstFreeVar{ctx}{\horstTypeContext}
    \horstFreeVar{x}{\horstTypeLValue}
    \horstFreeVar{gtN}{\horstTypeHomInit{\horstTypeLValue}{\horstOpAppgs{}{}}}
    \horstFreeVar{ltbl}{\horstTypeFlowLabel}
    \horstFreeVar{rgt}{\horstTypeHomInit{\horstTypeLValue}{\horstOpAppgs{}{}}}
    \horstPremise{\horstPredAppMState{\horstParVarfid,\horstParVarpc}{\horstFreeVarctx,\horstCONS{\horstFreeVarx}{\horstFreeVarst},\horstFreeVargt,\horstFreeVarlt,\horstFreeVarmem,\horstConstructorAppTbl{\horstConstructorAppTblPrecise{},\horstFreeVarltbl},\horstFreeVaratN,\horstFreeVargtN,\horstFreeVarmemN}}
    \horstPremise{\horstEQ{\horstOpApplabelOfCtx{}{\horstFreeVarctx}}{\horstConstructorAppLegal{}}}
    \horstPremise{\horstEQ{\horstOpAppflub{2}{\horstTUPINIT{\horstOpApplabelOf{}{\horstFreeVarx},\horstFreeVarltbl}}}{\horstConstructorAppIllegal{}}}
    \horstPremise{\horstOpAppabseq{}{\horstOpAppvalueOf{}{\horstFreeVarx},\horstOpAppmkConst{\horstParVaridx}{}}}
    \horstPremise{\horstOpAppoverApproximateCallArguments{\horstParVarcid}{\horstOpAppreverse{\horstOpAppas{\horstParVarcid}{}}{\horstSLICE{\horstFreeVarst}{}{\horstOpAppas{\horstParVarcid}{}}},\horstFreeVarat}}
    \horstPremise{\horstOpAppoverApproximateCallGlobals{\horstParVarcid}{\horstFreeVargt,\horstFreeVarngt}}
    \horstPremise{\horstOpAppoverApproximateCallMemory{\horstParVarcid}{\horstFreeVarmem,\horstFreeVarnmem}}
    \horstPremise{\horstEQ{\horstFreeVarp}{\horstOpAppperspectiveOfCtx{}{\horstFreeVarctx}}}
    \horstPremise{\horstPredAppReturn{\horstParVarcid}{\horstConstructorAppCtx{\horstFreeVarp,\horstConstructorAppIllegal{},\horstFreeVarrfrom},\horstFreeVarrt,\horstFreeVarrgt,\horstFreeVarrmem,\horstFreeVarrtbl,\horstFreeVarat,\horstFreeVarngt,\horstFreeVarnmem}}
    \horstPremise{\horstLT{\horstFreeVarrfrom}{0}}
    \horstConclusion{\horstPredAppMStateToJoin{\horstParVarfid,\horstADD{\horstParVarpc}{1}}{\horstFreeVarctx,\horstCONCAT{\horstFreeVarrt}{\horstSLICE{\horstFreeVarst}{\horstOpAppas{\horstParVarcid}{}}{}},\horstFreeVarrgt,\horstFreeVarlt,\horstFreeVarrmem,\horstFreeVarrtbl,\horstFreeVaratN,\horstFreeVargtN,\horstFreeVarmemN}}
  \end{horstClause}
  \begin{horstClause}
    \horstFreeVar{atN}{\horstTypeHomInit{\horstTypeLValue}{\horstOpAppas{\horstParVarfid}{}}}
    \horstFreeVar{st}{\horstTypeHomInit{\horstTypeLValue}{\horstSUB{\horstOpAppss{\horstParVarfid,\horstParVarpc}{}}{1}}}
    \horstFreeVar{rt}{\horstTypeHomInit{\horstTypeLValue}{\horstOpApprs{\horstParVarcid}{}}}
    \horstFreeVar{ngt}{\horstTypeHomInit{\horstTypeLValue}{\horstOpAppgs{}{}}}
    \horstFreeVar{rtbl}{\horstTypeTable}
    \horstFreeVar{lt}{\horstTypeHomInit{\horstTypeLValue}{\horstOpAppls{\horstParVarfid}{}}}
    \horstFreeVar{cl}{\horstTypeFlowLabel}
    \horstFreeVar{mem}{\horstTypeMemory}
    \horstFreeVar{rfrom}{\horstTypeint}
    \horstFreeVar{gt}{\horstTypeHomInit{\horstTypeLValue}{\horstOpAppgs{}{}}}
    \horstFreeVar{nmem}{\horstTypeMemory}
    \horstFreeVar{rmem}{\horstTypeMemory}
    \horstFreeVar{memN}{\horstTypeMemory}
    \horstFreeVar{p}{\horstTypeLabel}
    \horstFreeVar{at}{\horstTypeHomInit{\horstTypeLValue}{\horstOpAppas{\horstParVarcid}{}}}
    \horstFreeVar{ctx}{\horstTypeContext}
    \horstFreeVar{rl}{\horstTypeFlowLabel}
    \horstFreeVar{x}{\horstTypeLValue}
    \horstFreeVar{gtN}{\horstTypeHomInit{\horstTypeLValue}{\horstOpAppgs{}{}}}
    \horstFreeVar{ltbl}{\horstTypeFlowLabel}
    \horstFreeVar{rgt}{\horstTypeHomInit{\horstTypeLValue}{\horstOpAppgs{}{}}}
    \horstPremise{\horstPredAppMState{\horstParVarfid,\horstParVarpc}{\horstFreeVarctx,\horstCONS{\horstFreeVarx}{\horstFreeVarst},\horstFreeVargt,\horstFreeVarlt,\horstFreeVarmem,\horstConstructorAppTbl{\horstConstructorAppTblPrecise{},\horstFreeVarltbl},\horstFreeVaratN,\horstFreeVargtN,\horstFreeVarmemN}}
    \horstPremise{\horstOpAppabseq{}{\horstOpAppvalueOf{}{\horstFreeVarx},\horstOpAppmkConst{\horstParVaridx}{}}}
    \horstPremise{\horstOpAppoverApproximateCallArguments{\horstParVarcid}{\horstOpAppreverse{\horstOpAppas{\horstParVarcid}{}}{\horstSLICE{\horstFreeVarst}{}{\horstOpAppas{\horstParVarcid}{}}},\horstFreeVarat}}
    \horstPremise{\horstOpAppoverApproximateCallGlobals{\horstParVarcid}{\horstFreeVargt,\horstFreeVarngt}}
    \horstPremise{\horstOpAppoverApproximateCallMemory{\horstParVarcid}{\horstFreeVarmem,\horstFreeVarnmem}}
    \horstPremise{\horstEQ{\horstFreeVarp}{\horstOpAppperspectiveOfCtx{}{\horstFreeVarctx}}}
    \horstPremise{\horstEQ{\horstOpApplabelOfCtx{}{\horstFreeVarctx}}{\horstOpAppflub{2}{\horstTUPINIT{\horstOpApplabelOf{}{\horstFreeVarx},\horstFreeVarltbl}}}}
    \horstPremise{\horstEQ{\horstFreeVarcl}{\horstOpApplabelOfCtx{}{\horstFreeVarctx}}}
    \horstPremise{\horstPredAppReturn{\horstParVarcid}{\horstConstructorAppCtx{\horstFreeVarp,\horstFreeVarrl,\horstFreeVarrfrom},\horstFreeVarrt,\horstFreeVarrgt,\horstFreeVarrmem,\horstFreeVarrtbl,\horstFreeVarat,\horstFreeVarngt,\horstFreeVarnmem}}
    \horstPremise{\horstOR{\horstAND{\horstEQ{\horstFreeVarcl}{\horstConstructorAppLegal{}}}{\horstLE{0}{\horstFreeVarrfrom}}}{\horstAND{\horstEQ{\horstFreeVarcl}{\horstConstructorAppIllegal{}}}{\horstLT{\horstFreeVarrfrom}{0}}}}
    \horstConclusion{\horstPredAppMState{\horstParVarfid,\horstADD{\horstParVarpc}{1}}{\horstFreeVarctx,\horstCONCAT{\horstFreeVarrt}{\horstSLICE{\horstFreeVarst}{\horstOpAppas{\horstParVarcid}{}}{}},\horstFreeVarrgt,\horstFreeVarlt,\horstFreeVarrmem,\horstFreeVarrtbl,\horstFreeVaratN,\horstFreeVargtN,\horstFreeVarmemN}}
  \end{horstClause}
\end{horstRule}
\begin{horstRule}{brTableRule}
  \horstParVar{fid}{\horstTypeint}
  \horstParVar{pc}{\horstTypeint}
  \horstParVar{sz}{\horstTypeint}
  \horstParVar{n}{\horstTypeint}
  \horstParVar{idx}{\horstTypeint}
  \horstParVar{br}{\horstTypeint}
  \horstSelectorFunctionInvocation{\horstSelectorFunctionAppfunctionIds{\horstParVarfid}{},\horstSelectorFunctionApppcsForFunctionIdAndOpcode{\horstParVarpc}{\horstParVarfid,\horstConstBRTABLE},\horstSelectorFunctionAppsizeOfBreakTable{\horstParVarsz}{\horstParVarfid,\horstParVarpc},\horstSelectorFunctionAppgetAmountOfReturnValuesInBlock{\horstParVarn}{\horstParVarfid,\horstParVarpc},\horstSelectorFunctionAppinterval{\horstParVaridx}{0,\horstSUB{\horstParVarsz}{1}},\horstSelectorFunctionAppbreakTableDestinations{\horstParVarbr}{\horstParVarfid,\horstParVarpc,\horstParVaridx}}
  \begin{horstClause}
    \horstFreeVar{memN}{\horstTypeMemory}
    \horstFreeVar{atN}{\horstTypeHomInit{\horstTypeLValue}{\horstOpAppas{\horstParVarfid}{}}}
    \horstFreeVar{st}{\horstTypeHomInit{\horstTypeLValue}{\horstSUB{\horstOpAppss{\horstParVarfid,\horstParVarpc}{}}{1}}}
    \horstFreeVar{tbl}{\horstTypeTable}
    \horstFreeVar{ctx}{\horstTypeContext}
    \horstFreeVar{lt}{\horstTypeHomInit{\horstTypeLValue}{\horstOpAppls{\horstParVarfid}{}}}
    \horstFreeVar{mem}{\horstTypeMemory}
    \horstFreeVar{x}{\horstTypeLValue}
    \horstFreeVar{gtN}{\horstTypeHomInit{\horstTypeLValue}{\horstOpAppgs{}{}}}
    \horstFreeVar{gt}{\horstTypeHomInit{\horstTypeLValue}{\horstOpAppgs{}{}}}
    \horstFreeVar{nst}{\horstTypeHomInit{\horstTypeLValue}{\horstOpAppss{\horstParVarfid,\horstParVarbr}{}}}
    \horstPremise{\horstPredAppMState{\horstParVarfid,\horstParVarpc}{\horstFreeVarctx,\horstCONS{\horstFreeVarx}{\horstFreeVarst},\horstFreeVargt,\horstFreeVarlt,\horstFreeVarmem,\horstFreeVartbl,\horstFreeVaratN,\horstFreeVargtN,\horstFreeVarmemN}}
    \horstPremise{\horstOpAppabseq{}{\horstOpAppvalueOf{}{\horstFreeVarx},\horstOpAppmkConst{\horstParVaridx}{}}}
    \horstConclusion{\horstPredAppMState{\horstParVarfid,\horstParVarbr}{\horstOpAppraiseCtxTo{\horstParVarpc}{\horstFreeVarctx,\horstOpApplabelOf{}{\horstFreeVarx}},\horstOpAppdrop{\horstSUB{\horstOpAppss{\horstParVarfid,\horstParVarpc}{}}{1},\horstParVarn,\horstSUB{\horstSUB{\horstOpAppss{\horstParVarfid,\horstParVarpc}{}}{1}}{\horstOpAppss{\horstParVarfid,\horstParVarbr}{}}}{\horstFreeVarst},\horstFreeVargt,\horstFreeVarlt,\horstFreeVarmem,\horstFreeVartbl,\horstFreeVaratN,\horstFreeVargtN,\horstFreeVarmemN}}
  \end{horstClause}
  \begin{horstClause}
    \horstFreeVar{memN}{\horstTypeMemory}
    \horstFreeVar{atN}{\horstTypeHomInit{\horstTypeLValue}{\horstOpAppas{\horstParVarfid}{}}}
    \horstFreeVar{st}{\horstTypeHomInit{\horstTypeLValue}{\horstSUB{\horstOpAppss{\horstParVarfid,\horstParVarpc}{}}{1}}}
    \horstFreeVar{p}{\horstTypeLabel}
    \horstFreeVar{tbl}{\horstTypeTable}
    \horstFreeVar{ctx}{\horstTypeContext}
    \horstFreeVar{lt}{\horstTypeHomInit{\horstTypeLValue}{\horstOpAppls{\horstParVarfid}{}}}
    \horstFreeVar{mem}{\horstTypeMemory}
    \horstFreeVar{x}{\horstTypeLValue}
    \horstFreeVar{gtN}{\horstTypeHomInit{\horstTypeLValue}{\horstOpAppgs{}{}}}
    \horstFreeVar{gt}{\horstTypeHomInit{\horstTypeLValue}{\horstOpAppgs{}{}}}
    \horstFreeVar{from}{\horstTypeint}
    \horstPremise{\horstPredAppMState{\horstParVarfid,\horstParVarpc}{\horstFreeVarctx,\horstCONS{\horstFreeVarx}{\horstFreeVarst},\horstFreeVargt,\horstFreeVarlt,\horstFreeVarmem,\horstFreeVartbl,\horstFreeVaratN,\horstFreeVargtN,\horstFreeVarmemN}}
    \horstPremise{\horstOpAppabseq{}{\horstOpAppvalueOf{}{\horstFreeVarx},\horstOpAppmkConst{\horstParVaridx}{}}}
    \horstPremise{\horstEQ{\horstOpAppraiseCtxTo{\horstParVarpc}{\horstFreeVarctx,\horstOpApplabelOf{}{\horstFreeVarx}}}{\horstConstructorAppCtx{\horstFreeVarp,\horstConstructorAppIllegal{},\horstFreeVarfrom}}}
    \horstConclusion{\horstPredAppScopeExtend{\horstParVarfid}{\horstFreeVarfrom,\horstOpAppmax{}{\horstParVarpc,\horstParVarbr}}}
  \end{horstClause}
\end{horstRule}
\begin{horstRule}{selectRule}
  \horstParVar{fid}{\horstTypeint}
  \horstParVar{pc}{\horstTypeint}
  \horstSelectorFunctionInvocation{\horstSelectorFunctionAppfunctionIds{\horstParVarfid}{},\horstSelectorFunctionApppcsForFunctionIdAndOpcode{\horstParVarpc}{\horstParVarfid,\horstConstSELECT}}
  \begin{horstClause}
    \horstFreeVar{memN}{\horstTypeMemory}
    \horstFreeVar{atN}{\horstTypeHomInit{\horstTypeLValue}{\horstOpAppas{\horstParVarfid}{}}}
    \horstFreeVar{st}{\horstTypeHomInit{\horstTypeLValue}{\horstSUB{\horstOpAppss{\horstParVarfid,\horstParVarpc}{}}{3}}}
    \horstFreeVar{tbl}{\horstTypeTable}
    \horstFreeVar{ctx}{\horstTypeContext}
    \horstFreeVar{lt}{\horstTypeHomInit{\horstTypeLValue}{\horstOpAppls{\horstParVarfid}{}}}
    \horstFreeVar{mem}{\horstTypeMemory}
    \horstFreeVar{x}{\horstTypeLValue}
    \horstFreeVar{gtN}{\horstTypeHomInit{\horstTypeLValue}{\horstOpAppgs{}{}}}
    \horstFreeVar{y}{\horstTypeLValue}
    \horstFreeVar{z}{\horstTypeLValue}
    \horstFreeVar{gt}{\horstTypeHomInit{\horstTypeLValue}{\horstOpAppgs{}{}}}
    \horstPremise{\horstPredAppMState{\horstParVarfid,\horstParVarpc}{\horstFreeVarctx,\horstCONS{\horstFreeVarx}{\horstCONS{\horstFreeVary}{\horstCONS{\horstFreeVarz}{\horstFreeVarst}}},\horstFreeVargt,\horstFreeVarlt,\horstFreeVarmem,\horstFreeVartbl,\horstFreeVaratN,\horstFreeVargtN,\horstFreeVarmemN}}
    \horstPremise{\horstOpAppabseq{}{\horstOpAppvalueOf{}{\horstFreeVarx},\horstOpAppmkConst{0}{}}}
    \horstConclusion{\horstPredAppMState{\horstParVarfid,\horstADD{\horstParVarpc}{1}}{\horstFreeVarctx,\horstCONS{\horstOpAppraiseTo{}{\horstFreeVary,\horstOpAppflub{2}{\horstTUPINIT{\horstOpApplabelOf{}{\horstFreeVarx},\horstOpApplabelOfCtx{}{\horstFreeVarctx}}}}}{\horstFreeVarst},\horstFreeVargt,\horstFreeVarlt,\horstFreeVarmem,\horstFreeVartbl,\horstFreeVaratN,\horstFreeVargtN,\horstFreeVarmemN}}
  \end{horstClause}
  \begin{horstClause}
    \horstFreeVar{memN}{\horstTypeMemory}
    \horstFreeVar{atN}{\horstTypeHomInit{\horstTypeLValue}{\horstOpAppas{\horstParVarfid}{}}}
    \horstFreeVar{st}{\horstTypeHomInit{\horstTypeLValue}{\horstSUB{\horstOpAppss{\horstParVarfid,\horstParVarpc}{}}{3}}}
    \horstFreeVar{tbl}{\horstTypeTable}
    \horstFreeVar{ctx}{\horstTypeContext}
    \horstFreeVar{lt}{\horstTypeHomInit{\horstTypeLValue}{\horstOpAppls{\horstParVarfid}{}}}
    \horstFreeVar{mem}{\horstTypeMemory}
    \horstFreeVar{x}{\horstTypeLValue}
    \horstFreeVar{gtN}{\horstTypeHomInit{\horstTypeLValue}{\horstOpAppgs{}{}}}
    \horstFreeVar{y}{\horstTypeLValue}
    \horstFreeVar{z}{\horstTypeLValue}
    \horstFreeVar{gt}{\horstTypeHomInit{\horstTypeLValue}{\horstOpAppgs{}{}}}
    \horstPremise{\horstPredAppMState{\horstParVarfid,\horstParVarpc}{\horstFreeVarctx,\horstCONS{\horstFreeVarx}{\horstCONS{\horstFreeVary}{\horstCONS{\horstFreeVarz}{\horstFreeVarst}}},\horstFreeVargt,\horstFreeVarlt,\horstFreeVarmem,\horstFreeVartbl,\horstFreeVaratN,\horstFreeVargtN,\horstFreeVarmemN}}
    \horstPremise{\horstOpAppabsneq{}{\horstOpAppvalueOf{}{\horstFreeVarx},\horstOpAppmkConst{0}{}}}
    \horstConclusion{\horstPredAppMState{\horstParVarfid,\horstADD{\horstParVarpc}{1}}{\horstFreeVarctx,\horstCONS{\horstOpAppraiseTo{}{\horstFreeVarz,\horstOpAppflub{2}{\horstTUPINIT{\horstOpApplabelOf{}{\horstFreeVarx},\horstOpApplabelOfCtx{}{\horstFreeVarctx}}}}}{\horstFreeVarst},\horstFreeVargt,\horstFreeVarlt,\horstFreeVarmem,\horstFreeVartbl,\horstFreeVaratN,\horstFreeVargtN,\horstFreeVarmemN}}
  \end{horstClause}
\end{horstRule}
\begin{horstRule}{globalGetRule}
  \horstParVar{fid}{\horstTypeint}
  \horstParVar{pc}{\horstTypeint}
  \horstParVar{idx}{\horstTypeint}
  \horstSelectorFunctionInvocation{\horstSelectorFunctionAppfunctionIds{\horstParVarfid}{},\horstSelectorFunctionApppcsForFunctionIdAndOpcode{\horstParVarpc}{\horstParVarfid,\horstConstGLOBALGET},\horstSelectorFunctionAppimmediateForFunctionIdAndPc{\horstParVaridx}{\horstParVarfid,\horstParVarpc}}
  \begin{horstClause}
    \horstFreeVar{memN}{\horstTypeMemory}
    \horstFreeVar{atN}{\horstTypeHomInit{\horstTypeLValue}{\horstOpAppas{\horstParVarfid}{}}}
    \horstFreeVar{st}{\horstTypeHomInit{\horstTypeLValue}{\horstOpAppss{\horstParVarfid,\horstParVarpc}{}}}
    \horstFreeVar{tbl}{\horstTypeTable}
    \horstFreeVar{ctx}{\horstTypeContext}
    \horstFreeVar{lt}{\horstTypeHomInit{\horstTypeLValue}{\horstOpAppls{\horstParVarfid}{}}}
    \horstFreeVar{mem}{\horstTypeMemory}
    \horstFreeVar{gtN}{\horstTypeHomInit{\horstTypeLValue}{\horstOpAppgs{}{}}}
    \horstFreeVar{gt}{\horstTypeHomInit{\horstTypeLValue}{\horstOpAppgs{}{}}}
    \horstPremise{\horstPredAppMState{\horstParVarfid,\horstParVarpc}{\horstFreeVarctx,\horstFreeVarst,\horstFreeVargt,\horstFreeVarlt,\horstFreeVarmem,\horstFreeVartbl,\horstFreeVaratN,\horstFreeVargtN,\horstFreeVarmemN}}
    \horstConclusion{\horstPredAppMState{\horstParVarfid,\horstADD{\horstParVarpc}{1}}{\horstFreeVarctx,\horstCONS{\horstOpAppraiseTo{}{\horstACCESS{\horstFreeVargt}{\horstParVaridx},\horstOpApplabelOfCtx{}{\horstFreeVarctx}}}{\horstFreeVarst},\horstFreeVargt,\horstFreeVarlt,\horstFreeVarmem,\horstFreeVartbl,\horstFreeVaratN,\horstFreeVargtN,\horstFreeVarmemN}}
  \end{horstClause}
\end{horstRule}
\begin{horstRule}{testNoninterferenceTableSat}
  \horstParVar{fid}{\horstTypeint}
  \horstParVar{cid}{\horstTypeint}
  \horstSelectorFunctionInvocation{\horstSelectorFunctionAppstartFunctionId{\horstParVarfid}{},\horstSelectorFunctionApptableLeak{\horstParVarcid}{}}
  \begin{horstClause}
    \horstFreeVar{memN}{\horstTypeMemory}
    \horstFreeVar{atN}{\horstTypeHomInit{\horstTypeLValue}{\horstOpAppas{\horstParVarfid}{}}}
    \horstFreeVar{ptbl}{\horstTypeTablePrecision}
    \horstFreeVar{ctx}{\horstTypeContext}
    \horstFreeVar{lt}{\horstTypeHomInit{\horstTypeLValue}{\horstOpAppls{\horstParVarfid}{}}}
    \horstFreeVar{mem}{\horstTypeMemory}
    \horstFreeVar{gtN}{\horstTypeHomInit{\horstTypeLValue}{\horstOpAppgs{}{}}}
    \horstFreeVar{gt}{\horstTypeHomInit{\horstTypeLValue}{\horstOpAppgs{}{}}}
    \horstPremise{\horstPredAppMState{\horstParVarfid,0}{\horstFreeVarctx,\horstTUPINIT{},\horstFreeVargt,\horstFreeVarlt,\horstFreeVarmem,\horstConstructorAppTbl{\horstFreeVarptbl,\horstConstructorAppIllegal{}},\horstFreeVaratN,\horstFreeVargtN,\horstFreeVarmemN}}
    \horstPremise{\horstSimpleSumExp{OR}{\horstSelectorFunctionApptableOutLabel{\horstParVarct,\horstParVarit}{}}{\horstOpAppflowsTo{}{\horstOpAppmkLabel{}{\horstParVarct,\horstParVarit},\horstOpAppperspectiveOfCtx{}{\horstFreeVarctx}}}{{ct}{it}}}
    \horstConclusion{\horstPredApptestNoninterferenceTableSat{\horstParVarfid,\horstParVarcid}{}}
  \end{horstClause}
\end{horstRule}
\begin{horstRule}{blockRule}
  \horstParVar{fid}{\horstTypeint}
  \horstParVar{pc}{\horstTypeint}
  \horstSelectorFunctionInvocation{\horstSelectorFunctionAppfunctionIds{\horstParVarfid}{},\horstSelectorFunctionAppblocksForFunctionId{\horstParVarpc}{\horstParVarfid}}
  \begin{horstClause}
    \horstFreeVar{memN}{\horstTypeMemory}
    \horstFreeVar{atN}{\horstTypeHomInit{\horstTypeLValue}{\horstOpAppas{\horstParVarfid}{}}}
    \horstFreeVar{st}{\horstTypeHomInit{\horstTypeLValue}{\horstOpAppss{\horstParVarfid,\horstParVarpc}{}}}
    \horstFreeVar{tbl}{\horstTypeTable}
    \horstFreeVar{ctx}{\horstTypeContext}
    \horstFreeVar{lt}{\horstTypeHomInit{\horstTypeLValue}{\horstOpAppls{\horstParVarfid}{}}}
    \horstFreeVar{mem}{\horstTypeMemory}
    \horstFreeVar{gtN}{\horstTypeHomInit{\horstTypeLValue}{\horstOpAppgs{}{}}}
    \horstFreeVar{gt}{\horstTypeHomInit{\horstTypeLValue}{\horstOpAppgs{}{}}}
    \horstPremise{\horstPredAppMState{\horstParVarfid,\horstParVarpc}{\horstFreeVarctx,\horstFreeVarst,\horstFreeVargt,\horstFreeVarlt,\horstFreeVarmem,\horstFreeVartbl,\horstFreeVaratN,\horstFreeVargtN,\horstFreeVarmemN}}
    \horstConclusion{\horstPredAppMState{\horstParVarfid,\horstADD{\horstParVarpc}{1}}{\horstFreeVarctx,\horstFreeVarst,\horstFreeVargt,\horstFreeVarlt,\horstFreeVarmem,\horstFreeVartbl,\horstFreeVaratN,\horstFreeVargtN,\horstFreeVarmemN}}
  \end{horstClause}
\end{horstRule}
\begin{horstRule}{testNoninterferenceGlobalUnsat}
  \horstParVar{fid}{\horstTypeint}
  \horstParVar{cid}{\horstTypeint}
  \horstSelectorFunctionInvocation{\horstSelectorFunctionAppstartFunctionId{\horstParVarfid}{},\horstSelectorFunctionAppglobalSafe{\horstParVarcid}{}}
  \begin{horstClause}
    \horstFreeVar{memN}{\horstTypeMemory}
    \horstFreeVar{atN}{\horstTypeHomInit{\horstTypeLValue}{\horstOpAppas{\horstParVarfid}{}}}
    \horstFreeVar{tbl}{\horstTypeTable}
    \horstFreeVar{rt}{\horstTypeHomInit{\horstTypeLValue}{\horstOpApprs{\horstParVarfid}{}}}
    \horstFreeVar{ctx}{\horstTypeContext}
    \horstFreeVar{mem}{\horstTypeMemory}
    \horstFreeVar{gtN}{\horstTypeHomInit{\horstTypeLValue}{\horstOpAppgs{}{}}}
    \horstFreeVar{gt}{\horstTypeHomInit{\horstTypeLValue}{\horstOpAppgs{}{}}}
    \horstPremise{\horstPredAppReturnCall{\horstParVarfid}{\horstFreeVarctx,\horstFreeVarrt,\horstFreeVargt,\horstFreeVarmem,\horstFreeVartbl,\horstFreeVaratN,\horstFreeVargtN,\horstFreeVarmemN}}
    \horstPremise{\horstSimpleSumExp{OR}{\horstSelectorFunctionAppinterval{\horstParVaridx}{0,\horstOpAppgs{}{}},\horstSelectorFunctionAppglobalOutLabelForPosition{\horstParVarcg,\horstParVarig}{\horstParVaridx}}{\horstAND{\horstOpAppflowsTo{}{\horstOpAppmkLabel{}{\horstParVarcg,\horstParVarig},\horstOpAppperspectiveOfCtx{}{\horstFreeVarctx}}}{\horstEQ{\horstOpApplabelOf{}{\horstACCESS{\horstFreeVargt}{\horstParVaridx}}}{\horstConstructorAppIllegal{}}}}{{idx}{cg}{ig}}}
    \horstConclusion{\horstPredApptestNoninterferenceGlobalUnsat{\horstParVarfid,\horstParVarcid}{}}
  \end{horstClause}
\end{horstRule}
\begin{horstRule}{testNoninterferenceImportedFunctionGlobalsSat}
  \horstParVar{fid}{\horstTypeint}
  \horstParVar{cid}{\horstTypeint}
  \horstSelectorFunctionInvocation{\horstSelectorFunctionAppimportedFunctionIds{\horstParVarfid}{},\horstSelectorFunctionAppimportCallGlobalLeak{\horstParVarcid}{\horstParVarfid}}
  \begin{horstClause}
    \horstFreeVar{memN}{\horstTypeMemory}
    \horstFreeVar{atN}{\horstTypeHomInit{\horstTypeLValue}{\horstOpAppas{\horstParVarfid}{}}}
    \horstFreeVar{tbl}{\horstTypeTable}
    \horstFreeVar{ctx}{\horstTypeContext}
    \horstFreeVar{lt}{\horstTypeHomInit{\horstTypeLValue}{\horstOpAppls{\horstParVarfid}{}}}
    \horstFreeVar{mem}{\horstTypeMemory}
    \horstFreeVar{gtN}{\horstTypeHomInit{\horstTypeLValue}{\horstOpAppgs{}{}}}
    \horstFreeVar{gt}{\horstTypeHomInit{\horstTypeLValue}{\horstOpAppgs{}{}}}
    \horstPremise{\horstPredAppMState{\horstParVarfid,0}{\horstFreeVarctx,\horstTUPINIT{},\horstFreeVargt,\horstFreeVarlt,\horstFreeVarmem,\horstFreeVartbl,\horstFreeVaratN,\horstFreeVargtN,\horstFreeVarmemN}}
    \horstPremise{\horstSimpleSumExp{OR}{\horstSelectorFunctionAppinterval{\horstParVaridx}{0,\horstOpAppgs{}{}},\horstSelectorFunctionAppglobalInLabelForImportedFunctionAndPosition{\horstParVarcg,\horstParVarig}{\horstParVarfid,\horstParVaridx}}{\horstAND{\horstOpAppflowsTo{}{\horstOpAppmkLabel{}{\horstParVarcg,\horstParVarig},\horstOpAppperspectiveOfCtx{}{\horstFreeVarctx}}}{\horstEQ{\horstOpApplabelOf{}{\horstACCESS{\horstFreeVargtN}{\horstParVaridx}}}{\horstConstructorAppIllegal{}}}}{{idx}{cg}{ig}}}
    \horstConclusion{\horstPredApptestNoninterferenceImportedFunctionGlobalsSat{\horstParVarfid,\horstParVarcid}{}}
  \end{horstClause}
\end{horstRule}
\begin{horstRule}{localSetRule}
  \horstParVar{fid}{\horstTypeint}
  \horstParVar{pc}{\horstTypeint}
  \horstParVar{idx}{\horstTypeint}
  \horstSelectorFunctionInvocation{\horstSelectorFunctionAppfunctionIds{\horstParVarfid}{},\horstSelectorFunctionApppcsForFunctionIdAndOpcode{\horstParVarpc}{\horstParVarfid,\horstConstLOCALSET},\horstSelectorFunctionAppimmediateForFunctionIdAndPc{\horstParVaridx}{\horstParVarfid,\horstParVarpc}}
  \begin{horstClause}
    \horstFreeVar{memN}{\horstTypeMemory}
    \horstFreeVar{atN}{\horstTypeHomInit{\horstTypeLValue}{\horstOpAppas{\horstParVarfid}{}}}
    \horstFreeVar{st}{\horstTypeHomInit{\horstTypeLValue}{\horstSUB{\horstOpAppss{\horstParVarfid,\horstParVarpc}{}}{1}}}
    \horstFreeVar{tbl}{\horstTypeTable}
    \horstFreeVar{ctx}{\horstTypeContext}
    \horstFreeVar{lt}{\horstTypeHomInit{\horstTypeLValue}{\horstOpAppls{\horstParVarfid}{}}}
    \horstFreeVar{mem}{\horstTypeMemory}
    \horstFreeVar{x}{\horstTypeLValue}
    \horstFreeVar{gtN}{\horstTypeHomInit{\horstTypeLValue}{\horstOpAppgs{}{}}}
    \horstFreeVar{gt}{\horstTypeHomInit{\horstTypeLValue}{\horstOpAppgs{}{}}}
    \horstPremise{\horstPredAppMState{\horstParVarfid,\horstParVarpc}{\horstFreeVarctx,\horstCONS{\horstFreeVarx}{\horstFreeVarst},\horstFreeVargt,\horstFreeVarlt,\horstFreeVarmem,\horstFreeVartbl,\horstFreeVaratN,\horstFreeVargtN,\horstFreeVarmemN}}
    \horstConclusion{\horstPredAppMState{\horstParVarfid,\horstADD{\horstParVarpc}{1}}{\horstFreeVarctx,\horstFreeVarst,\horstFreeVargt,\horstOpAppset{\horstOpAppls{\horstParVarfid}{},\horstParVaridx}{\horstOpAppraiseTo{}{\horstFreeVarx,\horstOpApplabelOfCtx{}{\horstFreeVarctx}},\horstFreeVarlt},\horstFreeVarmem,\horstFreeVartbl,\horstFreeVaratN,\horstFreeVargtN,\horstFreeVarmemN}}
  \end{horstClause}
\end{horstRule}
\begin{horstRule}{sizeRule}
  \horstParVar{fid}{\horstTypeint}
  \horstParVar{pc}{\horstTypeint}
  \horstSelectorFunctionInvocation{\horstSelectorFunctionAppfunctionIds{\horstParVarfid}{},\horstSelectorFunctionApppcsForFunctionIdAndOpcode{\horstParVarpc}{\horstParVarfid,\horstConstMEMORYSIZE}}
  \begin{horstClause}
    \horstFreeVar{memN}{\horstTypeMemory}
    \horstFreeVar{atN}{\horstTypeHomInit{\horstTypeLValue}{\horstOpAppas{\horstParVarfid}{}}}
    \horstFreeVar{st}{\horstTypeHomInit{\horstTypeLValue}{\horstOpAppss{\horstParVarfid,\horstParVarpc}{}}}
    \horstFreeVar{tbl}{\horstTypeTable}
    \horstFreeVar{ctx}{\horstTypeContext}
    \horstFreeVar{v}{\horstTypeLValue}
    \horstFreeVar{lt}{\horstTypeHomInit{\horstTypeLValue}{\horstOpAppls{\horstParVarfid}{}}}
    \horstFreeVar{gtN}{\horstTypeHomInit{\horstTypeLValue}{\horstOpAppgs{}{}}}
    \horstFreeVar{i}{\horstTypeValue}
    \horstFreeVar{gt}{\horstTypeHomInit{\horstTypeLValue}{\horstOpAppgs{}{}}}
    \horstFreeVar{size}{\horstTypeLValue}
    \horstPremise{\horstPredAppMState{\horstParVarfid,\horstParVarpc}{\horstFreeVarctx,\horstFreeVarst,\horstFreeVargt,\horstFreeVarlt,\horstConstructorAppMem{\horstFreeVari,\horstFreeVarv,\horstFreeVarsize},\horstFreeVartbl,\horstFreeVaratN,\horstFreeVargtN,\horstFreeVarmemN}}
    \horstConclusion{\horstPredAppMState{\horstParVarfid,\horstADD{\horstParVarpc}{1}}{\horstFreeVarctx,\horstCONS{\horstOpAppraiseTo{}{\horstFreeVarsize,\horstOpApplabelOfCtx{}{\horstFreeVarctx}}}{\horstFreeVarst},\horstFreeVargt,\horstFreeVarlt,\horstConstructorAppMem{\horstFreeVari,\horstFreeVarv,\horstFreeVarsize},\horstFreeVartbl,\horstFreeVaratN,\horstFreeVargtN,\horstFreeVarmemN}}
  \end{horstClause}
\end{horstRule}
\begin{horstRule}{invokeRule}
  \horstParVar{fid}{\horstTypeint}
  \horstParVar{hasMem}{\horstTypebool}
  \horstParVar{importMem}{\horstTypebool}
  \horstParVar{min}{\horstTypeint}
  \horstParVar{max}{\horstTypeint}
  \horstParVar{cd}{\horstTypebool}
  \horstParVar{id}{\horstTypebool}
  \horstParVar{cs}{\horstTypebool}
  \horstParVar{is}{\horstTypebool}
  \horstParVar{impreciseTbl}{\horstTypebool}
  \horstParVar{ct}{\horstTypebool}
  \horstParVar{it}{\horstTypebool}
  \horstSelectorFunctionInvocation{\horstSelectorFunctionAppstartFunctionId{\horstParVarfid}{},\horstSelectorFunctionAppisMemoryPresent{\horstParVarhasMem}{},\horstSelectorFunctionAppisMemoryImported{\horstParVarimportMem}{},\horstSelectorFunctionAppgetMemoryMin{\horstParVarmin}{},\horstSelectorFunctionAppgetMemoryMax{\horstParVarmax}{},\horstSelectorFunctionAppmemoryDataInLabel{\horstParVarcd,\horstParVarid}{},\horstSelectorFunctionAppmemorySizeInLabel{\horstParVarcs,\horstParVaris}{},\horstSelectorFunctionAppisTableImprecise{\horstParVarimpreciseTbl}{},\horstSelectorFunctionApptableInLabel{\horstParVarct,\horstParVarit}{}}
  \begin{horstClause}
    \horstFreeVar{p}{\horstTypeLabel}
    \horstFreeVar{at}{\horstTypeHomInit{\horstTypeLValue}{\horstOpAppas{\horstParVarfid}{}}}
    \horstFreeVar{v}{\horstTypeLValue}
    \horstFreeVar{i}{\horstTypeValue}
    \horstFreeVar{arr}{\horstTypeArray{\horstTypeValue}}
    \horstFreeVar{gt}{\horstTypeHomInit{\horstTypeLValue}{\horstOpAppgs{}{}}}
    \horstFreeVar{size}{\horstTypeLValue}
    \horstPremise{\horstSimpleSumExp{AND}{\horstSelectorFunctionAppinterval{\horstParVaridx}{0,\horstOpAppgs{}{}},\horstSelectorFunctionAppvalueAndTopOfGlobal{\horstParVarv,\horstParVartop}{\horstParVaridx},\horstSelectorFunctionAppbitwidthForGlobal{\horstParVarbw}{\horstParVaridx},\horstSelectorFunctionAppglobalInLabelForPosition{\horstParVarcg,\horstParVarig}{\horstParVaridx}}{\horstAND{\horstAND{\horstEQ{\horstOpAppvalueOf{}{\horstACCESS{\horstFreeVargt}{\horstParVaridx}}}{\horstOpAppval{\horstParVartop,\horstParVarv}{}}}{\horstOpAppisInRange{\horstParVarbw}{\horstOpAppvalueOf{}{\horstACCESS{\horstFreeVargt}{\horstParVaridx}}}}}{\horstEQ{\horstOpApplabelOf{}{\horstACCESS{\horstFreeVargt}{\horstParVaridx}}}{\horstOpAppmkFlowLabel{}{\horstOpAppmkLabel{}{\horstParVarcg,\horstParVarig},\horstFreeVarp}}}}{{idx}{v}{top}{bw}{cg}{ig}}}
    \horstPremise{\horstSimpleSumExp{AND}{\horstSelectorFunctionAppinterval{\horstParVaridx}{0,\horstOpAppas{\horstParVarfid}{}},\horstSelectorFunctionAppbitwidthForArgument{\horstParVarbw}{\horstParVarfid,\horstParVaridx},\horstSelectorFunctionAppargumentLabelForPosition{\horstParVarca,\horstParVaria}{\horstParVaridx}}{\horstAND{\horstOpAppisInRange{\horstParVarbw}{\horstOpAppvalueOf{}{\horstACCESS{\horstFreeVarat}{\horstParVaridx}}}}{\horstEQ{\horstOpApplabelOf{}{\horstACCESS{\horstFreeVarat}{\horstParVaridx}}}{\horstOpAppmkFlowLabel{}{\horstOpAppmkLabel{}{\horstParVarca,\horstParVaria},\horstFreeVarp}}}}{{idx}{bw}{ca}{ia}}}
    \horstPremise{\horstOpAppisVal{}{\horstFreeVari}}
    \horstPremise{\horstOpAppiltu{64}{\horstFreeVari,\horstOpAppmkConst{\horstMUL{\horstParVarmax}{\horstOpApppow{16}{2}}}{}}}
    \horstPremise{\horstCOND{\horstParVarimportMem}{\horstOpAppisInRange{8}{\horstOpAppvalueOf{}{\horstFreeVarv}}}{\horstEQ{\horstOpAppvalueOf{}{\horstFreeVarv}}{\horstCustomSumExp{\horstSelectorFunctionAppdatasegmentsWithPositions{\horstParVarpos,\horstParVarval}{}}{acc}{\horstCOND{\horstEQ{\horstOpAppmkConst{\horstParVarpos}{}}{\horstFreeVari}}{\horstOpAppmkConst{\horstParVarval}{}}{\horstVaracc}}{\horstOpAppmkConst{0}{}}{{pos}{val}}}}}
    \horstPremise{\horstEQ{\horstOpApplabelOf{}{\horstFreeVarv}}{\horstCustomSumExp{\horstSelectorFunctionAppmemoryDataInLabelExceptions{\horstParVarstart,\horstParVarendInclusive,\horstParVarced,\horstParVaried}{}}{x}{\horstCOND{\horstAND{\horstOpAppabsle{}{\horstOpAppmkConst{\horstParVarstart}{},\horstFreeVari}}{\horstOpAppabslt{}{\horstFreeVari,\horstOpAppmkConst{\horstParVarendInclusive}{}}}}{\horstOpAppmkFlowLabel{}{\horstOpAppmkLabel{}{\horstParVarced,\horstParVaried},\horstFreeVarp}}{\horstVarx}}{\horstOpAppmkFlowLabel{}{\horstOpAppmkLabel{}{\horstParVarcd,\horstParVarid},\horstFreeVarp}}{{start}{endInclusive}{ced}{ied}}}}
    \horstPremise{\horstEQ{\horstOpAppvalueOf{}{\horstFreeVarsize}}{\horstOpAppmkConst{\horstParVarmin}{}}}
    \horstPremise{\horstEQ{\horstOpApplabelOf{}{\horstFreeVarsize}}{\horstOpAppmkFlowLabel{}{\horstOpAppmkLabel{}{\horstParVarcs,\horstParVaris},\horstFreeVarp}}}
    \horstConclusion{\horstPredAppInit{\horstParVarfid}{\horstFreeVarp,\horstFreeVarat,\horstFreeVargt,\horstConstructorAppMem{\horstFreeVari,\horstFreeVarv,\horstFreeVarsize}}}
  \end{horstClause}
  \begin{horstClause}
    \horstFreeVar{tbl}{\horstTypeTable}
    \horstFreeVar{p}{\horstTypeLabel}
    \horstFreeVar{at}{\horstTypeHomInit{\horstTypeLValue}{\horstOpAppas{\horstParVarfid}{}}}
    \horstFreeVar{mem}{\horstTypeMemory}
    \horstFreeVar{gt}{\horstTypeHomInit{\horstTypeLValue}{\horstOpAppgs{}{}}}
    \horstPremise{\horstPredAppInit{\horstParVarfid}{\horstFreeVarp,\horstFreeVarat,\horstFreeVargt,\horstFreeVarmem}}
    \horstPremise{\horstEQ{\horstFreeVartbl}{\horstConstructorAppTbl{\horstCOND{\horstParVarimpreciseTbl}{\horstConstructorAppTblImprecise{}}{\horstConstructorAppTblPrecise{}},\horstOpAppmkFlowLabel{}{\horstOpAppmkLabel{}{\horstParVarct,\horstParVarit},\horstFreeVarp}}}}
    \horstConclusion{\horstPredAppMState{\horstParVarfid,0}{\horstOpAppmkCtx{}{\horstFreeVarp,\horstConstructorAppLegal{}},\horstTUPINIT{},\horstFreeVargt,\horstCONCAT{\horstFreeVarat}{\horstHOMINIT{\horstOpAppmkLConst{0}{}}{\horstSUB{\horstOpAppls{\horstParVarfid}{}}{\horstOpAppas{\horstParVarfid}{}}}},\horstFreeVarmem,\horstFreeVartbl,\horstFreeVarat,\horstFreeVargt,\horstFreeVarmem}}
  \end{horstClause}
\end{horstRule}
\begin{horstRule}{functionScopeExtendRule}
  \horstParVar{fid}{\horstTypeint}
  \horstSelectorFunctionInvocation{\horstSelectorFunctionAppfunctionIds{\horstParVarfid}{}}
  \begin{horstClause}
    \horstPremise{\horstTrue}
    \horstConclusion{\horstPredAppScopeExtend{\horstParVarfid}{-1,\horstOpApppcmax{\horstParVarfid}{}}}
  \end{horstClause}
\end{horstRule}
\begin{horstRule}{testNoninterferenceMemoryDataUnsat}
  \horstParVar{fid}{\horstTypeint}
  \horstParVar{cid}{\horstTypeint}
  \horstSelectorFunctionInvocation{\horstSelectorFunctionAppstartFunctionId{\horstParVarfid}{},\horstSelectorFunctionAppmemoryDataSafe{\horstParVarcid}{}}
  \begin{horstClause}
    \horstFreeVar{memN}{\horstTypeMemory}
    \horstFreeVar{atN}{\horstTypeHomInit{\horstTypeLValue}{\horstOpAppas{\horstParVarfid}{}}}
    \horstFreeVar{tbl}{\horstTypeTable}
    \horstFreeVar{rt}{\horstTypeHomInit{\horstTypeLValue}{\horstOpApprs{\horstParVarfid}{}}}
    \horstFreeVar{ctx}{\horstTypeContext}
    \horstFreeVar{mem}{\horstTypeMemory}
    \horstFreeVar{gtN}{\horstTypeHomInit{\horstTypeLValue}{\horstOpAppgs{}{}}}
    \horstFreeVar{gt}{\horstTypeHomInit{\horstTypeLValue}{\horstOpAppgs{}{}}}
    \horstPremise{\horstPredAppReturnCall{\horstParVarfid}{\horstFreeVarctx,\horstFreeVarrt,\horstFreeVargt,\horstFreeVarmem,\horstFreeVartbl,\horstFreeVaratN,\horstFreeVargtN,\horstFreeVarmemN}}
    \horstPremise{\horstSimpleSumExp{OR}{\horstSelectorFunctionAppmemoryDataOutLabel{\horstParVarcd,\horstParVarid}{}}{\horstAND{\horstOpAppflowsTo{}{\horstOpAppmkLabel{}{\horstParVarcd,\horstParVarid},\horstOpAppperspectiveOfCtx{}{\horstFreeVarctx}}}{\horstEQ{\horstOpApplabelOf{}{\horstOpAppvalueOfMem{}{\horstFreeVarmem}}}{\horstConstructorAppIllegal{}}}}{{cd}{id}}}
    \horstConclusion{\horstPredApptestNoninterferenceMemoryDataUnsat{\horstParVarfid,\horstParVarcid}{}}
  \end{horstClause}
\end{horstRule}
\begin{horstRule}{propagateReturnRule}
  \horstParVar{fid}{\horstTypeint}
  \horstSelectorFunctionInvocation{\horstSelectorFunctionAppstartFunctionId{\horstParVarfid}{}}
  \begin{horstClause}
    \horstFreeVar{memN}{\horstTypeMemory}
    \horstFreeVar{atN}{\horstTypeHomInit{\horstTypeLValue}{\horstOpAppas{\horstParVarfid}{}}}
    \horstFreeVar{tbl}{\horstTypeTable}
    \horstFreeVar{rt}{\horstTypeHomInit{\horstTypeLValue}{\horstOpApprs{\horstParVarfid}{}}}
    \horstFreeVar{p}{\horstTypeLabel}
    \horstFreeVar{mem}{\horstTypeMemory}
    \horstFreeVar{gtN}{\horstTypeHomInit{\horstTypeLValue}{\horstOpAppgs{}{}}}
    \horstFreeVar{gt}{\horstTypeHomInit{\horstTypeLValue}{\horstOpAppgs{}{}}}
    \horstFreeVar{l}{\horstTypeFlowLabel}
    \horstFreeVar{from}{\horstTypeint}
    \horstPremise{\horstPredAppReturn{\horstParVarfid}{\horstConstructorAppCtx{\horstFreeVarp,\horstFreeVarl,\horstFreeVarfrom},\horstFreeVarrt,\horstFreeVargt,\horstFreeVarmem,\horstFreeVartbl,\horstFreeVaratN,\horstFreeVargtN,\horstFreeVarmemN}}
    \horstPremise{\horstPredAppInit{\horstParVarfid}{\horstFreeVarp,\horstFreeVaratN,\horstFreeVargtN,\horstFreeVarmemN}}
    \horstPremise{\horstLE{0}{\horstFreeVarfrom}}
    \horstConclusion{\horstPredAppReturnCall{\horstParVarfid}{\horstConstructorAppCtx{\horstFreeVarp,\horstFreeVarl,\horstFreeVarfrom},\horstFreeVarrt,\horstFreeVargt,\horstFreeVarmem,\horstFreeVartbl,\horstFreeVaratN,\horstFreeVargtN,\horstFreeVarmemN}}
  \end{horstClause}
\end{horstRule}
\begin{horstRule}{testNoninterferenceResultUnsat}
  \horstParVar{fid}{\horstTypeint}
  \horstParVar{cid}{\horstTypeint}
  \horstSelectorFunctionInvocation{\horstSelectorFunctionAppstartFunctionId{\horstParVarfid}{},\horstSelectorFunctionAppresultSafe{\horstParVarcid}{}}
  \begin{horstClause}
    \horstFreeVar{memN}{\horstTypeMemory}
    \horstFreeVar{atN}{\horstTypeHomInit{\horstTypeLValue}{\horstOpAppas{\horstParVarfid}{}}}
    \horstFreeVar{tbl}{\horstTypeTable}
    \horstFreeVar{rt}{\horstTypeHomInit{\horstTypeLValue}{\horstOpApprs{\horstParVarfid}{}}}
    \horstFreeVar{ctx}{\horstTypeContext}
    \horstFreeVar{mem}{\horstTypeMemory}
    \horstFreeVar{gtN}{\horstTypeHomInit{\horstTypeLValue}{\horstOpAppgs{}{}}}
    \horstFreeVar{gt}{\horstTypeHomInit{\horstTypeLValue}{\horstOpAppgs{}{}}}
    \horstPremise{\horstPredAppReturnCall{\horstParVarfid}{\horstFreeVarctx,\horstFreeVarrt,\horstFreeVargt,\horstFreeVarmem,\horstFreeVartbl,\horstFreeVaratN,\horstFreeVargtN,\horstFreeVarmemN}}
    \horstPremise{\horstSimpleSumExp{OR}{\horstSelectorFunctionAppinterval{\horstParVaridx}{0,\horstOpApprs{\horstParVarfid}{}},\horstSelectorFunctionAppresultLabelForPosition{\horstParVarcr,\horstParVarir}{\horstParVaridx}}{\horstAND{\horstOpAppflowsTo{}{\horstOpAppmkLabel{}{\horstParVarcr,\horstParVarir},\horstOpAppperspectiveOfCtx{}{\horstFreeVarctx}}}{\horstEQ{\horstOpApplabelOf{}{\horstACCESS{\horstFreeVarrt}{\horstParVaridx}}}{\horstConstructorAppIllegal{}}}}{{idx}{cr}{ir}}}
    \horstConclusion{\horstPredApptestNoninterferenceResultUnsat{\horstParVarfid,\horstParVarcid}{}}
  \end{horstClause}
\end{horstRule}
\begin{horstRule}{testNoninterferenceImportedFunctionMemorySizeSat}
  \horstParVar{fid}{\horstTypeint}
  \horstParVar{cid}{\horstTypeint}
  \horstSelectorFunctionInvocation{\horstSelectorFunctionAppimportedFunctionIds{\horstParVarfid}{},\horstSelectorFunctionAppimportCallMemorySizeLeak{\horstParVarcid}{\horstParVarfid}}
  \begin{horstClause}
    \horstFreeVar{memN}{\horstTypeMemory}
    \horstFreeVar{atN}{\horstTypeHomInit{\horstTypeLValue}{\horstOpAppas{\horstParVarfid}{}}}
    \horstFreeVar{tbl}{\horstTypeTable}
    \horstFreeVar{ctx}{\horstTypeContext}
    \horstFreeVar{v}{\horstTypeLValue}
    \horstFreeVar{lt}{\horstTypeHomInit{\horstTypeLValue}{\horstOpAppls{\horstParVarfid}{}}}
    \horstFreeVar{gtN}{\horstTypeHomInit{\horstTypeLValue}{\horstOpAppgs{}{}}}
    \horstFreeVar{i}{\horstTypeValue}
    \horstFreeVar{gt}{\horstTypeHomInit{\horstTypeLValue}{\horstOpAppgs{}{}}}
    \horstFreeVar{size}{\horstTypeLValue}
    \horstPremise{\horstPredAppMState{\horstParVarfid,0}{\horstFreeVarctx,\horstTUPINIT{},\horstFreeVargt,\horstFreeVarlt,\horstConstructorAppMem{\horstFreeVari,\horstFreeVarv,\horstFreeVarsize},\horstFreeVartbl,\horstFreeVaratN,\horstFreeVargtN,\horstFreeVarmemN}}
    \horstPremise{\horstSimpleSumExp{OR}{\horstSelectorFunctionAppmemoryDataInLabelForImportedFunction{\horstParVarcs,\horstParVaris}{\horstParVarfid}}{\horstAND{\horstOpAppflowsTo{}{\horstOpAppmkLabel{}{\horstParVarcs,\horstParVaris},\horstOpAppperspectiveOfCtx{}{\horstFreeVarctx}}}{\horstEQ{\horstOpApplabelOf{}{\horstFreeVarsize}}{\horstConstructorAppIllegal{}}}}{{cs}{is}}}
    \horstConclusion{\horstPredApptestNoninterferenceImportedFunctionMemorySizeSat{\horstParVarfid,\horstParVarcid}{}}
  \end{horstClause}
\end{horstRule}
\begin{horstRule}{returnJoinRule}
  \horstParVar{fid}{\horstTypeint}
  \horstSelectorFunctionInvocation{\horstSelectorFunctionAppfunctionIds{\horstParVarfid}{}}
  \begin{horstClause}
    \horstFreeVar{gtNI}{\horstTypeHomInit{\horstTypeLValue}{\horstOpAppgs{}{}}}
    \horstFreeVar{gtNII}{\horstTypeHomInit{\horstTypeLValue}{\horstOpAppgs{}{}}}
    \horstFreeVar{memNI}{\horstTypeMemory}
    \horstFreeVar{memNII}{\horstTypeMemory}
    \horstFreeVar{rtIII}{\horstTypeHomInit{\horstTypeLValue}{\horstOpApprs{\horstParVarfid}{}}}
    \horstFreeVar{rtI}{\horstTypeHomInit{\horstTypeLValue}{\horstOpApprs{\horstParVarfid}{}}}
    \horstFreeVar{rtII}{\horstTypeHomInit{\horstTypeLValue}{\horstOpApprs{\horstParVarfid}{}}}
    \horstFreeVar{atNII}{\horstTypeHomInit{\horstTypeLValue}{\horstOpAppas{\horstParVarfid}{}}}
    \horstFreeVar{memI}{\horstTypeMemory}
    \horstFreeVar{memII}{\horstTypeMemory}
    \horstFreeVar{atNI}{\horstTypeHomInit{\horstTypeLValue}{\horstOpAppas{\horstParVarfid}{}}}
    \horstFreeVar{ctx}{\horstTypeContext}
    \horstFreeVar{memIII}{\horstTypeMemory}
    \horstFreeVar{gtII}{\horstTypeHomInit{\horstTypeLValue}{\horstOpAppgs{}{}}}
    \horstFreeVar{gtIII}{\horstTypeHomInit{\horstTypeLValue}{\horstOpAppgs{}{}}}
    \horstFreeVar{gtI}{\horstTypeHomInit{\horstTypeLValue}{\horstOpAppgs{}{}}}
    \horstFreeVar{tblI}{\horstTypeTable}
    \horstFreeVar{tblII}{\horstTypeTable}
    \horstPremise{\horstPredAppReturnToJoin{\horstParVarfid}{\horstFreeVarctx,\horstFreeVarrtI,\horstFreeVargtI,\horstFreeVarmemI,\horstFreeVartblI,\horstFreeVaratNI,\horstFreeVargtNI,\horstFreeVarmemNI}}
    \horstPremise{\horstPredAppReturnToJoin{\horstParVarfid}{\horstFreeVarctx,\horstFreeVarrtII,\horstFreeVargtII,\horstFreeVarmemII,\horstFreeVartblII,\horstFreeVaratNII,\horstFreeVargtNII,\horstFreeVarmemNII}}
    \horstPremise{\horstOpApplowEq{\horstOpAppas{\horstParVarfid}{}}{\horstFreeVaratNI,\horstFreeVaratNII}}
    \horstPremise{\horstOpApplowEq{\horstOpAppgs{}{}}{\horstFreeVargtNI,\horstFreeVargtNII}}
    \horstPremise{\horstOpApplowEqMem{}{\horstFreeVarmemNI,\horstFreeVarmemNII}}
    \horstPremise{\horstOpAppjoinTuples{\horstOpApprs{\horstParVarfid}{}}{\horstFreeVarrtI,\horstFreeVarrtII,\horstFreeVarrtIII}}
    \horstPremise{\horstOpAppjoinTuples{\horstOpAppgs{}{}}{\horstFreeVargtI,\horstFreeVargtII,\horstFreeVargtIII}}
    \horstPremise{\horstOpAppjoinMem{}{\horstFreeVarmemI,\horstFreeVarmemII,\horstFreeVarmemIII}}
    \horstConclusion{\horstPredAppReturn{\horstParVarfid}{\horstFreeVarctx,\horstFreeVarrtIII,\horstFreeVargtIII,\horstFreeVarmemIII,\horstFreeVartblI,\horstFreeVaratNI,\horstFreeVargtNI,\horstFreeVarmemNI}}
  \end{horstClause}
\end{horstRule}
\begin{horstRule}{testNoninterferenceImportedFunctionGlobalsUnsat}
  \horstParVar{fid}{\horstTypeint}
  \horstParVar{cid}{\horstTypeint}
  \horstSelectorFunctionInvocation{\horstSelectorFunctionAppimportedFunctionIds{\horstParVarfid}{},\horstSelectorFunctionAppimportCallGlobalSafe{\horstParVarcid}{\horstParVarfid}}
  \begin{horstClause}
    \horstFreeVar{memN}{\horstTypeMemory}
    \horstFreeVar{atN}{\horstTypeHomInit{\horstTypeLValue}{\horstOpAppas{\horstParVarfid}{}}}
    \horstFreeVar{tbl}{\horstTypeTable}
    \horstFreeVar{ctx}{\horstTypeContext}
    \horstFreeVar{lt}{\horstTypeHomInit{\horstTypeLValue}{\horstOpAppls{\horstParVarfid}{}}}
    \horstFreeVar{mem}{\horstTypeMemory}
    \horstFreeVar{gtN}{\horstTypeHomInit{\horstTypeLValue}{\horstOpAppgs{}{}}}
    \horstFreeVar{gt}{\horstTypeHomInit{\horstTypeLValue}{\horstOpAppgs{}{}}}
    \horstPremise{\horstPredAppMState{\horstParVarfid,0}{\horstFreeVarctx,\horstTUPINIT{},\horstFreeVargt,\horstFreeVarlt,\horstFreeVarmem,\horstFreeVartbl,\horstFreeVaratN,\horstFreeVargtN,\horstFreeVarmemN}}
    \horstPremise{\horstSimpleSumExp{OR}{\horstSelectorFunctionAppinterval{\horstParVaridx}{0,\horstOpAppgs{}{}},\horstSelectorFunctionAppglobalInLabelForImportedFunctionAndPosition{\horstParVarcg,\horstParVarig}{\horstParVarfid,\horstParVaridx}}{\horstAND{\horstOpAppflowsTo{}{\horstOpAppmkLabel{}{\horstParVarcg,\horstParVarig},\horstOpAppperspectiveOfCtx{}{\horstFreeVarctx}}}{\horstEQ{\horstOpApplabelOf{}{\horstACCESS{\horstFreeVargtN}{\horstParVaridx}}}{\horstConstructorAppIllegal{}}}}{{idx}{cg}{ig}}}
    \horstConclusion{\horstPredApptestNoninterferenceImportedFunctionGlobalsUnsat{\horstParVarfid,\horstParVarcid}{}}
  \end{horstClause}
\end{horstRule}
\begin{horstRule}{binOpRule}
  \horstParVar{fid}{\horstTypeint}
  \horstParVar{op}{\horstTypeint}
  \horstParVar{pc}{\horstTypeint}
  \horstSelectorFunctionInvocation{\horstSelectorFunctionAppfunctionIds{\horstParVarfid}{},\horstSelectorFunctionAppbinOps{\horstParVarop}{},\horstSelectorFunctionApppcsForFunctionIdAndOpcode{\horstParVarpc}{\horstParVarfid,\horstParVarop}}
  \begin{horstClause}
    \horstFreeVar{memN}{\horstTypeMemory}
    \horstFreeVar{atN}{\horstTypeHomInit{\horstTypeLValue}{\horstOpAppas{\horstParVarfid}{}}}
    \horstFreeVar{st}{\horstTypeHomInit{\horstTypeLValue}{\horstSUB{\horstOpAppss{\horstParVarfid,\horstParVarpc}{}}{2}}}
    \horstFreeVar{tbl}{\horstTypeTable}
    \horstFreeVar{ctx}{\horstTypeContext}
    \horstFreeVar{lt}{\horstTypeHomInit{\horstTypeLValue}{\horstOpAppls{\horstParVarfid}{}}}
    \horstFreeVar{mem}{\horstTypeMemory}
    \horstFreeVar{x}{\horstTypeLValue}
    \horstFreeVar{gtN}{\horstTypeHomInit{\horstTypeLValue}{\horstOpAppgs{}{}}}
    \horstFreeVar{y}{\horstTypeLValue}
    \horstFreeVar{gt}{\horstTypeHomInit{\horstTypeLValue}{\horstOpAppgs{}{}}}
    \horstPremise{\horstPredAppMState{\horstParVarfid,\horstParVarpc}{\horstFreeVarctx,\horstCONS{\horstFreeVarx}{\horstCONS{\horstFreeVary}{\horstFreeVarst}},\horstFreeVargt,\horstFreeVarlt,\horstFreeVarmem,\horstFreeVartbl,\horstFreeVaratN,\horstFreeVargtN,\horstFreeVarmemN}}
    \horstConclusion{\horstPredAppMState{\horstParVarfid,\horstADD{\horstParVarpc}{1}}{\horstFreeVarctx,\horstCONS{\horstOpAppraiseTo{}{\horstOpApplabelledBinOp{\horstParVarop}{\horstFreeVary,\horstFreeVarx},\horstOpApplabelOfCtx{}{\horstFreeVarctx}}}{\horstFreeVarst},\horstFreeVargt,\horstFreeVarlt,\horstFreeVarmem,\horstFreeVartbl,\horstFreeVaratN,\horstFreeVargtN,\horstFreeVarmemN}}
  \end{horstClause}
\end{horstRule}

\renewcommand{\horstTypeArray}[1]{
  \mleft( \horstTypeint \mapsto #1 \mright)
}

\newcommand{\horstTypeint}{\mathbb{Z}}
\newcommand{\horstTypeValue}{\mathbb{B}_{64}}
\newcommand{\horstTypeBVLXIV}{\mathbb{B}_{64}}
\newcommand{\horstTypeLValue}{\mathbb{V}}
\newcommand{\horstTypeFlowLabel}{\mathbb{L}}
\newcommand{\horstTypeLabel}{\mathcal{L}}
\newcommand{\horstTypeTablePrecision}{\textit{Precision}}
\newcommand{\horstConstructorAppTblPrecise}[1]{\textit{Precise}}
\newcommand{\horstConstructorAppTblImprecise}[1]{\textit{Imprecise}}
\newcommand{\horstParVarfunctionId}{\mathbf{fid}}
\newcommand{\horstParVarreturnSize}{\mathbf{rs}}
\newcommand{\horstParVarglobalsSize}{\mathbf{gs}}
\newcommand{\horstParVarargumentsSize}{\mathbf{as}}
\newcommand{\horstConstructorAppTopB}{\top_B}
\newcommand{\fromSI}[2]{\textsf{SI.#1}\mleft(#2\mright)}
\newcommand{\horstOpAppss}[2]{\fromSI{ss}{#1}}
\newcommand{\horstOpAppls}[2]{\fromSI{ls}{#1}}
\newcommand{\horstOpAppas}[2]{\fromSI{as}{#1}}
\newcommand{\horstOpAppcsmaxFor}[2]{\textsf{SI.cs}_\textsf{max}\mleft(#1\mright)}
\newcommand{\horstOpAppgs}[2]{\fromSI{gs}{#1}}
\newcommand{\horstOpApprs}[2]{\fromSI{rs}{#1}}

\ExplSyntaxOn
\NewDocumentCommand{\ruleGroup}{m m o}{
  \group_begin:
  \IfNoValueF{#1}{
    \seq_set_from_clist:Nn \l_tmpa_seq {#3}
   }
    \begin{clauseGroup}{#2}
      \group_begin:
        \int_set:Nn \l_tmpa_int {0}
        \int_set:Nn \l_tmpb_int {\int_use:c {g_horstRule_#1_clauseCount_int}}

        \int_while_do:nNnn {\l_tmpa_int}<{\l_tmpb_int}
        {
          \group_begin:
          \par
          \smallskip
          \seq_if_in:NVTF \l_tmpa_seq \l_tmpa_int {
            \addedInWanilla{
              \clauseBox{#1}{\int_use:N \l_tmpa_int}
            }
          }{
            \clauseBox{#1}{\int_use:N \l_tmpa_int}
          } 
          \group_end:
          \int_incr:N \l_tmpa_int
        }
      \group_end:
    \end{clauseGroup}
  \medskip
  \group_end:
}
\ExplSyntaxOff

\ExplSyntaxOn
\DefineBinaryInfixOperation{\horstBINOP}{70}{\mathbin{\circledast \c_math_subscript_token {\textbf{op}}}}
\DefineBinaryInfixOperation{\horstTRAPPINGBINOP}{70}{\mathbin{\circledast \c_math_subscript_token {\textbf{op}}}}
\DefineBinaryInfixOperation{\horstINSET}{30}{\in}
\DefineBinaryInfixOperation{\horstLABSNEQ}{30}{\mathrel{\neq}}
\DefineBinaryInfixOperation{\horstLABSEQ}{30}{\mathrel{=}}
\DefineBinaryInfixOperation{\horstLUB}{80}{\sqcup}
\DefineBinaryInfixOperation{\horstLUBV}{80}{\mathbin{\overset{v}{\sqcup}}}
\DefineBinaryInfixOperation{\horstLUBPC}{80}{\mathbin{\overset{\textbf{pc}}{\sqcup}}}
\DefineBinaryInfixOperation{\horstSHL}{45}{\mathrel{\ll}}
\DefineBinaryInfixOperation{\horstSHR}{45}{\mathrel{\gg}}
\DefineBinaryInfixOperation{\horstBWAND}{30}{\mathrel{\&}}
\DefineBinaryInfixOperation{\horstBWOR}{29}{\mathrel{\parallel}}

\newcommand{\horstOpApplabelledBinOp}[2]{
  \group_begin:
    \clist_set:Nn \l_tmpa_clist { #2 }

    \horstBINOP{\clist_item:Nn \l_tmpa_clist { 1 }}{\clist_item:Nn \l_tmpa_clist { 2 }}
  \group_end:
}

\DefineBinaryInfixOperation{\horstOALOOPGL}{30}{\mathrel{{}\c_math_subscript_token \textbf{fid}\mkern-3mu\overset{G}{\lesssim}\mkern-3mu \c_math_subscript_token \textbf{pc}}}
\DefineBinaryInfixOperation{\horstOALOOPLO}{30}{\mathrel{{}\c_math_subscript_token \textbf{fid}\mkern-3mu\overset{L}{\lesssim}\mkern-3mu \c_math_subscript_token \textbf{pc}}}
\DefineBinaryInfixOperation{\horstOALOOPME}{30}{\mathrel{{}\c_math_subscript_token \textbf{fid}\mkern-3mu\overset{M}{\lesssim}\mkern-3mu \c_math_subscript_token \textbf{pc}}}

\DefineBinaryInfixOperation{\horstOACALLGL}{30}{\mathrel{{}\c_math_subscript_token \textbf{cid}\mkern-3mu\overset{G}{\lesssim}}}
\DefineBinaryInfixOperation{\horstOACALLAR}{30}{\mathrel{{}\c_math_subscript_token \textbf{cid}\mkern-3mu\overset{L}{\lesssim}}}
\DefineBinaryInfixOperation{\horstOACALLME}{30}{\mathrel{{}\c_math_subscript_token \textbf{cid}\mkern-3mu\overset{M}{\lesssim}}}

\newcommand{\horstOpAppoverApproximateLoopGlobals}[2]{
  \group_begin:
    \clist_set:Nn \l_tmpa_clist { #2 }

    \horstOALOOPGL{\clist_item:Nn \l_tmpa_clist { 1 }}{\clist_item:Nn \l_tmpa_clist { 2 }}
  \group_end:
}

\newcommand{\horstOpAppoverApproximateLoopLocals}[2]{
  \group_begin:
    \clist_set:Nn \l_tmpa_clist { #2 }

    \horstOALOOPLO{\clist_item:Nn \l_tmpa_clist { 1 }}{\clist_item:Nn \l_tmpa_clist { 2 }}
  \group_end:
}

\newcommand{\horstOpAppoverApproximateLoopMemory}[2]{
  \group_begin:
    \clist_set:Nn \l_tmpa_clist { #2 }

    \horstOALOOPME{\clist_item:Nn \l_tmpa_clist { 1 }}{\clist_item:Nn \l_tmpa_clist { 2 }}
  \group_end:
}

\newcommand{\horstOpAppoverApproximateCallGlobals}[2]{
  \group_begin:
    \clist_set:Nn \l_tmpa_clist { #2 }

    \horstOACALLGL{\clist_item:Nn \l_tmpa_clist { 1 }}{\clist_item:Nn \l_tmpa_clist { 2 }}
  \group_end:
}

\newcommand{\horstOpAppoverApproximateCallArguments}[2]{
  \group_begin:
    \clist_set:Nn \l_tmpa_clist { #2 }

    \horstOACALLAR{\clist_item:Nn \l_tmpa_clist { 1 }}{\clist_item:Nn \l_tmpa_clist { 2 }}
  \group_end:
}

\newcommand{\horstOpAppoverApproximateCallMemory}[2]{
  \group_begin:
    \clist_set:Nn \l_tmpa_clist { #2 }

    \horstOACALLME{\clist_item:Nn \l_tmpa_clist { 1 }}{\clist_item:Nn \l_tmpa_clist { 2 }}
  \group_end:
}

\newcommand{\horstOpApptrappingBinOp}[2]{
  \group_begin:
    \clist_set:Nn \l_tmpa_clist { #2 }

  \horstTRAPPINGBINOP{\clist_item:Nn \l_tmpa_clist { 1 }}{\clist_item:Nn \l_tmpa_clist { 2 }}
  \group_end:
}

\newcommand{\horstOpAppraiseCtxTo}[2]{
  \group_begin:
    \clist_set:Nn \l_tmpa_clist { #2 }

    \horstLUBPC{\clist_item:Nn \l_tmpa_clist { 1 }}{\clist_item:Nn \l_tmpa_clist { 2 }}
  \group_end:
}

\newcommand{\horstOpAppjoinTuples}[2]{
  \clist_set:Nn \l_tmpa_clist { #2 }

  \clist_item:Nn \l_tmpa_clist { 3 } = \horstOpApp{join}{}{\clist_item:Nn \l_tmpa_clist { 1 }, \clist_item:Nn \l_tmpa_clist { 2 }}
}

\newcommand{\horstOpAppjoinMem}[2]{
  \clist_set:Nn \l_tmpa_clist { #2 }

  \clist_item:Nn \l_tmpa_clist { 3 } = \horstOpApp{join}{}{\clist_item:Nn \l_tmpa_clist { 1 }, \clist_item:Nn \l_tmpa_clist { 2 }}
}

\newcommand{\horstOpAppflub}[2]{
  \clist_set:Nn \l_tmpa_clist { #2 }

  \group_begin:

  \renewcommand{\horstTUPINIT}[1]{ 
    \clist_set:Nn \l_tmpb_clist { ##1 }
    \clist_use:Nn \l_tmpb_clist {\sqcup}
  }

  \clist_item:Nn \l_tmpa_clist { 1 }

  \group_end:
}

\newcommand{\horstOpAppvalueOf}[2]{
  \group_begin:
    \clist_set:Nn \l_tmpa_clist { #2 }
    \access{\clist_item:Nn \l_tmpa_clist { 1 }}{val}
  \group_end:
}

\newcommand{\horstOpAppvaluesOf}[2]{
  \group_begin:
    \clist_set:Nn \l_tmpa_clist { #2 }
    \clist_item:Nn \l_tmpa_clist { 1 } . \textsf{value}
  \group_end:
}

\newcommand{\horstOpAppval}[2]{
  \group_begin:
    \clist_set:Nn \l_tmpa_clist { #1 }
    \clist_item:Nn \l_tmpa_clist { 2 }
  \group_end:
}

\newcommand{\horstOpAppdrop}[2]{
  \group_begin:
    \clist_set:Nn \l_tmpa_clist { #2 }
    \clist_set:Nn \l_tmpb_clist { #1 }
    \horstOpApp{unwind}{\clist_item:Nn \l_tmpb_clist { 2 }}{\clist_item:Nn \l_tmpa_clist { 1 }}
  \group_end:
}

\newcommand{\horstOpApplabelOf}[2]{
  \group_begin:
    \clist_set:Nn \l_tmpa_clist { #2 }
    \access{\clist_item:Nn \l_tmpa_clist { 1 }}{label}
  \group_end:
}

\newcommand{\horstOpApplabelsOf}[2]{
  \group_begin:
    \clist_set:Nn \l_tmpa_clist { #2 }
    \clist_item:Nn \l_tmpa_clist { 1 } . \textsf{label}
  \group_end:
}

\newcommand{\horstOpAppload}[2]{
  \group_begin:
    \horstOpApp{combine}{}{#2}
  \group_end:
}

\newcommand{\horstOpAppflowsTo}[2]{
  \group_begin:
    \clist_set:Nn \l_tmpa_clist { #2 }
    \clist_item:Nn \l_tmpa_clist { 1 } \sqsubseteq \clist_item:Nn \l_tmpa_clist { 2 }
  \group_end:
}

\newcommand{\horstOpApplub}[2]{
  \group_begin:
    \bigsqcup \mleft( #2 \mright)
  \group_end:
}

\newcommand{\horstOpAppglb}[2]{
  \group_begin:
    \clist_set:Nn \l_tmpa_clist { #2 }
    \clist_item:Nn \l_tmpa_clist { 1 } \sqcap \clist_item:Nn \l_tmpa_clist { 2 } 
  \group_end:
}

\newcommand{\horstOpAppset}[2]{
  \group_begin:
    \clist_set:Nn \l_tmpa_clist { #1 }
    \clist_set:Nn \l_tmpb_clist { #2 }
    \clist_item:Nn \l_tmpb_clist { 2 } \mleft[ \clist_item:Nn \l_tmpa_clist { 2 } \leftarrow \clist_item:Nn \l_tmpb_clist { 1 } \mright]
  \group_end:
}

\newcommand{\horstOpAppraiseTo}[2]{
  \group_begin:
    \clist_set:Nn \l_tmpa_clist { #2 }
    \horstLUBV{\clist_item:Nn \l_tmpa_clist { 1 }}{\clist_item:Nn \l_tmpa_clist { 2 }}
  \group_end:
}

\newcommand{\horstOpApplabelOfCtx}[2]{
  \group_begin:
    \clist_set:Nn \l_tmpa_clist { #2 }
    \access{\clist_item:Nn \l_tmpa_clist { 1 }}{label}
  \group_end:
}

\newcommand{\horstOpAppfromJustV}[2]{
  \group_begin:
    \clist_set:Nn \l_tmpa_clist { #2 }
    \clist_item:Nn \l_tmpa_clist { 1 }
  \group_end:
}

\newcommand{\horstOpAppisJustV}[2]{
  \group_begin:
    \clist_set:Nn \l_tmpa_clist { #2 }
    \horstINSET{\clist_item:Nn \l_tmpa_clist { 1 }}{ \horstTypeValue }
  \group_end:
}

\newcommand{\horstOpApplubIII}[2]{
  \group_begin:
    \clist_set:Nn \l_tmpa_clist { #2 }
    \clist_item:Nn \l_tmpa_clist { 1 } \sqcup \clist_item:Nn \l_tmpa_clist { 2 }  \sqcup \clist_item:Nn \l_tmpa_clist { 3 }
  \group_end:
}

\newcommand{\horstOpAppmkLConst}[2]{
  \group_begin:
    \clist_set:Nn \l_tmpa_clist { #1 }
    \clist_item:Nn \l_tmpa_clist { 1 }
  \group_end:
}

\newcommand{\horstOpAppmkConst}[2]{
  \group_begin:
    \clist_set:Nn \l_tmpa_clist { #1 }
    \clist_item:Nn \l_tmpa_clist { 1 }
  \group_end:
}

\newcommand{\horstOpAppiadd}[2]{
  \group_begin:
    \clist_set:Nn \l_tmpa_clist { #2 }
    \horstADD{\clist_item:Nn \l_tmpa_clist { 1 }}{\clist_item:Nn \l_tmpa_clist { 2 }}
  \group_end:
}

\newcommand{\horstOpAppimul}[2]{
  \group_begin:
    \clist_set:Nn \l_tmpa_clist { #2 }
    \horstMUL{\clist_item:Nn \l_tmpa_clist { 1 }}{\clist_item:Nn \l_tmpa_clist { 2 }}
  \group_end:
}

\newcommand{\horstOpAppisub}[2]{
  \group_begin:
    \clist_set:Nn \l_tmpa_clist { #2 }
    \horstSUB{\clist_item:Nn \l_tmpa_clist { 1 }}{\clist_item:Nn \l_tmpa_clist { 2 }}
  \group_end:
}

\newcommand{\horstOpAppiltu}[2]{
  \group_begin:
    \clist_set:Nn \l_tmpa_clist { #2 }
    \horstLT{\clist_item:Nn \l_tmpa_clist { 1 }}{\clist_item:Nn \l_tmpa_clist { 2 }}
  \group_end:
}

\newcommand{\horstOpAppishl}[2]{
  \group_begin:
    \clist_set:Nn \l_tmpa_clist { #2 }
    \horstSHL{\clist_item:Nn \l_tmpa_clist { 1 }}{\clist_item:Nn \l_tmpa_clist { 2 }}
  \group_end:
}

\newcommand{\horstOpAppilshr}[2]{
  \group_begin:
    \clist_set:Nn \l_tmpa_clist { #2 }
    \horstSHR{\clist_item:Nn \l_tmpa_clist { 1 }}{\clist_item:Nn \l_tmpa_clist { 2 }}
  \group_end:
}

\newcommand{\horstOpAppiand}[2]{
  \group_begin:
    \clist_set:Nn \l_tmpa_clist { #2 }
    \horstBWAND{\clist_item:Nn \l_tmpa_clist { 1 }}{\clist_item:Nn \l_tmpa_clist { 2 }}
  \group_end:
}

\newcommand{\horstOpApplabsneq}[2]{
  \group_begin:
    \clist_set:Nn \l_tmpa_clist { #2 }
    \horstLABSNEQ{\clist_item:Nn \l_tmpa_clist { 1 }}{\clist_item:Nn \l_tmpa_clist { 2 }}
  \group_end:
}

\newcommand{\horstOpAppabsneq}[2]{
  \group_begin:
    \clist_set:Nn \l_tmpa_clist { #2 }
    \horstNEQ{\clist_item:Nn \l_tmpa_clist { 1 }}{\clist_item:Nn \l_tmpa_clist { 2 }}
  \group_end:
}

\newcommand{\horstOpAppabseq}[2]{
  \group_begin:
    \clist_set:Nn \l_tmpa_clist { #2 }
    \horstEQ{\clist_item:Nn \l_tmpa_clist { 1 }}{\clist_item:Nn \l_tmpa_clist { 2 }}
  \group_end:
}

\newcommand{\horstOpApplabseq}[2]{
  \group_begin:
    \clist_set:Nn \l_tmpa_clist { #2 }
    \horstLABSEQ{\clist_item:Nn \l_tmpa_clist { 1 }}{\clist_item:Nn \l_tmpa_clist { 2 }}
  \group_end:
}

\newcommand{\horstOpAppreverse}[2]{
  \group_begin:
    \horstOpApp{reverse}{}{#2}
  \group_end:
}

\newcommand{\horstOpApprestrictLabel}[2]{
  \group_begin:
    \horstOpApp{restrictLabel}{}{#2}
  \group_end:
}

\newcommand{\horstOpApplowEq}[2]{
  \clist_set:Nn \l_tmpa_clist { #2 }

  \loweq{\clist_item:Nn \l_tmpa_clist { 1 }}{\clist_item:Nn \l_tmpa_clist { 2 }}
}

\newcommand{\horstOpApplowEqMem}[2]{
  \clist_set:Nn \l_tmpa_clist { #2 }

  \loweq{\clist_item:Nn \l_tmpa_clist { 1 }}{\clist_item:Nn \l_tmpa_clist { 2 }}
}

\newcommand{\horstOpAppperspectiveOfCtx}[2]{
  \clist_set:Nn \l_tmpa_clist { #2 }

  \access{\clist_item:Nn \l_tmpa_clist { 1 }}{attacker}
}

\newcommand{\horstConstructorAppLegal}{\L}
\newcommand{\horstConstructorAppIllegal}{\H}

\ExplSyntaxOff

\horstDeclarePredicateApplications
\horstDeclareSelectorFunctionApplications
\horstDeclareTypes
\horstDeclareOperationApplications

\ExplSyntaxOn
\newcommand{\horstFreeVarfrom}{\textit{from}}
\newcommand{\horstFreeVarresI}{\textit{res} \c_math_subscript_token {1}}
\newcommand{\horstFreeVarresII}{\textit{res} \c_math_subscript_token {2}}
\newcommand{\horstFreeVarctxI}{\textit{ctx} \c_math_subscript_token {1}}
\newcommand{\horstFreeVarctxII}{\textit{ctx} \c_math_subscript_token {2}}
\newcommand{\horstFreeVarctxIII}{\textit{ctx} \c_math_subscript_token {3}}
\newcommand{\horstFreeVarmemI}{\textit{mem} \c_math_subscript_token {1}}
\newcommand{\horstFreeVarmemII}{\textit{mem} \c_math_subscript_token {2}}
\newcommand{\horstFreeVarmemIII}{\textit{mem} \c_math_subscript_token {3}}
\newcommand{\horstFreeVarbrI}{\textit{br} \c_math_subscript_token {1}}
\newcommand{\horstFreeVarbrII}{\textit{br} \c_math_subscript_token {2}}
\newcommand{\horstFreeVarbrIII}{\textit{br} \c_math_subscript_token {3}}
\newcommand{\horstFreeVargt}{\textit{gt}}
\newcommand{\horstFreeVargtI}{\textit{gt} \c_math_subscript_token {1}}
\newcommand{\horstFreeVargtII}{\textit{gt} \c_math_subscript_token {2}}
\newcommand{\horstFreeVargtIII}{\textit{gt} \c_math_subscript_token {3}}
\newcommand{\horstFreeVarltI}{\textit{lt} \c_math_subscript_token {1}}
\newcommand{\horstFreeVarltII}{\textit{lt} \c_math_subscript_token {2}}
\newcommand{\horstFreeVarltIII}{\textit{lt} \c_math_subscript_token {3}}
\newcommand{\horstFreeVarstI}{\textit{st} \c_math_subscript_token {1}}
\newcommand{\horstFreeVarstII}{\textit{st} \c_math_subscript_token {2}}
\newcommand{\horstFreeVarstIII}{\textit{st} \c_math_subscript_token {3}}
\newcommand{\horstFreeVaratI}{\textit{at} \c_math_subscript_token {1}}
\newcommand{\horstFreeVaratII}{\textit{at} \c_math_subscript_token {2}}
\newcommand{\horstFreeVaratIII}{\textit{at} \c_math_subscript_token {3}}
\newcommand{\horstFreeVaratN}{\textit{at} \c_math_subscript_token {0}}
\newcommand{\horstFreeVargtN}{\textit{gt} \c_math_subscript_token {0}}
\newcommand{\horstFreeVarrtI}{\textit{rt} \c_math_subscript_token {1}}
\newcommand{\horstFreeVarrtII}{\textit{rt} \c_math_subscript_token {2}}
\newcommand{\horstFreeVarrtIII}{\textit{rt} \c_math_subscript_token {3}}
\newcommand{\horstFreeVarmemN}{\textit{mem} \c_math_subscript_token {0}}
\newcommand{\horstFreeVargtNI}{\textit{gt} \c_math_subscript_token {0 \c_math_subscript_token 1}}
\newcommand{\horstFreeVargtNII}{\textit{gt} \c_math_subscript_token {0 \c_math_subscript_token 2}}
\newcommand{\horstFreeVaratNI}{\textit{at} \c_math_subscript_token {0 \c_math_subscript_token 1}}
\newcommand{\horstFreeVaratNII}{\textit{at} \c_math_subscript_token {0 \c_math_subscript_token 2}}
\newcommand{\horstFreeVarmemNN}{\textit{mem} \c_math_subscript_token {0 \c_math_subscript_token 0}}
\newcommand{\horstFreeVarmemNI}{\textit{mem} \c_math_subscript_token {0 \c_math_subscript_token 1}}
\newcommand{\horstFreeVarmemNII}{\textit{mem} \c_math_subscript_token {0 \c_math_subscript_token 2}}
\newcommand{\horstFreeVarmemNIII}{\textit{mem} \c_math_subscript_token {0 \c_math_subscript_token 3}}
\newcommand{\horstFreeVartblI}{\textit{tbl} \c_math_subscript_token {1}}
\newcommand{\horstFreeVartblII}{\textit{tbl} \c_math_subscript_token {2}}
\newcommand{\horstFreeVartblIII}{\textit{tbl} \c_math_subscript_token {3}}

\newcommand{\horstFreeVarngt}{\textit{gt}'}
\newcommand{\horstFreeVarnlt}{\textit{lt}'}
\newcommand{\horstFreeVarnmem}{\textit{mem}'} 

\newcommand{\horstFreeVarrgtI}{\textit{rgt} \c_math_subscript_token {1}}
\newcommand{\horstFreeVarrgtII}{\textit{rgt} \c_math_subscript_token {2}}

\newcommand{\horstFreeVarreslabel}{\textit{reslabel}}
\newcommand{\horstFreeVarreslabelI}{\textit{reslabel} \c_math_subscript_token {1}}
\newcommand{\horstFreeVarreslabelII}{\textit{reslabel} \c_math_subscript_token {2}}
\ExplSyntaxOff

\newcommand*\circled[1]{\tikz[baseline=(char.base)]{
  \node[shape=rectangle,rounded corners=1.6mm, minimum width=3.2mm, minimum height=3.2mm, fill, inner ysep=0.5mm, inner xsep=0.8mm, 
  black] (char) {\normalsize{\textbf{\textsf{\color{white}{#1}}}}};}}

\ExplSyntaxOn
\int_new:N \g_clause_group_clause_counter_int
\ExplSyntaxOff


\ExplSyntaxOn
\int_gset:Nn\g_clause_group_clause_counter_int{1}
\newenvironment{clauseGroup}[1]{
    \newcommand{\clauseGroupLabel}{#1}
    \NewCommandCopy{\oldClauseBox}{\clauseBox}
    \renewcommand{\clauseBox}[2]{
      \begin{minipage}[t]{\textwidth}
        \begin{minipage}{5mm}
          \hfill
          \circled{\int_use:N \g_clause_group_clause_counter_int}
        \end{minipage}

      \iow_now:cx { @auxout }
      {
        \token_to_str:N \ExplSyntaxOn
        ^^J
        \token_to_str:N\newlabel{cls: ##1 : ##2}{{\token_to_str:N \circled{\int_use:N \g_clause_group_clause_counter_int}}{\thepage}{}{}{}}
        ^^J
        \token_to_str:N \ExplSyntaxOff
      }

      \hspace{0.1cm}
      \oldClauseBox{##1}{##2}
      \int_gincr:N\g_clause_group_clause_counter_int
      \end{minipage}
    }

  \begin{minipage}{.80\textwidth}
}
{
  \end{minipage}
  \vrule
  \begin{minipage}{.005\textwidth}
    \phantom{}
  \end{minipage}
  \begin{minipage}{.18\textwidth}
    \clauseGroupLabel
  \end{minipage}
}
\ExplSyntaxOff

\begin{horstRule}{brIfPseudoRule}
  \horstParVar{fid}{\horstTypeint}
  \horstParVar{pc}{\horstTypeint}
  \horstParVar{br}{\horstTypeint}
  \horstParVar{n}{\horstTypeint}
  \horstSelectorFunctionInvocation{\horstSelectorFunctionAppfunctionIds{\horstParVarfid}{},\horstSelectorFunctionApppcsForFunctionIdAndOpcode{\horstParVarpc}{\horstParVarfid,\horstConstBRIF},\horstSelectorFunctionAppbreakDestinations{\horstParVarbr}{\horstParVarfid,\horstParVarpc},\horstSelectorFunctionAppgetAmountOfReturnValuesInBlock{\horstParVarn}{\horstParVarfid,\horstParVarpc}}
  \begin{horstClause}
    \horstFreeVar{memN}{\horstTypeMemory}
    \horstFreeVar{atN}{\horstTypeHomInit{\horstTypeLValue}{\horstOpAppas{\horstParVarfid}{}}}
    \horstFreeVar{st}{\horstTypeHomInit{\horstTypeLValue}{\horstSUB{\horstOpAppss{\horstParVarfid,\horstParVarpc}{}}{1}}}
    \horstFreeVar{tbl}{\horstTypeTable}
    \horstFreeVar{ctx}{\horstTypeContext}
    \horstFreeVar{lt}{\horstTypeHomInit{\horstTypeLValue}{\horstOpAppls{\horstParVarfid}{}}}
    \horstFreeVar{mem}{\horstTypeMemory}
    \horstFreeVar{x}{\horstTypeLValue}
    \horstFreeVar{gtN}{\horstTypeHomInit{\horstTypeLValue}{\horstOpAppgs{}{}}}
    \horstFreeVar{gt}{\horstTypeHomInit{\horstTypeLValue}{\horstOpAppgs{}{}}}
    \horstPremise{\horstLT{\horstParVarbr}{\horstParVarpc}}
    \horstPremise{\horstPredAppMState{\horstParVarfid,\horstParVarpc}{\horstFreeVarctx,\horstCONS{\horstFreeVarx}{\horstFreeVarst},\horstFreeVargt,\horstFreeVarlt,\horstFreeVarmem,\horstFreeVartbl,\horstFreeVaratN,\horstFreeVargtN,\horstFreeVarmemN}}
    \horstPremise{\horstEQ{\horstOpApplabelOfCtx{}{\horstFreeVarctx}}{\horstConstructorAppLegal{}}}
    \horstPremise{\horstEQ{\horstOpApplabelOf{}{\horstFreeVarx}}{\horstConstructorAppIllegal{}}}
    \horstPremise{\horstOpAppabseq{}{\horstOpAppvalueOf{}{\horstFreeVarx},\horstOpAppmkConst{0}{}}}
    \horstConclusion{\horstPredAppMState{\horstParVarfid,\horstParVarpc}{\horstOpAppraiseCtxTo{\horstParVarpc}{\horstFreeVarctx,\horstOpApplabelOf{}{\horstFreeVarx}},\horstCONS{\horstFreeVarx}{\horstFreeVarst},\horstFreeVargt,\horstFreeVarlt,\horstFreeVarmem,\horstFreeVartbl,\horstFreeVaratN,\horstFreeVargtN,\horstFreeVarmemN}}
  \end{horstClause}
\end{horstRule}

\newcommand{\WA}{WebAssembly\xspace}
\newcommand{\wasm}{Wasm\xspace}

\newcommand{\etal}{et al.\xspace}
\newcommand{\etc}{etc.\xspace}
\newcommand{\ie}{i.e.\xspace}
\newcommand{\eg}{e.g.\xspace}
\newcommand{\challenge}[1]{\textbf{#1.}}

\newcommand{\mkfit}[1]{
\noindent
\resizebox{0.49\textwidth}{!}{
\begin{minipage}{\linewidth}
#1
\end{minipage}
}
}

\begin{abstract}
\WA (\wasm) is rapidly gaining popularity as a distribution format for software components embedded in various security-critical domains.
Unfortunately, despite its prudent design, \WA's primary use case as a compilation target for memory-unsafe languages leaves some possibilities for memory corruption.
Independently of that, \wasm is an inherently interesting target for information flow analysis due to its interfacing role.

Both the information flows between a Wasm module and its embedding context, as well as the memory integrity within a module, can be described by the hyperproperty noninterference. 
So far, no sound, fully static noninterference analysis for \wasm has been presented, but sound reachability analyses were. 
This work presents a novel and general approach to lift reachability analyses to noninterference by tracking taints on values and using value-sensitive, relational reasoning to remove them when appropriate.
We implement this approach in \tool{}, the first automatic, sound, and fully static noninterference analysis for WebAssembly, and demonstrate its performance and precision by verifying memory integrity and other noninterference properties with several synthetic and real-world benchmarks.
\end{abstract}

\maketitle
\renewcommand{\horstBoxWidth}{13.5cm}

\section{Introduction}

\WA (often shortened to \wasm)~\cite{DBLP:conf/pldi/HaasRSTHGWZB17,wasm-core-1.0,wasm-live-document} is a new compilation target rapidly gaining popularity.
Originally designed to speed up programs run in web browsers, its isolation guarantees and prudent design choices have helped to popularize it in a variety of fields,
such as plugin systems \cite{WasmPlugins1, WasmPlugins2}, smart contracts \cite{DBLP:conf/internetware/HuangJC20,WasmSmartContract1,WasmSmartContract2}, and edge computing \cite{DBLP:conf/eurosys/NiekeAK21, WasmEdge1}.

\WA's isolation and memory safety guarantees, however, only govern its interaction with the embedding context:
Within a \WA module, all functions of a module share the same mutable memory, which is dangerous if the \WA code was compiled from a memory-unsafe language. 
As Lehmann et al.\ have shown in previous research~\cite{DBLP:conf/uss/0002KP20}, this can be used to corrupt arbitrary (including ``constant'') data stored in the memory, which can break the layout and control flow assumptions of the high-level language compiled to \WA, leading to a variety of possible exploits.
A lack of mitigation strategies such as ASLR and read-only memory pages exacerbates this.

Consider, for example, the code in \autoref{fig:integrityc}, which might implement part of a game that is played on an 8-by-8 grid (chess, checkers, etc.)~\cite{wappler}.
Due to the insufficient restrictions on the values in \autoref{line:integrity-if} (\ccode{x} and \ccode{y} could be negative), the assignment in \autoref{line:integrity-write} could overwrite \ccode{trusted}, which shares the (mutable) memory with \ccode{game_state}. 
Later, in \autoref{line:integrity-eval}, \ccode{trusted} is passed to \ccode{eval}, a stand-in for any privileged function that expects to be called with trusted inputs.
In the context of a web application (where \ccode{eval} could, for example, modify the DOM), such a vulnerability could lead to an XSS attack.
Properly restricting \ccode{x} and \ccode{y} (e.g., by declaring them \ccode{unsigned}) would mitigate all unintended flows to \ccode{trusted}.

\begin{figure}
\begin{lstlisting}[language=C, xleftmargin=5.0ex, escapechar=|,]
static const char* trusted = "TRUSTED";
char game_state[64] = ...;
extern void eval(const char*);
void update_game_state(int x, int y, char c) {
  if (x < 8 && y < 8) { |\label{line:integrity-if}|
    game_state[y * 8 + x] = c; |\label{line:integrity-write}|
  }
  eval(trusted); |\label{line:integrity-eval}|
}
\end{lstlisting}
  \caption{Example of a vulnerable function (C).}
  \label{fig:integrityc}
\end{figure}

This kind of integrity problem can be described in terms of \emph{noninterference}~\cite{DBLP:conf/sp/GoguenM82a}, an important hyperproperty\cite{DBLP:conf/csfw/ClarksonS08} concerning information flows.
Noninterference, however, can not only verify the integrity of information but can also be used to verify its confidentiality, a crucial property for several of \wasm's application domains. 
As an example, consider the code in \autoref{fig:confidentialityc}, which illustrates a (secure) edge computing function.
The runtime populates a struct \ccode{message}, which may contain a (secret) 128-bit session ID and 384 bits of (public) payload, and calls \ccode{process}.
\ccode{process} can be executed on testing instances (\ccode{IS_PROD} is \ccode{false}) or in production (\ccode{IS_PROD} is \ccode{true}), where \ccode{IS_PROD} is considered trusted and public.
When \ccode{message} has session data attached, an imported, trusted function \ccode{authorize_and_execute} is executed. 
\ccode{authorize_and_execute} can access the entire memory (including \ccode{message}) and uses the secret data to check if the session is authorized to execute the operation described in \ccode{payload}.
To increase developer productivity, we allow for executing messages without session data on testing instances.
\ccode{untrusted_log} might be an untrusted third-party component that is not supposed to be handed any secret information and can also access the entire memory.
Therefore, we zero out all secret information, if present, before \ccode{untrusted_log} is called.
An attacker potentially monitoring the output of \ccode{untrusted_log} cannot learn the contents of \ccode{session} as they are zero in any case.

\begin{figure}
\begin{lstlisting}[language=C, xleftmargin=5.0ex]
struct {
  uint64_t session[2];    // secret-untrusted
  uint8_t payload[384];   // public-untrusted
} message;
void sanitize(void) {
  message.session[0] = 0;
  message.session[1] = 0;
}
uint64_t has_session(void) {
  return message.session[0] | message.session[1];
}
extern void authorize_and_execute(void); //mem: SU
extern void untrusted_log(void);         //mem: PU
extern uint32_t IS_PROD; 

void process(void) {
  if(has_session() || !IS_PROD) {
    authorize_and_execute();
    sanitize();
  }
  untrusted_log();
}

\end{lstlisting}
  \caption{A difficult-to-analyze, noninterferent function (C).}
  \label{fig:confidentialityc}
\end{figure}

\subsection{Related Work}
\label{sec:introduction-related-work}
In the following, we summarize the state-of-the-art in information flow analysis and its limitations, describing why they hinder the deployment of off-the-shelf techniques in the context of \wasm. 

Broadly speaking, sound noninterference analyses can be categorized into two families:
The first forbids writing to ``low'' (i.e., public or trusted) locations after branching on ``high'' (i.e., secret or untrusted) values (like with security type systems \cite{DBLP:journals/jsac/SabelfeldM03,DBLP:journals/mscs/BarthePR13,DBLP:conf/aplas/KobayashiS02} and dependency analyses \cite{DBLP:journals/ijisec/HammerS09,DBLP:journals/tosem/SneltingRK06}), or forbids branching on ``high'' values in general (like with synchronous product programs \cite{DBLP:conf/fm/ZaksP08, DBLP:conf/fm/BartheCK11}).
This effectively means that safe programs like the one in \autoref{fig:confidentialityc} have to be rejected, as a write (\ccode{sanitize()}) is executed depending on a secret (\ccode{has_session()}).
The other family, including (sequential) self-composition \cite{DBLP:conf/csfw/BartheDR04,DBLP:conf/spc/DarvasHS05,DBLP:conf/sas/TerauchiA05}, allows arbitrary control flows but is already prohibitively expensive when applied to high-level languages \cite{DBLP:conf/fm/BartheCK11,DBLP:conf/sas/TerauchiA05}, which limits its applicability for low-level targets such as \WA.
Additionally, these approaches crucially rely on precise reasoning, thus disallowing the introduction of abstractions in the name of runtime performance--a strategy often necessary in the static analysis of low-level code.
Taint tracking\footnote{The term ``taint tracking'' is not used homogeneously in the literature. In the formulation of Schoepe et al.~\cite{DBLP:conf/eurosp/SchoepeBPS16}, taints only attach to values, not to the ``control flow'', which means that only explicit flows can be tracked. Austin et al.~\cite{DBLP:conf/popl/AustinF12} demonstrate that some implicit flows are ignored even when tainting the control flow (and all values updated while the control flow is tainted). The approach of Bond et al.~\cite{DBLP:conf/uss/BondHKLLPRST17} is sound, as it uses taints to forbid branches on secret values. This is appropriate for their use case of preventing \emph{timing side-channels} and places it firmly in the first family. Bernardeschi's and De Francesco's \cite{DBLP:conf/vmcai/BernardeschiF02} approach is not precisely ``taint tracking'', but conceptually similar. It is sound, but value-insensitive and thus unable to verify the code in \autoref{fig:confidentialityc}.
 } is an approach that has been successfully applied in low-level contexts~\cite{DBLP:conf/uss/BondHKLLPRST17} but is unfortunately unsound in the presence of implicit flows \cite{DBLP:conf/eurosp/SchoepeBPS16, DBLP:conf/popl/AustinF12}.

For \WA in particular, information flow has, to our knowledge, been treated in two publications:
Wassail~\cite{DBLP:conf/scam/StievenartR20}, which provides summaries that describe the information flow within and between functions in the style of taint tracking, 
and SecWasm~\cite{DBLP:conf/sas/BastysASS22}, which introduces a hybrid verification technique combining a type system with a dynamic component to handle memory labels.
These works fall short of tackling a few fundamental challenges in Wasm noninterference analysis, which prevents them from correctly analyzing the two example programs, as we discuss below.

\challenge{Reasoning on the memory level}
\WA supports named and typed variables, but only for four primitive types (integral and floating point values in 32- and 64-bit);
all other values (such as \ccode{game_state}) have to be stored in the untyped linear memory.
This lack of structure in the linear memory complicates static analyses of information flow, as the memory is big (the smallest possible size is 64 KiB), can grow during execution, and allows unaligned accesses for various, differently sized types.
If the exact addresses of accessed memory cells are not modeled, one risks either missing genuine or reporting spurious flows.
In fact,  Wassail ignores some flows to the memory for the sake of performance and reducing false positives, thereby losing soundness, whereas SecWasm proposes a dynamic component that handles memory labels, thereby introducing runtime overhead and losing the ability to analyze programs statically.
Consider, for example, \autoref{fig:integrityc}.
Given that \ccode{x}, \ccode{y}, and \ccode{c} are untrusted, an attacker can always modify part of the memory.
To uncover the real problem (namely that \ccode{trusted} can be modified when only \ccode{game_state} is meant to, due to the insufficient restrictions on \ccode{x} and \ccode{y}) and to prove the fixed version noninterferent, it is necessary to differentiate between different (linear) memory regions and precisely model the values of \ccode{x} and \ccode{y}.

\challenge{Reasoning about values} 
Tracking values precisely not only enables us to reason over memory on a more granular level but is, at times, also necessary to verify noninterference, as shown in \autoref{fig:confidentialityc}.
The control flow in \ccode{process} branches on secret information (the return value of \ccode{has_session}).
Therefore, any assignments in the body of the \ccode{if} (including \ccode{sanitize}) could potentially lead to an implicit flow, which type-system-related approaches such as SecWasm would flag.
Only by tracking the values in \ccode{message}, which neither SecWasm nor Wassail support, can we ensure that, indeed, all secrets are set to zero, and an attacker who can observe the memory when \ccode{untrusted_log} is called can learn no secrets. 

\challenge{Balancing performance, soundness, and precision}
Besides precision and soundness, a critical challenge for a noninterference analysis in Wasm is performance.
Specifically, noninterference requires relating different executions of a program, which is notoriously expensive by itself. Doing that while ensuring value sensitivity in presence of an unstructured memory adds insult to\penalty100{}~injury.

Given the limitations of existing approaches, we pose the following research question:
\emph{
    Can we design an efficient yet sound and precise static analysis technique for information flow, which can, in particular, handle the specific challenges of \WA?
    }

\subsection{Our contributions}
In this work, we answer affirmatively by introducing a fundamentally novel and generic static analysis technique for noninterference, which we specifically instantiate for \WA.

Our approach can be summarized as follows: we extend an existing flow-sensitive, value-sensitive reachability analysis in the style of an abstract interpretation (in our case, \wappler~\cite{wappler}) with taint-tracking.
More specifically, we taint the outputs of computations with the combined taints of their inputs and the control flow if it is conditioned on a tainted value. 
This part of the analysis tells us which values can appear during the execution at different locations and whether they are tainted. 
The innovation of our approach lies in the fact that at certain points during the execution (e.g., after exiting an $\mycode{if}\cdot\mycode{then}\cdot\mycode{else}$-statement), we examine taints and values in all different related runs and adjust them.
This allows us to achieve soundness in the case of missing taints and increase precision by removing spurious taints if all related runs agree on a tainted value, such as in \autoref{fig:confidentialityc}. 
Additionally, the labeling system allows differentiating between ostensible flows introduced by imprecisions in the reachability analysis (which often are a desirable tradeoff to enhance runtime performance) and actual flows.

While we instantiate our analysis technique in Wasm, given its growing popularity and lack of effective verification techniques, the approach presented here is generic and applicable to arbitrary programming languages and, thus, of independent interest.
To summarize, the contributions of this work are as follows:

\begin{itemize}
  \item we devise a novel approach to lift value-sensitive, flow-sensitive reachability analyses to noninterference analysis; 
  \item we implement this approach in \tool, the first automatic, sound, fully static noninterference analysis for \wasm, and make it available to the community~\cite{wanilla-artifacts};
  \item we successfully establish \tool's performance and precision on two noninterference datasets from literature; in particular, we compare favorably in terms of termination with a general hyperproperty checker on its noninterference benchmark \cite{DBLP:conf/fmcad/BartheEGGKM19} and improve upon two WebAssembly analyzers \cite{DBLP:conf/internetware/HuangJC20, DBLP:conf/scam/StievenartR20} in terms of both precision and soundness when analyzing a realistic benchmark of 82 smart contracts for the use of bad randomness (an integrity property);
  \item we additionally introduce a noninterference benchmark comprising \wanillaBenchmarkModuleCount{} \wasm modules and \wanillaBenchmarkTestCaseCount{} test cases and make it available to the community~\cite{wanilla-artifacts}
\end{itemize}

\section{Preliminaries}

\subsection{\WA}

While a complete discussion of \wasm's design would go beyond the scope of this article\footnote{Interested readers are referred to the initial paper~\cite{DBLP:conf/pldi/HaasRSTHGWZB17}, the specification of \wasm 1.0~\cite{wasm-core-1.0} (the subject of this article) and the developing version of the specification~\cite{wasm-live-document}.}, in the following, we give a broad overview of \wasm's essential features, further elaborating on them when discussing the relevant parts of the analysis.

The basic model of \WA is a \emph{stack machine}.
The stack is required to be \emph{well-typed}, which means that the size and content type of the stack are known at every point during the execution.
Therefore, certain classes of runtime errors (stack over- and underflows, most type errors) can be statically ruled out.
\WA 1.0 supports four different types: integers and floating point numbers in 32- and 64-bit variants.
The instruction $\Add{i32}$, for example, takes two $\watype{i32}$ values from the value stack and pushes their sum (again, an $\watype{i32}$ value) to the stack.
The program $(\Const{f32}{3.0})\allowbreak(\Const{i32}{5})~\allowbreak\Add{i32}$ would not pass the type checker, as the two topmost stack elements when \Add{i32} is executed are not both of type $\watype{i32}$ (\Const{t}{x} pushes a value $x$ of type \watype{t} to the stack).

The type system requires the use of \emph{structured control flow primitives}, which is a distinctive feature among low-level languages (we will illustrate how some of these primitives work in \autoref{sec:prelim-example}).
Another feature not found in other low-level languages are \emph{functions}, which are modeled as \emph{activations} on a call stack.
The functions in \WA are typed and can carry a list of typed and numbered variables, the so-called \emph{locals}.
The \emph{globals}, which are also typed and numbered, are available in all functions.
Since functions take their arguments from and leave their results on the stack, they are subject to the type checker.
The \Call{x} instruction     puts an activation for the statically known (and type-checked) function $x$ on the call stack.
The \CallIndirect{} instruction takes an index value from the stack and uses it to look up a function to be called in the \emph{table}.
The table is a designated memory region that holds opaque pointers to functions;
if the looked-up function does not have the right type, the execution traps (aborts exceptionally).
Such traps cannot be caught in \wasm and also appear in other contexts such as out-of-bounds memory accesses or undefined arithmetic.

All values that do not fit into the four primitive types have to be stored in the \emph{linear memory}, a mutable byte array that can grow (but not shrink) during the execution.
As \WA is meant to be embedded into different contexts, its designers took care to isolate \WA's context from the embedding context.
All components (e.g., functions, memories, or globals) that can interact with the embedding context have to be explicitly imported or exported.
For our analysis, these last points lead to two insights:
First, \wasm's memory-safety guarantees must be interpreted in relation to the embedding context: \wasm modules may not corrupt the memory of the embedder or other \wasm modules but may, however, corrupt their own memory.
Second, we have to take special care when handling imported functions (which we do in \autoref{sec:lift_soundness}).

\subsection{Noninterference}
\label{sec:prelim_noninterference}
Noninterference~\cite{DBLP:conf/sp/GoguenM82a} is an important hyperproperty~\cite{DBLP:conf/csfw/ClarksonS08} that expresses the inability of an attacker to tamper with a system.
It comprises \emph{Confidentiality}---the inability of an attacker to learn secrets by observing the public outputs of a system---%
and \emph{Integrity}---the inability of an attacker to alter a system's trusted state by modifying its untrusted inputs.
The different information handled by a system and the abilities of an attacker w.r.t.\ that information can be expressed by a lattice \Lat{} (see \autoref{fig:nilattice} for the lattice that is used in the following examples).
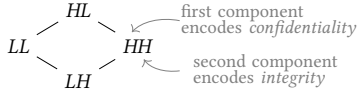
\begin{figure}[t]
  \center
  \resizebox{5.455cm}{!}{
  \begin{tikzpicture}[yscale=0.54, xscale=0.9]
    \node (SU) at (0,1)  {\SU};
    \node (PU) at (-1,0) {\PU};
    \node (ST) at (1,0)  {\ST};
    \node (PT) at (0,-1) {\PT};

    \node (FA) at (1.6,0.7)  [anchor=west,text width=32mm,align=left,font=\small,color=gray] {first component\\[-1mm] encodes \emph{confidentiality}};
    \node (SA) at (1.8,-0.7) [anchor=west,text width=30mm,align=left,font=\small,color=gray] {second component\\[-1mm] encodes \emph{integrity}};

    \draw ($(FA.west) + (0.1,0.0)$) edge [->,bend right=40,color=gray] ($(ST) + (-0.10,0.38)$);
    \draw ($(SA.west) - (0.1,0.0)$) edge [->,bend left=40,color=gray] ($(ST) + (0.10,-0.38)$);

    \draw (SU) -- (PU);
    \draw (SU) -- (ST);
    \draw (PU) -- (PT);
    \draw (ST) -- (PT);
 \end{tikzpicture}
}
  \caption{Example lattice \Lat{} with $H$ = ``high'' and $L$=``low''.}
  \label{fig:nilattice}
\end{figure}
Specifically, data of level $\ell$ can flow to locations of level $\ell$ or higher, whereas an attacker at $\ell \in \Lat$ is allowed to modify/observe inputs/outputs that are at or below $\ell$ in \Lat.
Formally, we write $\ell_1 \flowsto \ell_2$ to express the fact that $\ell_1$ flows to $\ell_2$.
$\ell_1 \sqcup \ell_2$ denotes the least upper bound of $\ell_1$ and $\ell_2$ ($\ell_1 \sqcup \ell_2 = \ell_2 \implies \ell_1 \flowsto \ell_2$), while $\ell_1 \sqcap \ell_2$ denotes their greatest lower bound ($\ell_1 \sqcap \ell_2 = \ell_1 \implies \ell_1 \flowsto \ell_2$).
The assignment of memory positions to labels is done by so-called security policies $\Gamma$, which are formally introduced in \autoref{sec:security_policies}.
Given an equivalence relation $\elleq{\cdot}{\cdot}$, which relates all possible inputs/ outputs if they agree on the positions labeled $\ell$ or lower, we call $S$ noninterferent if for all $\ell$-equal input pairs $i_1, i_2$ $S$ produces an $\ell$-equal output (i.e., $\forall i_1\, i_2\,.\, \elleq{i_1}{i_2} \implies \elleq{S(i_1)}{S(i_2)}$).

\subsection{Example}
\label{sec:prelim-example}
As an example, consider the function \ccode{process} from \autoref{fig:confidentialityc}.
A possible translation%
\footnote{We chose the presented translation to illustrate the ideas of our analysis.
Real-world compilers will translate the function differently (e.g., by inlining calls to simple functions), but not in a way that impacts our analysis negatively.}
of \ccode{process} to \wasm is shown in \autoref{fig:symbolic-runs}, together with a symbolic representation of all possible runs.
In this example, we are concerned with flows to \ccode{untrusted_log}: in particular, we want to prove that \ccode{untrusted_log} does not leak any secrets (high confidentiality) while not caring about untrusted (low integrity) data.
As described in the introduction, we assume that \ccode{untrusted_log} can access the whole linear memory but nothing else.
To verify confidentiality, we thus have to show that when \ccode{untrusted_log} is executed, all runs that agree on the initially public data (here: \ccode{payload} and \ccode{IS_PROD}) will agree on their linear memory (corresponding to the high-level variables \ccode{session} and \ccode{payload}).

\begin{figure}[t]
  \includegraphics[width=8.5cm]{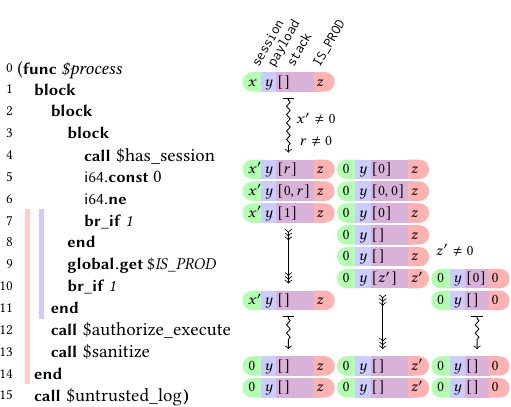}
  \caption{Translation of \ccode{process} from \autoref{fig:confidentialityc} to \wasm{}.}
  \label{fig:symbolic-runs}
\end{figure}

\autoref{fig:symbolic-runs} shows simplified memory configurations as they evolve during the program execution.
As establishing confidentiality (and noninterference in general) requires reasoning about universally quantified inputs, we do not show concrete values (unless a value is uniquely determined) but instead variables like $x$, $y$, and $z$.
Unless noted otherwise, these variables range over all possible values (e.g., $x$ could be any $128$-bit bit vector, $z$ any $32$-bit integer, etc.).

Until \Call{\textdollar has\_session} (whose execution is straightforward and thus here omitted for brevity) returns, all possible runs look the same.%
\footnote{To keep the figure clean, we use the squiggly arrow to symbolize progress that does not change the program state.}
Afterward, we have to distinguish two cases:
If \ccode{session} has a non-zero value, \Call{\textdollar has\_session} will put a non-zero value on the stack.
Otherwise, the session and the topmost stack value are both $0$.
In the figure, we illustrate these two scenarios by splitting the execution and introducing the restricted variables $x'$ and $r$.

In either case, the execution pushes $0$ to the stack (\Const{i64}{0}) and checks if the two topmost stack elements are unequal (\AnyCmd{i64}{ne}).
As we know that $r \neq 0$ in column 1, we know that \AnyCmd{i64}{ne} will leave $1$ on the stack, while in column 2, the result will be $0$.

\BrIf{1} executes a conditional branch.
If the topmost stack entry is equal to $0$, we continue with the next line in the program; otherwise, we continue at the \End{} of the first surrounding \Block{} (i.e., in \li{11}\footnote{In the following \li{x} will denote line $x$ in \autoref{fig:symbolic-runs}}; we use the barbed arrow to signify a branch like this in the diagram).%
\footnote{``First'' here is taking into account a ``zeroth'' block, which is the block immediately containing \BrIf{1}.
To illustrate the \WA semantics: In this position, \BrIf{0} would have continued the execution in \li{8}, \BrIf{2} in \li{14}.
If the surrounding structure would begin with $\Loop{}$ instead of $\Block{}$, branching instructions would instead take the execution to the \emph{start} of the block.
}
Conditional control flow is essential in noninterference analyses, as state modifications that depend on, for example, secret data may lead to the revelation of secrets.
It is, therefore, crucial to determine the program segments that are only conditionally executed.
We will discuss the details of how we calculate these segments in \autoref{sec:scope_extensions_and_implicit_flows}; for now, note that the region marked red is a sound approximation of the conditionally executed segments.

The execution in column 2 leaves the block and pushes \ccode{IS_PROD}'s contents to the stack.
Again, we distinguish two cases w.r.t.\ equality to $0$.
If $z$ is unequal to $0$, we jump directly to \li{14} (column 2); otherwise (column 3), we continue at \li{11} together with column 1.
Columns 1 and 3 execute \Call{\textdollar authorize\_execute} and \Call{\textdollar sanitize}. The latter sets the contents of \ccode{session} to 0.
In \li{15}, where \ccode{untrusted_log} is called, we have two scenarios w.r.t. to the initial public values:
either the \ccode{IS_PROD} was $0$ or not, where we do not have to restrict \ccode{payload} in any way.
Columns 1 and 3 show the first case, and columns 1 and 2 the show second (column 1 has no restriction on \ccode{IS_PROD} and therefore appears twice).
In any case, the memory contents are the same, which shows that the analyzed program does preserve confidentiality.

\section{From Reachability to Noninterference}
\label{sec:background}
\renewcommand{\example}[1]{#1}

In this section, we will describe our approach to generically lift reachability analyses to noninterference.
In \autoref{sec:computational_model}, we describe a simple model of computation, which is meant  to capture  real-world programming languages.
\autoref{sec:reachbility_analyses} likewise introduces a generic  interface to reason about flow-sensitive reachability analyses.
In \autoref{sec:security_policies}, we introduce the security policies and taint labels that are at the core of our noninterference analysis, which we introduce in \autoref{sec:constructing_nia} and whose soundness we discuss in \autoref{sec:lift_soundness}. 
\example{Throughout, we will illustrate the technical concepts by means of examples from  \autoref{fig:symbolic-runs}}.

\subsection{Semantics and our Computational Model}
\label{sec:computational_model}

Our semantic model relies  on two sets, $\MMPos$ and $\MPPos$, to describe memory and program positions, respectively.
\example{
\wasm has several memory regions, such as the linear memory, the globals, and the function pointer table, which are global to the module, as well as memory regions that are addressed relative to the current activation, such as the locals and the stack.
Since we only consider one activation in our examples, we will omit the activation's identifier for clarity's sake.
In the following, \mpst{i} will denote the $i$-th stack position (counting from the bottom), \mpgl{i} will denote the $i$-th global, and \mplm{i} will denote the $i$-th byte in the linear memory.
Program positions in \wasm can be modeled as sequences of activations and program counters, that is, call stacks that track where to continue the execution upon returning.
As we are only describing one activation in our examples, we will again simplify this and identify the program position with the line number of the example program.
}
Program configurations are then tuples of program positions and memories ($\MMem$), where memories are functions from memory positions to values\footnote{\example{In \wasm, values are integers or floating point numbers in either 32- or 64-bit variants.}} ($\MVal$).
These tuples represent the program's state at a particular position.
The only restriction we introduce on these types is that we can decide equality for $\MVal$, $\MMPos$, and $\MPPos$.

The semantics of the analyzed program is given by $\Mstep$.
For a configuration $c$, $\Mstep[c]$ returns the next configuration if it exists.
$\Mstep[c][*]$ denotes $\Mstep[c]$'s reflexive
and transitive closure, i.e., all configurations that are reachable from $c$. 
$\Mstep[c][n]$ denotes the configuration reached by applying $\Mstep$ $n$ times, starting with $c$.\footnote{For readability, we write $\Mstep[pp, m]$ to mean $\Mstep[(pp,m)]$ when discussing a configuration's constituent parts and adopt the same convention for similar functions.}

We assume the existence of special program positions $\{\terminatepc, \nojoinpc\} \subseteq \MPPos$.
$\terminatepc$ denotes that the execution has terminated, i.e.\ $\nexists pp'\, m\, m'\,.\,\allowbreak (pp', m') = \Mstep[\terminatepc, m]$ and $\forall pp\,m \,.\, pp \neq \terminatepc \implies \exists pp'\,m'\,.\,(pp',m') \allowbreak = \Mstep[pp, m]$.
$\nojoinpc$, conversely, denotes a program position that is never reached, i.e., $\nexists pp\, m\, m'\,.\,(\nojoinpc, m') = \Mstep[pp, m]$.

\subsection{Reachability Analyses}
\label{sec:reachbility_analyses}

For this semantics, we assume the existence of a reachability analysis $\Mreachability$, with the following signature:

\begin{myalign}
\Mreachability:& \MPPos \times \MMem \mapsto \powerset{\MPPos \times \MMem}
\end{myalign}
For a start configuration (program position and memory), $\Mreachability$ returns a set of possibly reachable configurations.%
\footnote{
\label{footnote:interface}
Please note that this set has to be understood as an ``interface'' to capture arbitrary reachability analyses:
\shortened{
If an analysis $\Mfunc{ra'}: \MPPos \times \MMem \mapsto (\MPPos \mapsto (\MMPos \mapsto \powerset{\MVal}))$ outputs a set of values for every memory position
(in the context of abstract interpretation, this is sometimes called a ``non-relational'' abstraction), we can, for example, convert $\Mfunc{ra'}$ to the correct signature as follows:
\begin{align*}  
\Mfunc{ra''}[pp, m] & = \{ (pp', m') | pp' \in \MPPos, m' \in G_{\MMPos}(\Mfunc{ra'}[pp, m](pp')) \} \\
G_{\{p\}}(f) & = \{ \lambda x . v ~|~ v \in f(p) \} \\
G_{S}(f) & = \{ \lambda x . \begin{cases} v & \text{if } x = p \\ f'(x) & \text{otherwise}\end{cases} ~|~ v \in f(p), p \in S, \\
& \hspace{3.7cm} f' \in G_{S \backslash \{p\}}(f) \} \\
\end{align*} 
}{}
analyses usually do not return arbitrary $x \in \powerset{\MPPos \times \MMem}$ but efficiently representable objects (e.g., intervals for values instead of arbitrary sets.).
For generality, we do not put any such restriction on our underlying analysis.
}
$\Mreachability$ is characterized by an abstract step function $\Mupdate$ that, given a configuration, returns a set of possible updates consisting of
a follow-up program position,
a set of memory positions that might have influenced this program position\footnotenumber[fn1],
plus a set of triples of updated values, each consisting of the updated memory position, the new value, and a set of memory positions that influenced the calculation of this value\footnotenumber[fn2].
\multifootnote[fn1,fn2]{We will refer to these sets in the following as input sets.}
\begin{myalign}
	&\Mupdate: \MPPos \times \MMem \mapsto \\  &\quad\powerset{\MPPos \times \powerset{\MMPos} \times \powerset{\MMPos \times \MVal \times \powerset{\MMPos}}}
\end{myalign}
\begin{examplethm}[Abstract Steps]
\label{example:abstract_steps}
\example{
Consider the following possible return values for $\Mupdate$ in \autoref{fig:symbolic-runs}. 

$\Mupdate[5, m]= \{(6, \emptyset, \{(\mpst1, 0, \emptyset)\})\}$: \Const{i64}{0} in \li{5} continues the execution at \li{6} and pushes $0$ to the top of the stack (which is \mpst1 in \li{6}).
No value influences the execution or the pushed value; therefore, the corresponding input sets are empty. 

$\Mupdate[6, m]= \{(7, \emptyset, \{(\mpst0, 0, \{\mpst0, \mpst1 \})\},             (\mpst1, \bot, \emptyset)\})$ (if $m(\mpst0) \neq m(\mpst1)$, like in column 1) or 
$\Mupdate[6, m]= \{(7, \emptyset, \{(\mpst0, 1, \{\mpst0, \mpst1 \})\}, \allowbreak (\mpst1, \bot, \emptyset)\})$ (if $m(\mpst0) = m(\mpst1)$, like in column 2):
\AnyCmd{i64}{ne} in \li{6} continues in any case in at \li{7}, writes, depending on \mpst0 and \mpst1, either 0 or 1 to the top of the stack (\mpst0 in \li{7}).
The former top of the stack \mpst1 will be deleted, regardless of any inputs.  

$\Mupdate[7, m]= \{(11, \{\mpst0\}, \{(\mpst0, \bot, \emptyset)\})\}$ (if $m(\mpst0) \neq 0$, like in column 1),
$\Mupdate[7, m]= \{( 8, \{\mpst0\}, \{(\mpst0, \bot, \emptyset)\})\}$ (if $m(\mpst0) = 0$, like in column 2): 
\BrIf{1} in \li{7} continues, depending on \mpst0, either in \li{8} or \li{11} and deletes \mpst0 regardless.
}
\end{examplethm}

\begin{defthm}[$\Mreachability$ in terms of $\Mupdate$]
\label{def:ra_in_terms_of_astep}
$\Mreachability(pp, m)$ can be defined in terms of $\Mupdate$ as the least fixed point of
\begin{myalign}
	f(X) = &\{(pp,m)\} \cup \{(pp'', \Mapplyupdate[m', U]) ~|~ (pp',m') \in X, \\
	& \phantom{\{(pp,m)\} \cup {}\{} (pp'', \_, U) \in \Mupdate[pp', m'] \}
\end{myalign}
where $\Mapplyupdate$ is defined as
\begin{myalign}
\Mapplyupdate[m,U] = \lambda p . \begin{cases}
 v    & \text{ if } (p,v,\_) \in U \\
 m(p) & \text{ otherwise}
 \end{cases}
\end{myalign}
\end{defthm}

$\Mupdate$ allows for imprecision in the analysis as it returns a set of possible ways to update the state;
as long as it returns one tuple that correctly captures the semantics of $\Mstep$ (as formalized below), it is sound.
\begin{figure}
\mkfit{
\begin{align*}
	(1)~&\forall pp\,pp'\,m\,m'. \\
	    &\quad \Big((pp', m') = \Mstep[pp,m] \implies \\
	(2)~&\quad\quad \exists I_{pp}\,U . (pp', I_{pp}, U) \in \Mupdate[pp,m] \land {}\\
	(3)~&\quad\quad\quad \big(\forall U'\,\ppi\,\mi . (\ppi, \mi) = \Mstep[pp, \Mapplyupdate[m, U']] \land {} \\
		&\quad\quad\quad\quad pp' \neq \ppi' \implies \exists p' . (p',\_,\_) \in U' \land p' \in I_{pp}\big) \land{}\\
	(4)~&\quad\quad\quad \big(m' = \Mapplyupdate[m,U] \big) \land {}\\
	(5)~&\quad\quad\quad \big(\forall p\,v\,v'\,I\,I' . (p, v, I) \in U \land (p,v',I') \in U\\
	    &\quad\quad\quad\quad\implies v = v' \land I = I' \big) \land {}\\
	(6)~&\quad\quad\quad \big(\forall p\,I_U\,v\,U'\,\ppi\,\mi . (p,v,I_U) \in U \land {} \\
	    &\quad\quad\quad\quad (\ppi, \mi) = \Mstep[pp, \Mapplyupdate[m, U']] \land m'(p) \neq \mi(p) \\
	    &\quad\quad\quad\quad\quad {} \implies \exists p' . (p',\_,\_) \in U' \land p' \in I_U\big) \land{}\\
	(7)~&\quad\quad\quad \big(\forall p . \big( \exists m''\,\ppi\,\mi . (\ppi, \mi) = \Mstep[pp, m''] \land{}\\
	    &\quad\quad\quad\quad \mi(p) \neq m''(p)\big) \implies \exists v'\,I_{pp} . (p,v',I_{pp}) \in U\big)\Big)
\end{align*}
}

\caption{$\Mupdate$'s soundness condition.}
\label{fig:soundness_astep}
\end{figure}
\begin{defthm}[Soundness of $\Mupdate$]
\label{def:soundness_astep}
$\Mupdate$ is sound w.r.t.\ $\Mstep$ if the condition in \autoref{fig:soundness_astep} (which we discuss immediately below) holds.
\end{defthm}
The soundness condition requires that for every concrete step $\Mstep$ produces $(1)$, $\Mupdate$ produces at least one abstract step $(2)$ that correctly captures it. 
What ``correctly'' capturing means is subject to conditions $(2)$--$(7)$.  
Firstly, to capture at least the same states as $\Mstep$, the abstract step has to contain the correct follow-up program position $pp'$ $(2)$ and an update set $U$ that creates the correct memory $m'$ from $m$ $(4)$.
To correctly track information flows, we have to be able to determine which memory positions influenced the follow-up program position and each value in the update set.
This is ensured by conditions $(3)$ and $(6)$.
If $\Mstep$ returns a different follow-up program position when applying any arbitrary update $U'$ to $m$ (while $pp$ remains unchanged), we require that $I_{pp}$ contains at least one of the memory positions in $U'$.
Since this holds for all possible $U'$, it holds in particular for the minimal $U'$, meaning that if changing a single memory position could lead to a different follow-up program position, then this memory position has to be included in $I_{pp}$.
Condition $(6)$ expresses the same notion for the updated values in $U$. 
As we will see later, it is in the interest of precision to assign $I_{pp}$ and $I_U$ to the minimal set that fulfills (3), respectively, (6).
As $\Mapplyupdate$ is not defined for ambiguous update sets (i.e., ones that update a particular memory position more than once), condition $(5)$ forbids them.
Lastly, condition $(7)$ tells us that $U$ has to contain an update for every memory position $p$ that \emph{might} change at $pp$.
This means that even if $\Mstep$ does not change the value at a specific position, $\Mupdate$ has to ``update'' it to the value it already holds.
As we will see soon, there are also flows by data \emph{not} being written, which is what $(7)$ addresses by requiring the enumeration of all possibly updated memory positions.

\begin{examplethm}[Input Sets and Updated Values]
To illustrate the importance of input sets (conditions $(3)$ and $(6)$) and tracking potentially updated values (condition $(7)$), we anticipate the basic principle of our noninterference analysis and give an example.
The analysis works by tracking which values are ``tainted'' and propagating this taint  during the program execution.
The taint of a value updated in a step is determined by the taint labels of all values in the input set.\footnote{In case of the $I_{pp}$, the execution itself will be tainted, as we will see. Our example, however, will only concern values.}
A particularly interesting case for the computation of these taints are program constructs that value-dependently update the memory, such as array accesses or, in \wasm, accesses to the linear memory.
Consider the \AnyCmd{i32}{store8} instruction, which takes two values $i$ (e.g., stored in \mpst0) and $c$ (e.g., stored in \mpst1) from the stack and writes the least significant byte of $c$ to the memory cell at $i$. 
In the case of confidentiality analysis, an attacker does not necessarily have to learn the value of \mplm{i} to learn something about $i$, as all other memory cells \emph{not} being written can already reveal it.
Therefore, our analysis has to taint all potentially updated values in such a case.
When we are not considering the surrounding code \AnyCmd{i32}{store8} appears in, this means we have to taint the whole linear memory.
This is encoded by condition $(7)$.
The minimal input sets for each updated memory position, however, depend on the exact value at \mpst0.
$\Mupdate$ has to return $(\mplm{i}, \{\mpst0, \mpst1\}, c \mathbin{\text{mod}} 256)$, as two minimal updates can change the value at \mplm{i}:
Either the value at \mpst0 is changed (which means that the resultant value at \mplm{i} changes by \emph{not} being updated), or the value at \mpst1 is changed (thereby changing the value directly). 
For all other memory positions \mplm{j} (with $j \neq i$), the input set can be $\{\mpst0\}$, as \mpst0 \emph{has} to change in order to change the resulting value at \mplm{j}.
As this very general formalization would mean that writing to a secret offset would taint the whole linear memory, we run a simple preanalysis to determine if the offset is constant.
If it is, this drastically shrinks the update set, as only one memory cell per written byte has to be updated.
Furthermore, it removes \mpst0 from the respective input sets, as its contents are fixed.
Therefore, the attacker can learn nothing from it, and we can ignore its taint.
If we cannot determine if the index is constant, we apply the general rule.
\end{examplethm}

\begin{defthm}[Soundness for reachability]
\label{def:soundness_reachability}
We call $\Mreachability$ sound w.r.t.\ to $\Mstep$ if it fulfills the following formula:
\begin{myalign}
&\forall pp\, pp'\, m\, m' . (pp', m') \in \Mstep[pp,m][*] \implies (pp', m') \in \Mreachability[pp, m]
\end{myalign}
\end{defthm}
\begin{theorem}[Soundness for reachability ($\Mupdate$)]
$\Mreachability$ in \autoref{def:ra_in_terms_of_astep} is sound w.r.t.\ \autoref{def:soundness_reachability} if $\Mupdate$ fulfills the conditions in \autoref{def:soundness_astep}. 
\end{theorem}
\begin{proof}
By induction on the number of applications of $\Mstep$ and applying \checkme{$2$}, \checkme{$4$}, and \checkme{$5$} from \autoref{def:soundness_astep} to \autoref{def:ra_in_terms_of_astep}.
\end{proof}

\subsection{Security Policies and Taint Labels}
\label{sec:security_policies}

We formalize the concept of security policies introduced in \autoref{sec:prelim_noninterference} as functions from memory positions to elements of the security lattice: $\MMPos \mapsto \Lat$.
Equality w.r.t.\ an attacker $\ell$ and security policy $\Gamma$, $\elleqg$,  is defined as: 
\begin{myalign}
&\elleqg[m_1][m_2] \iff \forall p . \Gamma(p) \sqsubseteq \ell \implies m_1(p) = m_2(p)
\end{myalign}

\begin{examplethm}[Security Policies]
\label{example:security_policies}
\example{
The initial security policy for the code in \autoref{fig:symbolic-runs} is
\begin{myalign}
\Gamma_0(p) &=
\begin{cases}
\PT & \text{ if }p = \mpgl0 \\
\SU & \text{ if }p \in \{\mplm{i}~|~0   \le i < 128\}\\
\PU & \text{ otherwise }
\end{cases}
\end{myalign}

\ccode{IS_PROD} (assumed to be stored in \mpgl0) is a public ($L$) and trusted ($H$) value; therefore, it is assigned $\PT$.
\ccode{session}, which we assume to be stored in the first 128 bytes of the linear memory, is secret ($H$) and untrusted ($L$), while all other values are assumed to be public ($L$) and untrusted ($L$) (mainly \ccode{payload}, although the compiled \WA module will contain other memory locations).
At \li{15}, where \ccode{untrusted_log} is called, we require the linear memory not to contain any secrets, which corresponds to the security policy
\begin{myalign}
\Gamma_{15}(p) &=
\begin{cases}
\PU & \text{ if }p \in \{\mplm{i}~|~i \in \mathbb{N}\}\\
\SU & \text{ otherwise }
\end{cases}
\end{myalign}
}
\end{examplethm}

Our noninterference analysis $\Mnoninterference$, which is introduced in \autoref{sec:constructing_nia}, does not directly reason about security policies but about \emph{taint labels}, which it assigns to every memory position.
The set of taint labels $\horstTypeFlowLabel$ contains $\L$ and $\H$.
The meaning of $\L$ is that the labeled value cannot possibly be influenced by a value whose initial label (per a security policy) is higher than the currently analyzed attacker in $\Lat$, while this is not true for a value labeled $\H$. 
$(\horstTypeFlowLabel,\sqcup,\flowsto)$ forms a lattice. We have $\L \flowsto \H$, $x \sqcup \H = \H$, and $x \sqcap \L = \L$.
As $\sqcup$ and $\sqcap$ are associative and commutative, we sometimes write $\bigsqcup_{x \in X}$ and $\bigsqcap_{x \in X}$ to denote the least upper/greatest lower bound of all elements in $X$.
To translate from security policies to taint labels, we introduce the type of label maps $\MLMap: \MMPos \mapsto \horstTypeFlowLabel$ that assigns taint labels to memory positions.
Given an attacker $\ell$ and a security policy $\Gamma$, we can define the corresponding label map $\Lmap$ as follows: 
\begin{myalign}
\Lmap[p] &= \begin{cases}
\L & \text{if }\Gamma(p) \sqsubseteq \ell \\
\H & \text{otherwise}
\end{cases}
\end{myalign}

\begin{examplethm}[Label Maps]
\label{example:label_maps}
To analyze the program in \autoref{fig:symbolic-runs} for confidentiality (but not integrity), we have to analyze it from the perspective of the $\PU$ attacker, who can manipulate untrusted (low integrity) inputs but does not know secret inputs when the execution starts.
Thus, $\Lambda^{\Gamma_0}_\PU$ has to taint exactly all memory positions that correspond to \ccode{session}, which happens when we combine the definitions of $\Lmap$, $\Lat$, and $\Gamma_0$. 

\begin{myalign}
\Lambda^{\Gamma_0}_\PU & \begin{cases}
\H & \text{ if }p \in \{\mplm{i}~|~0   \le i < 128\}\\
\L & \text{ otherwise}
\end{cases}
\end{myalign}
\end{examplethm}

\subsection{Constructing $\Mnoninterference$}
\label{sec:constructing_nia}

To motivate the design of our noninterference analysis, $\Mnoninterference$, let us first reiterate the design objectives introduced in \autoref{sec:introduction-related-work}:
Our analysis should be sound for noninterference, which comprises reasoning about explicit and implicit flows. 
At the same time, it should be expressive enough to allow control flow that is dependent on ``high'' ($\H$-labeled) values, respectively allow writing to low positions while the control flow is tainted (like in the red area following \li{7} in \autoref{fig:symbolic-runs}).
Additionally, it should be value-sensitive to be able to verify programs like in \autoref{fig:symbolic-runs}, but not solely rely on value-based reasoning, as pure value-based analysis of noninterference (i.e., deriving a flow from any two differing values reachable at an attacker-observable location) disallows any kind of performance-enhancing imprecision. 
Consider the following:
A technique to efficiently analyze the reachable configurations is abstract interpretation \cite{CousotPatrick2021Poai}, which fundamentally computes an \emph{abstraction} for every program configuration, i.e., a set of values that this configuration can possibly hold at different memory locations.
Naively, this approach is, however, not suitable for noninterference analysis, as determining the (in)equality of two memory locations of two abstractions is an imprecise business: they can only definitely be judged equal if their abstractions contain exactly and only the same value.
When analyzing \autoref{fig:symbolic-runs}, a precise enough abstract interpretation could determine that \ccode{session} is $0$ in any execution (indeed, our underlying analysis does, although maintaining this level of precision using statement-wise abstraction (especially in low-level code) can be difficult \cite{DBLP:conf/vmcai/JiangCWW17}).
The problem here lies in the contents of \ccode{payload}: a naive abstract interpretation will (correctly) derive that \ccode{payload} can hold any value in \li{15}.
Using only this information, we cannot verify that these (attacker-observable) values did not change from their initial values and are thus the same in all executions that started with the same public values.

$\Mnoninterference$ thus combines ideas from both taint-tracking systems and abstract interpretation.
In the following, we will go through the program in \autoref{fig:symbolic-runs} and explain the rules presented in \autoref{fig:niarules}, which lift a reachability analysis $\Mreachability$ (characterized by $\Mupdate$) to a noninterference analysis. 

\subsubsection{Initialization}
\phantom{}
Like a regular taint-tracking system, $\Mnoninterference$ has to apply an initial taint according to some security policy.
This is done by \rniainit, which applies the initial label map $\Lambda_{\PU}^{\Gamma_0}$ as already discussed in Examples \ref{example:security_policies} and \ref{example:label_maps}.  
As can be seen, $\Mnoninterference$ takes an initial labeling, an initial program position, and two initial memories as arguments.%
\footnote{In the spirit of \autoref{footnote:interface}, $\Mnoninterference$ is not meant to be run once for every possible pair $m_1, m_2$ but to take and return efficiently representable abstractions as input.
The rules as presented here are to be understood as the most general interface for an abstract interpretation.
To verify, for example, the program in \autoref{fig:symbolic-runs} in the most efficient way, one would run $\Mnoninterference$ with three abstract memories $\hat{m}_1$, $\hat{m}_2$, and $\hat{m}_3$ that capture all memories with $m_1(x) \neq 0$, $m_2(x) = 0 \land m_2(z) \neq 0$, and $m_3(x) = m_3(z) = 0$, respectively.
For these abstract memories, we would (abstractly) have $\elleqg[\hat{m}_1][\hat{m}_2][\Gamma_0][\PU]$ and $\elleqg[\hat{m}_1][\hat{m}_3][\Gamma_0][\PU]$ but not $\elleqg[\hat{m}_2][\hat{m}_3][\Gamma_0][\PU]$, since they necessarily disagree on $z$ (which is stored at $\mpgl{0}$), where $\Gamma_0(\mpgl{0}) = \PT$ and $\PT \flowsto \PU$. For $\elleqg[\hat{m}_1][\hat{m}_2]$ (or $\ellneqg[\hat{m}_1][\hat{m}_2]$) to hold, $\elleqg[m_1][m_2]$ has to hold for any two (concrete) memories $m_1$, $m_2$ that are abstracted by $\hat{m}_1$ and $\hat{m}_2$, respectively.
}
It returns a set of tuples that contain a program position, a memory, and the corresponding label map in the last three components and a context object in the first one.
We will explain the context object in detail later on;
here, we only mention that $\Mlctx$ expresses that the current program position does not depend on any information labeled $\H$, which is the case, as the execution has not started yet.

\subsubsection{Explicit Flows}
From \li{1} to \li{7}, no conditional control flow happens, and the analysis only has to propagate the taint labels.
This is achieved by \rniaprop, in particular, by the calculation of $\Lambda'$.
Ignoring $\ctxlabel_L$ for now, we see that every updated memory position is assigned the least upper bound of the labels of its input set.
We can thus determine that \ccode{has_session}'s return value in \li{5} will be labeled $\H$, as it is directly computed from a secret.
Referring back to \autoref{example:abstract_steps}, we see that $0$ (pushed in \li{5}) will be labeled $\L$ (as it has an empty input set), and the result of \AnyCmd{i64}{ne} will be labeled $\H$ as one of its inputs (at $\mpst{0}$) is labeled $\H$.
The memory is updated exactly the same way as in $\Mreachability$ by applying $U$.
The conditions in the third line prevent \rniaprop from being executed if an other rule applies.

\subsubsection{Implicit Flows}
All flows so far are explicit flows; to handle implicit flows (which some analyses~\cite{DBLP:conf/scam/StievenartR20} forgo for increased performance and lower false positives), we have to reason about the control flow.
Control flow  appears in $\Mnoninterference$ in three constructs which we overview shortly here and formally describe in \autoref{sec:lift_soundness}.
$\Mccf[pp, pp', m_1, m_2]$ returns true for all program positions $pp'$ that hold a conditional control flow instruction such as $\If$.
Given a conditional control flow instruction at $pp'$, $\Mjoinat[pp, pp', m_1, m_2]$ is used to determine when the program's control flow does not depend on said instruction anymore.
$\Mjoinat$ does this by returning a pair of program positions $(pp_{j_1}, pp_{j_2})$ such that all executions starting at $pp'$ necessarily encounter $pp_{j_1}$ and $pp_{j_2}$ consecutively at some point.
$(pp_{j_1}, pp_{j_2})$ can be thought of as an edge in the control flow graph of $\Mstep$.
As all of our rules describe the transition from one state into the next one, just identifying a state that all executions encounter would not suffice (as from one state, different, diverging transitions would be possible).
Since all executions starting at $pp'$ execute this so-called \emph{joining transition}, this means that at this point, the execution no longer depends on the control flow decision taken at $pp'$.
Lastly, we introduce the context object:
it  is either $\lctx$ (which, as mentioned above, expresses that the control flow does not depend on tainted information) or a tuple $(\H, pp, (pp_{j_1}, pp_{j_2}))$.
The latter variant expresses that the control flow may depend on tainted information.
It started depending on tainted information when executing the instruction at $pp$ and will cease to do so after encountering the joining transition $(pp_{j_1}, pp_{j_2})$.
Raising the context label is done by \rniarcst (and its variant \rniarcns), as it is applied in \li{7}.
If the control flow is not yet tainted (i.e., the context object equals $\Mlctx$), the executed instruction is a conditional control flow instruction (expressed by $\Mccf$), and the next program counter $pp''$ depends on tainted data, then we have to taint the context object, i.e., set its first coordinate to $\H$.
Recall that $I_{pp}$ contains all memory positions that could influence $pp''$'s value, thus the last condition can be assessed by calculating the least upper bound of $I_{pp}$ in $\Lambda$.
Additionally, the updated context object $\ctxlabel'$ will hold the current program position $pp'$ in the second coordinate and joining transition $(pp_{j_1}, pp_{j_2})$ in the third.
A valid joining transition for \li{7} in \autoref{fig:symbolic-runs} is $(14,15)$ (the transition leaving the red area). 
\rniarcns is a special case to handle divergent control flows where one of the executions does not actually progress before rejoining.
This is expressed by $\Mjoinat$ returning $(pp', pp'')$, i.e., the current and next program position.
The analysis of \autoref{fig:symbolic-runs} does not contain an application of this rule, but a characteristic example of this behavior would be \mycode{do}-\mycode{while}-loops, where one execution can leave the loop without iterating and the other one cannot (depending on an $\H$-labeled condition).

Revisiting \rniaprop, we can now explain the role of the context's taint $\ctxlabel_L$ in the computation of $\Lambda'$: every updated memory position's label is raised to at least $\ctxlabel_L$.
This means that $\Mnoninterference$ taints all values it updates while the context is tainted to track implicit flows.
Examples of this behavior appear in \li{9} and during the execution of \ccode{sanitize}, where public and constant values are tainted. 

\begin{figure}[t]
\begin{lstlisting}[language=C++, xleftmargin=5.0ex]
bool f(bool h) {
	bool x = true, l = true;
	if (h) x = false;
	if (x) l = false; |@\label{line:conditional-update}@|
	return l;
}
\end{lstlisting}
  \caption{Adapted from~\cite{DBLP:conf/popl/AustinF12}: The tainted input \ccode{h} is returned via \ccode{l}, but without being careful, \ccode{l} will not be tainted.}
  \label{fig:taintsoundnesproblem}
\end{figure}

\subsubsection{Soundness and Precision Concerns}
With the rules presented so far, $\Mnoninterference$ would be a pure taint tracking system, which would reject our example program:
In \li{15}, the values corresponding to \ccode{session} would be tainted in any case: either because the originally tainted values are not overwritten (column 2) or because the values are overwritten in a tainted context (column 1 and 3).
Additionally, naive tainting systems (such as $\Mnoninterference$ is so far) have a soundness problem as illustrated in \autoref{fig:taintsoundnesproblem}:
Implicit flows are only tracked if a value is updated.
If now a portion of the program is only executed if a certain memory location is \emph{not} updated (such as line \ref{line:conditional-update} in \autoref{fig:taintsoundnesproblem}), the control flow in this portion will not be tainted, which enables arbitrary implicit flows.
To address these two points, $\Mnoninterference$ contains \rniajoin, which is executed when going from \li{14} to \li{15}.
\rniajoin relates all so far derived results that have the same $\H$-labeled context object.
Then, \rniajoin computes $\Lambda'$ as follows:
For every memory position $p$ that is not updated, we check if either of the two intermediate label maps $\Lambda_1$ and $\Lambda_2$ considers it $\H$.
If this is the case \emph{and} the two memories disagree on $p$, $p$ is labeled $\H$.
This also raises the taints of unequal values that were only updated in one execution.
If the value at $p$ is updated, the least upper bound of either of the input sets is $\H$, and the updated value differs, $p$ is labeled $\H$.
In \wasm, this does not happen as the only update potentially joining instructions execute is deleting values from the stack.\footnote{If we treat all deleted values as a special symbol $\bot$, deleted values cannot differ, and therefore not cause any flows.
The type system of \wasm ensures that the shape of the stack is known at every program position.
Thus, it is impossible that a value is deleted from the stack in one branch and not deleted in the other.
}
If neither of these two options is the case, the label is set to $\L$.
In particular, that means two things:
Firstly, two differing $\L$-labeled values do not change their label since their inequality is the result of imprecisions in the reachability analysis and not indicative of an actual flow. 
Secondly, two equal values will be labeled $\L$ in any case; if $\Mreachability$ derives only one value for a specific memory position, this means that this memory position will certainly be $\L$ in $\Lambda'$.
If $\Mreachability$ is precise enough, this happens when going from \li{14} to \li{15}: the contents of \ccode{session} are labeled $\H$, but at the same time, all possible runs agree on the value ($0$).
Therefore, the label is lowered.
It is easy to see that a pure (value-insensitive) labeling system could not remove the taint here.

\begin{figure*}
\resizebox{0.97\linewidth}{!}
{
\small
\begin{mathpar}

\infer[Init]{
}{
\{(\Mlctx, pp, \Lambda_0, m_1), (\Mlctx, pp, \Lambda_0, m_2) \} \subseteq \Mnoninterference[\Lambda_0, pp, m_1, m_2] 
}

\infer[Propagate-Taint]{
(\ctxlabel, pp', \Lambda, m') \in \Mnoninterference[\Lambda_0, pp, m_1, m_2] \\
\ctxlabel = (\ctxlabel_L, \_, (pp_{j_1},pp_{j_2})) \\
(pp'', I_{pp}, U) \in \Mupdate[pp', m'] \\
\hbox{$\Lambda' = \lambda p . \begin{cases}
\textstyle\bigsqcup_{i \in I_U} \Lambda(i) \sqcup \ctxlabel_L & \text{ if } (p,\_,I_U) \in U \\
\Lambda(p)                                      & \text{ otherwise}
\end{cases}$} \\
\lnot \Mccf[pp, pp', m_1, m_2] \lor \big(\Mccf[pp, pp',m_1,m_2] \land (\ctxlabel_L = \L \land \textstyle\bigsqcup_{i \in I_{pp}} \Lambda(i) = \L) \lor \ctxlabel_L = \H \big)\\
pp_{j_1} \neq pp' \lor  pp_{j_2} \neq pp'' \\
}{
\{(\ctxlabel, pp'', \Lambda', \Mapplyupdate[m',U])\} \subseteq \Mnoninterference[\Lambda_0, pp, m_1, m_2] 
}

\infer[Raise-Context-Step]{
\hbox{$
(\Mlctx, pp', \Lambda, m') \in \Mnoninterference[\Lambda_0, pp, m_1, m_2] \hspace{0.8cm}
(pp'', I_{pp}, U) \in \Mupdate[pp', m']
$}\\
(pp_{j_1}, pp_{j_2}) = \Mjoinat[pp, pp', m_1, m_2] \\
\hbox{$\Lambda' = \lambda p . \begin{cases}
\textstyle\bigsqcup_{i \in I_U} \Lambda(i)                & \text{ if } (p,\_,I_U) \in U \\
\Lambda(p)                                      & \text{ otherwise}
\end{cases}$} \\
\ctxlabel' = (\H, pp', (pp_{j_1}, pp_{j_2}))\\ 
pp_{j_1} \neq pp' \lor  pp_{j_2} \neq pp'' \\
\H = \textstyle\bigsqcup_{i \in I_{pp}} \Lambda(i)\\
\Mccf[pp, pp', m_1, m_2]}{
\{(\ctxlabel', pp'', \Lambda', \Mapplyupdate[m', U])\} \subseteq \Mnoninterference[\Lambda_0, pp, m_1, m_2] \\
}

\infer[Raise-Context-No-Step]{
(\Mlctx, pp', \Lambda, m') \in \Mnoninterference[\Lambda_0, pp, m_1, m_2] \\
(pp'', I_{pp}, U) \in \Mupdate[pp', m'] \\
(pp', pp'') = \Mjoinat[pp, pp', m_1, m_2]\\
\ctxlabel' = (\H, pp', (pp', pp''))\\ 
\H = \textstyle\bigsqcup_{i \in I_{pp}} \Lambda(i)\\
\Mccf[pp, pp', m_1, m_2]}{
\{(\ctxlabel', pp', \Lambda, m' )\} \subseteq \Mnoninterference[\Lambda_0, pp, m_1, m_2] \\
} 

\infer[Join]{
\hbox{$
((\H, pp_d, (pp', pp'')), pp', \Lambda_1, m_1') \in \Mnoninterference[\Lambda_0, pp, m_1, m_2] \hspace{0.5cm}
(pp'', I_{pp_1}, U_1) \in \Mupdate[pp_1', m_1'] 
$}\\
\hbox{\rule[-.5\baselineskip]{0pt}{1.4\baselineskip}$
((\H, pp_d, (pp', pp'')), pp', \Lambda_2, m_2') \in \Mnoninterference[\Lambda_0, pp, m_1, m_2] \hspace{0.5cm}
(pp'', I_{pp_2}, U_2) \in \Mupdate[pp_2', m_2']
$}\\
\hbox{$\Lambda' = \lambda p . \begin{cases}
\H & \text{ if } (\Lambda_1(p) \sqcup \Lambda_2(p) = \H) \land m_1'(p) \neq m_2'(p) \land (p,\_,\_) \not \in U_1 \land (p,\_,\_) \not \in U_2 \\
\H & \text{ if } (\textstyle\bigsqcup_{i \in I_{U_1}} \Lambda_1(i) \sqcup \textstyle\bigsqcup_{i \in I_{U_2}} \Lambda_2(i) = \H) \land v_1 \neq v_2 \land (p,v_1,I_{U_1}) \in U_1 \land (p,v_2,I_{U_2}) \in U_2\\
\L & \text{ otherwise}
\end{cases}$}
}{
\{(\Mlctx, pp'', \Lambda', \Mapplyupdate[m_1', U_1]), (\Mlctx, pp'', \Lambda', \Mapplyupdate[m_2', U_2])\} \subseteq \Mnoninterference[\Lambda_0, pp, m_1, m_2] 
}
\end{mathpar}
}
\caption{Rules to lift $\Mreachability$ to $\Mnoninterference$.}
\label{fig:niarules}
\end{figure*}

\subsection{Soundness}
\label{sec:lift_soundness}

To formalize the intuitions given above and prove $\Mnoninterference$'s soundness, we define $\Mnoninterference$'s signature as:
\begin{myalign}
	\Mnoninterference &: \MLMap \times \MPPos \times \MMem \times \MMem \mapsto \powerset{\MCtx \times \MPPos \times \MLMap \times \MMem }
\end{myalign}
Given an initial label map, an initial program position, and two initial memories, $\Mnoninterference$ gives us a set of possible labelings for a memory at a certain program counter and a certain context. 
Contexts are tuples $\MCtx: \horstTypeFlowLabel \times \MPPos \times (\MPPos\times\MPPos)$.
We write $\Mlctx$ for $(\L, \bot, \bot)$.

As seen above, $\Mnoninterference$ crucially relies on the ability to reason about the possible control flows of $\Mstep$ to handle implicit flows.
This reasoning is done by two functions, $\Mjoinat$ and $\Mccf$. 
$\Mccf$ encodes the intuition of conditional control flow instructions (instructions that have more than one possible follow-up program position),
while $\Mjoinat$ encodes the intuition of entering a part of an execution shared among all possible executions after possibly diverging before.

Formally, $\Mjoinat$ is a function that, given an initial program position, a queried program position, and two initial memories, returns a pair of program positions that describes a transition that post-dominates the queried program position, i.e., a transition that is surely encountered in all executions starting at $pp'$\footnote{This follows an idea first published by Denning and Denning~\cite{DBLP:journals/cacm/DenningD77}.} (from configurations that are either reachable from $(pp, m_1)$ or $(pp,m_2)$).
If such a transition cannot be calculated (e.g., because the executions can diverge permanently at the queried program position), the second component of the transition is $\nojoinpc$.
\begin{myalign}
	\Mjoinat:& \MPPos \times \MPPos \times \MMem \times \MMem \mapsto \MPPos \times \MPPos
\end{myalign}

\begin{defthm}[Soundness of $\Mjoinat$]
\label{def:joinat-soundness}
$\Mjoinat$ is sound w.r.t.\ $\Mstep$ if it fulfills the following formula:
\begin{myalign}
	&\forall pp\,pp'\,pp_{j_1}\,pp_{j_2}\,m\,m'\,m_1\,m_2\,m_1'\,m_2'\,.\, \\
	&\quad (pp_{j_1},pp_{j_2}) = \Mjoinat[pp, pp', m_1, m_2] \land pp_{j_2} \neq \nojoinpc \land {} \\
	&\quad (m = m_1 \lor m = m_2) \land (pp', m') \in \Mstep[pp, m][*] \implies \exists k\,m_{j_1}\,m_{j_2}. \\
	&\quad\quad          (pp_{j_1}, m_{j_1}) = \Mstep[pp',m'][k] \land (pp_{j_2}, m_{j_2}) = \Mstep[pp',m'][k+1]
\end{myalign}
\end{defthm}
Formally, $\Mccf$ is a function that takes the same arguments as $\Mjoinat$ but returns whether the queried program position can have multiple follow-up program positions.
\begin{myalign}
	\Mccf:& \MPPos \times \MPPos \times \MMem \times \MMem \mapsto \{\textit{true},\textit{false}\}
\end{myalign}

\begin{defthm}[Soundness of $\Mccf$]
\label{def:ccf-soundness}
$\Mccf$ is sound w.r.t.\ $\Mstep$ if it fulfills the following formula:
\begin{myalign}
    &\forall pp\,pp'\,m_1\,m_2 \,.\, |\{\ppi ~|~ (pp',m') \in \Mstep[pp,m][*] \land {}\\ 
	&\quad (m = m_1 \lor m = m_2) \land (\ppi, \_) = \Mstep[pp',m'] \}|  > 1 \\
	&\quad \implies \Mccf[pp, pp', m_1, m_2]
\end{myalign}
\end{defthm}

Assuming that $\Mupdate$, $\Mjoinat$, and $\Mccf$ fulfill their soundness conditions, we can define various soundness notions that $\Mnoninterference$ fulfills.

\begin{theorem}[Soundness w.r.t.\ reachability]
For any two initial memories, $\Mnoninterference$ and $\Mreachability$ compute the same state.
\begin{myalign}
&\forall pp\,m_1\,m_2\,\Lambda .  (\Mreachability[pp, m_1] \cup \Mreachability[pp, m_2]) = \\
&\quad\quad    \{(pp',m') ~|~ (\_,pp',\_,m') \in \Mnoninterference[\Lambda, pp, m_1, m_2]\} 
\end{myalign}
\end{theorem}

\begin{proof}
\rniainit returns the initial memories as required, all other rules (except \textsc{Raise-Context-No-Step}, which is always directly followed by \rniajoin) $\Mapplyupdate$ the updates required by $\Mupdate$ directly. 
\end{proof}

$\Mnoninterference$'s soundness w.r.t.\ noninterference is captured by \autoref{thm:nia_sound_memories} and \autoref{thm:nia_sound_traces}.
\ifextended
The full proofs are given in ~\autoref{sec:proofs} for space reasons. 
\fi

\begin{theorem}[Soundness of $\Mnoninterference$ (w.r.t. memories)]
\label{thm:nia_sound_memories}
Given an attacker $\ell$ and a security policy $\Gamma$:
For any two executions starting at the same program position $pp$ from two memories $m_1$ and $m_2$ that are equal w.r.t.\ $\ell$ and $\Gamma$ and terminating with the updated memories $m_1'$ and $m_2'$, it holds that
for every memory position $p$ that $m_1'$ and $m_2'$ disagree on, $\Mnoninterference$ can derive a labeling $\Lambda$ at $\terminatepc$ for either $m_1'$ or $m_2'$ that labels $p$ $\H$ from $\ell$'s and $\Gamma$'s label map, $pp$, $m_1$, and $m_2$.
\begin{myalign}	
&\forall \ell\,\Gamma\,pp\,m_1\,m_2\,m_1'\,m_2'\,p~.~\elleqg[m_1][m_2] \land m_1'(p) \ne m_2'(p)\land{}\\
&\quad(\terminatepc,m_1') \in \Mstep[pp, m_1][*] \land (\terminatepc,m_2') \in \Mstep[pp, m_2][*]{} \\
&\quad\quad \implies \exists \ctxlabel\, m'\, \Lambda\,.\,  (\ctxlabel, \terminatepc, m', \Lambda) \in \Mnoninterference[\Lmap, pp, m_1, m_2] \land {} \\
&\quad\quad\quad\Lambda(p) = \H \land (m_1' = m' \lor m_2' = m')
\end{myalign}
\end{theorem}

\begin{proof}
\ifextended
We prove the stronger \autoref{thm:nia_general_memory_soundness} in \autoref{sec:proofs}.
\fi
In short, we observe that the pairwise execution can either proceed in lock-step or not.
We show that if it does not proceed in lock-step, we can derive $\H$ as context label, and that if we can only derive $\L$ as context label, we have the same correct labeling for both executions.
\end{proof}

In some contexts, it is useful not only to define noninterference in terms of indistinguishability of memories but also to model that an attacker can observe certain events.
\example{In the case of WebAssembly, this might, for example, be imported functions with observable side effects, such as \ccode{untrusted_log}.}

To capture this notion, we first define the (possibly infinite) trace of an execution as
\begin{myalign}
&\Tr{pp, m} = (pp,m) \cdot \Tr{pp', m'}\\
&\quad\quad \text{ where }(pp',m') = \Mstep[pp,m]
\end{myalign}

To assess which events are observable by an attacker, we introduce \emph{security policies for events}.
$\Gamma_{E}: \MPPos \mapsto \Lat$ assigns to every program position an element from the security lattice.
If, for an attacker $\ell$, we have $\Gamma_E(pp) \flowsto \ell$, an attacker can observe the fact that an event has happened at the given program position.
\example{In our example, we have $\PU \flowsto \Gamma_E(15)$, as the attacker at $\PU$ is able to observe the call to \ccode{untrusted_log}.}
Additionally, we assume a function that tells us for every memory position if an attacker can observe it as part of the event, i.e., a function that assigns security policies to program positions:
$\Gamma_{\textit{ME}}: \MPPos \mapsto (\MMPos \mapsto \Lat)$.
\example{In our example, we have $\Gamma_{\textit{ME}}(15) = \Gamma_{15}$, as at the call site of \ccode{untrusted_log}, we require the linear memory to only contain public data, which is exactly what $\Gamma_{15}$ encodes.
}

With these definitions, we can define trace equality w.r.t.\ an attacker $\ell$, $\traceelleqg$ (and its negation $\traceellneqg$), as follows:
\begin{myalign}
  \traceelleqg[t_1][t_2] =
  \begin{cases}
    \textit{true}            & \text{if }          t_1 = \epsilon \land t_2 = \epsilon \\
    \traceelleqg[t_1'][t_2]  & \text{elif }          t_1 = (pp_1, m_1) \cdot t_1' \land \Gamma_E(pp_1) \notflowsto \ell \\
    \traceelleqg[t_1][t_2']  & \text{elif }          t_2 = (pp_2, m_2) \cdot t_2' \land \Gamma_E(pp_2) \notflowsto \ell \\
    \traceelleqg[t_1'][t_2'] & \text{elif }          t_1 = (pp, m_1) \cdot t_1' \land t_2 = (pp, m_2) \cdot t_2' \land {} \\
                       &\phantom{\text{elif }} \elleqg[m_1][m_2][\Gamma_\textit{ME}(pp)]\\
	\textit{false} & \text{otherwise}
  \end{cases}
\end{myalign}

$\Mnoninterference$'s soundness w.r.t.\ establishing the equality of event traces is captured by the following theorem.

\begin{theorem}[Soundness of $\Mnoninterference$ (w.r.t.\ event traces)]
\label{thm:nia_sound_traces}
Given an attacker $\ell$:
For any two traces originating from the same program position $pp$ from two memories $m_1$ and $m_2$ that are equal w.r.t.\ $\ell$ and $\Gamma$, it holds that if the
traces are not equal w.r.t.\ $\ell$, there exists a configuration $(pp', m')$ for which $\Mnoninterference$ can derive a context $\ctxlabel$ and labeling $\Lambda$, such that $\ell$ can observe the event at $(pp',m')$ and either $\ctxlabel$ or a memory position that does not flow to $\ell$ is labeled $\H$ in $\Lambda$. 
\begin{myalign}
&\forall \ell\,pp\,m_1\,m_2. \traceellneqg[\Tr{pp,m_1}][\Tr{pp,m_2}] \land\elleqg[m_1][m_2] \implies {}\\
&\quad\exists \ctxlabel\, pp'\, m'\, \Lambda\,.\,   (\ctxlabel, pp', m', \Lambda) \in \Mnoninterference[\Lmap, pp, m_1, m_2] \land \Gamma_E(pp',m') \flowsto \ell \land {} \\
&\quad\quad \big((\ctxlabel = (\H,\_,\_))  \lor \exists p . \Lambda(p) = \H \land \Gamma_\textit{ME}(pp,p) \notflowsto \ell \big)
\end{myalign}

\end{theorem}

\begin{proof}
\ifextended
See \autoref{sec:proofs}.
\fi
\ifextended
In short, if an event appears in only one trace, it has to be executed in a context labeled $\H$; if it appears in both traces, but they disagree on an attack-observable memory position, this position has to be labeled $\H$ by \autoref{thm:nia_general_memory_soundness}.
\fi
\end{proof}

\section{Analysis}
\label{sec:analysis}

\begin{figure*}
\small
\newcommand{\opc}{\overline{\textit{pc}}}
\begin{align*}
\Wjoinat[pp, \ppi' \cdot (\fid, \ipc), m_1, m_2] &= 
\begin{cases}
(\ppi' \cdot (\fid, pc'), \ppi' \cdot (\fid, pc' + 1)) &\text{if }    \ipc' < |\access{\access{M}{\arraccess{funcs}{\fid}}}{body}| \land  \access{c_1}{pp} = \access{c_2}{pp} = \ppi' \cdot (\fid, \ipc) \land {} \\
                                           &\phantom{\text{ if }}     m_1 = \access{c_1}{mem} \land m_2 = \access{c_2}{mem} \land \left(\derive{\Delta}{\h{ScopeExtend}[\fid][\ipc, pc']}\right) \land{} \\
                                           &\phantom{\text{ if }}     \Delta = \Wabstract[c_1] \cup \Wabstract[c_2] \land \left(\forall \opc . pc' < \opc \implies \nderive{\Delta}{\h{ScopeExtend}[\fid][\ipc, \opc]}\right) \\
(\ppi' \cdot (\fid, \ipc), \ppi' \cdot (\fid, \ipc + 1)) &\text{if }  \arraccess{\access{\access{M}{\arraccess{funcs}{\fid}}}{body}}{\ipc} = \CallIndirect \\
(\nojoinpc, \nojoinpc) &  \text{otherwise}
\end{cases}\\
\Wccf[pp, \ppi' \cdot pp', m_1, m_2] &= \begin{cases}
\textit{true} & \text{if } \arraccess{\access{\access{M}{\arraccess{funcs}{\access{\mathit{pp}'}{fid}}}}{body}}{\access{pp'}{pc}} \in \{\BrIf{}, \BrTable{}, \CallIndirect, \If{} \} \\
\textit{false} & \text{otherwise}
\end{cases} 
\end{align*}
\caption{Definitions of the specific instances of $\Mjoinat$ and $\Mccf$ for \tool{} (with fixed $\Gamma_0$ and $M$). Note that the program positions in \WA are modeled as stack of pairs of function identifiers and program counter.}
\label{fig:joinatccf}
\end{figure*}
\begin{figure*}[ht]
\ExplSyntaxOn
\makefitcolumn{
\newcommand{\myspace}{\hspace{0.35cm}}
\begin{smallalign*}
  & \mathrlap{
    \addedInWanilla{ \horstTypeFlowLabel = \mathTypeDefinitionConstructors{FlowLabel}} \myspace
    \addedInWanilla{ \horstTypeContext = \mathTypeDefinitionConstructors{Context}} \myspace
    \horstTypeLValue = \horstTypeValue \times \addedInWanilla{ \horstTypeFlowLabel } \myspace
    \horstTypeMemory = \mathTypeDefinitionConstructors{Memory} \myspace
    \addedInWanilla{ \horstTypeTablePrecision = \mathTypeDefinitionConstructors{TablePrecision}} \myspace
    \addedInWanilla{ \horstTypeTable = \mathTypeDefinitionConstructors{Table}}
  }\\
  & \mathPredicateSignatureName{MState} : &&
  \group_begin:
    \declarePredicateVariables{MState}
    \horstDeclareParVars{\thisPredicateParVarNameList}
    \addedInWanilla{\seq_item:cn { \thisPredicateArgTypeList } { 1 }} \times
    \seq_item:cn { \thisPredicateArgTypeList } { 2 } \times
    \seq_item:cn { \thisPredicateArgTypeList } { 3 } \times
    \seq_item:cn { \thisPredicateArgTypeList } { 4 } \times
    \seq_item:cn { \thisPredicateArgTypeList } { 5 } \times
    \addedInWanilla{\seq_item:cn { \thisPredicateArgTypeList } { 6 }} \times
    \seq_item:cn { \thisPredicateArgTypeList } { 7 } \times
    \seq_item:cn { \thisPredicateArgTypeList } { 8 } \times
    \seq_item:cn { \thisPredicateArgTypeList } { 9 }
  \group_end:
  &&& \addedInWanilla{ \mathPredicateSignatureName{ScopeExtend} } :&& \addedInWanilla{ \mathPredicateSignatureArgumentTypes{ScopeExtend} } \\
  & \addedInWanilla{ \mathPredicateSignatureName{MStateToJoin} }:&& \addedInWanilla{ \mathPredicateSignatureArgumentTypes{MStateToJoin} }
\end{smallalign*}
}
\ExplSyntaxOff

  \caption{Subset of type and predicate signature of \tool{}. Changes from \wappler are highlighted.}
  \label{fig:wanilla_signature}
\end{figure*}

Following the approach described in the previous section, we implemented a prototype analysis, \tool{}\footnote{\textbf{W}eb\textbf{A}ssembly \textbf{N}on\textbf{i}nterference \textbf{L}ayered \textbf{L}abel \textbf{A}nalysis}, which we base on \wappler{}~\cite{wappler}, an  existing reachability analysis for \WA.
In the model presented in \autoref{sec:background}, \wappler is an instance of $\Mreachability$ (respectively $\Mupdate$), while \tool is an instance of $\Mnoninterference$.
Going forth, we will overview \wappler and \tool's basic analysis approach in \autoref{sec:analysis_horn_clause_abstraction} and then highlight some interesting implementation details in \autoref{sec:from_wappler_to_wanilla}. 
\ifextended
The interested reader can find an extended technical discussion with additional examples in \autoref{sec:additional_details}.
\fi

\subsection{Horn-clause-based Abstractions}
\label{sec:analysis_horn_clause_abstraction}
Both \wappler and \tool implement a so-called Horn-clause-based abstraction. Intuitively, the concrete configurations of the analyzed program are translated into sets of first-order predicates, the so-called \emph{abstract configurations}.
As the name suggests, these sets may represent many concrete configurations, as some of the details of the concrete configuration are abstracted away.
The behavior of the analyzed program is translated into (constrained) Horn clauses, i.e., first-order implications with a single (non-negated) predicate application in the conclusion and a conjunction of (non-negated) predicate applications (and possibly some terms over a theory) in the premise. 
\ifextended
The translation to facts and Horn clauses is formalized by the function $\Wabstract$, which is described in detail in Appendix \autoref{sec:proofs-on-wanilla}
\fi
Queries of the form $(pp', m') \in \Mreachability[pp, m]$ are thus translated into (first-order) satisfiability problems of the shape ``Is $\Delta'$ (an abstraction of $(pp', m')$) logically entailed by $\Delta$ (an abstraction of $(pp, m)$ and the behavior of the analyzed program)'' or, symbolically, $\derive{\Delta}{\Delta'}$.
These satisfiability problems can then be solved by SMT solvers, such as Z3~\cite{DBLP:conf/tacas/MouraB08}.
SMT solvers can universally quantify over inputs, which means that, given a sound abstraction, an UNSAT result in the abstract (i.e., a proof that an abstract configuration with certain properties is not entailed by an abstraction of any initial configuration) proves that concrete configurations with these properties cannot appear in concrete executions of the program. 
For \tool in particular, we can additionally universally quantify over all possible attackers and determine if any one of them may cause unintended flows.

\subsection{From \wappler to \tool}
\label{sec:from_wappler_to_wanilla}
\subsubsection{Predicate Signature}
\autoref{fig:wanilla_signature} shows an illustrative subset of the predicate signature of \tool (with the changes from \wappler{} \addedInWanilla{highlighted}).
$\mathPredicateSignatureName{MState}$ is the principal predicate of the analysis.
A fact $\h{MState}[\hfid,\hpc][\hspace{-0.6mm}\ctx,\hspace{-0.1mm} \st,\hspace{-0.1mm} \gt,\hspace{-0.1mm} \lt,\hspace{-0.4mm} \h{Mem}[i,\hspace{-0.1mm} \vm,\hspace{-0.1mm} \sz],\hspace{-0.1mm} t,\hspace{-0.1mm} \oat,\hspace{-0.1mm} \ogt, \hspace{-0.4mm} \h{Mem}[i,\hspace{-0.1mm} \ovm,\osz]\hspace{-0.6mm}]$ expresses that the (abstract) execution encountered an activation of the function $\hfid$ that executed the instruction at program counter $\hpc$ on the call stack.
Furthermore, the values on the value stack and the local variables of the activations hold the values in $\st$ and $\lt$, while $\gt$ holds the globals' values.%
\footnote{$\textsf{SI}.*$ is a family of functions that returns information on the module, such as the number of globals ($\fromSI{gs}{}$) or the stack size ($\fromSI{ss}{}$) for a program position $(\hpc, \hfid)$.}
The linear memory is handled slightly differently: since we do not necessarily know its size upfront, and it would be infeasible to generate it as a tuple for its size, we instead model it by having an $\mathPredicateSignatureName{MState}$-fact for every possible byte-sized memory cell at $i$. 
Being able to derive the fact from above means that the value $v$ is stored at index $i$ and the current size of the memory (in pages) is $\sz$.
$t$ describes the state of the \emph{table}, a special region in the memory that holds function pointers for dynamically dispatched calls.
It is relevant for the $\CallIndirect$ instruction, which takes the top-of-stack value to look up which function to call in the table (opposed to the $\Call{}$ instruction for compile-time-known call targets).
$t$ has two components.
The first tracks if the function table has been changed since the execution started.
If it has not, we can track call targets precisely (corresponding to $\h{TblPrecise}$); if it has, we have to overapproximate the call targets (corresponding to $\h{TblImprecise}$).
The second tracks whether tainted information has been stored to the table.
This is a slight change from \wappler{}, where the table contents are modeled as a separate predicate.
We did this since it simplifies the implementation, and Wasm 1.0 does not provide means to modify the table from within a module (imported functions, however, can change the table).
The last three parameters of $\mathPredicateSignatureName{MState}$ hold the arguments (a subset of the locals), the globals, and the memory cell (at the same offset as above) at the time the $\hfid$ was called.%
\footnote{\wappler holds these initial values to increase precision, as we will see, they are crucial for relating runs that started in attacker-indistinguishable configurations.}

To reconstruct the set of concrete configurations that are abstracted by an abstract configuration $\Delta$, one has to first find subsets of $\Delta$ that agree on all values but the memory and contain exactly one $\mathPredicateSignatureName{MState}$-fact for each possible memory index.
\ifextended
We call these subsets abstract activations and generate them with the function $\preframes$ (described in Appendix \autoref{sec:proofs-on-wanilla}).
\fi
Due to the condition mentioned above, accessing the program position or the value of a memory position $p$ in an abstract activation $\Delta'$ is well-defined.
We write $\access{\Delta'}{pp}$ and $\mpaccess{\Delta'}{p}$, respectively.
These abstract activations can be arranged into call stacks by matching the values of $\st$, $\gt$, $\vm$, and $\sz$ in the abstract activation describing the caller with the values of $\oat$, $\ogt$, $\ovm$, and $\osz$ in the abstract activation describing the callee.
The globally accessible data (corresponding to the table, the linear memory, and the globals) appear once for each abstract activation; the ones used in the concretization are those contained in the abstract activation put at the top of the call stack.
\ifextended
The formal description of $\gamma$, the concretization function, is relegated to Appendix~\autoref{sec:proofs-on-wanilla}.
\fi
Since the functions characterizing the program's semantics/reachability analysis ($\Mnoninterference$, $\Mccf$, $\Mjoinat$) are defined in terms of concrete configurations, it is, however, important to be able to convert between abstract and concrete configurations.

Excluding the minor change in modeling the table mentioned above, all changes between \wappler{}'s and \tool{}'s predicate signature are \emph{additions}:
The first parameter of $\mathPredicateSignatureName{MState}$ holds a \addedInWanilla{context} object, which corresponds to $\Mnoninterference$'s context object and is further discussed below.
In \tool{}, $\Mnoninterference$'s label maps are not instantiated as separate objects.
Instead, we change the definition of \wappler{}'s value type $\horstTypeLValue$ to carry a taint label.
The remaining predicates shown in \autoref{fig:wanilla_signature} will be explained when they are used in the following paragraphs.

\subsubsection{Transition Rules and Abstract Semantics}
\label{analysis_transition_rules}
\wappler works by assessing which \WA instructions can be executed at a given program position and then generating specialized Horn clauses accordingly.
Applying possibly multiple Horn clauses to an abstract configuration corresponds to an application of $\Mupdate$.
The next program position is reflected in the $\hfid$ and $\hpc$ component of the $\mathPredicateSignatureName{MState}$-fact in the conclusion, the update set in the changed values.
\ifextended
A definition of \wappler's (and \tool's) $\Mupdate$ function ($\Wupdate$) can be found in Appendix \ref{sec:proofs-on-wanilla}.
\fi
For \tool, we take the same approach of generating specialized Horn clauses for each instruction but additionally have to assess which rule from \autoref{fig:symbolic-runs} might be applicable, given the semantics of the instruction.
All instructions that do not pertain to control flow are instances of \rniaprop,
local control flow instructions, such as \BrIf{k}, are, depending on their position and the current labeling, instances of any rule (but \rniainit), $\End$, the instruction executed when leaving blocks, is either an instance of \rniaprop or \rniajoin, depending on the calculated \emph{scope extension}.

\subsubsection{Scope Extensions and Implicit Flows}
\label{sec:scope_extensions_and_implicit_flows}
Like previous work~\cite{DBLP:conf/sas/BastysASS22}, we use the syntactical structure to determine where implicit flows might happen and, conversely, where to apply \rniajoin.
In our predicate signature, there is a predicate $\mathPredicateSignatureName{ScopeExtend}$ that assigns for each function from which start program counter to which end program counter implicit flows might happen.
We can use $\mathPredicateSignatureName{ScopeExtend}$ to define $\Wjoinat$, which is $\Mjoinat$ instantiated to \tool for a fixed module description $M$ and security policy $\Gamma_0$, in \autoref{fig:joinatccf}.
$\mathPredicateSignatureName{ScopeExtend}$ is populated the following way:
When a conditional control flow instruction is executed with an $\H$-labeled condition and an $\L$-labeled context, we derive a $\mathPredicateSignatureName{ScopeExtend}$ predicate from the current program counter to the program counter where the control flow converges \emph{at the earliest}, as this follows from the syntactical structure of the program and is statically known.
In \autoref{fig:symbolic-runs}, this happens in \li{7}, where an initial scope from \li{7} to \li{11} (marked in blue) is calculated.
If, however, a conditional control flow instruction is executed with an already $\H$-labeled context, the scope starting at the divergence point stored in the context object is extended.
This happens in \li{10} and corresponds to the area marked in red.
For every function, we furthermore imply a $\mathPredicateSignatureName{ScopeExtension}$ fact from $\nojoinpc$\footnote{As the program counters in \wappler are non-negative, we use $-1$ to designate a program counter that cannot be reached.} to the end of the function.
When entering a function in an $\H$-labeled context, we set the context object to $\horstConstructorAppCtx{\ell, \H, \nojoinpc}$, which means that within the function, the context label can effectively not be lowered.
When the execution selects which function to enter based on tainted data, the context is similarly tainted for the whole execution until this function returns.
This happens when a $\CallIndirect$ instruction is executed and either the table or the top element of the value stack is tainted, and is handled in the second case in $\Wjoinat$'s definition.%
\footnote{$\arraccess{\access{\access{M}{\arraccess{funcs}{\fid}}}{body}}{\ipc} = \CallIndirect$ means that the instruction at program counter $\ipc$ in function $\fid$ is $\CallIndirect$, as described by the module description $M$.
The notation follows the Wasm specification~\cite{wasm-core-1.0}.}
For every potential transition where \rniajoin should possibly be executed (i.e., instances of $\End$ and instances of $\BrIf{}$ that leave loops), we generate Horn clauses with the following behavior: 
If the scope of the program counter stored in the context object extends over the next program counter, the context label will stay $\H$ (this corresponds to an application of \rniaprop).
\rniajoin, however, is executed in any case.
This is done because Horn clause reasoning does not allow us to negate premises (as in ``There is no scope that does extend over the next program counter.''). 
Our soundness statement is not affected by these spurious joins, as it speaks about the existence of certain results;
applying additional rules (in this case, applying \rniaprop and \rniajoin in parallel) only allows for the derivation of additional results.

\tool{}'s context object looks slightly different from $\Mnoninterference$'s.
Since the $\Mjoinat$'s result is encoded in the clauses, we do not need to track it in the context object.
Instead, we carry an element of $\Lat$ there to describe the current attacker's perspective.
This way, we can quantify over multiple attackers during one analysis run.

\rniajoin is implemented via the $\mathPredicateSignatureName{MStateToJoin}$-predicate and a special join rule that is generated for all potentially joining instructions.
This rule has two predicate applications of $\mathPredicateSignatureName{MState}$ in its premises that must have the same values in their initial values (i.e., their last three arguments) if they are labeled $\L$.
In this way, \rniajoin is only executed for configurations that are related by having been derived from starting configurations related by $\elleqg$.

\subsubsection{Overapproximations}
\label{analysis_overapproximations}

\tool inherits from \wappler its imprecise treatment of floating point numbers: every time a floating point number is involved in a computation, the result of said computation is indeterminate.
They are, however, still subject to the labeling system, meaning that the result of a computation involving only $\L$-labeled values still is labeled $\L$, even if, by looking at the results of two related executions, we could not determine if they have indeed the same values.
\tool additionally includes a conservative preanalysis to determine which memory positions could be modified in loops and overapproximates their values (while, again, keeping their labels).
This (optional) optimization helps SMT solvers, who often struggle with recursive Horn clauses (like they are generated by loops), to terminate.

\subsubsection{Soundness of \tool}
Given these definitions, we are now ready to formalize the soundness of \tool.
The following theorem states that
if two concrete executions\footnote{We assume hereafter that all executions are concerned with a single module $M$.} start out in configurations that are indistinguishable to an attacker and end up in configurations,
which disagree on a memory position the attacker can observe,
this memory position will be tainted.
In particular, the theorem states that for any \WA module description $M$ and two configurations $c_1$, $c_2$,
whose memories are equal w.r.t.\ an attacker $\ell$ and a security policy $\Gamma$ ($\elleqg[\access{c_1}{mem}][\access{c_2}{mem}]$)\footnote{Since a \WA configuration $c$ is not (syntactically) a tuple of program positions and memory, we instead write $\access{c}{pp}$ to access its program position as extracted by $\exconfpp[c]$ and $\access{c}{mem}$ to mean $\exconfmem[c]$.
\ifextended
Both $\exconfpp$ and $\exconfmem$ are defined in Appendix~\autoref{sec:proofs-on-wanilla}.
\fi
},
which have the same program position ($\access{c_1}{pp} = \access{c_2}{pp}$),
whose executions end at the same program position $pp'$ ($\nstep{*}{c_1}{c_1'} \land \nstep{*}{c_2}{c_2'} \land pp' = \access{c_1'}{pp} = \access{c_2'}{pp} \land \terminalconf[pp']$),
we can derive an abstract configuration $\Delta$ from the abstractions of $c_1$ and $c_2$ ($\derive{\Wabstract[c_1][\Gamma][M], \Wabstract[c_2][\Gamma][M]}{\Delta}$),
such that $\Delta$ contains an abstract activation $\Delta'$ at $pp'$ ($\Delta' \in \preframes[\Delta] \land \access{\Delta'}{pp} = pp'$),
such that $\Delta'$ models the attacker $\ell$ ($\access{\Delta'}{attacker} = \ell$) and the value that $\Delta'$ holds at $p$ is tainted ($\mpaccess{\Delta'}{p} = (v, \H)$),
and equal to $p$ in either $c_1'$ or $c_2'$ ($\mcall{\access{c_1'}{mem}}{p} = v \lor \mcall{\access{c_2'}{mem}}{p} = v$).

\begin{theorem}[Soundness of \tool]
\hfill{}\\
\label{thm:soundness-wanilla}
\makefitcolumn
{
\small
\begin{align*}
&\forall \ell\,\Gamma\,M\,c_1\,c_2\,c_1'\,c_2'\,pp'\, p~.~\elleqg[\access{c_1}{mem}][\access{c_2}{mem}] \land \access{c_1'}{mem}(p) \ne \access{c_2'}{mem}(p)\land{}\\
&\hspace{0.2cm} \access{c_1}{pp} = \access{c_2}{pp} \land \nstep{*}{c_1}{c_1'} \land \nstep{*}{c_2}{c_2'} \land pp' = \access{c_1'}{pp} = \access{c_2'}{pp} \land \terminalconf[pp'] \\
&\hspace{0.4cm} \implies \exists \Delta\, \Delta'\, v \,.\, \derive{\Wabstract[c_1][\Gamma][M], \Wabstract[c_2][\Gamma][M]}{\Delta} \land \Delta' \in \preframes[\Delta] \land \access{\Delta'}{attacker} = \ell \land {}\\
&\hspace{0.6cm} \mpaccess{\Delta'}{p} = (v, \H) \land \access{\Delta'}{pp} = pp'  \land \left( \mcall{\access{c_1'}{mem}}{p} = v \lor \mcall{\access{c_2'}{mem}}{p} = v\right)
\end{align*}
}
\end{theorem}
\begin{proof}
\ifextended
This follows directly from \autoref{thm:nia_sound_memories} and the soundness of $\Wupdate$, $\Wjoinat$, and $\Wccf$. For details see Appendix \ref{sec:proofs-on-wanilla}.
\fi
\end{proof}

If (with the same preconditions as before) \tool cannot derive an $\H$-label for any attacker-observable position $p$ ($\call{\Gamma}{p} \flowsto \ell$), the memories of $c_1'$ and $c_2'$ have to be equal w.r.t.\ to $\Gamma$ and $\ell$.
\begin{corollary}[Noninterference Proofs]
\hfill{}\\
\makefitcolumn{
{
\small
\begin{align*}
&\forall \ell\,\Gamma\,M\,c_1\,c_2\,c_1'\,c_2'\,pp'\,p~.~\elleqg[\access{c_1}{mem}][\access{c_2}{mem}] \land \access{c_1'}{mem}(p) \ne \access{c_2'}{mem}(p)\land{}\\
&\quad \access{c_1}{pp} = \access{c_2}{pp} \land \nstep{*}{c_1}{c_1'} \land \nstep{*}{c_2}{c_2'} \land pp' = \access{c_1'}{pp} = \access{c_2'}{pp} \land \terminalconf[pp'] \land{} \\
&\quad \big(\nexists \Delta\, \Delta'\, v \, p .\, \derive{\Wabstract[c_1][\Gamma][M], \Wabstract[c_2][\Gamma][M]}{\Delta} \land \Delta' \in \preframes[\Delta] \land \access{\Delta'}{pp} = pp' \land {}\\
&\quad \mpaccess{\Delta'}{p} = (v, \H) \land \access{\Delta'}{attacker} = \ell \land \call{\Gamma}{p} \flowsto \ell \big)
\implies \elleqg[m_1'][m_2']
\end{align*}
}
}
\end{corollary}

\subsubsection{Comparison of \wappler{} and \tool{}'s Expressiveness}
As \tool{} is not only calculating taints but also reachable program configurations, it can also be used as reachability analysis, and can potentially even combine the two aspects (e.g., to answer queries of the form ``Is it possible to derive a taint when the configuration is restricted to these values?'').
Concerning the expressiveness of \tool{}-as-reachability-analysis and \wappler, there are only two relevant differences, which we reiterate here:
First, \wappler in prin-\penalty-10000ciple allows queries of the form ``Is a certain value derivable in the table?'' while \tool{} only remembers if the table has changed.
In practice, this has a negligible impact as \wasm 1.0  offers no language-internal means to modify the table: only host functions, which both \wappler and \tool{} have to overapproximate, can do so.

More significant in terms of expressiveness is \tool{}'s overapproximation of loops described above.
This means that \wappler{} can reason over loops more precisely (to the degree the underlying SMT solver supports it).

\section{Implementation and Evaluation}
As the underlying analysis, \tool{} is implemented using the \horst framework~\cite{DBLP:conf/ccs/SchneidewindGSM20,DBLP:conf/isola/SchneidewindSM20}, which compiles declarative descriptions of static analyses to SMT-LIB~\cite{BarFT-RR-17} files.
Indeed, most rules did not need major modifications except for adding the execution context and labeled value types.
In sum, the specification of the abstract semantics and properties takes
\checkme{around 3k} lines of \horst, while parsing test specifications and other preprocessing is implemented in \checkme{approximately the same amount} of \java.

To assess the precision and soundness of \tool{}, we focussed on existing noninterference benchmarks that \emph{come with a ground truth}.
Since annotating code with a security policy and assessing the ground truth is a labor-intensive process (especially if the programs are not available in a high-level language), the number of available benchmarks is still limited despite \WA's rising popularity.
To remediate this, we translated one existing noninterference benchmark to \WA (see \autoref{sec:rapid_benchmarks}) and introduce our own benchmark suite (see \autoref{sec:wanilla_benchmarks}).
Additionally, we evaluate \tool{} on a benchmark of real-world smart contracts against the tool that introduced the benchmark~\cite{DBLP:conf/internetware/HuangJC20} and Wassail~\cite{DBLP:conf/scam/StievenartR20} in \autoref{sec:eosfuzzer_benchmarks}\footnote{We did not evaluate against SecWasm~\cite{DBLP:conf/sas/BastysASS22} as it is not a fully static approach (which would necessitate some fuzzing component for judging a (non-instrumented) program noninterferent) and, more importantly, does not come with an implementation.}.

In the following, noninterferent programs will be marked \nimark, interferent ones \inmark. \cmark{} denotes the correct result, \xmark{} an incorrect one.

\subsection{On the Used Solvers}

\tool{}'s outputs SMT-LIB files, specifically using its Constrained Horn Clause fragment with integer and bit vector arithmetic.
A solver supporting these fragments is the Z3 solver~\cite{DBLP:conf/tacas/MouraB08}, which we use to solve the benchmarks in \autoref{sec:wanilla_benchmarks} and \autoref{sec:rapid_benchmarks}.
For the programs in \autoref{sec:eosfuzzer_benchmarks}, which have instruction counts (marked ``ic'' in the tables)  in the thousands and for which \tool{} generates megabytes of SMT-LIB files, Z3 does not terminate within 7 hours (or terminates with an error) for some queries (in particular, the ones expected to show an attack).
We, therefore, implemented our own experimental solver that supports the requisite subset of SMT-LIB.

In essence, our solver implements an abstract interpretation of the SMT-LIB program by assigning a single valuation to each predicate that captures all possible valuations.
If, e.g., Z3 could derive exactly $P(x)$ for $x \in \{3,7\}$, our solver, using an interval abstraction, would be able to derive $P(x)$ for $3 \le x \le 7$.
This comes with a precision penalty (in the example above, we would be able to spuriously derive $P(x)$ for $x \in \{4,5,6\}$), but maintains soundness and allows us to analyze programs for which classical SMT solvers time out.
To quantify the effect of the imprecision, we report the number of failed test cases for the benchmarks in \autoref{sec:wanilla_benchmarks} and \autoref{sec:rapid_benchmarks}.

\subsection{\tool{} Benchmark}
\label{sec:wanilla_benchmarks}

In order to evaluate our implementation for precision and correctness, we created a benchmark (including the code presented in the figures of this paper) of \wanillaBenchmarkModuleCount{} WebAssembly modules and \wanillaBenchmarkTestCaseCount{} test cases that cover all constructs relevant to noninterference.
The benchmarks are synthetic and designed to give minimal examples of information flow behavior as combination of programs and security policies.
Each test case might contain multiple queries corresponding to different memory regions.
The average instruction count in this benchmark is \wanillaBenchmarkAverageInstructionCount (min: \wanillaBenchmarkMinimalInstructionCount, max: \wanillaBenchmarkMaximalInstructionCount).
When using the Z3 as the solver, we are able to solve all benchmarks precisely.
Using our experimental solver, we get false positives in \wanillaBenchmarkAsmtImpreciseTestcases{} test cases (\wanillaBenchmarkAsmtImpreciseQueries{} queries) and solve \wanillaBenchmarkAsmtPreciseTestcases{} (\wanillaBenchmarkAsmtPreciseQueries{} queries) precisely.
We make this benchmark available to the community at~\cite{wanilla-artifacts}.

\subsection{Rapid Benchmarks}
\label{sec:rapid_benchmarks}

To compare our expressiveness with existing solutions, we translated the noninterference benchmark from RAPID~\cite{DBLP:conf/fmcad/BartheEGGKM19} to \WA and ran \tool on it.
The programs are all noninterferent and show a variety of potential explicit and implicit flows that are meant to be difficult for \emph{sound} tools to handle.\footnote{Wassail, for example, also marks all programs in this benchmark noninterferent. This unsurprising, as it ignores implicit flows. For an analysis that is, in principle, capable of detecting implicit flows, verifying these benchmarks is more difficult.}
We translated the benchmarks from RAPID's while-like programming language to \WA{} by hand.
The reason for doing this (instead of writing the benchmark programs in, say, C and compiling them to \WA{}) was that, in many cases, the compiler produced unreasonably difficult-to-analyze code when instructed not to apply any optimizations (mostly modeling the call stack in the linear memory instead of using \WA{}'s dedicated mechanisms).
When instructed to use optimizations, however, the compiler would often remove the interesting, potentially interfering parts by applying general-purpose optimizations.

The results (presented in \autoref{fig:wanilla_vs_rapid}) are not directly comparable, as RAPID is a general hyperproperty checker for a simplified model language (using a special version of vampire\cite{DBLP:conf/cav/KovacsV13, SpecialVampire} as a solver), while \tool{} is specialized in noninterference on a real-world language (and uses Z3\cite{DBLP:conf/tacas/MouraB08}).
The result of our comparison is that \tool{} can prove all noninterference examples of RAPID's benchmark noninterferent, while we were not able to get results with RAPID for two of the examples.\footnote{We asked RAPID's authors how to make these benchmarks terminate in February 2024, but they have not come back to us in time of submission.}
W.r.t.\ to the time needed by the underlying solvers to solve the generated SMT-LIB files, RAPID is faster in most terminating cases, while \tool{} is still reasonably fast.
The one (terminating) example where \tool{} is faster (10-ni-rsa-exponentiation) contains a loop, which indicates that the loop abstraction bears fruits.

Our experimental solver is imprecise for \rapidBenchmarkAsmtImpreciseTestcases{} test cases (\rapidBenchmarkAsmtImpreciseQueries{} queries) and solves \rapidBenchmarkAsmtPreciseTestcases{} (\rapidBenchmarkAsmtPreciseQueries{} queries) precisely.

\begin{figure}
  \begin{tabular}{p{1mm} p{40mm} r  r r p{5mm} r}
            &  \textbf{Benchmark}             & \multicolumn{2}{c}{RAPID}& \multicolumn{3}{c}{\tool{}}\\
            &                                 &        & t [ms] &       & \mbox{t [ms]}       & ic\\
    \hline
    \nimark &  1-ni-assign-to-high            & \cmark & 3      &\cmark & 33           & 8 \\
    \nimark &  2-ni-branch-on-high            & \cmark & 3      &\cmark & 364          & 18\\
    \nimark &  2-ni-branch-on-high-twice      & \cmark & 17     &\cmark & 2239         & 35\\
    \nimark &  \mbox{3-ni-high-guard-equal-branches} & \cmark & 3      &\cmark & 158          & 13\\
    \nimark &  5-ni-temp-impl-flow            & \cmark & 16     &\cmark & 445          & 21\\
    \nimark &  6-ni-branch-assign-equal-val   & \cmark & 193    &\cmark & 2485         & 12\\
    \nimark &  7-ni-explicit-flow             & \cmark & 8      &\cmark & 248          & 20\\
    \nimark &  8-ni-explicit-flow-while       & -      & $\bot$ &\cmark & 291          & 21\\
    \nimark &  9-ni-equal-output              & -      & $\bot$ &\cmark & 490          & 18\\
    \nimark &  10-ni-rsa-exponentiation       & \cmark & 12402  &\cmark & 2604         & 46
\end{tabular}

  \caption{Times used to solve the files generated by RAPID (with vampire) and \tool{} (with Z3).}
  \label{fig:wanilla_vs_rapid}
\end{figure}

\begin{figure}
  \begin{tabular}{l p{16mm} l r l r l r r}
       &  \textbf{Benchmark}           & \multicolumn{2}{c}{EOSFuzzer}  & \multicolumn{2}{c}{Wassail} & \multicolumn{3}{c}{\tool{}} \\
       &                               &        & t [s]               &          &  t [s]           &             & t [s]          & ic \\
\hline
\inmark &  coingame                   & \cmark   &        -            & \xmark   &    0.5           & \cmark      & 8.2           & 3140\\
\inmark &  eosfun                     & \cmark   &        -            & \xmark   &    0.5           & \cmark      & 7.5           & 3178\\
\inmark &  lottery1                   & \xmark   &        -            & \cmark   &  104.6           & \cmark      & 221.9         & 16538\\
\hline
\nimark &  random                     & -        &        -            & \cmark   &    0.7           & \cmark      & 3.4           & 2575\\
\end{tabular}

  \caption{Times (in seconds) used to analyze the \WA modules from the EOSFuzzer benchmark. }
  \label{fig:eosfuzzer_benchmark_eval}
\end{figure}

\subsection{EOSFuzzer Benchmarks}
\label{sec:eosfuzzer_benchmarks}

EOSFuzzer~\cite{DBLP:conf/internetware/HuangJC20} is a fuzzing tool for the EOSIO blockchain platform, which uses \wasm for its smart contracts.
In such platforms, computations have to be deterministic, as multiple so-called miners have to compute the same state.
Some contracts, notably lotteries, use block information, which is not trivial to predict but still deterministic, to model ``randomness''.
A sufficiently advanced attacker may predict the outcomes of such lotteries and thus claim the prize while bereaving honest users of their chance to win.

We can express the absence of this problem in terms of noninterference by treating block information as low-integrity while requiring high-integrity inputs and contexts for functions transferring tokens.
As all blockchain-specific functionality in EOSIO is handled by imported functions, we simply have to provide the correct security policies to solve the problem with \tool{}.
Orienting ourselves on EOSFuzzer, we treat the outputs of \ccode{tapos_block_num} and \ccode{tapos_block_prefix} as low-integrity while requiring high-integrity inputs and contexts for \ccode{send_inline} and \ccode{send_deferred}.
Of the 82 contracts with sources that EOSFuzzer includes, 79 are trivially noninterferent according to this model, as they do not contain both calls to functions producing low-integrity output and functions requiring high-integrity input.
We also included one contract from a different part of EOSFuzzer's dataset (\ccode{atom/random}) for which no results are reported\footnote{We assume EOSFuzzer can judge this contract correctly, as it is a directory with testing contracts, but we cannot verify it, as the docker image for EOSFuzzer is unavailable as of April 2025~\cite{eosfuzzerdockernotavailable}, and the project does not contain detailed build instructions.}, as it is not trivially safe.

In \autoref{fig:eosfuzzer_benchmark_eval}, we report EOSFuzzer's results from~\cite{DBLP:conf/internetware/HuangJC20} and the results of running Wassail and \tool{} (using our experimental solver) on the benchmark ourselves.
Wassail fails to analyze two smart contracts, as the relevant flows are implicit flows.\footnote{In the context of these lotteries this means that whether the payment happens depends on bad randomness but not the payment amount.}
EOSFuzzer does not find the vulnerability in lottery1, as the set of concrete values that trigger the vulnerability is very small and thus unlikely to be discovered by random fuzzing.
\tool{} correctly labels all flows in the benchmark.

\subsection{Limitations}
The most difficult-to-analyze part of WebAssembly is the linear memory, which means modules with complicated memory interactions are sometimes challenging for \tool{}.
As demonstrated in the preceding section, this can be mitigated by sacrificing precision for performance.
As \wappler, \tool can only analyze one module at a time and overapproximates functions imported from other \WA modules as if they were host functions.
An improvement in \wappler in this regard would directly (modulo the changes in \autoref{sec:from_wappler_to_wanilla}) carry over to \tool.

The analysis presented here is termination-insensitive.
Handling the exceptional halt that \wasm specifies (which we have not mentioned so far, as it was not needed for the presentation of the analysis) does not present a fundamental problem for our approach, as reasoning about the potential reachability of exceptions is very similar to reasoning about the reachability of regular program states.
Possible non-termination, on the other hand, requires different analysis techniques.
In the absence of additional machinery for termination-checking, one would have to consider every loop that is conditioned on an $\H$ variable as a potential termination channel, which is possible but imprecise.

Lastly, while the \WA specification tries to minimize nondeterminism, some instructions, e.g., $\mycode{memory.grow}$ to grow the memory, are nondeterministic.
While it is possible to model this nondeterminism in the analysis, we assume that for noninterference, both executions will make the same nondeterministic choice and that this does not offer an information channel for a potential attacker.
If such an information channel was present, it would be easy to label nondeterministic values $\H$ by default.

\section{Conclusion and Future Work}
We presented a novel approach to lift sound reachability analyses to sound noninterference analyses and instantiated it for \WA in \tool{}.
We evaluated \tool{} on a variety of benchmarks against existing analyzers for \WA and improved upon them in terms of precision and soundness.

In terms of future work, multiple avenues are possible:
As noninterference is needlessly strict in some contexts, one could extend our approach to more lenient forms of information flow control, such as robust declassification~\cite{DBLP:conf/csfw/ZdancewicM01, DBLP:conf/ccs/CecchettiMA17} or the ``abstract noninterference'' of Giacobazzi and Mastroeni~\cite{DBLP:conf/popl/GiacobazziM04}.

Making the approach more composable, e.g., by deriving interfaces that specify which input values/taint labels could lead to a flow on the function level, would help to analyze bigger code bases. 

Any approach that could extract ``representative runs'' (such as the three columns in \autoref{fig:symbolic-runs}) from a program could obviate the use of SMT-Solvers (which construct them implicitly) and improve the precision and performance of our approach when using more classical abstract interpretations as underlying analyses.
If similar approaches could be used to verify hyperproperties that are unrelated to noninterference remains an open research question. 

For \wasm in particular, any improvement of the underlying analysis (such as support for Wasm 2.0, floating point number, or an analysis to gather non-overlapping linear memory segments) would directly carry over to \tool{}.

\section*{Acknowledgements}
We thank our anonymous reviewers for their valuable feedback and Quentin Stiévenart for promptly answering our questions about Wassail.
This work was partially supported by the European Research Council (ERC; grant 101141432-BlockSec);
by the Austrian Science Fund (FWF) through the SpyCode SFB project F8510-N;
by the Austrian Research Promotion Agency (FFG) through the SBA-K1 NGC;
by the Vienna Science and Technology Fund (WWTF) through the project ForSmart;
and by the Christian Doppler Research Association through the Christian Doppler Laboratory Blockchain Technologies for the Internet of Things (CDL-BOT).

\bibliographystyle{abbrv}
\bibliography{literature}

\ifextended
\appendix

\section{Additional Details}
\label{sec:additional_details}
\begin{figure*}[t]
  \newcommand{\horstFreeVarctx}{\addedInWanilla{\textit{ctx}}}
  \let\oldhorstOpAppraiseCtxTo\horstOpAppraiseCtxTo
  \renewcommand{\horstOpAppraiseCtxTo}[2]{\addedInWanilla{\oldhorstOpAppraiseCtxTo{#1}{#2}}}
  \let\oldhorstOpAppraiseTo\horstOpAppraiseTo
  \renewcommand{\horstOpAppraiseTo}[2]{\addedInWanilla{\oldhorstOpAppraiseTo{#1}{#2}}}
  \renewcommand{\horstConstructorAppCtx}[1]{\addedInWanilla{\horstConstructorApp{Ctx}{#1}}}
  \renewcommand{\horstPredAppScopeExtend}[2]{\addedInWanilla{\horstPredApp{ScopeExtend}{#1}{#2}}}
  \let\oldhorstOpApplabelOf\horstOpApplabelOf
  \renewcommand{\horstOpApplabelOf}[2]{\addedInWanilla{\oldhorstOpApplabelOf{#1}{#2}}}
  \let\oldhorstOpApplabelOfCtx\horstOpApplabelOfCtx
  \renewcommand{\horstOpApplabelOfCtx}[2]{\addedInWanilla{\oldhorstOpApplabelOfCtx{#1}{#2}}}
  \let\oldhorstConstructorAppIllegal\horstConstructorAppIllegal
  \renewcommand{\horstConstructorAppIllegal}[1]{\addedInWanilla{\oldhorstConstructorAppIllegal{#1}}}
  \let\oldhorstConstructorAppLegal\horstConstructorAppLegal
  \renewcommand{\horstConstructorAppLegal}[1]{\addedInWanilla{\oldhorstConstructorAppLegal{#1}}}
  \let\oldhorstOpAppoverApproximateLoopGlobals\horstOpAppoverApproximateLoopGlobals 
  \renewcommand{\horstOpAppoverApproximateLoopGlobals}[2]{\addedInWanilla{\oldhorstOpAppoverApproximateLoopGlobals{#1}{#2}}}
  \let\oldhorstOpAppoverApproximateLoopLocals\horstOpAppoverApproximateLoopLocals 
  \renewcommand{\horstOpAppoverApproximateLoopLocals}[2]{\addedInWanilla{\oldhorstOpAppoverApproximateLoopLocals{#1}{#2}}}
  \let\oldhorstOpAppoverApproximateLoopMemory\horstOpAppoverApproximateLoopMemory 
  \renewcommand{\horstOpAppoverApproximateLoopMemory}[2]{\addedInWanilla{\oldhorstOpAppoverApproximateLoopMemory{#1}{#2}}}
  
  \center
  \resizebox{!}{9.1cm}{
  \begin{minipage}{\textwidth}
  \ruleGroup{binOpRule}{\Add{i32},\Mul{i64}, \ldots}
  \ruleGroup{blockRule}{\Block{}}
  \ruleGroup{loopRule}{\Loop{}}
  \ruleGroup{brIfRule}{\BrIf{k}}[1,4,5]
  \ruleGroup{brIfPseudoRule}{pseudo \BrIf{k}}
  {
  \let\oldhorstLT\horstLT 
  \renewcommand{\horstLT}[2]{\addedInWanilla{\oldhorstLT{#1}{#2}}}
  \ruleGroup{endRule}{\End{}}[1]
  }
  \ruleGroup{functionScopeExtendRule}{for each function $\h{fid}<p>$}[0]
  \ruleGroup{joinRule}{potential points of converging control flow, e.g. \End{}, \BrIf{k}}[0]
  \end{minipage}
  }
  \caption{Illustrative subset of \tool{}'s abstract rules.}
  \label{fig:wanilla_rules}
\end{figure*}
\subsection{Extensions to the Type and Predicate Signature}

\autoref{fig:wanilla_signature} shows part of the predicate signature of \tool, with changes from \wappler{} \addedInWanilla{highlighted}.
We introduce \h{FlowLabel} and  \h{Context}
to formalize the taint labels and context objects from \autoref{sec:background} and \autoref{sec:analysis}.
$\lctx[\ell]$ corresponds to $\h{Ctx}[\ell, \L, -1]$, where $\ell$ is the current attacker's perspective and $-1$ corresponds to $\nojoinpc$.
Given $\ctx \in \h{Context}$, we write $\h{labelOfCtx}[\h{ctx}<f>]$ to access its taint label and \h{perspectiveOfCtx}[\h{ctx}<f>] to access its attacker perspective.
The principle value type of the analysis, \h{LValue}, is changed from an alias of \h{Value} (64-bit bitvectors) to a tuples of bitvector and taint label.
Given $x \in \h{LValue}$, we write $\h{valueOf}[x]$ for its value component and $\h{labelOf}[x]$ for its label.
We sometime write $x_=$ for $(x,\L) \in \h{LValue}$ and $x_?$ for $(x,\H) \in \h{LValue}$. 
The memory abstraction of \wappler unchanged: for every memory cell there is a named tuple comprising its index, value, and size.
Note, however, that since the value abstraction changed, the value and size fields carry a taint label now.
While the reason for the former is obvious, the reason for the latter is more subtle: an attacker could, e.g., learn a secret if the linear memory was grown by a secret value or while the control flow depended on a secret.
For this reason, we treat the memory size are potential flow.
\wappler models the \emph{table} via a predicate.
For \tool, we instead carry an abstraction of the table (of the type \h{Table}) in the relevant predicates.
The first component of \h{Table} describes if the contents of the table are initialized to known values and can thus be modeled precisely (\h{TblPrecise}), or if the initial values have been overwritten (\h{TblImprecise}).%
\footnote{
Please note, that the version of \WA that we handle in this paper does not provide language-internal means to modify the table: all modifications have to be done in imported functions.
Imported functions in \tool do not only come with the security policies $\Gamma_{E}$ and $\Gamma_{ME}$ but also with an annotation that describes their behavior, i.e., if they modify globals, the memory or the table (similar to the policies \wappler provides in such cases).
We forwent modeling one particular behavior of imported functions in \tool that is modeled in \wappler:
An imported function can add functions to the \wasm store and modify the table to point towards these functions.
In order to model these newly added functions, we would have to specify a security policy for them, which would be possible to implement but cumbersome and detrimental to the clarity of exposition.
}

In the second line of \autoref{fig:wanilla_signature}, we describe the principal predicate of the analysis: 
$\mathPredicateSignatureName{MState}$ encodes part of an abstract configuration of the function identified by $\textbf{fid}$ at program counter $\textbf{pc}$.%
\footnote{$\textsf{SI}.*$ is a family of functions that returns information on the module, such as the number of globals ($\fromSI{gs}{}$) or the stack size ($\fromSI{ss}{}$) for a program position $(\hpc, \hfid)$.}
It holds as arguments \addedInWanilla{the context},
tuples of values for the value stack, the globals and locals,
an instance of the $\horstTypeMemory$ type for a single memory cell,
\addedInWanilla{a table},
and tuples for the initial arguments and the globals and a memory object for when the function is called.

Assuming \ccode{session} and \ccode{payload} were stored at memory addresses $0$ and $128$, the linear memory currently held $s$ pages, $t$ was a table object, and \ccode{main} had the function ID $0$, an abstraction of the first state in column 1 from \autoref{fig:symbolic-runs} would include the following set:\footnote{We write $x_0$ for the byte least-significant byte $x$, $y_1$ for the second-least-significant byte in $y$ and so on} 
\ 
\makefitcolumn{
\renewcommand{\MState}[2]{\textit{MState}_{#1}\mkern-2mu(\mkern-1mu#2\mkern-2mu)}
\begin{smallalign*}
&\big\{
 \MState{0,0}{\lctx[\PU],\es, [1_=], \es, \h{Mem}[0, {x_0}_?, s], t, \es, [1_=], \h{Mem}[0, {x_0}_=, s]}, \\
&\MState{0,0}{\lctx[\PU],\es, [1_=], \es, \h{Mem}[128, {y_0}_=, s], t, \es, [1_=], \h{Mem}[128, {y_0}_=, s]}, \\
&\MState{0,0}{\lctx[\PU],\es, [1_=], \es, \h{Mem}[129, {y_1}_=, s], t, \es, [1_=], \h{Mem}[129, {y_1}_=, s]}\big\}
\end{smallalign*}
}

Under the same assumptions, an abstraction of column 1 right before \ccode{untrusted_log} was called would include:

\makefitcolumn{
\renewcommand{\MState}[2]{\textit{MState}_{#1}\mkern-2mu(\mkern-1mu#2\mkern-2mu)}
\begin{smallalign*}
&\big\{
 \MState{0,15}{\lctx[\PU],\es, [1_=], \es, \h{Mem}[0, 0_=, s], t, \es, [1_=], \h{Mem}[0, {x_0}_?, s]}, \\
&\MState{0,15}{\lctx[\PU],\es, [1_=], \es, \h{Mem}[128, {y_0}_=, s], t, \es, [1_=], \h{Mem}[128, {y_0}_=, s]}, \\
&\MState{0,15}{\lctx[\PU],\es, [1_=], \es, \h{Mem}[129, {y_1}_=, s], t, \es, [1_=], \h{Mem}[129, {y_1}_=, s]}\big\}
\end{smallalign*}
}

Note that while the memory cell describing \ccode{session} changed (the first \h{Mem} object in the first fact) between program counters $0$ and $15$, the initial memory cell (the second \h{Mem} object) did not change.
The fields describing the input values of a function at call time are used to relate executions when joining them.
The $\mathPredicateSignatureName{MStateToJoin}$ is used exactly for this reason, while the $\mathPredicateSignatureName{ScopeExtend}$ is used to correctly taint the context.
Both are described below in \autoref{sec:scope_extension} and \autoref{sec:adjustment_opportunistic_lowering_flow_labels}.

\subsection{Explicit and Implicit Flows}

To describe the interactions between taint labels, contexts and labeled values, we introduce the following notations:
Given $x \in \h{LValue}$ and $l \in \h{FlowLabel}$, we can \emph{raise} $x$'s label by $l$.
We write $\h{raiseTo}[x,l]$ to mean $(\h{valueOf}[x], \horstLUB{\h{labelOf}[x]}{l})$.
We introduce a similar operation for contexts:
Given $\h{ctx}<f> \in \h{Context}$, $l \in \h{FlowLabel}$ and a constant $\ppc$, we define \emph{raising a context at a program program counter} as

\begin{smallalign*}
  \horstLUBPC{\h{ctx}<f>}{l} =& \begin{cases}
    \h{ctx}<f>                            & \text{if }\h{labelOfCtx}[\h{ctx}<f>] = \H \lor l = \L\\
    \h{Ctx}[\h{perspectiveOfCtx}[\h{ctx}<f>], \H, \textbf{pc}] & \text{otherwise} 
  \end{cases}
\end{smallalign*}

If the context label is already $\H$ or $l$ is $\L$, we do not change $\h{ctx}<f>$, otherwise we keep the attacker perspective but set to taint label to $\H$ and the program counter to \ppc.

\paragraph{Numeric Instructions}

Operations that only operate on the stack (that includes all numeric operations, comparisons, and conversions) can be modeled as illustrated by \ref{cls:binOpRule:0} in \autoref{fig:wanilla_rules}.
The topmost stack elements are fed into an operation $\horstOpApplabelledBinOp{}{}$ that applies an operation to the values and computes the least upper bound of its operands labels.
Additionally, we raise this label to the label of the current context.
An instance of this rule is for example applied in \li{6} in \autoref{fig:symbolic-runs} when \AnyCmd{i64}{ne} is executed.

\paragraph{Raising the Context Label}

Whenever the control flow depends on an $\H$-labeled value, the context label has to be raised, using the $\horstLUBPC{\cdot}{\cdot}$-operator.
As an exemplary control flow instruction, we showcase the abstract rules for $\BrIf{}$---others are handled similarly.   
To model $\BrIf{}$'s semantics, we assume $\h{br}<p>$ holds the program counter at which the execution continues if the branch is taken.
In \autoref{fig:symbolic-runs}, $\h{br}<p>$ would be $11$ for $\BrIf{1}$ in \li{7} and $14$ for  $\BrIf{1}$ in \li{10}. 

Clause \ref{cls:brIfRule:0} describes a branch taken.
If a branch is not taken, we distinguish different cases:
\ref{cls:brIfRule:2} For forward branches, we handle the changes to the context just like in \ref{cls:brIfRule:0}.
In the case of not taking a backward branch (via \ref{cls:brIfRule:4}), which corresponds to exiting a loop (and therefore a join point), we continue differently, as we might be able to lower the context (see below).
\ref{cls:brIfRule:3} If the context and conditional labels are $\L$, then we continue without changing the context. 
\ref{cls:brIfPseudoRule:0} raises the context of an abstract configuration to $\H$ without progressing the program counter if we \emph{do not} execute a branch to the start of a loop conditioned on an $\H$-labeled value.
The theoretical description of the analysis includes this rule, as otherwise \rniarcns would be missing. 
In the actual implementation, however, does not include it as it is subsumed by \ref{cls:brIfRule:4} in all cases of practical interest.

\subsection{Scope Extensions}
\label{sec:scope_extension}
Similar to previous work~\cite{DBLP:conf/sas/BastysASS22}, we make use of \WA's syntactical structure to determine where implicit flows might happen.
To illustrate \WA's peculiarities, consider again \autoref{fig:symbolic-runs}.
When executing a branching instruction conditioned on an $\H$-labeled value, such as \BrIf{1} in \li{7}, the context has to be tainted until the corresponding \End.
We say that the scope of the \BrIf{k} extends from 7 to 11 (similarly, we say that the scope of $\hctx[7]$ extends to 11).
This region is marked in blue in \autoref{fig:symbolic-runs}.
In our abstract configurations, this is expressed by being able to derive a fact $\h{ScopeExtend}[0][7,11]$ via clause~\ref{cls:brIfRule:1}, as the condition $x$ is labeled $\H$.
Column 1 and 3, for example converge indeed at \li{7} (and would be related, if \ccode{IS_PROD} in column 1 had $0_=$ as value, which would not change the execution).
Such convergences are handled by the rule for \End{} and \BrIf{} in clauses \ref{cls:brIfRule:4} and \ref{cls:endRule:1} by implying $\h{MStateToJoin}<P>$-facts.

The blue region, however, is too small to soundly verify against implicit flows, as witnessed by column 2:
Since it branches beyond \li{11} by executing \BrIf{1} in \li{10}, implicit flows could happen in columns 1 and 3 in \li{12} to \li{14}. 
The scope of \checkme{7} should therefore be extended to \checkme{14} (this is marked in red).
This is, again, accomplished via \ref{cls:brIfRule:1} by implying new \h{ScopeExtend}<P> predicates when a branching instruction is executed with an $\H$-labeled context, in this case $\h{ScopeExtend}[0][7,14]$.
Clauses \ref{cls:brIfRule:5} and \ref{cls:endRule:2} allow us to derive follow-up states with an $\H$-labeled context if the context label is already $\H$ and the scope of the context extends over the next program counter.
As long as we can derive an $\H$-labeled context whenever an implicit flow might happen, it is sound to lower the context indiscriminately at every possible join point (as done in \ref{cls:brIfRule:4} and \ref{cls:endRule:1}).
This part of our analysis does not only take into account the structure of the program but also its reachable values.
If, for example, due to an additional assumption or analyzing preceding code, we were able to deduce that \ccode{IS_PROD} was $0$ in all runs (i.e., all runs would look similar to columns 1 and 3), the scope of \BrIf{1} in \li{7} would not extend beyond 11 (which is more precise than previous approaches \cite{DBLP:conf/sas/BastysASS22}).

If a function is entered in an $\H$-labeled context, the context cannot be lowered while executing this function.
The context's program counter of such a function is set to $-1$.
Together with rule~\ref{cls:functionScopeExtendRule:0}, this ensures that we can always have $\H$ as context label in such a case. 

\subsection{Adjustment and Opportunistic Lowering of Taint Labels}
\label{sec:adjustment_opportunistic_lowering_flow_labels}
As mentioned in \autoref{sec:scope_extensions_and_implicit_flows}, lowering the context label requires us to adjust the taint labels in related executions and provides an opportunity to regain precision.
In \tool, this is done via the \h{MStateToJoin}<P> predicate, which just holds exactly the same information as \h{MState}<P> and is implied at points where a divergent control flow can converge again (see \autoref{sec:scope_extension}). 
For every rule that implies \h{MStateToJoin}<P>-facts, we instantiate clause~\ref{cls:joinRule:0}. 
If two \h{MStateToJoin}<P> facts that describe the same program counter and function ID and are related by having the $\L$-labeled initial values (this is written as $\h{lowEq}[\cdot,\cdot]$), a new \h{MState}<P> fact can be derived that has an $\L$-labeled context and whose values (in the stack, the globals, the memory cell, etc.) are \emph{joined} by the following operation.

\makefitcolumn{
\begin{smallalign*}
  \call{join}{x,y} & = \begin{cases}
    (\h{valueOf}[x], \H) & \text{ if } \h{valueOf}[x] \neq \h{valueOf}[y] \land (\h{labelOf}[x] \sqcup \h{labelOf}[y]) = \H  \\
    (\h{valueOf}[x], \L) & \text{ otherwise }
  \end{cases}
\end{smallalign*}
}

In our running example we cannot derive a scope that extends over the join point in \li{14}, which means that only joined executions will continue after this point.
As all related executions agree on \ccode{session}, the only possible value after joining them is $0_=$.

\subsection{Overapproximations}
\label{sec:overapproximations}

In some cases, namely in the presence of loops and recursive functions, SMT-solvers struggle to terminate, if the input programs become too complex.
For these cases, we can provide overapproximations of the bit vector components of the values we track.
This is done  by introducing a set of overapproximation relations 
${}_\textbf{fid}\mkern-3mu\overset{.}{\lesssim}\mkern-3mu{}_\textbf{pc}$ for the different components of the abstract configuration, namely the locals, the globals and the memory, and can be seen in clause \ref{cls:loopRule:0}.
Any relation that does not change the value (i.e., only possibly allows additional values) and does not lower the taint label is sound to use at this position (in particular, simple equality).
For our implementation, we determine which globals or locals might be modified in the loop that starts at $\textbf{pc}$ in function $\textbf{fid}$, respectively if the memory is modified at all in a static preanalysis step (this is possible since the index of accessed globals/locals is known at analysis time).
The overapproximation relation then is fulfilled if unmodified values are the same in both arguments, while potentially modified values only have to share the same label.
For the globals tuple, the overapproximation relation is formally defined as follows, the relations for locals and the memory are handled similarly:

\mkfit{
\begin{smallalign*}
  \horstOpAppoverApproximateLoopGlobals{}{\horstFreeVargt, \horstFreeVarngt} = \begin{cases}
    \horstACCESS{\horstFreeVargt}{i} = \horstACCESS{\horstFreeVarngt}{i}&\text{ if  }i\text{ cannot be modified} \\
    \horstOpApplabelOf{}{\horstACCESS{\horstFreeVargt}{i}} = \horstOpApplabelOf{}{\horstACCESS{\horstFreeVarngt}{i}}&\text{ otherwise} \\
  \end{cases}
\end{smallalign*}
}

We relegate finding more sophisticated overapproximations to future research, but mention here that the abstract interpretation concept of widening (see, e.g.,~\cite{RivalXavier2020Itsa}) is related.
Note that we still handle the ``back-edge'' of the loop: if, due to the overapproxmation the labels change within the loop, these changes are taken into account.

\FloatBarrier

\section{Proofs}
\label{sec:proofs}
\subsection{Proofs on $\Mnoninterference$}

To define our notion of noninterference, we first define a family of pair-wise step functions.
$\Mstep[][][s]$ describes a \emph{synchronous step}.
If we can take a step from $pp$ to $pp'$ starting with either $m_1$ or $m_2$ as the memory (while updating the memory to $m_1'$ and $m_2'$, respectively), we can also take a synchronous step.
\begin{align*}
&\Mstep[pp,m_1,m_2][][s] = (pp', m_1', m_2')\\
&         {} \iff    (pp', m_1') = \Mstep[(pp,m_1)] \land {} \\
&\phantom{{} \iff{}} (pp', m_2') = \Mstep[(pp,m_2)]
\end{align*}

\emph{Diverging steps}, defined by $\Mstep[][][d]$, also start and end in the same program position in each memory configuration but diverge inbetween:
Starting from $pp$, they have to disagree on the next program position.
In order to derive the next memory configuration, the may take arbitrarily many steps, as long as at least one step is taken in either of the two executions, and the set of intermediate program positions has no intersections.
\begin{align*}
&\Mstep[pp,m_1,m_2][][d] = (pp', m_1', m_2')\\
&         {} \iff    \exists \ppi_1\, \ppi_2\, k\, l\,.  \\
&\phantom{{} \iff{}} (\ppi_1, \_) = \Mstep[(pp,m_1)] \land {} \\
&\phantom{{} \iff{}} (\ppi_2, \_) = \Mstep[(pp,m_2)] \land \ppi_1 \neq \ppi_2 \land{} \\
&\phantom{{} \iff{}} (pp', m_1') = \Mstep[(pp, m_1)][k] \land {} \\
&\phantom{{} \iff{}} (pp', m_2') = \Mstep[(pp, m_2)][l] \land \Mfunc{max}[k,l] > 0 \land{} \\
&\phantom{{} \iff{}} \big(\forall k'\,l'\,\ppi_1'\,\ppi_2'\,.\, 0 < k' < k \land 0 \le l' \le l  \land{} \\
&\phantom{{} \iff{}}\quad (\ppi_1', \_) = \Mstep[(pp,m_1)][k'] \land {} \\
&\phantom{{} \iff{}}\quad (\ppi_2', \_) = \Mstep[(pp,m_2)][l'] \implies \ppi_1' \neq \ppi_2'\big) \land{} \\
&\phantom{{} \iff{}} \big(\forall k'\,l'\,\ppi_1'\,\ppi_2'\,.\, 0 \le k' \le k \land 0 < l' < l  \land{} \\
&\phantom{{} \iff{}}\quad (\ppi_1', \_) = \Mstep[(pp,m_1)][k'] \land {} \\
&\phantom{{} \iff{}}\quad (\ppi_2', \_) = \Mstep[(pp,m_2)][l'] \implies \ppi_1' \neq \ppi_2'\big)
\end{align*}

$\Mstep[][][s,d]$ describes either a synchronous or a diverging step. 
\begin{align*}
&\Mstep[pp,m_1,m_2][][s,d] = (pp', m_1', m_2')\\
&         {} \iff    (pp', m_1', m_2') = \Mstep[pp,m_1,m_2][][s] \lor {} \\
&\phantom{{} \iff{}} (pp', m_1', m_2') = \Mstep[pp,m_1,m_2][][d]
\end{align*}

\begin{proof}[Proof of \autoref{thm:nia_sound_memories}]
By \autoref{corollary:terminating_execution_pairwise_reachable} and \autoref{lemma:pairwise_step_weak_noninterference}.
\end{proof}

\begin{proof}[Proof of \autoref{thm:nia_sound_traces}]
If we have $\Tr{pp, m_1}\traceellneqg\Tr{pp, m_2}$, we can by the definition of $\traceelleqg[\cdot][\cdot]$ decompose $\Tr{pp, m_1}$ and $\Tr{pp, m_2}$ into (possibly empty) prefixes $s_1$ and $s_2$ with $s_1 \traceelleqg s_2$ and suffixes $r_1$ and $r_2$ for which one of the following conditions hold:
\begin{itemize}
  \item attacker-distinguishable event positions: $r_1 = (\ppi_1, \mi_1) \cdot r'_1$ such $r_2 = (\ppi_2, \mi_2) \cdot r'_2$ such that $\ppi_1 \neq \ppi_2$, $\Gamma_E(\ppi_1) \flowsto \ell$, and $\Gamma_E(\ppi_2) \flowsto \ell$; or
  \item attacker-distinguishable premature termination: $\big(r_1 = \epsilon \land (\exists \ppi\,\mi . (\ppi,\mi) \in r_2 \land \Gamma_E(\ppi) \flowsto \ell)\big) \lor \big(r_2 = \epsilon \land \allowbreak (\exists \ppi\,\mi . \allowbreak (\ppi,\mi) \in r_1 \land \Gamma_E(\ppi) \flowsto \ell) \big)$ 
  \item attacker-distinguishable memories attached to event: $r_1 = (\ppi, \mi_1) \cdot r'_1$ and $r_2 = (\ppi, \mi_2) \cdot r'_2$ such that $\ellneqg[\mi_1][\mi_2][\Gamma_\textit{ME}(\ppi)]$ and $\Gamma_E(\ppi) \flowsto \ell$.
\end{itemize}
We proceed by distinguishing these three cases.

\emph{Event Positions}. In this case either $\ppi_1$ or $\ppi_2$ (or both; w.l.o.g. $\ppi_1$) have to be reached by an application of $\Mstep$ that was has no ``synchronous'' counterpart (as $\Mstep[][][s]$ always generates the same program position in lock-step).
This means $\ppi_1$ is reached as an intermediate configuration within $\Mstep[][][d]$.
Thus, \autoref{lemma:diverging_steps_context} applies (we can derive an $\H$-labeled context).

\emph{Premature Termination}. This case is handled similarly: if one execution already terminated, and the other is still running, the last step will be a diverging step and \autoref{lemma:diverging_steps_context} applies (we can derive an $\H$-labeled context)..

\emph{Memories}. In this case, there exists a memory position $p$ with $\Gamma_\textit{ME}(\ppi)(p) \flowsto \ell$ and $\mi_1(p) \neq \mi_2(p)$.
By \autoref{lemma:equal_prefix_pairwise_execution_reachable} we can apply \autoref{lemma:pairwise_step_weak_noninterference}, which tells us that $p$ can be labeled $\H$.
\end{proof}

\begin{lemma}
\label{lemma:general_execution_pairwise_reachable}
Given two executions starting in $(pp, m_1)$ and $(pp, m_2)$ and ending at program counter $pp'$ with memories $m_1'$ and $m_2'$, and the $pp'$ is not reached in an intermediate configuration, we can be sure that $(pp', m_1', m_2')$ is contained in $\Mstep[pp,m_1,m_2][*][\{s,d\}]$.
\begin{align*}
&\forall pp\,pp'\,m_1\,m_2\,m_1'\,m_2'\,n\,l.\\
&\quad (pp', m_1') = \Mstep[pp, m_1][n] \land (pp', m_2') = \Mstep[pp, m_2][l] \\
&\quad (\nexists n'\,\mi_1 . (pp',\mi_1) = \Mstep[pp,m_1][n'] \land n' < n) \land {} \\
&\quad (\nexists l'\,\mi_2 . (pp',\mi_2) = \Mstep[pp,m_2][l'] \land l' < l)  \implies {}\\
&\quad\quad (pp', m_1', m_2') \in \Mstep[pp, m_1, m_2][*][s,d]
\end{align*}
\end{lemma}

\begin{proof}
By Noetherian induction on $k = n + l$.

$k = 0$. This case trivially holds by setting $pp' = pp$, $m_1' = m_1$, and $m_2' = m_2$.

$k' = k + 1$. Either the computation of $(pp', m_1')$ and $(pp', m_2')$ share an intermediate program position $\ppi$ that can be reached by $\Mstep[pp,m_1,m_2][*][\{s,d\}]$, or not.

\quad $\exists \ppi\,\mi_1\,\mi_2\,\dot{n}\,\dot{l} \,.\, (\ppi, \mi_1') = \Mstep[pp, m_1][\dot{n}] \land (\ppi, \mi_2') = \penalty-60000 \Mstep[pp, m_2][\dot{l}] \land 0 < \dot{n} + \dot{l} \le k$.
Let $(\ppi, \mi_1)$ and $(\ppi, \mi_2)$ be the first configurations at $\ppi$ derivable from $(pp, m_1)$ and $(pp, m_2)$, respectively. 
As they are the first configurations at $\ppi$ (i.e., there is no configuration at $\ppi$ derivable in fewer steps), we can apply the induction hypothesis to $(\ppi, \mi_1') = \Mstep[pp, m_1][\dot{n}]$ and $(\ppi, \mi_2') = \Mstep[pp, m_2][\dot{l}]$ and, therefore, have $(\ppi, \mi_1, \mi_2) \in \penalty-60000 \Mstep[pp, m_1, m_2][*][\{s,d\}]$.
We can also apply the induction hypothesis on the second part of the execution, $(pp',m_1') = \Mstep[\ppi, \mi_1][n - \dot{n}]$ and $(pp',m_2') = \Mstep[\ppi, \mi_2][l - \dot{l}]$ to derive $(pp', m_1', m_2') \in \penalty-100000 \Mstep[\ppi, \mi_1, \mi_2][*][\{s,d\}]$.
By combining these two facts, we can derive $(pp', m_1', m_2') \in \Mstep[pp, m_1, m_2][*][\{s,d\}]$ 

\quad $\nexists \ppi\,\mi_1\,\mi_2\,\dot{n}\,\dot{l} . (\ppi, \mi_1') = \Mstep[pp, m_1][\dot{n}] \land (\ppi, \mi_2') = \penalty-10000 \Mstep[pp, m_2][\dot{l}] \land 0 < \dot{n} + \dot{l} \le k$.
In this case, we can by the definition of $\Mstep[][][d]$ go from $(pp,m_1,m_2)$ to $(pp',m_1',m_2')$ in a single, divergent step.
\end{proof}

\begin{corollary}
\label{corollary:terminating_execution_pairwise_reachable}
Given two executions starting in $(pp, m_1)$ and $(pp, m_2)$ and terminating with memories $m_1'$ and $m_2'$, we can be sure that $(\terminatepc, m_1', m_2')$ is contained in $\Mstep[pp,m_1,m_2][*][\{s,d\}]$.
\begin{align*}
&\forall pp\,m_1\,m_2\,m_1'\,m_2'\,n\,l.\\
&\quad (\terminatepc, m_1') = \Mstep[pp, m_1][n] \land (\terminatepc, m_2') = \Mstep[pp, m_2][l] \implies {}\\
&\quad\quad (\terminatepc, m_1', m_2) \in \Mstep[pp, m_1, m_2][*][s,d]
\end{align*}
\end{corollary}

\begin{proof}
By definition, $\terminatepc$ cannot appear as intermediate program position in any execution.
This allows us to apply \autoref{lemma:general_execution_pairwise_reachable}.
\end{proof}

\begin{lemma}
\label{lemma:pairwise_step_weak_noninterference}
Given an attacker $\ell$ and a security policy $\Gamma$:
When starting two executions consisting of arbitrarily many synchronous and diverging steps, starting at program position $pp$ from two memories $m_1$ and $m_2$ that are equal w.r.t.\ $\ell$ and $\Gamma$ and ending at $pp'$ with the updated memories $m_1'$ and $m_2'$:
for every memory position $p$ that $m_1'$ and $m_2'$ disagree on, $\Mnoninterference$ can derive a labeling $\Lambda$ at $pp'$ for either $m_1'$ or $m_2'$ that labels $p$ $\H$ from corresponding label map of $\ell$ and $\Gamma$.

\begin{align*}	
&\forall \ell\,\Gamma\,pp\,pp'\,m_1\,m_2\,m_1'\,m_2'\,p.\\
&\quad(pp',m_1', m_2') \in \Mstep[pp, m_1, m_2][*][s,d] \land {}\\ 
&\quad\elleqg[m_1][m_2] \land m_1'(p) \ne m_2'(p) \\
&\quad\quad \implies \exists \ctxlabel\, m'\, \Lambda\,.\,  (\ctxlabel, pp', m', \Lambda) \in \Mnoninterference[\Lmap, pp, m_1, m_2] \land {} \\
&\quad\quad\quad\Lambda(p) = \H \land (m_1' = m' \lor m_2' = m')
\end{align*}
\end{lemma}
\begin{proof}
By \autoref{thm:nia_general_memory_soundness} and \autoref{lemma:weakening_soundness}
\end{proof}

\begin{theorem}[General Soundness of Noninterference w.r.t. memories] 
\label{thm:nia_general_memory_soundness}
If from two attacker-indistinguishable memories we can derive two memories that disagree on all memory positions in the set $P$, we can derive a labeling for one of the two resulting memories that labels the disagreeing position $\H$.
Additionally, if we can derive a labeling for one of the result memories with $\Mlctx$ as the context, we can also derive the same labeling for the other one and this labeling labels the disagreeing position $\H$. 

\begin{align*}	
&\forall \ell\,\Gamma\,pp\,pp'\,m_1\,m_2\,m_1'\,m_2'\,n.\\
&\quad(pp',m_1', m_2') = \Mstep[pp, m_1, m_2][n][s,d] \land {}\\ 
&\quad\elleqg[m_1][m_2] \implies {} \exists P\,\ctxlabel_L\,s\,j\,\Lambda_1\,\Lambda_2.\\
&\quad\quad P \supseteq \{p ~|~ m_1'(p) \ne m_2'(p)\} \land {} \\
&\quad\quad ((\ctxlabel_L,s,j), pp', m_1', \Lambda_1) \in \Mnoninterference[\Lmap, pp, m_1, m_2] \land {} \\
&\quad\quad ((\ctxlabel_L,s,j), pp', m_2', \Lambda_2) \in \Mnoninterference[\Lmap, pp, m_1, m_2] \land {} \\
&\quad\quad (\ctxlabel_L = \H \implies (\forall p . p \in P \iff \Lambda_1(p) \sqcup \Lambda_2(p) = \H)) \land {} \\
&\quad\quad (\ctxlabel_L = \L \implies (\forall p . p \in P \iff \Lambda_1(p) \sqcap \Lambda_2(p) = \H))
\end{align*}
\end{theorem}

\begin{proof}
The argument of soundness for noninterference works by the induction on the number of steps $n$ in $\Mstep[pp,m][n][\{a,d\}]$.

$n=0$.
We have $m_1 = m_1'$, $m_2 = m_2'$, and $pp = pp'$, since $\Mstep$ has not been applied.
This allows us to apply \autoref{lemma:initial_labeling}.

$n' = n + 1$. By induction hypothesis, the theorem holds for executions of length $n$ where $(pp', m_1', m_2') = \Mstep[pp, m_1, m_2][n][\{a,d\}]$.
We proceed by case distinction on the nature of the last step.

\quad $(pp'', m_1'', m_2'') = \Mstep[pp', m_1', m_2'][][d]$. By \autoref{lemma:diverging_steps}.

\quad $(pp'', m_1'', m_2'') = \Mstep[pp', m_1', m_2'][][s]$. By \autoref{lemma:synchronous_steps}.
\end{proof}

\begin{lemma}
\label{lemma:weakening_soundness}
If $\Mnoninterference$ can derive two labelings $\Lambda_1$ for memory $m_1'$ and $\Lambda_2$ for memory $m_2'$ and a context label $\ctxlabel_L$ such that in case that $\ctxlabel_L = \H$, either $\Lambda_1$ or $\Lambda_2$ label any position on which $m_1'$ and $m_2'$ differ $\H$, and in case that $\ctxlabel_L = \L$, both $\Lambda_1$ and $\Lambda_2$ label any position on which $m_1'$ and $m_2'$ differ $\H$, then $\Mnoninterference$ can also derive a labeling $\Lambda$ for either $m_1'$ or $m_2'$ for any $p$ at which $m_1'(p) \neq m_2'(p)$ such that $\Lambda(p) = \H$.
\begin{align*}
&\forall P\,\ctxlabel_L\,s\,j\,\Lambda_1\,\Lambda_2.\\
&\quad P \supseteq \{p ~|~ m_1'(p) \ne m_2'(p)\} \land {} \\
&\quad ((\ctxlabel_L,s,j), pp', m_1', \Lambda_1) \in \Mnoninterference[\Lmap, pp, m_1, m_2] \land {} \\
&\quad ((\ctxlabel_L,s,j), pp', m_2', \Lambda_2) \in \Mnoninterference[\Lmap, pp, m_1, m_2] \land {} \\
&\quad (\ctxlabel_L = \H \implies (\forall p . p \in P \iff \Lambda_1(p) \sqcup \Lambda_2(p) = \H)) \land {} \\
&\quad (\ctxlabel_L = \L \implies (\forall p . p \in P \iff \Lambda_1(p) \sqcap \Lambda_2(p) = \H)) \\
&\quad {} \implies \forall p. \big( m_1'(p) \ne m_2'(p) \\
&\quad\quad \implies \exists \ctxlabel\, m'\, \Lambda\,.\,  (\ctxlabel, pp', m', \Lambda) \in \Mnoninterference[\Lmap, pp, m_1, m_2] \land {} \\
&\quad\quad\quad\Lambda(p) = \H \land (m_1' = m' \lor m_2' = m') \big)
\end{align*}
\end{lemma}

\begin{proof}
By the hypothesis we have that any $p$ such that $m_1'(p) \neq m_2'(p)$ is in $P$.
We distinguish the $\ctxlabel_L = \L$ and $\ctxlabel_L = \H$. 

\quad $\ctxlabel_L = \L$. In this case, we can derive $\Lambda_1$ and $\Lambda_2$ such that all memory positions in $P$ are labeled $\H$.
Therefore, we can satisfy the conclusion by setting $\Lambda$ to either one of them and $m'$ to either $m_1'$ or $m_2'$. 

\quad $\ctxlabel_L = \H$. In this case, for any $p$ in $P$ either $\Lambda_1(p) = \H$ or $\Lambda_2(p) = \H$.
 Therefore, we can satisfy the conclusion by setting $\Lambda$ to the one that labels $p$ $\H$ and $m'$ to the corresponding memory (either $m_1'$ or $m_2'$).
\end{proof}

\begin{lemma}[Initial Labeling]
\label{lemma:initial_labeling}
\begin{align*}
&\forall \ell\,\Gamma\,pp\,m_1\,m_2\,.\, \elleqg[m_1][m_2] \implies \big( \exists P\,\ctxlabel_L\,s\,j\,\Lambda_1\,\Lambda_2 . \\ 
&\quad P \supseteq \{p ~|~ m_1(p) \ne m_2(p)\} \land {} \\
&\quad ((\ctxlabel_L,s,j), pp, m_1, \Lambda_1) \in \Mnoninterference[\Lmap, pp, m_1, m_2] \land {} \\
&\quad ((\ctxlabel_L,s,j), pp, m_2, \Lambda_2) \in \Mnoninterference[\Lmap, pp, m_1, m_2] \land {} \\
&\quad (\ctxlabel_L = \H \implies (\forall p . p \in P \iff \Lambda_1(p) \sqcup \Lambda_2(p) = \H)) \land {} \\
&\quad (\ctxlabel_L = \L \implies (\forall p . p \in P \iff \Lambda_1(p) \sqcap \Lambda_2(p) = \H)) \big)
\end{align*}
\end{lemma}

\begin{proof}
We have $\elleqg[m_1][m_2]$, meaning that $m_1$ and $m_2$ agree on all $p$ with $\Gamma(p) \flowsto \ell$.
By \textsc{Init}, $\Mnoninterference$ can return $(\Mlctx,pp,m_1,\Lmap)$ and $(\Mlctx,pp,m_2,\Lmap)$.
$\Lmap$ is constructed such that all $p$ with $\Gamma(p) \flowsto \ell$ are labeled $\L$, which means that all potentially disagreeing memory positions are labeled $\H$.
\end{proof}

\begin{lemma}[Diverging Steps have $\H$-labeled context]
\label{lemma:diverging_steps_context}
For all intermediate configurations $(\ppi_1, \mi_1)$ and $(\ppi_2, \mi_2)$ that appear as results of $\Mstep$ in and application of $\Mstep[][][d]$ $\Mnoninterference$ can derive an $\H$-labeled context. 

\mkfit{
\begin{align*}	
&\forall \ell\,\Gamma\,pp\,pp'\,pp''\,m_1\,m_2\,m_1'\,m_2'\,m_1''\,m_2''\,n.\\
&\quad(pp'',m_1'', m_2'') = \Mstep[pp', m_1', m_2'][][d] \land {} \\
&\quad(pp'',m_1'') = \Mstep[pp', m_1'][k] \land  (pp'',m_2'') = \Mstep[pp', m_2'][l] \land {} \\
&\quad(\nexists k' . k' < k \land (pp',m_1') = \Mstep[pp, m_1][k']) \land {} \\
&\quad(\nexists l' . l' < l \land (pp',m_2') = \Mstep[pp, m_2][l']) \land {}\\
&\quad\big( \exists P\,\ctxlabel_L\,s\,j\,\Lambda_1\,\Lambda_2 . \\ 
&\quad\quad P \supseteq \{p ~|~ m_1'(p) \ne m_2'(p)\} \land {} \\
&\quad\quad ((\ctxlabel_L,s,j), pp', m_1', \Lambda_1) \in \Mnoninterference[\Lmap, pp, m_1, m_2] \land {} \\
&\quad\quad ((\ctxlabel_L,s,j), pp', m_2', \Lambda_2) \in \Mnoninterference[\Lmap, pp, m_1, m_2] \land {} \\
&\quad\quad (\ctxlabel_L = \H \implies (\forall p . p \in P \iff \Lambda_1(p) \sqcup \Lambda_2(p) = \H)) \land {} \\
&\quad\quad (\ctxlabel_L = \L \implies (\forall p . p \in P \iff \Lambda_1(p) \sqcap \Lambda_2(p) = \H)) \big) \\
&\quad                {} \implies  \big( \forall k' . 0 < k' \le k \implies \exists s'\,j'\,\Lambda_1.\\
&\quad\quad \hphantom{{} \implies {}}(\ppi_1, \mi_1') = \Mstep[pp, m_1][k'] \land {} \\
&\quad\quad \hphantom{{} \implies {}}((\H,s',j'), \ppi, \mi_1, \Lambda_1) \in \Mnoninterference[\Lmap, pp, m_1, m_2] \big) \land {} \\
&\quad      \hphantom{{} \implies {}}\big( \forall l' . 0 < l' \le l \implies \exists s'\,j'\,\Lambda_2.\\
&\quad\quad \hphantom{{} \implies {}}(\ppi_2, \mi_2') = \Mstep[pp, m_2][l'] \land {} \\
&\quad\quad \hphantom{{} \implies {}}((\H,s',j'), \ppi, \mi_2, \Lambda_2) \in \Mnoninterference[\Lmap, pp, m_1, m_2] \big)
\end{align*}
}
\end{lemma}

\begin{proof}
Either $\Mnoninterference$ returns an $\H$-labeled context for $m_1'$ and $m_2'$ or not.

\quad $\ctxlabel_L = \H$.
In this case, we can be sure that (by the soundness condition of $\Mjoinat$ and the fact that asynchronous steps have no shared program positions between the initial and the final program position) \textsc{Join} is not applicable for any of the intermediate configurations $(\ppi_1, \mi_1)$ and $(\ppi_2, \mi_2)$.
Indeed, the only applicable rule is \rniaprop, which does never lower the context label.

\quad $\ctxlabel_L = \L$.
In this case, by the definition of $\Mstep[][][d]$ and the soundness condition on $\Mccf$, we know that $\Mccf[pp, pp', m_1, m_2]$ is true.
The soundness condition on $\Mupdate$ necessitates that if the next program positions of $(pp',m_1')$ and $(pp',m_2')$ differ, that the memory position(s) that causes this divergence has to be included in the $I_{pp}$-component of the tuple returned by $\Mupdate$.
By the hypothesis, differing values in $m_1'$ and $m_2'$ can be labeled $\H$, which means that $\bigsqcup_{i \in I_{pp}}i = \H$, and that therefore the first rule applied by $\Mnoninterference$ is either \rniarcst or \rniarcns.
All following applications will be, by the same argument as in the former case, \rniaprop, which does not lower the context label.
\end{proof}

\begin{lemma}[Diverging Steps]
\label{lemma:diverging_steps}
\ \\
\mkfit{
\begin{align*}	
&\forall \ell\,\Gamma\,pp\,pp'\,pp''\,m_1\,m_2\,m_1'\,m_2'\,m_1''\,m_2''\,n.\\
&\quad(pp'',m_1'', m_2'') = \Mstep[pp', m_1', m_2'][][d] \land \big( \exists P\,\ctxlabel_L\,s\,j\,\Lambda_1\,\Lambda_2 . \\ 
&\quad\quad P' \supseteq \{p ~|~ m_1'(p) \ne m_2'(p)\} \land {} \\
&\quad\quad ((\ctxlabel_L,s,j), pp', m_1', \Lambda_1) \in \Mnoninterference[\Lmap, pp, m_1, m_2] \land {} \\
&\quad\quad ((\ctxlabel_L,s,j), pp', m_2', \Lambda_2) \in \Mnoninterference[\Lmap, pp, m_1, m_2] \land {} \\
&\quad\quad (\ctxlabel_L = \H \implies (\forall p . p \in P \iff \Lambda_1(p) \sqcup \Lambda_2(p) = \H)) \land {} \\
&\quad\quad (\ctxlabel_L = \L \implies (\forall p . p \in P \iff \Lambda_1(p) \sqcap \Lambda_2(p) = \H)) \big) \\
&\quad \implies {} \exists P'\,\ctxlabel_L'\,s'\,j'\,\Lambda_1'\,\Lambda_2'.\\
&\quad\quad P' \supseteq \{p ~|~ m_1''(p) \ne m_2''(p)\} \land {} \\
&\quad\quad ((\ctxlabel_L',s',j'), pp'', m_1'', \Lambda_1') \in \Mnoninterference[\Lmap, pp, m_1, m_2] \land {} \\
&\quad\quad ((\ctxlabel_L',s',j'), pp'', m_2'', \Lambda_2') \in \Mnoninterference[\Lmap, pp, m_1, m_2] \land {} \\
&\quad\quad (\ctxlabel_L' = \H \implies (\forall p . p \in P' \iff \Lambda_1(p) \sqcup \Lambda_2(p) = \H)) \land {} \\
&\quad\quad (\ctxlabel_L' = \L \implies (\forall p . p \in P' \iff \Lambda_1(p) \sqcap \Lambda_2(p) = \H))
\end{align*}
}

\end{lemma}

\begin{proof}
By the soundness condition of $\Mupdate$, there is an update that correctly generates all intermediate configurations, in particular $\Mnoninterference$ can generate $(pp', m_1', m_2')$ from $(pp, m_1, m_2)$. 

Either $\Mnoninterference$ returns an $\H$-labeled context for $m_1'$ and $m_2'$ or not.

\quad $\ctxlabel_L = \H$.
By \autoref{lemma:diverging_steps_context} we know that we can derive an $\H$-labeled context for all intermediate configurations and that the only rule that is intermediately applied is \rniaprop.
If there is a memory position $p$ for which $m_1''(p) \neq m_2''(p)$, then either we already had $m_1'(p) \neq m_2'(p)$ (in that case we have $\Lambda_1(p) \sqcup \Lambda_2(p) = \H$ by the hypothesis) or we only have $m_1''(p) \neq m_2''(p)$.
In the latter case, $p$'s value (in either the first or the second execution) has been changed in a configuration for which an $\H$ context label is derivable.
This means, that $p$'s label in the returned label map is $\H$, as the least upper bound of all values changed by \rniaprop is the context label (which is $\H$). 
Since \rniaprop does not change does not change the context object at all, $s = s'$ and $j = j'$. 

\quad $\ctxlabel_L = \L$.
In this case, by the definition of $\Mstep[][][d]$ and the soundness condition on $\Mccf$, we know that $\Mccf[pp, pp', m_1, m_2]$ is true.
The soundness condition on $\Mupdate$ necessitates that if the next program positions of $(pp',m_1')$ and $(pp',m_2')$ differ, that the memory position(s) that causes this divergence has to be included in the $I_{pp}$-component of the tuple returned by $\Mupdate$.
By the hypothesis, differing values in $m_1'$ and $m_2'$ can be labeled $\H$, which means that $\bigsqcup_{i \in I_{pp}}i = \H$, and that therefore the first rule applied by $\Mnoninterference$ is either \rniarcst or \rniarcns.
If the next transition that the executions starting at $(pp',m_1')$ and $(pp',m_2')$ share is already the immediately next one (by the definition of $\Mstep[][][d]$ this can only happen if one of the two execution does not do a step), we apply \rniarcns to raise the context label without applying any other changes.
In any other case, the first rule that $\Mnoninterference$ applies will be \rniarcst, which sets the context label to $\H$.
\rniarcst propagates the taints of all updated memory positions in the same way as \rniaprop, which means that all memory positions the two follow-up configurations disagree on will be labeled $\H$. 
By the same argument as above, all subsequent rule applications will be \rniaprop and all memory positions $m_1''$ and $m_2''$ disagree on can be labeled $\H$.
Similarly, we observe that the context object generated in \rniarcst or \rniarcns is the same in any case and not changed by the subsequent applications of \rniaprop.
Therefore, the same context object can be derived for $m_1''$ and $m_2''$. 
\end{proof}

\begin{lemma}[Synchronous Steps]
\label{lemma:synchronous_steps}
\ \\
\mkfit{
\begin{align*}	
&\forall \ell\,\Gamma\,pp\,pp'\,pp''\,m_1\,m_2\,m_1'\,m_2'\,m_1''\,m_2''\,n.\\
&\quad(pp'',m_1'', m_2'') = \Mstep[pp', m_1', m_2'][][s] \land \big( \exists P\,\ctxlabel_L\,s\,j\,\Lambda_1\,\Lambda_2 . \\ 
&\quad\quad P' \supseteq \{p ~|~ m_1'(p) \ne m_2'(p)\} \land {} \\
&\quad\quad ((\ctxlabel_L,s,j), pp', m_1', \Lambda_1) \in \Mnoninterference[\Lmap, pp, m_1, m_2] \land {} \\
&\quad\quad ((\ctxlabel_L,s,j), pp', m_2', \Lambda_2) \in \Mnoninterference[\Lmap, pp, m_1, m_2] \land {} \\
&\quad\quad (\ctxlabel_L = \H \implies (\forall p . p \in P \iff \Lambda_1(p) \sqcup \Lambda_2(p) = \H)) \land {} \\
&\quad\quad (\ctxlabel_L = \L \implies (\forall p . p \in P \iff \Lambda_1(p) \sqcap \Lambda_2(p) = \H)) \big) \\
&\quad \implies {} \exists P'\,\ctxlabel_L'\,s'\,j'\,\Lambda_1'\,\Lambda_2'.\\
&\quad\quad P' \supseteq \{p ~|~ m_1''(p) \ne m_2''(p)\} \land {} \\
&\quad\quad ((\ctxlabel_L',s',j'), pp'', m_1'', \Lambda_1') \in \Mnoninterference[\Lmap, pp, m_1, m_2] \land {} \\
&\quad\quad ((\ctxlabel_L',s',j'), pp'', m_2'', \Lambda_2') \in \Mnoninterference[\Lmap, pp, m_1, m_2] \land {} \\
&\quad\quad (\ctxlabel_L' = \H \implies (\forall p . p \in P' \iff \Lambda_1(p) \sqcup \Lambda_2(p) = \H)) \land {} \\
&\quad\quad (\ctxlabel_L' = \L \implies (\forall p . p \in P' \iff \Lambda_1(p) \sqcap \Lambda_2(p) = \H))
\end{align*}
}

\end{lemma}

\begin{proof}
By the soundness condition of $\Mupdate$, there is an update that correctly generates all the next generations.
When talking about the components of the applied update ($I_{pp}, U, \Lambda$ etc.), we assume these components to fulfill $\Mupdate$'s soundness condition.
To distinguish between a variable $X$ that is bound in a rule application that concerns a configuration derived from $m_1$, we will call it $X_1$ and conversely $X_2$ for a configuration derived from $m_2$.

Either $\Mnoninterference$ returns an $\H$-labeled context for $m_1'$ and $m_2'$ or not.

\quad$\ctxlabel_L = \L$.
In latter case, the only applicable rules are \rniaprop and \rniarcst (if the control flow does not diverge, but the label map or the $I_{pp}$-component returned by $\Mupdate$ are imprecise).
As the calculation of the new label map is the same in both cases, the argument is the same.
If $m_1''(p) \neq m_2''(p)$ then we either have $m_1'(p) \neq m_2'(p)$ which means that we can apply the induction hypothesis (as neither \rniaprop nor \rniarcst ever lower the flow labels in returned in $\Lambda$) or the value at $p$ was modified in the last step in either of the two executions.
In this case (by the soundness condition of $\Mupdate$) an update for $p$ is included in both $U_1$ and $U_2$. 
At least one of the inputs used to calculate the new values at $m_1''(p)$ (respectively $m_2''(p)$) has to differ between $m_1'$ and $m_2'$.
By the soundness condition of $\Mupdate$, we know that the input sets for $p$, $I_{U_1}$ and $I_{U_2}$, contain all memory positions that might cause the value at $p$ in $m_1''$ or $m_2''$ to change, so in particular they contain the aforementioned position(s) where $m_1'$ and $m_2'$ differ.
Because $\ctxlabel_L = \L$, we know (by the hypothesis) that all values $m_1'$ and $m_2'$ disagree on are labeled $\H$ in both $\Lambda_1$ and $\Lambda_2$, indeed that $\Lambda_1 = \Lambda_2$.
This means that $p$ is labeled $\H$ in both $\Lambda_1'$ and $\Lambda_2'$ (because its label equals $\bigsqcup_{i \in I_{U_1}} \Lambda_1(i) = \bigsqcup_{i \in I_{U_2}} \Lambda_2(i)$), which is (as $\ctxlabel_L'$ will be $\L$ for both $m_1''$ and $m_2''$) exactly what is to be proved.
If the applied rule was \rniaprop, we have $s = s'$ and $j = j'$.
Otherwise, if the applied rule was \rniarcst, we update the context object in the same way for both executions.

\quad $\ctxlabel_L = \H$.
In this case, the rules that can be applied are either \rniaprop or \rniajoin.
In case of \rniaprop, $\ctxlabel_L'$ will also be $\H$, so we only have to prove that any value that differs between $m_1''$ and $m_2''$ has to be labeled $\H$ in one of $\Lambda_1$ and $\Lambda_2$.
If $m_1''(p) \neq m_2''(p)$ then we either have $m_1'(p) \neq m_2'(p)$ which means that we can apply the induction hypothesis (as \rniaprop never lowers the flow labels in returned in $\Lambda$) or the value at $p$ was modified in the last step in either of the two executions.
As all the labels of all memory positions updated by \rniaprop are raised to at least $\ctxlabel_L$, which is $\H$, any changed position will indeed be labeled $\H$. 
As before, the context object is not changed \rniaprop.

In case of \rniajoin, $\ctxlabel_L'$ will be $\L$, so we have to prove that all positions that differ between $m_1''$ and $m_2''$ have to be labeled $\H$ in both $\Lambda_1'$ and $\Lambda_2'$.
From the induction hypothesis, we know that all values that differ in $m_1'$ and $m_2'$ are labeled $\H$ in either $\Lambda_1$ or $\Lambda_2$.
For each memory position $p$, we distinguish two different cases: either $p$ is included in $U_1$ and $U_2$ (by the soundness condition of $\Mupdate$, we know that each $p$ is updated in either both or neither).
If $p$ is not updated (i.e., $m_1''(p) = m_1'(p)$ and $m_2''(p) = m_2'(p)$), and $m_1'(p) \neq m_2'(p)$, we know, by the hypothesis, that $\Lambda_1(p) \sqcup \Lambda_2(p) = \H$, which means that in $\Lambda'$ $p$ will be labeled $\H$.
If $p$ is updated to different values (i.e., $m_1''(p) = v_1$ and $m_2''(p) = v_2$ with $v_1 \neq v_2$), we know by the hypothesis and the soundness condition of $\Mupdate$ that the upper bound of all labels in $\Lambda_1$ for positions in $I_{U_1}$ and in $\Lambda_2$ for positions in $I_{U_2}$ will be $\H$.
Therefore, $p$'s label in $\Lambda'$ will be $\H$.
The context object for both $m_1''$ and $m_2''$ will be $\Mlctx$. 
\end{proof}

\begin{lemma}
\label{lemma:equal_prefix_pairwise_execution_reachable}

If two traces starting at $(pp, m_1)$ and $(pp, m_2)$, respectively, arrive at intermediate configurations $(\ppi, \mi_1)$ and $(\ppi, \mi_2)$ and the respective sequence of events before arriving at these configurations are attacker-indistinguishable, and $\ppi$ is attacker-observable, then $(\ppi, \mi_1, \mi_2)$ can be derived from $(pp, m_1, m_2)$ via $\Mstep[][][\{s,d\}]$.
\begin{align*}
&\forall pp\,m_1\,m_2\,\mi_1\,\mi_2\,\ppi\,r_1\,r_2 \\
&\quad \Tr{pp, m_1} = s_1 \cdot (\ppi, \mi_1) \cdot r_1 \land \Gamma_E(\ppi) \flowsto \ell \land {} \\
&\quad \Tr{pp, m_2} = s_2 \cdot (\ppi, \mi_2) \cdot r_2 \land s_1 \traceelleqg s_2 \\
&\quad \implies (\ppi, \mi_1, \mi_2) \in \Mstep[pp, m_1, m_2][*][\{s,d\}] 
\end{align*}
\end{lemma}

\begin{proof}
By induction on the number of attacker-observable events in $s_1$ and $s_2$, $k = \sum_{(\ppi,e) \in s_1} \begin{cases} 1 \text{ if } \Gamma_E(\ppi) \flowsto \ell \\ 0 \text{ otherwise} \end{cases} = \penalty-100000 \sum_{(\ppi,e) \in s_2} \begin{cases} 1 \text{ if } \Gamma_E(\ppi) \flowsto \ell \\ 0 \text{ otherwise} \end{cases}$.

$k = 0$. Since neither $s_1$ nor $s_2$ contain any attacker-observable events, they in particular contain no event at $\ppi$ (since we have $\Gamma_E(\ppi)$).
This allows us to apply \autoref{lemma:general_execution_pairwise_reachable}. 

$k' = k+1$. Let $(\ppi', \mi_1')$ and $(\ppi', \mi_2')$ be the last attacker-observable events before reaching $(\ppi, \mi_1)$ and $(\ppi, \mi_2)$.
This means that we can (for $i \in \{1,2,\}$) decompose $\Tr{pp, m_i}$ like this, \penalty-10000 $\Tr{pp, m_i} = s_i \cdot (\ppi', \mi_i') \cdot r_i' \cdot (\mi, \mi_i) \cdot r_i$, such that $s_i$ contains $k$ attacker-observable events and $r_i'$ contains no attacker-observable events.
By induction hypothesis, we then have $(\ppi', \mi_1', \mi_2') \in \Mstep[pp, m_1, m_2][*][\{s,d\}]$.
Since $r_i'$ contains no attacker-observable \allowbreak events, it in particular does not contain an event at $\ppi$.
This means, we can again apply \autoref{lemma:general_execution_pairwise_reachable}. 
\end{proof}

\subsection{Proofs on \tool}
 \label{sec:proofs-on-wanilla}

\begin{figure*}[t]
  \newcommand{\horstFreeVarctx}{\addedInWanilla{\textit{ctx}}}
  \let\oldhorstOpAppraiseCtxTo\horstOpAppraiseCtxTo
  \renewcommand{\horstOpAppraiseCtxTo}[2]{\addedInWanilla{\oldhorstOpAppraiseCtxTo{#1}{#2}}}
  \let\oldhorstOpAppraiseTo\horstOpAppraiseTo
  \renewcommand{\horstOpAppraiseTo}[2]{\addedInWanilla{\oldhorstOpAppraiseTo{#1}{#2}}}
  \renewcommand{\horstConstructorAppCtx}[1]{\addedInWanilla{\horstConstructorApp{Ctx}{#1}}}
  \renewcommand{\horstPredAppScopeExtend}[2]{\addedInWanilla{\horstPredApp{ScopeExtend}{#1}{#2}}}
  \let\oldhorstOpApplabelOf\horstOpApplabelOf
  \renewcommand{\horstOpApplabelOf}[2]{\addedInWanilla{\oldhorstOpApplabelOf{#1}{#2}}}
  \let\oldhorstOpApplabelOfCtx\horstOpApplabelOfCtx
  \renewcommand{\horstOpApplabelOfCtx}[2]{\addedInWanilla{\oldhorstOpApplabelOfCtx{#1}{#2}}}
  \let\oldhorstConstructorAppIllegal\horstConstructorAppIllegal
  \renewcommand{\horstConstructorAppIllegal}[1]{\addedInWanilla{\oldhorstConstructorAppIllegal{#1}}}
  \let\oldhorstConstructorAppLegal\horstConstructorAppLegal
  \renewcommand{\horstConstructorAppLegal}[1]{\addedInWanilla{\oldhorstConstructorAppLegal{#1}}}
  \let\oldhorstOpAppoverApproximateLoopGlobals\horstOpAppoverApproximateLoopGlobals 
  \renewcommand{\horstOpAppoverApproximateLoopGlobals}[2]{\addedInWanilla{\oldhorstOpAppoverApproximateLoopGlobals{#1}{#2}}}
  \let\oldhorstOpAppoverApproximateLoopLocals\horstOpAppoverApproximateLoopLocals 
  \renewcommand{\horstOpAppoverApproximateLoopLocals}[2]{\addedInWanilla{\oldhorstOpAppoverApproximateLoopLocals{#1}{#2}}}
  \let\oldhorstOpAppoverApproximateLoopMemory\horstOpAppoverApproximateLoopMemory 
  \renewcommand{\horstOpAppoverApproximateLoopMemory}[2]{\addedInWanilla{\oldhorstOpAppoverApproximateLoopMemory{#1}{#2}}}
  
  \center
  \begin{minipage}{\textwidth}
  \ruleGroup{brRule}{\Br{}}
  \ruleGroup{brTableRule}{\BrTable{}}
  \ruleGroup{brTableDefaultRule}{\BrTable{} (default case)}
  \ruleGroup{ifThenElseRule}{\If{}~\Then{}~\Else{}}
  \end{minipage}
  \caption{\tool's abstract rules for $\Br{}$, $\BrTable{}$, and $\If{}~\Then{}~\Else{}$.}
  \label{fig:wanilla_rules_brtable}
\end{figure*}

\begin{figure*}[t]
  \newcommand{\horstFreeVarctx}{\addedInWanilla{\textit{ctx}}}
  \let\oldhorstOpAppraiseCtxTo\horstOpAppraiseCtxTo
  \renewcommand{\horstOpAppraiseCtxTo}[2]{\addedInWanilla{\oldhorstOpAppraiseCtxTo{#1}{#2}}}
  \let\oldhorstOpAppraiseTo\horstOpAppraiseTo
  \renewcommand{\horstOpAppraiseTo}[2]{\addedInWanilla{\oldhorstOpAppraiseTo{#1}{#2}}}
  \renewcommand{\horstConstructorAppCtx}[1]{\addedInWanilla{\horstConstructorApp{Ctx}{#1}}}
  \renewcommand{\horstPredAppScopeExtend}[2]{\addedInWanilla{\horstPredApp{ScopeExtend}{#1}{#2}}}
  \let\oldhorstOpApplabelOf\horstOpApplabelOf
  \renewcommand{\horstOpApplabelOf}[2]{\addedInWanilla{\oldhorstOpApplabelOf{#1}{#2}}}
  \let\oldhorstOpAppperspectiveOfCtx\horstOpAppperspectiveOfCtx
  \renewcommand{\horstOpAppperspectiveOfCtx}[2]{\addedInWanilla{\oldhorstOpAppperspectiveOfCtx{#1}{#2}}}
  \let\oldhorstOpApplabelOfCtx\horstOpApplabelOfCtx
  \renewcommand{\horstOpApplabelOfCtx}[2]{\addedInWanilla{\oldhorstOpApplabelOfCtx{#1}{#2}}}
  \let\oldhorstConstructorAppIllegal\horstConstructorAppIllegal
  \renewcommand{\horstConstructorAppIllegal}[1]{\addedInWanilla{\oldhorstConstructorAppIllegal{#1}}}
  \let\oldhorstConstructorAppLegal\horstConstructorAppLegal
  \renewcommand{\horstConstructorAppLegal}[1]{\addedInWanilla{\oldhorstConstructorAppLegal{#1}}}
  \let\oldhorstOpAppoverApproximateLoopGlobals\horstOpAppoverApproximateLoopGlobals 
  \renewcommand{\horstOpAppoverApproximateLoopGlobals}[2]{\addedInWanilla{\oldhorstOpAppoverApproximateLoopGlobals{#1}{#2}}}
  \let\oldhorstOpAppoverApproximateLoopLocals\horstOpAppoverApproximateLoopLocals 
  \renewcommand{\horstOpAppoverApproximateLoopLocals}[2]{\addedInWanilla{\oldhorstOpAppoverApproximateLoopLocals{#1}{#2}}}
  \let\oldhorstOpAppoverApproximateLoopMemory\horstOpAppoverApproximateLoopMemory 
  \renewcommand{\horstOpAppoverApproximateLoopMemory}[2]{\addedInWanilla{\oldhorstOpAppoverApproximateLoopMemory{#1}{#2}}}

  \let\oldhorstOpAppoverApproximateCallGlobals\horstOpAppoverApproximateCallGlobals 
  \renewcommand{\horstOpAppoverApproximateCallGlobals}[2]{\addedInWanilla{\oldhorstOpAppoverApproximateCallGlobals{#1}{#2}}}
  \let\oldhorstOpAppoverApproximateCallArguments\horstOpAppoverApproximateCallArguments
  \renewcommand{\horstOpAppoverApproximateCallArguments}[2]{\addedInWanilla{\oldhorstOpAppoverApproximateCallArguments{#1}{#2}}}
  \let\oldhorstOpAppoverApproximateCallMemory\horstOpAppoverApproximateCallMemory 
  \renewcommand{\horstOpAppoverApproximateCallMemory}[2]{\addedInWanilla{\oldhorstOpAppoverApproximateCallMemory{#1}{#2}}}

  \let\oldhorstOpAppflub\horstOpAppflub
  \renewcommand{\horstOpAppflub}[2]{\addedInWanilla{\oldhorstOpAppflub{#1}{#2}}}
  
  \center
  \begin{minipage}{\textwidth}
  \ruleGroup{callRule}{\Call{}}
  \ruleGroup{callIndirectRule}{\CallIndirect{}}
  \ruleGroup{callIndirectHavokRule}{\CallIndirect{} (when table has been changed)}
  \end{minipage}
  \caption{\tool's abstract rules for $\Call{}$ and $\CallIndirect{}$.}
  \label{fig:wanilla_rules_brtable}
\end{figure*}

\begin{figure*}[t]
  \newcommand{\horstFreeVarctx}{\addedInWanilla{\textit{ctx}}}
  \let\oldhorstOpAppraiseCtxTo\horstOpAppraiseCtxTo
  \renewcommand{\horstOpAppraiseCtxTo}[2]{\addedInWanilla{\oldhorstOpAppraiseCtxTo{#1}{#2}}}
  \let\oldhorstOpAppraiseTo\horstOpAppraiseTo
  \renewcommand{\horstOpAppraiseTo}[2]{\addedInWanilla{\oldhorstOpAppraiseTo{#1}{#2}}}
  \renewcommand{\horstConstructorAppCtx}[1]{\addedInWanilla{\horstConstructorApp{Ctx}{#1}}}
  \renewcommand{\horstPredAppScopeExtend}[2]{\addedInWanilla{\horstPredApp{ScopeExtend}{#1}{#2}}}
  \let\oldhorstOpApplabelOf\horstOpApplabelOf
  \renewcommand{\horstOpApplabelOf}[2]{\addedInWanilla{\oldhorstOpApplabelOf{#1}{#2}}}
  \let\oldhorstOpAppperspectiveOfCtx\horstOpAppperspectiveOfCtx
  \renewcommand{\horstOpAppperspectiveOfCtx}[2]{\addedInWanilla{\oldhorstOpAppperspectiveOfCtx{#1}{#2}}}
  \let\oldhorstOpApplabelOfCtx\horstOpApplabelOfCtx
  \renewcommand{\horstOpApplabelOfCtx}[2]{\addedInWanilla{\oldhorstOpApplabelOfCtx{#1}{#2}}}
  \let\oldhorstConstructorAppIllegal\horstConstructorAppIllegal
  \renewcommand{\horstConstructorAppIllegal}[1]{\addedInWanilla{\oldhorstConstructorAppIllegal{#1}}}
  \let\oldhorstConstructorAppLegal\horstConstructorAppLegal
  \renewcommand{\horstConstructorAppLegal}[1]{\addedInWanilla{\oldhorstConstructorAppLegal{#1}}}
  \let\oldhorstOpAppoverApproximateLoopGlobals\horstOpAppoverApproximateLoopGlobals 
  \renewcommand{\horstOpAppoverApproximateLoopGlobals}[2]{\addedInWanilla{\oldhorstOpAppoverApproximateLoopGlobals{#1}{#2}}}
  \let\oldhorstOpAppoverApproximateLoopLocals\horstOpAppoverApproximateLoopLocals 
  \renewcommand{\horstOpAppoverApproximateLoopLocals}[2]{\addedInWanilla{\oldhorstOpAppoverApproximateLoopLocals{#1}{#2}}}
  \let\oldhorstOpAppoverApproximateLoopMemory\horstOpAppoverApproximateLoopMemory 
  \renewcommand{\horstOpAppoverApproximateLoopMemory}[2]{\addedInWanilla{\oldhorstOpAppoverApproximateLoopMemory{#1}{#2}}}

  \let\oldhorstOpAppoverApproximateCallGlobals\horstOpAppoverApproximateCallGlobals 
  \renewcommand{\horstOpAppoverApproximateCallGlobals}[2]{\addedInWanilla{\oldhorstOpAppoverApproximateCallGlobals{#1}{#2}}}
  \let\oldhorstOpAppoverApproximateCallArguments\horstOpAppoverApproximateCallArguments
  \renewcommand{\horstOpAppoverApproximateCallArguments}[2]{\addedInWanilla{\oldhorstOpAppoverApproximateCallArguments{#1}{#2}}}
  \let\oldhorstOpAppoverApproximateCallMemory\horstOpAppoverApproximateCallMemory 
  \renewcommand{\horstOpAppoverApproximateCallMemory}[2]{\addedInWanilla{\oldhorstOpAppoverApproximateCallMemory{#1}{#2}}}

  \let\oldhorstOpAppflub\horstOpAppflub
  \renewcommand{\horstOpAppflub}[2]{\addedInWanilla{\oldhorstOpAppflub{#1}{#2}}}
  
  \center
  \begin{minipage}{\textwidth}
  \ruleGroup{functionExitRule}{\Return{} and \End{} leaving functions}
  \ruleGroup{returnJoinRule}{\rniajoin on leaving functions}
  \end{minipage}
  \caption{\tool's abstract rules for exiting functions.}
  \label{fig:wanilla_rules_brtable}
\end{figure*}

\newcommand{\cci}{conditional control flow instruction\xspace}
\newcommand{\ccis}{conditional control flow instructions\xspace}

\begin{theorem}[Soundness of $\Wjoinat$]
\label{thm:soundness-wanilla-joinat}
Given a module description $M$ and a security policy $\Gamma_0$, 
$\Wjoinat$ is sound according to \autoref{def:joinat-soundness} w.r.t.\ the WebAssembly 1.0 semantics (formalized by $\step{\cdot}{\cdot}$~\cite{wappler}).
\end{theorem}
\begin{proof}
We first observe that \tool{}'s soundness claim only covers termination-insensitive noninterference and thus exclude initial memories, that generate infinite traces or halt exceptionally.
We thus have $|\Tr{pp, m_1}| \in \nat$ and $|\Tr{pp, m_2}| \in \nat$.
Furthermore, we observe that $\Wjoinat$'s results are only relevant in circumstances where the derived context is untainted.

We proceed by induction on $n$, the maximum number of control flow instructions (as marked by $\Wccf$ plus $\Br{}$, $\Call{}$, and $\Return$) executed in either $\Tr{pp, m_1}$ or $\Tr{pp, m_2}$.

\medskip

$n = 0$. Neither of the two run encountered a control flow instruction, thus they encounter exactly the same program positions in the same order.
If due to some imprecision in the analysis can derive $\mathPredicateSignatureName{ScopeExtend}$-facts that correspond to reachable program counters, they will be encountered by both executions.
If $\mathPredicateSignatureName{ScopeExtend}$ is not implied for any reachable program position,
this means $\Wjoinat$ returns $(\nojoinpc, \nojoinpc)$ in any case and the implication is trivially true.

We note that this means that for runs without \ccis \tool is sound (since its instantiations of $\Mupdate$, $\Mccf$, and $\Mjoinat$ are sound) and we can thus use \autoref{thm:nia_general_memory_soundness}. A similar result could alternatively be derived applying \autoref{lemma:synchronous_steps} inductively. 

\medskip

$n = 1$. This means one control flow instruction is encountered by both executions.
$\Wjoinat$ can be called either for a program position $pp'$ which holds a control flow instruction or not.
If not, $\Wjoinat$ returns $(\nojoinpc, \nojoinpc)$ and the soundness conditions holds again trivially.
If $pp'$ holds a control flow instruction, we first observe that the execution so far has encountered no control flow instruction, in particular no \emph{conditional} control flow instruction.
We thus know (by \autoref{thm:nia_general_memory_soundness} and the induction hypothesis) that all values the concrete configurations at disagree on will be tainted.

We then proceed by case distinction on the executed control flow instruction:

$\Br{}$: We know that the context object is untainted and no \cci is encountered before or after $pp'$.
Since both executions will continue in lockstep, there is no need to adjust $\Wjoinat$'s return values.
$\Br{}$'s abstraction \autoref{cls:brRule:1} indeed only implies $\mathPredicateSignatureName{ScopeExtend}$ if the context is tainted.

$\Call{}$: Same as $\Br{}$.

$\Return{}$: Same as $\Br{}$.

$\BrIf{}$, $\BrTable{}$: In this case we need to imply a new $\mathPredicateSignatureName{ScopeExtend}$ fact from $pp'$ to the where the two executions will join again.
If the two concrete executions diverge at $pp'$, they will disagree on their to-of-stack value, which will be tainted as noted above.
The two concrete executions will furthermore surely execute either the step out of the enclosing block (from $\hbr$ to the next program counter after the block) or, if the enclosing block is a $\Loop{}$, exit the loop by stepping from $\hpc$ to $\hpc+1$.
Other scenarios are not possible, as there are no control flow instructions in the execution following $pp'$.
We thus (depending on if we are in a $\Loop{}$ or a $\Block{}$, distinguished by the relation of $\hpc$ and $\hbr$) imply $\mathPredicateSignatureName{ScopeExtend}$ from $pp'$ to the program counter before we want to join in \autoref{cls:brIfRule:1}, \autoref{cls:brTableRule:1}, and \autoref{cls:brTableDefaultRule:1} (since the context is untainted, $\h{from}<f>$ will be $\hpc$, which is our $pp'$). 

$\If{}\cdot\Then{}\cdot\Else{}$: Same as $\BrIf{}$, except that we imply one big $\mathPredicateSignatureName{ScopeExtend}$ fact that encloses the whole structure in \autoref{cls:brIfRule:1}.

$\CallIndirect{}$: In this case we need to enter the called function with a tainted context that is tainted and remains so and join in the step after returning.
This is achieved by entering new function with a context tainted with least upper bound of the taint of top-of-stack value and table and its program counter set to $-1$ (this is what the $\h{mkCtx}[\cdot]$-function return), as done in \autoref{cls:callIndirectRule:0} and \autoref{cls:callIndirectHavokRule:0}.
This context is within the function unlowerable, as for every function there is a scope that extends from $-1$ to the end of the function by \autoref{cls:functionScopeExtendRule:0}.
If the function then returns, it has the same context object as it entered with (tainted and with a program counter $<0$), which mean that  \autoref{cls:callIndirectRule:1} and \autoref{cls:callIndirectHavokRule:1} imply $\h{MStateToJoin}<P>$ which executes \rniajoin.

\medskip

$n' = n + 1$.
Since both executions are regularly terminating there exist two final configurations $(\terminatepc, m_1') = \Mstep[pp, m_1][k_1]$ and $(\terminatepc, m_2') = \Mstep[pp, m_2][k_2]$.
Thus, we have by \autoref{corollary:terminating_execution_pairwise_reachable} $(\terminatepc, m_1', m_2') = \Mstep[pp, m_1, m_2][*][\{s,d\}]$.
Consider the case that there are multiple divergent steps within $\Mstep[pp, m_1, m_2][*][\{s,d\}]$ and \tool can precisely identify that the context can be lowered at one point in between (precisely meaning here that \emph{only} an untainted context can be derived at some program position between parts of the execution where a tainted context can be derived).
In this case we can apply the induction hypothesis to the subtraces, as they have at most $n$ executions of control flow instructions within them.

This means we can focus on the case where one divergent step starts at $pp'$.
Let $\ppi$ be the position of the last control flow instruction in the execution.
This means the $\Mnoninterference$'s soundness guarantees hold for \tool up to that point by the induction hypothesis.
In particular, we can be sure that the context will be tainted when $\ppi$ is executed by \autoref{lemma:diverging_steps_context}.
We proceed with case distinction on the kind of the executed control flow instruction:

$\Call{}$ or $\CallIndirect{}$: In this case the join point does not have to be changed as we can surely join after returning (possibly multiple times).
Let's say $((\fid, pp_{j_1}), (\fid, pp_{j_1}+1)) = \Wjoinat[pp, pp', m_1, m_2]$, i.e., the $pp_{j_1}$ is the largest number for which we can derive \penalty-10000 $\h{ScopeExtend}[\hfid][pp', pp_{j_1}]$.
The abstractly executing the $\Call{}$ (or $\CallIndirect{}$) could potentially to a change in the joining transition returned by $\Wjoinat$ by making it possible to derive additional $\mathPredicateSignatureName{ScopeExtend}$ facts.
This can only lead to the scope of $pp'$ growing, such that we can derive $\h{ScopeExtend}[\hfid][pp', pp_{j_1}']$ with $pp_{j_1}' > pp_{j_1}$.
As both executions executed $(pp_{j_1}, pp_{j_1}+1)$ (by the induction hypothesis), both executions will execute $(pp_{j_1}', pp_{j_1}'+1)$ (as there are no control flow instructions after the $\Call{}$ (or $\CallIndirect{}$) that could divert the control from just increasing the program counter to the end of the function.

$\If{}\cdot\Then{}\cdot\Else{}\cdot\End$: If the control flow is tainted, the scope of must extend after the $\End$ of the $\If{}\cdot\Then{}\cdot\Else{}\cdot\End$ already.
Since within either of the branches no control flow instructions are executed the current scope extension (or an extension as mentioned above) is fine.

$\Br{}$, $\BrIf{}$: Executing a branching instruction during a divergent step can change the joining transition.
Since the executed control flow instruction is the last one, the scope has to be extended to at most the program counter in the case of a forward jump (backward jumps do not have to extend the scope).
This is achieved by rules \autoref{cls:brIfRule:1}, \autoref{cls:brRule:1}, \autoref{cls:brTableRule:1}, and \autoref{cls:brTableDefaultRule:1}.

$\Return$: If a $\Return$ is executed in a high context, we have to extend the scope to the end of the function (meaning that the joining transition is stepping out of the function).
This is achieved by \autoref{cls:functionExitRule:3}.
\end{proof}

The soundness of \tool{}'s instance of $\Mccf$, $\Wccf$ is formalized by the following theorem.

\begin{theorem}[Soundness of $\Wccf$]
\label{thm:soundness-wanilla-ccf}
Given a module description $M$ and a security policy $\Gamma_0$, 
$\Wccf$ is sound (according to \autoref{def:joinat-soundness} w.r.t.\ the WebAssembly 1.0 semantics (formalized by $\step{\cdot}{\cdot}$~\cite{wappler}) when only considering regularly terminating executions.
\end{theorem}
\begin{proof}
By going over all rules in~\cite{wasm-core-1.0} and excluding all rules that rewrite to $\Trap$, we see that the only those listed in $\Wccf$'s definition in \autoref{fig:joinatccf} can lead to multiple different program positions.
\end{proof}

For an analysis to handle exceptional halts, all potentially trapping instructions would have to be handled as conditional control flow instructions (as, e.g., a \AnyCmd{t}{store} instruction can either go to the next program counter or to the exceptional state, depending on if the accessed index is out-of-bounds or not). 

\begin{theorem}[Soundness of $\Wupdate$]
\label{thm:soundness-wanilla-update}
Given a \WA configuration $c = S;F;\iinstr$ as defined in \cite{wappler}, with $(pp, m) = \exconftuple[c]$, \WA module description $M$, and a security policy $\Gamma$.
Let $\Delta_T$ the abstraction of the topmost activation of $c$ (the set returned from the last call of $\aframe{\cdot}$ in $\athread{c, \access{F}{stor}, 0, \Gamma}$) and $\Delta_B$ be the abstraction of all activations but the topmost.
Let $\Delta'$ be the abstract configurations derivable by a first-order calculus from $\Delta_T$ and $\afunctions{S}$, i.e, $\derive{\Delta, \afunctions{S}}{\Delta'}$.
If we define $\Wupdate[pp, m]$ being defined as the updates sets leading to the follow-up configurations in $\concrete{\Delta' \cup \Delta_B}$, where the input sets are defined by table in \ref{fig:wanilla-input-update-sets}, $\Wupdate$ fulfills the soundness conditions from \autoref{def:soundness_astep}. 
\end{theorem}
\begin{proof}
The returned sets are well-formed (condition $(5)$) by construction.
That the follow-up memory configurations/program positions are correct (conditions $(2)$ and $(4)$), directly follows from \wappler's soundness proof~\cite{wappler-tr}.
That the input sets are chosen correctly can be verified by comparing \ref{fig:wanilla-input-update-sets} with the \WA semantics in~\cite{wappler-tr}.
\end{proof}

\begin{proof}[Proof of \autoref{thm:soundness-wanilla}]
A fixed point computation such as $\Mnoninterference$ can be translated straight-forwardly to a constrained Horn-clause program.
By \autoref{thm:nia_general_memory_soundness} we know that an instance of $\Mnoninterference$ is sound if the implementation of the constituent functions $\Mupdate$, $\Mccf$, and $\Mjoinat$ are sound.
These are proven sound by \autoref{thm:soundness-wanilla-joinat}, \autoref{thm:soundness-wanilla-ccf}, and \autoref{thm:soundness-wanilla-update}, respectively.
\end{proof}

\begin{theorem}[Joins in $\L$-Contexts]
Replacing an application of \rniaprop in a low context with an application of \rniajoin (that is only changed by requiring $\lctx$ in the context object returned by $\Mnoninterference$) retains the soundness of $\Mnoninterference$.
\end{theorem}
\begin{proof}
Any instruction $\Cmd$ where such a behavior is wanted can be replaced by the instruction sequence $\If{\textit{true}}~\Cmd~\End$ (or the equivalent in the respective specification mechanism), which maintains the program behavior.

If we treat the condition as tainted, the effect in $\Mnoninterference$ (with appropriate return values for $\Mjoinat$) would be the same as executing \rniajoin instead of \rniaprop in the original program.
\end{proof}

\begin{figure*}
{
\small
\begin{minipage}{0.3\textwidth}
\begin{align*}
\avalue{x} = \begin{cases}
             \{ \call{zeropad}{c}  \}                       & \text{if } x = \Const{\watype{i32}}{c} \\
             \{c\}                                          & \text{if } x = \Const{\watype{i64}}{c} \\
             \{\call{zeropad}{y}~|~y \in \mathbb{B}_{32} \} & \text{if } x = \Const{\watype{f32}}{c} \\
             \mathbb{B}_{64}                                & \text{if } x = \Const{\watype{f64}}{c}
             \end{cases}
\end{align*}
\end{minipage}
\hspace{0.5cm}
\begin{minipage}{0.5\textwidth}
\begin{align*}
\newcommand{\xs}{\textit{xs}}
\newcommand{\ys}{\textit{ys}}
\aseq{s, \addedInWanilla{\Lambda, p_i}} = \begin{cases}
           \{ \addedInWanilla{({\color{black}y}, \mcall{\Lambda}{p_i})} \consSYMBOL \ys ~|~ y \in \avalue{x}, \ys \in \aseq{\xs, \addedInWanilla{\Lambda, p_{i+1}}} \} & \text{if } s = x \consSYMBOL \xs \\
           \{\epsilon\}                                                             & \text{otherwise }
           \end{cases}\\
\end{align*}
\end{minipage}
\begin{align*}
\athread{\begin{minipage}{1.28cm}$S;F;\iinstrs, \\ \addedInWanilla{S_0, i, \Gamma}$\end{minipage}} &= 
\begin{cases}
  \aframe{\access{F}{fid}, \access{F'}{pc}, \access{F'}{stor}, F, \access{F'}{args}~\access{F'}{index}, \iinstrs, \addedInWanilla{S_0, i, \Gamma}} 
  \cup \athread{S;F';\iinstrs', \addedInWanilla{S_0, i+1, \Gamma}}
& \text{if } \iinstrs = \EC{\Frame{n}{F'}~\iinstrs'~\End}\\
\aframe{\access{F}{fid}, \hpc, S, F, \epsilon, \iinstrs, \addedInWanilla{S_0, i, \Gamma}} 
& \text{if } \iinstrs = \EC{\Cmd_{\hpc}}\\
\end{cases}\\
\aframe{\begin{minipage}{1.58cm} $\hfid, \hpc, S, F, p, \\\iinstrs, \addedInWanilla{S_0, i, \Gamma} $ \end{minipage}}
  &= \begin{minipage}{17.1cm} $\big\{
  \MState{\hfid, \hpc}{\addedInWanilla{\ctx}, \st, \gt, \lt, \mem, \addedInWanilla{\atable{S, S_0, \Lmap}}, \oat, \ogt, \omem} ~|~ \st \in \aseq{\exstack{\iinstrs} \concatSYMBOL p, \addedInWanilla{\Lmap, \mpst{i,0}}} , \gt \in \aseq{\exglobals{S,F}, \addedInWanilla{\Lmap, \mpgl{0}}}, \\
  \phantom{\big\{~} \lt \in \aseq{\access{F}{locals}, \addedInWanilla{\Lmap, \mplo{0}}}, 
  \oat \in \aseq{\access{F}{args}, \addedInWanilla{\Lmap, \mpar{0}}}, \ogt \in \aseq{\exglobals{\access{F}{stor}, F}, \addedInWanilla{\Lmap, \mpglo{0}}}, 0 \le i \le \exmemsize{S,F}, \\
  \phantom{\big\{~}\mem = \amem{\addedInWanilla{\mplm{i}},S,F, \addedInWanilla{\Lmap, \mpsz}}, \omem = \amem{\addedInWanilla{\mplmo{i}},\access{F}{stor}, F, \addedInWanilla{\Lmap, \mpszo}}, \addedInWanilla{\ctx = \h{Ctx}[\ell, L, -1],  \ell \in \Lat}
  \big\}
  $\end{minipage}
\\%
\addedInWanilla{ \atable{S, S_0, \Lambda} } &= \addedInWanilla{
  \begin{cases}
  \h{Tbl}[\h{TblPrecise}, \bigsqcup_{i} \mcall{\Lambda}{\mpta{i}} ] & \text{if } \access{S}{tables} = \access{S_0}{tables}\\
  \h{Tbl}[\h{TblImprecise}, \bigsqcup_{i} \mcall{\Lambda}{\mpta{i}} ] & \text{otherwise}\\
  \end{cases}
 }
\\
  \exstack{\iinstrs} &= 
    \begin{cases}
      \Const{t}{x} \consSYMBOL \exstack{\iinstrs'} & \text{if } \iinstrs = \Const{t}{x} \consSYMBOL \iinstrs' \\
      \exstack{\iinstrs'}                          & \text{if } \iinstrs = \Label{n}{} \consSYMBOL \iinstrs' \\
        \epsilon & \text{otherwise}
    \end{cases}
  \\
  \exglobals{S,F} &= \exglobals{S,F} =\eta_G'(S,F,0)\\
  \call{\eta_G'}{S,F,i}&=
    \begin{cases}
      x = \access{\arraccess{\access{S}{globals}}{\arraccess{\access{\access{F}{module}}{globaladdrs}}{i}}}{value} \consSYMBOL \call{\eta_G'}{S,F,i+1} & \text{if } i < |\arraccess{\access{\access{F}{module}}{globaladdrs}}{i}| \\
        \epsilon & \text{otherwise}
    \end{cases}\\
  \exmemsize{S,F} &=
  \begin{cases}
    2^{32}                    & \text {if }  \access{\exmem{S,F}}{max} = \epsilon \\
    \access{\exmem{S,F}}{max} \cdot 2^{16} & \text {otherwise}
  \end{cases}\\
  \exmem{S,F} &= \horstACCESS{\access{S}{mems}}{\horstACCESS{\access{\access{F}{moduleinst}}{memaddrs}}{0}}\\
  \amem{\addedInWanilla{p_i},S,F,\addedInWanilla{\Lambda, sz}} &= \h{Mem}[i, \addedInWanilla{({\color{black} \horstACCESS{d}{i},} \mcall{\Lambda}{p_i})}, \addedInWanilla{({\color{black}|d|/2^{16},} \mcall{\Lambda}{sz})} ]\\
  &\phantom{==}\text{ where } d =  \access{\exmem{S,F}}{data}\\
 \afunctions{S} &= \bigcup_{\substack{0 \le \hfid < |\access{S}{funcs}|,\\f = \horstACCESS{\access{S}{funcs}}{\hfid}}} \bigcup_{\Cmd_{\hpc} \in \annotate{\Block{\access{f}{type}}~\access{f}{code}~\End}}
 \begin{cases}
   \ainstr{\hpc}{\If} \cup  \ainstr{\hpc+1}{\Block} & \text{ if } \Cmd = \If \\
   \ainstr{\hpc}{\Block} \cup  \ainstr{\hpc-1}{\End} & \text{ if } \Cmd = \Else \\
   \ainstr{\hpc}{\Cmd} & \text{ otherwise}
 \end{cases}
\\%
\addedInWanilla{\Wabstract[S;F;\iinstrs][\Gamma]} &= \athread{S;F;\iinstrs, \addedInWanilla{\access{F}{stor}, 0, \Gamma} } \cup \afunctions{S}\\   
\end{align*}
}

  \caption{The abstraction function for configurations.}
  \label{fig:abstraction-function-configurations}
\end{figure*}

\begin{figure*}
\begin{center}
\begin{tabular}{ lllll }
\textbf{regular instructions} ($I_{pp}= \emptyset$)\\
$U$                                                                                                 && commands                                                            \\ 
$\{(\mpst{z+1}, v, \emptyset \}$                                                                    && \Const{t}{v}                                                        \\ 
$\{(\tos, v, \{\tos)\} \}$                                                                          && \AnyCmd{t}{unop}, \AnyCmd{t}{testop}, \AnyCmd{t2}{cvtop\_t_1\_sx^?} \\ 
$\{(\tos, \bot, \emptyset), (\mpst{z-1}, v, \{\tos, \mpst{z-1}\}) \}$                               && \AnyCmd{t}{binop}, \AnyCmd{t}{relop}                                \\ 
$\{(\tos, \bot, \emptyset) \}$                                                                      && \Drop                                                               \\ 
$\{(\tos, \bot, \emptyset), (\mpst{z-1}, \bot, \emptyset), (\mpst{z-2}, v, \{\tos, \mpst{z-1}\})\}$ && \Select                                                             \\ 
$\{(\tos, \bot, \emptyset), (\mpst{z-1}, \bot, \emptyset), (\mpst{z-2}, v, \{\tos, \mpst{z-2}\})\}$ && \Select                                                             \\ 
$\{(\mpst{z+1}, v, \{\mplo{x}\})\}$                                                                 && \LocalGet{x}                                                        \\ 
$\{(\tos, \bot, \emptyset), (\mplo{x}, v, \{\tos\})\}$                                              && \LocalSet{x}                                                        \\ 
$\{(\mplo{x}, v, \{\tos\})\}$                                                                       && \LocalTee{x}                                                        \\ 
$\{(\mpst{z+1}, v, \{\mpgl{x}\})\}$                                                                 && \GlobalGet{x}                                                       \\ 
$\{(\tos, \bot, \emptyset), (\mpgl{x}, v, \{\tos\})\}$                                              && \GlobalSet{x}                                                       \\ 
$\{(\tos, v, \{\tos\} \cup \{\mplm{i} ~|~ \access{\memarg}{offset} + \call{\mu}{\mpst{z}} + i < |t| \})\}$ && \Load{t}~\memarg                                                    \\ 
$\{(\tos, \bot, \emptyset),(\mpst{z-1}, \bot, \emptyset), (\mplm{i}, v, \{\tos, \mpst{z-1}\})\}$    && \Store{t}                                                           \\ 
$\{(\mpst{z+1}, v, \{\mpsz\})\}$                                                                    && \Size                                                               \\ 
$\{(\tos, v, \{\tos, \mpsz\}), (\mpsz, v', \{\tos, \mpsz\})\}$                                      && \Grow                                                               \\ 
$\emptyset$                                                                                         && \Nop, $\Block{t^n}~\iinstrs~\End{}$, $\Loop{t^?}~\iinstrs~\End{}$   \\ 
---                                                                                                 && \Unreachable                                                        \\ 
\textbf{control flow instructions}\\
$U$                               & $I_{pp}$    & commands\\
$\{(\tos, \bot, \emptyset)\}$     & $\{ \tos\}$ & $\If[t^?]~\iinstrs[1]~\Then~\iinstrs[2]~\End$ & \\
$\emptyset$                       & $\emptyset$ & \Br{l}                                        & \\
$\{(\tos, \bot, \emptyset)\}$     & $\{ \tos\}$ & \BrIf{l}                                      & \\
$\{(\tos, \bot, \emptyset)\}$     & $\{ \tos\}$ & \BrTable{\seq{l} l_N}                         & \\
\begin{minipage}{5.5cm}$\{ (\mpst{q,i}, \bot, \emptyset) ~|~ 0 \le i \le z \} \cup {}$~~~~~~~~~$\{ (\mpst{q-1,z'-j+n}, \memmp{\mpst{q,z-j}}, \{ \mpst{z-j} \} ) ~|~ 0 \le j < n  \}$\end{minipage} & $\emptyset$ & \Return{} ($n$ is the number of returned values)\\
$ \{(\mplo{i,q+1}, \memmp{\mpst{z-i}}, \{ \mpst{z-i}\})   ~|~ 0 \le i < m\} $     & $\emptyset$                               & \Call{x} ($m$ is the number of arguments) \\
$ \{(\mplo{i,q+1}, \memmp{\mpst{z-i-1}}, \{ \mpst{z-i-1}\})   ~|~ 0 \le i < m\} $ & $\{\tos\} \cup \{\mpta{i} ~|~ 0 \le i \}$ & \CallIndirect~{x} ($m$ is the number of arguments) \\
 \hline
\end{tabular}
\end{center}
\caption{List of update and input sets for each instruction. In cases where $v$ is not given it assumed to be correct by \wappler's soundness. $z$ and $q$ hold the respective indices of the topmost value and call stack element \emph{before} execution the described command.
If we omit a call stack index from a memory position that should carry one, it is assumed to be $q$ (e.g., $\mpst{0} = \mpst{q,0}$).
We write $\call{\mu}{p}$ for the value of $p$ in before executing the command.
}
\label{fig:wanilla-input-update-sets}
\end{figure*}

\begin{figure*}
\begin{align*}
\exconftuple[S;F;\iinstrs] & = (\exconfpp[\iinstrs], \exconfmem[F,\iinstrs]) \\
\exconfmem[S;F;\iinstrs]   &= \lambda p .
\begin{cases}
k                                                                                                                        & \text{if } p = \mpst{x,y} \land \arraccess{\getframe[\iinstrs, x]}{y} = \Const{t}{k}\\
\access{\getframe[\iinstrs, x]}{\arraccess{locals}{y}}                                                                   & \text{if } p = \mplo{x,y} \\
\access{\access{S}{\arraccess{global}{\access{F}{\access{moduleinst}{\arraccess{globaladdrs}{x}}}}}}{value}              & \text{if } p = \mpgl{x}   \\
\access{\access{S}{\arraccess{mems}{\access{F}{\access{moduleinst}{\arraccess{memaddrs}{0}}}}}}{\arraccess{data}{x}}     & \text{if } p = \mplm{x}   \\
|\access{\access{S}{\arraccess{mems}{\access{F}{\access{moduleinst}{\arraccess{memaddrs}{0}}}}}}{data}|/2^{16}           & \text{if } p = \mpsz      \\
\access{\access{S}{\arraccess{tables}{\access{F}{\access{moduleinst}{\arraccess{tableaddrs}{0}}}}}}{\arraccess{elem}{x}} & \text{if } p = \mpta{x}   \\
\bot                                                                                                                     & \text{otherwise}
\end{cases}\\
\exconfpp[F,\iinstrs] &= \begin{cases}
(\access{F}{fid}, \access{F'}{pc}) \cdot \exconfpp[F', \iinstrs']
& \text{if } \iinstrs = \EC{\Frame{n}{F'}~\iinstrs'~\End}\\
(\access{F}{fid}, \hpc)
& \text{if } \iinstrs = \EC{\Cmd_{\hpc}}\\
\end{cases}\\
\end{align*}
\begin{align*}
\getframe[\iinstrs, k] =
\begin{cases}
F'                        & \text{if } k = 0 \land \iinstrs = \Frame{n}{F'}~\iinstrs'~\End \\
\getframe[\iinstrs', k-1] & \text{if } k > 0 \land \iinstrs = E[\Frame{n}{F'}~\iinstrs'~\End] \\
\bot                      & \text{otherwise}
\end{cases}\\
\getstack[\iinstrs, k] =
\begin{cases}
\iinstrs'                 & \text{if } k = 0 \land \iinstrs = \Frame{n}{F'}~\iinstrs'~\End \\
\getstack[\iinstrs', k-1] & \text{if } k > 0 \land \iinstrs = E[\Frame{n}{F'}~\iinstrs'~\End] \\
\bot                      & \text{otherwise}
\end{cases}\\
\end{align*}
\caption{Conversion between \WA configurations and the $\MPPos \times \MMem$ view of configurations.}
\end{figure*}

\begin{figure*}
\small
\begin{align*}
\preframes[\Delta] &= \big\{ \Delta' ~|~ \Delta' \subseteq \Delta, \forall i . 0 \le i < \access{\access{\arraccess{\access{M}{mems}}{0}}{type}}{max} \implies \exists f . f \in \Delta' \land \h{MState}[\_, \_, \_, \_, \h{Mem}[i, \_, \_ ], \_, \_, \_, \h{Mem}[i, \_, \_]],\\
&\forall i\,f\,f' . \{f, f'\} \subseteq \Delta \land \memaccess{f}{\mplm{i}} \neq \bot \land \memaccess{f'}{\mplm{i}} \neq \bot \implies f = f', \\
&\forall i\,j\,f . \h{MState}[\_, \_, \_, \_, \h{Mem}[i, \_, \_ ], \_, \_, \_, \h{Mem}[j, \_, \_]] \in \Delta \implies j = i, \\
&\forall p\,f\,f' . \{f, f'\} \subseteq \Delta \land p \in (\MMPos \setminus \{ p' ~|~ p = \mplm{i} \lor \mplmo{i}, i \in \mathbb{B}_{32}\}) \implies \memaccess{f}{p} = \memaccess{f'}{p}, \\
&\forall f\,f' . \access{f}{pc} = \access{f'}{pc} \land \access{f}{fid} = \access{f'}{fid} \land \access{f}{table} = \access{f'}{table}
\\
 \big\}
\\
\abstractcallgraph[\Delta]&= \big(\preframes[\Delta], \big\{ (\Delta_a,\Delta_b) ~|~ \{\Delta_a,\Delta_b\} \subseteq \preframes[\Delta], \\
& \pca = \access{\Delta_a}{pc}, \pcb = \access{\Delta_b}{pc}, \fida = \access{\Delta_a}{fid}, \fidb = \access{\Delta_b}{fid}, f_a = \access{M}{\arraccess{funcs}{\fida}}, f_b = \access{M}{\arraccess{funcs}{\fidb}},\\
& (\Call{}[\pca] \in \annotate{\Block{\access{f_a}{type}}~\access{f_a}{body}~\End}
 \land o = 0) \lor\\
&\quad\quad
(\CallIndirect[t_b][\pca] \in \annotate{\Block{\access{f_b}{type}}~\access{f_b}{body}~\End} \land (\access{\Delta}{table} = \h{TblPrecise} \implies \memaccess{\Delta_a}{\tos} = \fidb) \land o = 1),\\
& \forall i . i \in \mathbb{B}_{32} \implies \memaccess{\Delta_a}{\mplm{i}} = \memaccess{\Delta_b}{\mplmo{i}} \land \memaccess{\Delta_a}{\mpgl{i}} = \memaccess{\Delta_b}{\mpglo{i}} \land  \memaccess{\Delta_a}{\mpsz} = \memaccess{\Delta_b}{\mpszo}, \\
& t_b = {t_1}^n \rightarrow {t_2}^m \land \forall i . 0 \le i < m \implies \memaccess{\Delta_a}{\tos[-i-o]} = \memaccess{\Delta_b}{\mpar{i}}\\
& \big\} \big)
\end{align*}
\caption{Concretization 1: Abstract activations and abstract frames.}
\end{figure*}

\begin{figure*}

{
\small
\begin{align*}
&\concretestack{\Delta} =&& \{ \Frame{n}{F}~\concretestackp{\Delta_1,\ldots,\Delta_n, \access{\Delta_0}{pc}, \annotate{\Block{\access{f}{type}}~\access{f}{body}~\End}, 0, \concretestacklocal{\Delta_0} }~\End ~|~ \\
&&&\quad \Delta_0,\ldots,\Delta_n \text{ is a path through }\abstractcallgraph[\Delta], f = \access{M}{\arraccess{funcs}{\access{\Delta_0}{fid}}} \}\\
&\concretestackp{\vec{\Delta}, \ipc, \iinstrs, \icnt, \ist} =&&
\begin{cases}
\ist~\iinstrs
\\\quad\quad \text{if } \iinstrs = \Cmd[\ipc]~\iinstrsrest \land \vec{\Delta} = \epsilon \\
\arraccess{\ist}{0:-|\access{\Delta_0}{args}|}~
\Frame{k}{F}~
\concretestackp{\Delta_1,\ldots,\Delta_n, \access{\Delta_0}{pc}, \annotate{\Block{\access{f}{type}}~\access{f}{code}~\End}, 0, \arraccess{\ist}{|\access{\Delta_0}{args}|:} }
~\End
~\iinstrsrest
\\\quad\quad \text{if } \iinstrs = \Cmd[\ipc]~\iinstrsrest \land \vec{\Delta} = \Delta_0,\ldots,\Delta_n \land f = \access{M}{\arraccess{funcs}{\access{\Delta_0}{fid}}}\\
\arraccess{\ist}{0:\icnt}~
\Label{n}{\epsilon}~
\concretestackp{\vec{\Delta}, \ipc, \iinstrs', 0, \arraccess{\ist}{\icnt:} }
\End{}~
~\iinstrsrest
\\\quad\quad \text{if } \iinstrs = \Block{t^n}[\ipc']~\iinstrs'~\End[\ipc'']~\iinstrsrest \land \ipc' < \ipc < \ipc''\\
\concretestackp{\vec{\Delta}, \ipc, \iinstrsrest, \icnt + n, \ist}
\\\quad\quad \text{if } \iinstrs = \Block{t^n}[\ipc']~\iinstrs'~\End[\ipc'']~\iinstrsrest \land \lnot(\ipc' < \ipc < \ipc'') \\
\arraccess{\ist}{0:\icnt}~
\Label{0}{\Loop{t^n}~\iinstrs'~\End{}}~
\concretestackp{\vec{\Delta}, \ipc, \iinstrs', 0, \arraccess{\ist}{\icnt:} }
\End{}~
~\iinstrsrest
\\\quad\quad  \text{if } \iinstrs = \Loop{t^n}[\ipc']~\iinstrs'~\End[\ipc'']~\iinstrsrest \land \ipc' < \ipc < \ipc''\\
\concretestackp{\vec{\Delta}, \ipc, \iinstrsrest, \icnt + n, \ist}
\\\quad\quad  \text{if } \iinstrs = \Loop{t^n}[\ipc']~\iinstrs'~\End[\ipc'']~\iinstrsrest \land \lnot(\ipc' < \ipc < \ipc'') \\
\arraccess{\ist}{0:\icnt}~
\Label{n}{\epsilon}~
\concretestackp{\vec{\Delta}, \ipc, \iinstrs_1', 0, \arraccess{\ist}{\icnt:} }
\End{}~
~\iinstrsrest
\\\quad\quad  \text{if } \iinstrs = \If[t^n][\ipc']~\iinstrs_1~\Else[\ipc'']~\iinstrs_2~\End[\ipc''']~\iinstrsrest \land \ipc' < \ipc < \ipc''\\
\arraccess{\ist}{0:\icnt}~
\Label{n}{\epsilon}~
\concretestackp{\vec{\Delta}, \ipc, \iinstrs_2', 0, \arraccess{\ist}{\icnt:} }
\End{}~
~\iinstrsrest
\\\quad\quad  \text{if } \iinstrs = \If[t^n][\ipc']~\iinstrs_1~\Else[\ipc'']~\iinstrs_2~\End[\ipc''']~\iinstrsrest \land \ipc'' < \ipc < \ipc''' \\
\concretestackp{\vec{\Delta}, \ipc, \iinstrsrest, \icnt + n, \ist}
\\\quad\quad  \text{if } \iinstrs = \If[t^n][\ipc']~\iinstrs_1~\Else[\ipc'']~\iinstrs_2~\End[\ipc''']~\iinstrsrest \land \lnot(\ipc' < \ipc < \ipc''') \\
\concretestackp{\vec{\Delta}, \ipc, \iinstrsrest, \icnt - n + m, \ist}
\\\quad\quad  \text{if } \iinstrs = \Cmd[\ipc']~\iinstrsrest \land \ipc' \neq \ipc \land \Cmd[\ipc'] \text{ takes }n\text{ values and returns }m\text{ values}%
\end{cases}
\end{align*}
\begin{align*}
\concretemem{\Delta} &= \record{\recordentry{data}{\concretememlocal{\Delta, 0} },\recordentry{max}{\access{M}{\access{\arraccess{mems}{0}}{\access{type}{max}}}}}\\
\concretememlocal{\Delta, i} &=
\begin{cases}
\memaccess{\Delta}{\mplm{i}} \cdot \concretemem{\Delta, i+1} & \text{ if } i < \mpaccess{\Delta}{\mpsz} \\
\epsilon                                                         & \text{ otherwise}
\end{cases} \\
\concreteglobals{\Delta} &= \concreteglobalslocal{\Delta, 0}\\
\concreteglobalslocal{\Delta, i} &=
\begin{cases}
\record{\recordentry{val}{\memaccess{\Delta}{\mpgl{i}}}, \recordentry{mut}{\access{\arraccess{\access{M}{globals}}{i}}{mut}}} \cdot \concreteglobals{\Delta, i+1} & \text{ if } i < |\access{\Delta}{globals}| \\
\epsilon                                                         & \text{ otherwise}
\end{cases} \\
\concretestore{\Delta} &= \{ \concretestorelocal{\Delta_n} ~|~\Delta_0,\ldots,\Delta_n \text{ is a path through }\abstractcallgraph[\Delta] \} \\
\concretestorelocal{\Delta} &= \{ \record{\recordentry{funcs}{f}, \recordentry{tables}{t}, \recordentry{mems}{\concretemem{\Delta}}, \recordentry{globals}{\concreteglobals{\Delta}}} ~|~ t \in \concretetable{\Delta}, f = \access{M}{funcs} \} \\
\concretetable{\Delta} &= \begin{cases}
\big\{\record{\recordentry{elem}{\access{\arraccess{\access{M}{elem}}{0}}{init}}, \recordentry{max}{\access{\arraccess{\access{M}{tables}}{0}}{max}}} \big\} & \text{if } \access{\Delta}{table} = \h{TblPrecise}\\
\big\{\record{\recordentry{elem}{t},                                              \recordentry{max}{\access{\arraccess{\access{M}{tables}}{0}}{max}}} ~|~ t \text{ is any possible table of correct size}\big\} & \text{if } \access{\Delta}{table} = \h{TblImprecise}\\
\end{cases}\\
\concrete{\Delta} &= \big\{S; \record{\recordentry{module}{\{\}}, \recordentry{local}{\epsilon}};\iinstrs ~|~ \iinstrs \in \concretestack{\Delta}, S \in \concretestore{\Delta}  \big\}
\end{align*}
}
\caption{Concretization 2: Mapping between abstract and concrete configurations.}
\end{figure*}

\fi
\end{document}